\documentclass[twocolumn,aps,american,pra,floatfix,nofootinbib,superscriptaddress,longbibliography]{revtex4-1}

\usepackage{amsfonts}
\usepackage{amsmath}
\usepackage{amssymb}
\usepackage{amsthm}
\usepackage{bbm}
\usepackage{bm}
\usepackage{mathrsfs}
\usepackage{MnSymbol}
\usepackage{accents}
\usepackage{braket}
\usepackage{multirow}
\usepackage[table]{xcolor}

\usepackage{epsfig}
\usepackage{graphics}
\usepackage{graphicx}
\usepackage{subfigure}
\usepackage{tikz}
\usetikzlibrary{arrows.meta}
\usetikzlibrary{quotes,fit}
\usetikzlibrary{datavisualization.formats.functions}
\usepackage[matrix,frame,arrow]{xy}
\usepackage[compat=0.6]{yquant}
\useyquantlanguage{groups}

\usepackage{times}
\usepackage{algorithm}
\usepackage{algpseudocode}
\algrenewcommand\algorithmicrequire{{Data:}}
\algrenewcommand\algorithmicensure{Initialize:}
\usepackage[pdfstartview=FitH]{hyperref}
\usepackage{color}
\usepackage{xcolor}

\usepackage{hypernat}
\hypersetup{
    colorlinks=true,   
    linkcolor=cyan,    
    citecolor=magenta, 
    filecolor=magenta, 
    urlcolor=cyan,     
    runcolor=cyan
}
\usepackage[capitalise]{cleveref} 

\newtheorem{mytheorem}{Theorem}
\newtheorem{mylemma}{Lemma}

\newcommand{\mc}{\mathcal}
\newcommand{\wSimon}[2]{\ensuremath{\mathsf{w}_{{#1}}\textnormal{Simon-}{#2}}} 
\newcommand{\NTSIQ}{\mathrm{NTS}_{\textnormal{IQ}}}
\newcommand{\EEB}{E_{\textnormal{EB}}}
\newcommand{\floor}[1]{\left\lfloor {#1} \right\rfloor} 

\newcommand{\bes}{\begin{subequations}}
\newcommand{\ees}{\end{subequations}}
\newcommand{\beq}{\begin{equation}}
\newcommand{\eeq}{\end{equation}}

\newcommand{\mcO}{\mathcal{O}} 
\newcommand{\mcOf}{\mathcal{O}_f} 

\newcommand{\expv}[1]{\langle #1\rangle} 
\newcommand{\abs}[1]{\ensuremath{\left|#1\right|}} 
\newcommand{\ave}[1]{\langle{#1}\rangle}
\newcommand{\ignore}[1]{}
\DeclareMathOperator{\HW}{HW} 

\def\Tr{\mathrm{Tr}}
\def\Pr{\mathrm{Pr}}
\newcommand{\ketbra}[1]{|{#1}\rangle\!\langle{#1}|}
\newcommand{\NTS}{\mathrm{NTS}}
\newcommand{\DKL}{D_{\mathrm{KL}}}

\newcommand{\vect}[1]{\boldsymbol{#1}}
\graphicspath{ {images/} }

\begin{document}

\title{Demonstration of Algorithmic Quantum Speedup for an Abelian Hidden Subgroup Problem}

\author{Phattharaporn Singkanipa}
\affiliation{Department of Physics,
University of Southern California, Los Angeles, CA 90089, USA}

\author{Victor Kasatkin}
\affiliation{Viterbi School of Engineering,
University of Southern California, Los Angeles, CA 90089, USA}
\author{Zeyuan Zhou}
\affiliation{William H. Miller III Department of
Physics \& Astronomy, Johns Hopkins University, Baltimore, Maryland 21218, USA}
\author{Gregory Quiroz}
\affiliation{William H. Miller III Department of
Physics \& Astronomy, Johns Hopkins University, Baltimore, Maryland 21218, USA}
\affiliation{Johns Hopkins University Applied Physics Laboratory, Laurel, Maryland 20723, USA}

\author{Daniel A. Lidar}
\affiliation{Departments of Electrical Engineering, Chemistry, Physics and Astronomy, and Center for Quantum Information Science \& Technology,
University of Southern California, Los Angeles, CA 90089, USA}

\date{\today}

\begin{abstract}
Simon's problem is to find a hidden period (a bitstring) encoded into an unknown $2$-to-$1$ function. It is one of the earliest problems for which an exponential quantum speedup was proven for ideal, noiseless quantum computers, albeit in the oracle model. Here, using two different $127$-qubit IBM Quantum superconducting processors, we demonstrate an algorithmic quantum speedup for a variant of Simon's problem where the hidden period has a restricted Hamming weight $w$. For sufficiently small values of $w$ and for circuits involving up to $58$ qubits, we demonstrate an exponential speedup, albeit of a lower quality than the speedup predicted for the noiseless algorithm. The speedup exponent and the range of $w$ values for which an exponential speedup exists are significantly enhanced when the computation is protected by dynamical decoupling. Further enhancement is achieved with measurement error mitigation. This constitutes a demonstration of a bona fide quantum advantage for an Abelian hidden subgroup problem.
\end{abstract}

\maketitle

\section{Introduction}

Quantum algorithms have been known to outperform classical algorithms for more than 30 years~\cite{Deutsch:92,bernsteinQuantumComplexityTheory1997,Simon:94,Grover:97a,Shor:97,childs2003exponential,Van-Dam:2006aa,Harrow:2009aa,montanaroQuantumAlgorithmsOverview2016,Bravyi:2017aa,Bravyi:2020aa,Bharti:2022aa,Daley:2022vu}, assuming that they run on ideal, noiseless quantum devices. However, today's noisy intermediate-scale quantum (NISQ) \cite{Preskill2018} devices are functional on a relatively small scale of several hundreds of qubits and are highly susceptible to performance degradation due to decoherence and control errors.
A central current focus is the experimental demonstration of an algorithmic quantum speedup on these devices, i.e., a scaling advantage for a quantum algorithm solving a computational problem. A variety of such demonstrations have been reported~\cite{Albash:2017aa,King:2019aa,Saggio:2021vh,Centrone:2021tq,Maslov:2021aa,Xia:2021ux,Huang:2021,Ebadi:22,zhouExperimentalQuantumAdvantage2022,King:22,Kim:2023aa}, but the classical hardness of the problems chosen in these demonstrations relied on computational complexity conjectures or the complexity of a restricted set of classical algorithms. These demonstrations are separate and distinct from recent quantum supremacy results~\cite{aaronson2016,Arute:2019aa,wu2021strong,Zhong:2020aa,Zhong:2021wv,morvan2023phase} and from benchmarking the performance of NISQ devices~\cite{figgattComplete3QubitGrover2017,wrightBenchmarking11qubitQuantum2019,royProgrammableSuperconductingProcessor2020,Pelofske2022,Lubinski2021}. The former does not corroborate a scaling advantage~\cite{Barak:spoofing,Zlokapa:2023aa,Aharonov:22}, while the latter focuses on demonstrating better-than-classical probabilities of success without resolving the question of scaling of the time to solve the computational problem with problem size, which is essential for the demonstration of an algorithmic quantum speedup~\cite{speedup}.

Recently, a conjecture-free algorithmic quantum speedup was demonstrated in the oracle model against the best possible classical algorithm~\cite{pokharel2022demonstration}. In particular, a polynomial algorithmic quantum scaling advantage for the single-shot version of the Bernstein-Vazirani algorithm was observed when implemented on a 27-qubit IBM Quantum processor with noise suppression via dynamical decoupling (DD)~\cite{Viola:98,Viola:99,Zanardi:1999fk,Vitali:99,Duan:98e}, a well-established error suppression method for NISQ devices~\cite{Pokharel2018, Arute:2019aa, souza2020process,tripathi2021suppression,jurcevicDemonstrationQuantumVolume2021,raviVAQEMVariationalApproach2021,Zeyuan:22, baumer2023efficient, seif2024suppressing, Shirizly:2024aa, Baumer2024, evert2024syncopated, brown2024efficient,tripathi2024quditdynamicaldecouplingsuperconducting}. 
The Bernstein-Vazirani algorithm was among the very first algorithms for which a quantum speedup was rigorously proven, and in this sense, it is of historical significance. Even more interesting would be an algorithmic quantum speedup for a problem that belongs to the class of Abelian hidden subgroup problems~\cite{Jozsa:2001aa}, which includes integer factorization.
Simon's problem~\cite{Simon:94}, a precursor to Shor's factoring algorithm~\cite{Shor:97}, involves finding a hidden subgroup of the group $(\mathbb{Z}_2^n, \oplus)$, which is the $n$-fold direct product of the cyclic group $\mathbb{Z}_2$ (the integers modulo 2) with the group operation being bitwise XOR ($\oplus$), and this group is Abelian. In Simon's problem, the Abelian hidden subgroup consists of the identity and a secret string $b$, and the goal is to determine $b$. 

Here, we revisit Simon's problem and demonstrate an unequivocal algorithmic quantum speedup for a restricted Hamming-weight version of this problem using a pair of IBM Quantum processors. Similar to Ref.~\cite{pokharel2022demonstration}, we find an enhanced quantum speedup when the computation is protected by DD. The use of measurement error mitigation (MEM) further enhances the scaling advantage we observe. Our result can be viewed as bringing the field of NISQ algorithms closer to a demonstration of a quantum speedup via Shor's algorithm. It also highlights the essential role of quantum error suppression methods in such a demonstration.

To set the stage, we now explain the restricted-weight Simon's problem that is the focus of this work. 
In the original formulation of Simon's problem, we are given a function
$f_b: \{0,1\}^n \rightarrow \{0,1\}^n$,
where $n$ is the problem size, and we are promised that $f_b$ is either 1-to-1 or 2-to-1, such that
$\forall x,y \in \{0,1\}^n$,
$f_b(x) = f_b(y)$
if and only if $x = y$ or $x = y \oplus b$ for a hidden bitstring $b\in \{0,1\}^n$. We wish to determine which condition holds for $f_b$ and, in the latter case, find $b$. Since the first condition (1-to-1) has one possible hidden bitstring $b=0^n$, in this work we only consider the 2-to-1 version of Simon's problem, and our goal is to determine the length-$n$ hidden bitstring $b$. It is well known (and we revisit this below in detail) that the classical and quantum query complexities are $O(2^{n/2})$ and $O(n)$, respectively.

Instead of allowing all possible binary hidden bitstrings $b\in\{0,1\}^n$ as in the original Simon's problem, we restrict their Hamming weight (HW), that is, the number of $1$'s in $b$. We refer to this modification as the `restricted-HW' version of Simon's algorithm, denoted \wSimon{w}{n}. Here $w\le n$ is the maximum allowed HW. We will show that this modification allows us to exhibit a quantum speedup for relatively shallow circuits, whose depth is set by $w$. The need for the restriction arises because current NISQ devices are still too noisy to fully solve the original, unrestricted ($w=n$) Simon's problem when $n$ becomes large.

We work in the setting of a guessing game, the rules of which are tuned in order to make the speedup possible and are explained in \cref{s:rules}. In particular, we introduce an oracle-query metric we call NTS (number-of-oracle-queries-to-solution), to quantify the performance of different players of this game. The optimal classical and quantum algorithms to solve Simon's problem are described in \cref{sec:clqm-algo}. We show that the classical algorithm requires $\Omega(n^{w/2})$ oracle queries, whereas the quantum algorithm requires $\sim w\log_2(n)$ queries. In \cref{sec:speedup} we quantify how a quantum speedup can be detected using the NTS metric. In \cref{sec:expt} we discuss the experimental setup and IBM Quantum devices we used to perform our algorithmic speedup tests, including the DD sequences we selected to enhance performance. Then, in \cref{sec:results} we discuss our results and the evidence for a quantum speedup in the \wSimon{w}{n} problem. We conclude in \cref{sec:conc}. The appendices contain additional technical details as well as supplemental experimental results.

Finally, we caveat the speedup result we find by noting that although the structure of the oracle is unknown to the player who wishes to solve Simon's problem, someone needs to play the role of the `verifier' who acts like a referee, that is, construct the oracle for the player and verify whether the player's guess is correct. The verifier needs to know the structure of the quantum oracle $f_b$ in order to construct the quantum circuit for each hidden bitstring $b$. Since it is a Clifford circuit, the oracle we construct in this work can be efficiently simulated by a classical computer. However, this does not destroy the claimed speedup because our setting assumes the `black box' scenario \cite{Zantema2022}, i.e., the players are not allowed to know the structure of the oracle and solve the problem in linear time by classically constructing the oracle on their own.

\section{Rules of the Game}
\label{s:rules}

The game is designed for a single player. If there are multiple players, they can play the game individually and then compare their scores. At a high level, the game works as follows. A function $f\colon \{0, 1\}^n \to \{0, 1\}^n$ that satisfies the 2-to-1 condition is chosen uniformly at random. The function is not known to the player, but the player has oracle access to compute $f(x)$ for any $x\in \{0, 1\}^n$. For classical players, oracle access means that the player can send a query $x$ to the oracle and receive $f(x)$ in return. For quantum players, we define a unitary $\mcOf$ such that $\mcOf \ket{x} \ket{a} = \ket{x} \ket{a \oplus f(x)}$. The player performs a certain number of such oracle queries and arbitrary classical computations and then guesses the hidden bitstring $b$. The correctness of the guess is checked, and a new round begins, i.e., a new function is chosen, and the game is repeated. The goal of the player is to maximize the number of correct guesses using the smallest number of oracle queries.

The rest of this section is devoted to specifying the details missing in the simplified description above: the set of functions $f$, the oracle access mechanism, and the exact formula used to score the player.

\subsection{Set of functions \texorpdfstring{$f$}{f}}
\label{sec:rules:f}

We would like to design a game that is hard for a classical computer and easy for a quantum computer. The classical complexity lower bound (see \cref{sec:cl-algo}) is based on the assumption that the only information about $b$ that the classical player can extract by evaluating $f$ at various points is the presence or absence of a match $f(x) = f(y)$ for pairs of queries $(x, y)$. In $\wSimon{w}{n}$, the conditions are as follows: (i) $\HW(b)\leq w$; (ii) $f(x) = f(y)$ if and only if $x = y$ or $x = y \oplus b$; and (iii) $b \neq 0^n$. The function $f$ is then chosen uniformly at random from the set of all functions that satisfy these conditions. Such a broad class of functions $f$ may seem overly complicated, but to illustrate the importance of this choice, consider an extreme alternative: $b$ is chosen uniformly from all bitstrings that satisfy conditions (i) and (iii), and $f(x) = \min(x, x \oplus b)$, where $\min$ yields the first bitstring in lexicographic ordering. In this case, since $f(1^n) = 1^n \oplus b$, we have $b = 1^n \oplus f(1^n)$. Therefore, the classical player can find $b$ in a single query: $x=1^n$.

\subsection{Oracle access mechanism}

Unfortunately, choosing $f$ uniformly at random from the set of all functions that meet conditions (i)--(iii) leads to a problem. To select a function we need to pick its value for each pair of inputs $x$ and $x\oplus b$. There are $2^{n-1}$ such pairs and $2^n$ possible outputs, all of which should be different. One way to compute the number of such functions is to first pick one of the $\binom{2^n}{2^{n-1}}$ subsets as the image of $f$ and then one of the $(2^{n-1})!$ ways to map $2^{n-1}$ inputs to $2^{n-1}$ outputs, giving the total number of functions $f$ for a fixed $b$ as $\binom{2^n}{2^{n-1}} (2^{n-1})! = 2^{n}! / (2^{n-1})! > 2^{(n-1)2^{n-1}}$, which means that the choice of $f$ requires at least $(n-1)2^{n-1}$ encoding bits. One would thus expect exponentially many gates to be required to implement the unitary quantum oracle $\mcOf$ in a quantum circuit. Such large circuits are infeasible on NISQ devices. To circumvent this problem, we introduce the notion of a compiler that preserves the shallowness of the quantum circuits implementing $\mcOf$ but prevents making the problem trivial by offloading all computations to classical post-processing. The details of this compiler are given in \cref{app:rules:oracle}. The corresponding Simon's oracle construction is described in \cref{app:simon-oracle}.

Under our rules of the game, the player does not see the inner workings of the oracle implementation. This is because otherwise, the player would be able to efficiently reconstruct $b$. Nor does the player get to see the inner workings of the compiler since an efficient reconstruction of $b$ would be possible if the player saw (1) the circuit that was actually sent to the quantum device and (2) the mapping between the qubits of that circuit and the qubits in the original circuit.

\subsection{Scoring}
\label{sec:scoring}

Any function $f$ satisfying (i)--(iii) can be decomposed as $f(x) = f_1(f_0(x))$, where $f_0$ is any 2-to-1 function satisfying condition (ii) with the same $b$, and $f_1$ is a permutation of bitstrings. A uniformly random $f$ can be obtained by fixing $f_0 = f_{0b}$ for every $b$, then picking $b$ and $f_1$ uniformly at random.    
Since exponentially many bits are required to record the permutation $f_1$, the program to compute $f_1$ occupies exponential memory. If no caching is involved, this program would need to be loaded into memory every time, which takes exponential time in $n$. Thus, classical post-processing may easily remove most of the quantum advantage. To circumvent this,
we introduce a scoring function that only accounts for the number
of oracle queries and ignores the time needed for all other steps of the protocol.

Given the goal of maximizing the number of correct guesses of $b$ while minimizing the number of oracle queries, the players' performance 
can be measured in terms of the average score per query, which we denote by $\NTS^{-1}$,
where
\begin{equation}
  \NTS = \frac{\expv{Q}}{\expv{P}} .
  \label{eq:NTS.def}
\end{equation}
Here NTS denotes the \emph{number-of-oracle-queries-to-solution},
$\expv{\bullet}$ denotes the expectation value (which can be obtained
as an average over many rounds of the game), $Q$ is the number of
oracle queries, and $P$ is the score obtained in a round.
The best player is the one with the highest $\NTS^{-1}$.

Naively, one could set $P = 1$ if the guess is correct and $P = 0$ otherwise.
However, this would allow the players to use the following strategy. Let
$N$ be the total number of options for $b$, and let $C$ be a constant satisfying
$N \geq C > 0$.
With probability $C / N$, the player executes a single query and discards
the result; otherwise, $0$ queries are executed. Then the player guesses
$b$ uniformly at random. The expected score is $1/N$, and the expected number
of queries is $C/N$. Thus, $\NTS = C$, and the player can make $\NTS$ arbitrarily
small (e.g., $\NTS = 1$) by picking $C$ small enough.

To avoid this, we introduce a penalty for incorrect guesses. Let $p_r = 1/N$
be the probability of a random $b$ being correct. Then we set $P = 1$ if the guess
is correct and $P = -p_r / (1 - p_r)$ otherwise. This ensures that the randomly
guessing player has $\expv{P} =  p_r \times 1 + (1-p_r) \times [-p_r / (1 - p_r)] = 0$ and $\NTS = \infty$.

To recap, we focus on counting the number of oracle calls 
for two main reasons:
(i) optimizations performed by the compiler may change the time significantly, and 
(ii) the data generated by NISQ devices, contrary to classical and ideal quantum algorithms, requires non-trivial post-processing, which will destroy any quantum speedup we might otherwise derive from such devices. Indeed, the classical post-processing algorithm we describe in \cref{app:post-processing} has a higher cost than that of the classical algorithm we describe in the next section.

\section{Classical and quantum algorithms}
\label{sec:clqm-algo}

We now explain in detail the algorithms for the different variants of the algorithm we consider:
  (i) classical,
  (ii) noiseless quantum, 
  (iiia) NISQ, and
  (iiib) NISQ without measurement errors, simulated using measurement error mitigation (MEM). 

We first comment briefly on the quantum variants. In the \emph{noiseless quantum} case, an ideal gate-based quantum computer executes the circuit the player designs. Whether any compilation is performed is irrelevant to determining the optimal score since such a compilation does not affect the number of oracle calls or the probability distribution of the outcomes of the circuit execution returned to the player. In the \emph{NISQ} case, we run the compiled circuit on a NISQ device and, after post-processing, return the result to the player. The score a player receives in this case depends not only on the player's strategy but also on the NISQ device used. In the \emph{NISQ without measurement errors} case, we run the circuit multiple times and use MEM to estimate the probability distribution we would have obtained from a single run if measurement errors were absent. Then we sample from that probability distribution, apply the post-processing, and return the result to the player.

From the standpoint of the player, the procedure for implementing the `NISQ' and `NISQ with MEM' variants remains identical. This is due to MEM being processed within the compiler, with the observer receiving solely the outcome, presumed to incorporate measurement errors (NISQ) or devoid of such errors (NISQ with MEM). 

\subsection{Classical algorithm}
\label{sec:cl-algo}

The original Simon's problem can be solved using a classical deterministic algorithm in $\Theta(2^{n/2})$ queries \cite{Simon:94,Cai2018}.\footnote{Recall that $f(x) \in \Theta[g(x)]$ means that $f$ and $g$ grow at the rate asymptotically: there exist two positive constants $c_1$ and $c_2$ and $x_0>0$ such that for all $x \ge x_0$: $c_1 g(x) \leq f(x) \leq c_2 g(x)$.}
We sharpen this bound here and generalize it to the setting of the \wSimon{w}{n} problem. 

Let $S$ be the set of all possible values for the hidden bitstring $b\neq 0^n$. In the original Simon's problem the size of this set is $N_n\equiv 2^n-1$.
We now consider \wSimon{w}{n}, where the set $S$ of possible values of $b$ is restricted by $\HW(b)=w<n$. The size of this set is
\begin{equation}
N_w \equiv \sum_{j=1}^{w}\binom{n}{j}.
\label{eq:N_w}
\end{equation}

\begin{mytheorem}
\label{th:Simon-classical}
A lower bound on the worst-case number of queries needed by a classical player to solve \wSimon{w}{n} is:
\begin{equation}
     k \geq \left\lceil \sqrt{2N_w-\frac{7}{4}}+\frac{1}{2}\right\rceil \equiv k_{\min}(N_w) .
     \label{eq:cl-worst}
\end{equation}
\end{mytheorem}
\noindent The proof is presented in \cref{app:alternative-deriv}.

Note that since this result is for a deterministic classical algorithm, it is also a lower bound on the worst-case classical NTS. 
\Cref{eq:NTS.def}, however, involves not the worst-case but the average number of queries $\expv{Q}$.
We show in \cref{app:alternative-deriv} that the lower bound on the expected number of classical queries $\expv{Q_C}$ needed to know $b$ exactly is: 
\begin{equation}
     \expv{Q_C} \geq 
     k-\frac{k(k-1)(k-2)}{6N_w} .
    \label{eq:cl-expect}
\end{equation}
As mentioned above, guessing $b$ when it is unknown to the classical player does not increase the $\NTS$; hence, 
the lower bound on the classical NTS $\equiv \NTS_C =  {\expv{Q_C}}/{\expv{P_C}} \ge {\expv{Q_C}}$ is given by the r.h.s. of \cref{eq:cl-expect}. Combining this with \cref{eq:cl-worst} gives
\begin{equation}
\NTS_C \ge \NTS^{\text{lb}}_C(N_w) \equiv k_{\min}-\frac{k_{\min}(k_{\min}-1)(k_{\min}-2)}{6N_w}
    \label{eq:cl-expect2}
\end{equation}
as the lower bound on the average case classical NTS. 

Note that $\binom{n}{w} \le N_w \le \min(w \binom{n}{w}, 2^n -1)$.
Thus, for constant $w$ (independent of $n$), $N_w \sim n^w$, $k_{\min} \sim n^{w/2}$, and \cref{eq:cl-expect2} yields $\NTS_C\sim n^{w/2}$. Without restricting $\HW(b)$, $\NTS_C=O(2^{n/2})$. However, with a fixed upper limit $w$ on $\HW(b)$, $\NTS_C$ is polynomial asymptotically (as $n\to \infty$).

Although \cref{eq:cl-expect2} lower-bounds $\NTS_C$, there is no guarantee that the player will find a sequence of bitstrings $x = \{x_1,x_2,\dots,x_k\}$ that achieves it.  Nevertheless, henceforth we use the lower bound given by \cref{eq:cl-expect2} as the metric for $\NTS_C$ because if the quantum algorithm defeats this lower bound, the implication is an unequivocal quantum speedup. In \cref{app:bounds-gap}, we derive an upper bound on $\NTS_C$ and discuss the gap between the lower and upper bounds.

\subsection{Noiseless quantum algorithm}
\label{sec:q-algo}

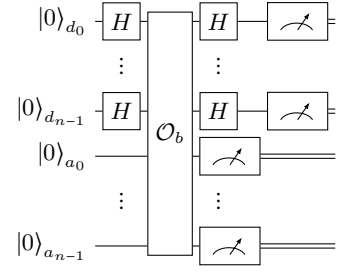
\begin{figure}
    \centering
    \begin{tikzpicture}
        \begin{yquant}
            qubit {$\ket{0}_{d_0}$} x1;
            nobit d0;
            qubit {$\ket{0}_{d_{n-1}}$} xn;
            qubit {$\ket{0}_{a_0}$} a1;
            nobit d1;
            qubit {$\ket{0}_{a_{n-1}}$} an;
            h x1;
            h xn;
            text {$\vdots$} d0;
            text {$\vdots$} d1;
            box {$\mc{O}_b$} (x1, d0, xn, d1, a1, an);
            h x1;
            h xn;
            measure a1;
            measure an;

            align x1, xn, a1, an;
            measure x1;
            measure xn;
            text {$\vdots$} d0;
            text {$\vdots$} d1;
        \end{yquant}
    \end{tikzpicture}
    \caption{Quantum circuit for solving Simon's problem with a length-$n$ hidden bitstring $b$. The structure of the oracle $\mc{O}_b$ is not visible to the player who wishes to guess $b$. The top $n$ measurement results form the bitstring $z$, and the bottom $n$ measurement results are discarded in the algorithm.}
    \label{fig:simon-circ}
\end{figure}

The original Simon-$n$ problem can also, in theory, be solved on a noiseless quantum computer (QC). \Cref{fig:simon-circ} shows a quantum circuit that executes this algorithm. The circuit uses $2n$ qubits for \wSimon{w}{n}; the first $n$ are data qubits, labeled $d_j$, and the other $n$ are ancilla qubits, labeled $a_j$, where $j=0,1,\dots,n-1$. The box $\mc{O}_b$ is the oracle whose action is $f_b(x)$ on an input $x$, but its circuit is hidden from the players. Each execution of the circuit produces a uniformly random $z$
such that $b \cdot z = 0$.
The well-known proof that $O(n)$ executions of this circuit (on average) on an ideal quantum computer are sufficient for solving Simon-$n$ is given in \cref{app:SimonQ} for completeness.

For our purposes, we need the $\NTS$ for the noiseless quantum \wSimon{w}{n} problem, for arbitrary $w$. More precisely, we need the $\NTS$ for an 
optimal, ideal quantum player executing the circuit in \cref{fig:simon-circ} until $b \in S$ can be uniquely determined, which we denote by $\NTSIQ$ for `ideal quantum'. To obtain this, we first consider a generalization of the \wSimon{w}{n} problem to the case where $b$ is known to belong to some subset $S$ of $\mathbb{Z}_2^n \setminus \{0^n\}$, and obtain $\NTSIQ$ exactly in the two limits $w=1,\infty$, after which we construct $\NTSIQ$ for arbitrary $1<w<\infty$ via interpolation.

\begin{mytheorem}
\label{th:1}
If $\abs{S} = 1$ then $\NTSIQ = 0$.
  If $\abs{S} \geq 2$ then
  \begin{equation}
    \label{eq:NTSIQ-bounds}
    \log_2 \abs{S} \leq \NTSIQ \leq \log_2 (\abs{S} - 1) + 2.
  \end{equation}
  For \wSimon{w}{n}, we have $\abs{S} = N_w$.
For Simon-$n$ (i.e., \wSimon{\infty}{n}), $\abs{S} = 2^{n} - 1$ and
  \begin{equation}
    \label{eq:NTSIQ-winf}
    \NTSIQ = \sum_{k=1}^{n-1} \frac{1}{1-2^{-k}} = n + \EEB - 1 + O(2^{-n}),
  \end{equation}
  where $\EEB = 1.60669\dots$ is the Erd\H{o}s
-Borwein constant~\cite{Erdos-Borwein}.
  For \wSimon{1}{n}, $\abs{S} = n$ and
  \begin{equation}
    \label{eq:NTSIQ-w1}
    \NTSIQ = \log_2(n) + \frac12 + \frac{\gamma}{\ln(2)} + \chi(\log_2(n)) + O\left(\frac{1}{n}\right),
  \end{equation}
  where $\gamma = 0.57721\dots$ is the Euler-Mascheroni constant
  and $\chi$ is a small ($\ll 1$) periodic function with mean $0$.
\end{mytheorem}

The proof is given in \cref{app:th1-proof}.
Numerically, we estimate the amplitude of $\chi$ in \cref{eq:NTSIQ-w1} as $1.58 \times 10^{-6}$.

We now construct the aforementioned interpolation. Let $t \equiv N_w / (2^n - 1)$; $t$ represents the density of $b$'s in the space of all non-zero bitstrings. In \cref{eq:NTSIQ-winf} $t=1$ and in \cref{eq:NTSIQ-w1} $t$ approaches $0$.
The interpolation between \cref{eq:NTSIQ-winf,eq:NTSIQ-w1} is then:
\begin{equation}
\NTSIQ(t) = \log_2(N_w) + \left(\frac12 + \frac{\gamma}{\ln(2)}\right) (1-t) + (\EEB - 1) t .
\label{eq:NTSIQ-interp}
\end{equation}
\Cref{app:NTSIQ} demonstrates that this interpolation is an accurate estimate of the actual $\NTSIQ$ for the \wSimon{w}{n} problem when $1<w<\infty$.

\subsection{NISQ algorithm}
\label{sec:NISQ-algo}

The `NISQ algorithm' executes the circuit depicted in \cref{fig:simon-circ} on a NISQ device. 
Recall that to determine the correct hidden bitstring $b$ in the context of a noiseless quantum algorithm, it is necessary to acquire a set of $n-1$ independent measurements $z$, all of which satisfy the orthogonality condition $b \cdot z = 0$. However, in the presence of noise, the measurement outcomes from the first $n$ qubits are expected to deviate from the noiseless case, potentially leading to incorrect values for $z$. The likelihood of obtaining at least one incorrect $z$ increases as the problem size $n$ increases, and given the prevailing noise levels in current NISQ devices, the probability of encountering at least one erroneous $z$ is high even for relatively small problem sizes. Consequently, attempting to determine the parameter $b$ based solely on $n-1$ measurements is highly likely to yield an incorrect solution.

As a remedy, we employ a modified algorithm, which does not have any explicit bound on the number of queries, but the theory predicts that for the Simon-$n$ problem $\ln(N_w) \exp(O(n^2))$ measurements are required on average, and for the $\wSimon{n}{w}$ problem $\ln(N_w) \exp(O(w^2))$ measurements are required on average. Our empirical data roughly supports these theoretical estimates.
We assume that the player has access to the probability distribution $\Pr(z\cdot b=0,z\neq 0)$ of the quantum device. An example of such data is presented in \cref{app:MEM}.
This data is obtained by running the quantum device in a separate experiment, distinct from the main experiment designed to solve Simon's problem. With this probability in hand, the player utilizes Bayesian statistics to compute prior and posterior probabilities after each call to the oracle. Based on these probabilities, the player decides whether to guess $b$ or to make an additional call. For a comprehensive description, see \cref{app:post-processing}.

We subject the outcomes from the `NISQ' and `NISQ with MEM' algorithms to identical post-processing methodologies. The speedup comparison between the classical and quantum algorithms uses the better of the outcomes observed between these two quantum variants.

\section{Quantum speedup}
\label{sec:speedup}

In \cref{sec:scoring}, we motivated and defined the NTS
metric, which we use to compare the number of oracle queries
for classical and NISQ solutions to Simon's problem. Instead of computing the NTS for all possible oracles $\mcOf$, or even for all $N_{w}$ values of $b$, we note that our NISQ implementation for $\wSimon{w}{n}$ depends only on $n$, $w$, and the Hamming weight $\HW(b)=i$, which allows us to focus only on $w\le n$ instances of $b$: one representative $b$ for each $i$. In addition, instead of rerunning our circuits for every new value of $n$, we extract all $n$ values from the largest-$n$ circuits we implement. For full details, see the reduction method described in \cref{app:circ-reduction}.

\subsection{Quantum speedup quantification}

Let $h_i = {n\choose i}$ be the corresponding number of length-$n$ bitstrings $b$ and let $Q_i$ be the corresponding total number of oracle queries; we can write the number of queries for a given $i$ as $h_i Q_i$, and thus the total number of queries per hidden bitstring as $\ave{Q} = \sum_{i=1}^{n} h_i Q_i/N_n$. Let $p_i$ be the probability of the correct guess for a representative bitstring $b$ with $\HW(b) = i$. That is, $p_i$ is the fraction of successful guesses of the representative $b$ based on the $Q_i$ oracle calls. Assuming all $h_i$ bitstrings of fixed Hamming weight $i$ are guessed with the same probability $p_i$, the total probability of successfully guessing all weight-$i$ bitstrings is $h_i p_i$, and the expected overall probability of success over all bitstrings is $\ave{p} = \sum_{i=1}^{n} h_i p_i/N_n$. In the restricted $\wSimon{w}{n}$ version, since the player knows $w$, a random guess succeeds with probability $p_r=1/N_w$, where $N_w$ replaces $N_n$.
Then the average score (which includes the penalty for incorrect guesses and was defined in \cref{sec:scoring}) can be written as $\ave{P} = \ave{p}\times 1 + (1-\ave{p})\times [-p_r/(1-p_r)] = \ave{p} - (1-\ave{p})/(N_w-1) = \frac{1}{N_w-1}(N_w\ave{p}-1)$. We can now express the quantum NTS for the unrestricted ($w=n$) or restricted ($w<n$) version of Simon's problem as:
\begin{equation}
  \NTS_Q(n; w)
  = \frac{\ave{Q}}{\ave{P}}
  = \frac{N_w-1}{N_w} \frac{\sum_{i=1}^{w} h_i Q_i}{\sum_{i=1}^{w} h_i p_i-1} .
\label{eq:NTS-wn}
\end{equation}
Recall that we need to compare $\NTS_Q$ to the lower bound on the average-case classical NTS, i.e., $\NTS^{\text{lb}}_C$ given by \cref{eq:cl-expect2}, which is a simple function of $k_{\min}(N_w)$, the lower bound on the worst-case classical NTS given by \cref{eq:cl-worst}. Since $k_{\min}(N_w) \sim N_w^{1/2}$ for $N_w \gg 1$, it is more convenient to define the NTS scaling with respect to $N_w$ than $n$.

We can now define an \emph{algorithmic quantum speedup} for the $\wSimon{w}{n}$ problem as a better scaling with $N_w$ of the function $\NTS_Q$ [\cref{eq:NTS-wn}] than the function $\NTS^{\text{lb}}_C$ [\cref{eq:cl-expect2}]. We determine the scaling by fitting two two-parameter models and one three-parameter model that is intermediate between the first two. All three models satisfy the constraint $\NTS(N_w=1)=0$, since when $N_w=1$ the only possible string is $b=1$, so no oracle calls are needed.

First, consider a polylogarithmic (polylog) model in $N_w$:
\begin{equation}
    \NTS_{\text{polylog}}=a  (\log_2{N_w})^{\alpha} .
    \label{eq:polylog-NTS}
\end{equation}
Second, consider a polynomial (poly) model in $N_w$:
\begin{equation}
    \NTS_{\text{poly}}= b [N_w^\beta-1] .
    \label{eq:poly-NTS}
\end{equation}
Third, consider an intermediate  (mixed) model parameterized by three parameters:
\begin{equation}
    \NTS_{\text{mix}} = c[e^{C (\log_2 N_w)^\gamma\ln(1+\log_2 N_w)^{1-\gamma}}-1].
    \label{eq:mixed-hw-NTS}
\end{equation}
This model is inspired the by scaling of the classical number sieve algorithm for integer factoring, and its scaling is known as $L$-notation~\cite{Lenstra:2011aa}. When $\gamma=0$, $\NTS_{\text{mix}}=c[(1+\log_2 N_w)^C-1]$, which has the same scaling as the polylog model. When $\gamma= 1$, $\NTS_{\text{mix}}=c[N_w^{C/\ln 2}-1]$, which is the poly model. When $\gamma > 1$, the mixed model grows even faster than the poly model. When $\gamma < 0$ the mixed model gives an unreasonable scaling (which decreases after increasing). Hence, we constrained $0\le \gamma \le 1$ and $C,c>0$ to ensure that the mixed model provides a meaningful intermediate between the polylog and poly models. 

Since the classical scaling, $\NTS_C \sim N_w^{1/2}$, follows the poly model, an \emph{exponential} algorithmic quantum speedup is observed when $\NTS_Q$ follows the polylog model. If instead, $\NTS_Q$ follows the poly model with $\beta_Q < \beta_C$ for a fixed $w$, a \emph{polynomial} speedup is observed. When $\NTS_Q$ follows the mixed model with $0\le \gamma \le 1$, a \emph{subexponential} (and superpolynomial) algorithmic quantum speedup is observed. 

We refer to $\alpha$, $\beta$, and $\gamma$ as the scaling exponents. The fitting parameters $a$, $b$, $c$ and $C$ do not matter for scaling purposes. We explain how we determine which model is the better fit in \cref{sec:results}.

\subsection{Theoretical derivation of quantum speedup on NISQ devices}
\label{sec:NISQ-speedup-expectation}

\subsubsection{Exponential error model}
\label{sssec:poisson}

We now present a simple error model to predict the NTS scaling and to assess whether a speedup should be expected in the presence of noise. We prove a theorem (\cref{thm:NTS-bound}) that is a key result of this work: \emph{the restricted-weight Simon's problem exhibits an exponential quantum speedup even in the NISQ setting}. 


Recall from \cref{eq:NTSIQ-interp} that with a noiseless quantum computer only $\log_2(N_w) + O(1)$ circuit executions would be needed to find $b$.
To understand how this is modified in the NISQ setting, we use a simple error model
where we only care about the probability $p$ of $\kappa=0$ non-benign errors, i.e.,
errors that may affect the value of $z \cdot b$.
This probability is $p = \Pr(\kappa=0) = \prod_{j=1}^{M} (1-p_j)$,
where $j$ indexes the error locations in the circuit (gates and idle times between them)
which contribute to the value of $z \cdot b$,
$p_j$ is the corresponding error probability,
and $M$ is the number of locations.
We can rewrite this as $p = \exp(\sum_{j=1}^{M} \ln(1-p_j)) = \exp(M \langle-\ln(1-p_j)\rangle)$.
I.e., an absence of errors occurs with probability $p = e^{-q}$,
where $q \equiv M \langle -\ln(1-p_j)\rangle$.
As the circuit (i.e., $M$) grows, we expect $\langle-\ln(1-p_j)\rangle$ to approach a constant value and $q$ to asymptotically become proportional to $M$.

This simple exponential error model is appropriate for the circuits used in our experiments but would need to be adjusted if circuits involving features improving the resilience to noise, such as error detection, were used instead, or if $q/M$ deviates from a constant,
e.g., due to error correlation with the timing of gates.

When at least one non-benign error occurs (i.e., with probability $1 - p$),
we assume that the output bitstring $z$ is uniformly random,
while in the absence of such errors (probability $p$),
we obtain $z$ uniformly at random from the set of bitstrings
$z'$ satisfying $b \cdot z' = 0$.

The specific circuit used determines how $q$ depends on $n$ and $i = \HW(b)$. We assume the same circuits as used by the NISQ player in the experiments we
present in this work, whose depth is $\Theta(i)$.
Since $b \cdot z = \sum_j b_j z_j$ is invariant under errors that flip any $z_j$ for which $b_j=0$,
only the qubits corresponding to $b_j=1$
and the corresponding output qubits are subject to non-benign errors.
Thus, the number of error locations $L$ is $\Theta(i)$ (the number of qubits)
times $\Theta(i)$ (the depth), i.e., $q \in \Theta(i^2)$ uniformly in $n$.
In particular, $q < \lambda i^2 + o(i^2)$ for some $\lambda > 0$.
For the unrestricted Simon's problem, the maximum Hamming weight of $b$ is $n$,
thus $q \in O(n^2)$.
The role of error suppression via DD is to reduce $\lambda$.

\subsubsection{Simplified exponential error model}
In the error model above $q$ depends on the Hamming weight $i$ of $b$, which makes the theoretical analysis of the speedup challenging.
Below, for simplicity, we assume a constant $q$
and call the corresponding error model the \emph{simplified exponential error model}.
In this model, $z$ is sampled uniformly from all bitstrings satisfying
$b \cdot z = 0$ with probability $p = e^{-q}$ and uniformly at random
from all bitstrings in $\mathbb{F}_2^n$ otherwise.
A self-contained theoretical analysis of learning $b$ in this model is
presented in \cref{app:dkl}. One of the key results can
be summarized as follows:

\begin{mytheorem}
\label{thm:NTS-bound}
Assume a noisy quantum computer whose non-benign errors obey the simplified exponential error model with parameter
$q$ (independent of $b$). Then the \wSimon{w}{n} problem is solved in time
\begin{equation}
  \NTS_Q(n; w) \leq \frac{8 \ln(N_w) e^{2q}}{1 - N_w^{-2} + O(e^{-q})}.
\end{equation}
\end{mytheorem}
This is an asymptotic exponential speedup [in $\ln(N_w)$]
in the NISQ setting for
the bounded Hamming weight case if $q$ has a bound independent of $n$.

Here, the conditions on $q$ are consistent with the simplified exponential error model
for circuits used in the experiments in this work but may not hold
for other circuits.

\begin{proof}
  Consider an agent who executes the circuit until the posterior
  probability of one of the options is at least $(1 + 1/N_w)/2$,
  then submits that guess. By construction, this agent
  has $\mathbb{E}(P) \geq 1/2$ [recall that $\mathbb{E}(P)$ is
  the denominator in the NTS formula \cref{eq:NTS.def}].
  In \cref{app:dkl}, \cref{lm:Qmax} and \cref{eq:Qmax.asymptotic}
  we obtain a bound on the expected number of queries
  [numerator in \cref{eq:NTS.def}]:
  \begin{equation}
    \label{eq:Qmax.asymptotic.copy}
    \mathbb{E}(Q) \leq \frac{2 \ln(2) p^{-2} H(P_B)}{1 - K_1 + O(K_1 p)}.
  \end{equation}
  Here $K_1 = (1 + N_w^{-2})/2$,
  $H(P_B) = \log_2(N_w)$ is the entropy of the prior distribution over $b$,
  and $p = e^{-q}$. Substituting these values into \cref{eq:Qmax.asymptotic.copy} we obtain
  \begin{equation}
    \label{eq:Qmax.asymptotic.q}
    \mathbb{E}(Q) \leq \frac{4 \ln(N_w) e^{2q}}{1 - N_w^{-2} + O(e^{-q})}.
  \end{equation}
\end{proof}

Note that these conclusions depend on the validity of the simplified exponential error model we have assumed. In addition, the calculation is information-theoretic and does not take into account the postprocessing cost (it only counts the number of queries).
We are unaware of any postprocessing algorithm that would solve the noisy problem for $b$ in polynomial time, whereas the classical complexity is $\sqrt{N_w}$.
However, according to the rules of the game we specified, this does not destroy the speedup because only the number of queries is counted.

\cref{thm:NTS-bound} gives a rigorous upper bound on $\NTS_Q$ but does not explain the dynamics of how information about $b$ is acquired and updated via the results of progressive circuit executions. We next provide an approximate, heuristic model that provides a more explicit model of this process.

Let $B$ be a random variable indicating the correct (unknown) value of $b$ and
let $X_b(0) = 1 / N_w$ for $b \in S$ (where $S$ is the set of all options for $b$ with $\abs{S} = N_w$)
be the prior (uniform) distribution of $b$. Then, our model is:
\begin{equation}
\label{eq:X_b-model}
  X_b(t) = \frac{X_b(0) e^{-Y_b(t)}}{\sum_{b' \in S} X_{b'}(0) e^{-Y_{b'}(t)}},
\end{equation}
where $t$ is time,\footnote{%
Assume that Alice learns information about $B$ by continuously obtaining identically distributed weak observations over  time, and these observations are independent conditional on the value of $B$. Assume further that, as in the simplified exponential error model, the distribution of observations is symmetric (i.e., for any permutation of $S$ there is a permutation of the space of observations preserving the joint probability distribution). We conjecture that this would necessarily lead to a model like in \cref{eq:X_b-model} (i.e., this model is universal in some sense) up to a rescaling of time, i.e., if Alice's time is $t'$ then we would need to rescale it to obtain $t$: $t = C t'$. Hence, $t$ can also be interpreted as the scaled number of queries.}
$Y_b(t) = W_b(t) + t \delta_{b, B}$ and $W_b(t)$ are i.i.d. standard Wiener processes.\footnote{A Wiener process $(W_{t})_{t \ge 0}$ is a continuous-time stochastic process
satisfying the following properties:
  (1) $W_{0} = 0$;
  (2) For any $0 \le s < t$, the increment $W_{t} - W_{s}$ is independent of the $\sigma$-algebra (defined below) generated by $\{W_{u}:u \le s\}$;
  (3) The increments are stationary and Gaussian, specifically $\forall 0 \le s < t$ we have $W_{t}-W_{s} \sim \mathcal{N}(0,t-s)$, where $\mathcal{N}(\mu,\sigma^{2})$ denotes a normal distribution with mean $\mu$ and variance $\sigma^{2}$;
  (4) With probability $1$, the sample paths $t \mapsto W_{t}$ are continuous.}
Here $\delta_{b, B}$ is $1$ if $b$ is the correct bitstring, $0$ otherwise.
It can be shown that $X_b(t)$ is a time-homogeneous Markov martingale;\footnote{This means that $X_b(t)$ is a Markov process with no explicit $t$-dependence in its transition probabilities and whose expected future given the present equals the present value: $\mathbb{E}[X_b(t+\Delta)|\mathcal{F}_t] = X_b(t)$ $\forall t\ge0,\Delta>0$}
that $X_b(t) = \Pr(b = B| \mathcal{F}_t)$, where $\mathcal{F}_t$ is the $\sigma$-algebra generated by the $X_b(s)$ for $s \leq t$
and measure $0$ sets;\footnote{A $\sigma$-algebra represents all the events about which we can, in principle, assign probabilities. A probability measure assigns a probability to each event in the $\sigma$-algebra. Including measure zero sets in the $\sigma$-algebra ensures that if an event $A$ is in the $\sigma$-algebra and $B$ differs from $A$ by a measure zero set, then $B$ is also measurable.} that $B$ is $\mathcal{F}_{\infty}$-measurable but not $\mathcal{F}_t$-measurable for any $t < \infty$;\footnote{This means that knowing the process $X_b(t)$ for all $t\ge 0$ allows one to determine $B$ with probability $1$, but knowledge of $X_b(s)$ up to any finite $t$ does not suffice.}
and that
\begin{equation}
  \label{eq:d-HJ-Xt}
  \frac{d}{dt}\mathbb{E}[ H_{J}(X(t))| \mathcal{F}_{t}] = -\frac{1}{2 \ln(2)} \left(1 - \sum_{j=1}^{n} X_j(t)^2\right)
\end{equation}
where $H_{J}(X(t)) = -\sum_{j=1}^{n} X_j(t) \log_2(X_j(t))$ is the Shannon entropy of $X(t)$ (interpreted as a distribution over $S$).

The right hand side of \cref{eq:d-HJ-Xt} is the amount of information gained per unit time about $B$. This can be compared with information learned per query
in the simplified exponential error model, which is
\begin{equation}
  \label{eq:I.B.Zp.asymptotic.copy}
  I(B; Z_p) = \frac{(1 - K) p^2}{2 \ln(2)} + O(p^4 + K p^3)
\end{equation}
[see \cref{eq:I.B.Zp.asymptotic} in \cref{app:dkl}
for the derivation of this expression],
where $K = \sum_{b\in S} P_B(b)^2$: the sum of the squares of prior probabilities of the correct bitstrings before the query was executed. In the limit $p \to 0$ \cref{eq:d-HJ-Xt,eq:I.B.Zp.asymptotic.copy} match if we execute $p^{-2} = e^{2q}$
queries per unit of time.

Hence, $X_b(t)$ can be used as a model for the posterior probability of $B$ after $t e^{2q}$ queries, so that $t$ is the scaled number of queries.
This toy model can be used to derive heuristic
predictions for the speedup, stopping criteria, etc.

In the next section, we present our experimental results testing the possibility of a speedup in the setting of the restricted Hamming-weight Simon's problem.

\section{Experimental Implementation}
\label{sec:expt}

Our experiments were conducted using the $127$-qubit devices Sherbrooke (ibm\_sherbrooke) and Brisbane (ibm\_brisbane), as well as the $27$-qubit devices Cairo (ibm\_cairo) and Kolkata (ibmq\_kolkata), whose specifications are detailed in \cref{app:device}.

The quantum circuit for each hidden bitstring $b$ was constructed according to \cref{fig:simon-circ}, then compiled to fit the architecture of the devices using the as-late-as-possible (ALAP) schedule. This schedule initializes each qubit just before its first operation. We performed our experiments in two main modes: (i) `with DD', i.e., error-suppressed experiments with different dynamical decoupling sequences, and (ii) `without DD,' i.e., the circuit as specified in \cref{fig:simon-circ} without any additional error suppression.  

\begin{figure}[t]
    \centering
    \includegraphics[width=0.99\linewidth]{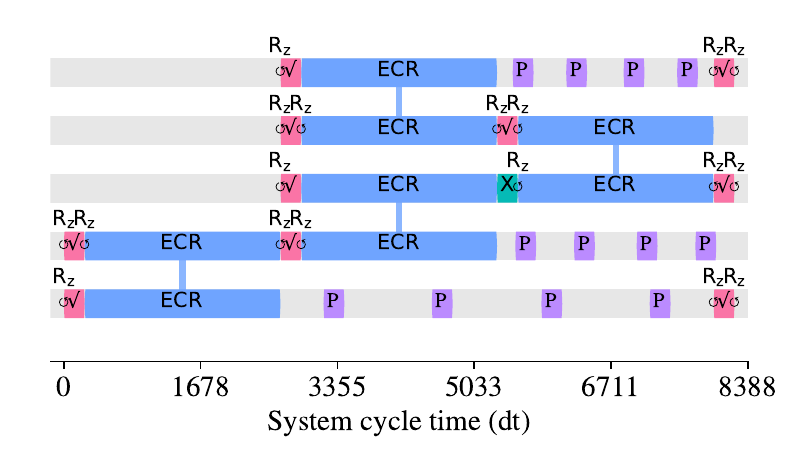}
    \caption{Circuit for Simon-3 ($b=111$) compiled into the Sherbrooke architecture using the ALAP schedule. The XY4 DD pulse sequences shown here as an example (the ``P'' makes inside the purple boxes) are placed such that they fill all available idle spaces after each qubit is initialized. Pulse intervals vary depending on the length of the idle period. circular arrows done $R_z$ gates and $\sqrt{}$ denotes the $\sqrt{X}$ gate. ECR denotes an echoed cross-resonance gate, which is equivalent to a CNOT up to single-qubit rotations. 
    }
    \label{fig:DD}
\end{figure}

We use a pre-compiled quantum circuit aligned with the underlying device architecture as our base circuit. The base circuit contains idle gap periods, suitable for the integration of DD sequences. We incorporated one iteration of a DD sequence into each idle gap, with the condition that if the duration of the period is insufficient for sequence inclusion, the gap remains unoccupied~\cite{pokharel2022demonstration}. The spacing of DD pulses is contingent upon both the length of the idle gap and the total number of pulses of the DD sequence. The compiled Simon-$3$ circuit ($b=111$) is depicted in \cref{fig:DD}. The idle intervals within this circuit are schematically populated with a $4$-pulse DD sequence.

We thoroughly examined a broad set of DD sequences: $\mathcal{D}=\{$CPMG, RGA$_{2x}$, XY4, UR$_6$, RGA$_{8a}$, UR$_{10}$, RGA$_{16b}$, RGA$_{32c}$, UR$_{32}\}$. These sequences were selected from a comprehensive array of sequences investigated in Ref.~\cite{DD-survey} on $1$-$5$ qubit IBM Quantum devices, which includes detailed descriptions of these sequences and concluded that sequence performance is highly device- and metric-dependent. Based on these findings, our sequence-selection process involved iterative experimentation, with smaller-scale trials conducted repeatedly to discern sequences that, on average, exhibited superior performance compared to others. We also investigated the crosstalk-robust sequences proposed in Ref.~\cite{Zeyuan:22}.

\begin{figure}[t]
    \centering
    \includegraphics[width=0.47\textwidth]{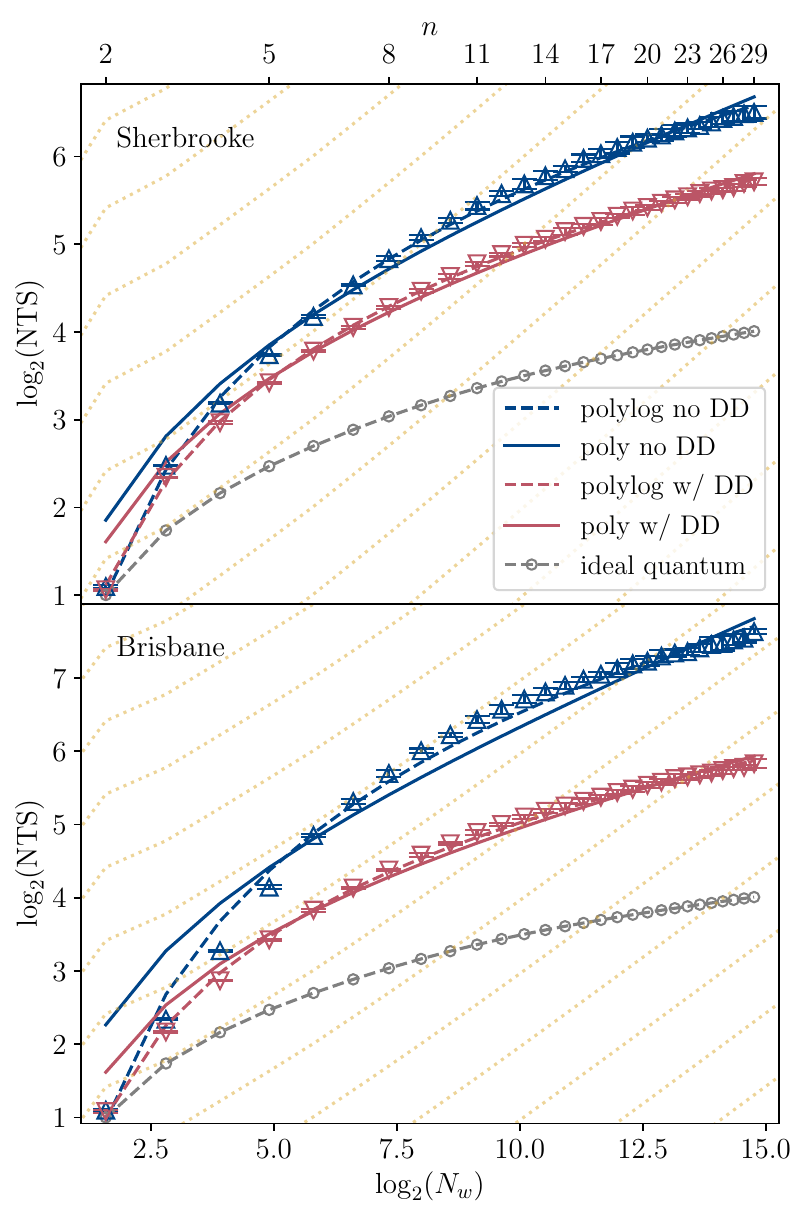}
    \caption{NTS as a function of $\log_2(N_w)$ and problem size $n$ on Sherbrooke (top) and Brisbane (bottom) for \wSimon{4}{29}, both with MEM. The blue lines represent the unprotected, no-DD form of the algorithm. The red lines represent the DD-protected form, i.e., a circuit where DD fills the idle gaps, with the optimal DD sequence from the set $\mathcal{D}$ used at each problem size $n$. The dashed lines denote fitting using the polylog model [\cref{eq:polylog-NTS}]; the solid lines denote fitting using the poly model [\cref{eq:poly-NTS}].
The error bars, representing confidence intervals derived through bootstrapping, extend $1\sigma$ in both directions from each data point. The yellow dotted lines represent the scaling of the lower bound on $\NTS_C$ and serve as a visual guide for the scaling of the poly model. The grey dashed line is the theoretical performance of the quantum algorithm running on a noiseless device, given by \cref{eq:NTSIQ-interp}.
}
    \label{fig:main-plot}
\end{figure}

Certain sequences, such as UR$_m$~\cite{Genov:2017aa} and RGA$_m$~\cite{Quiroz:2013fv}, are families depending on a parameter $m$, which (roughly) counts the number of pulses and the corresponding error suppression order; we optimized over $m$ and chose the best-performing sequences for the final experiment. The optimization procedure we used relies on the NTS metric. Subsequently, circuits incorporating idle periods, each filled with a specific sequence, were executed independently, and the corresponding NTS was computed. The DD sequence exhibiting the lowest NTS was chosen for each problem size $n$. Thus, our final results, discussed below, compare the unprotected `no-DD' circuits with DD-protected circuits where the NTS was minimized for each $n$ separately over all the sequences in $\mathcal{D}$. 

The results for all the sequences in $\mathcal{D}$, in support of our DD sequence ranking methodology, are presented in \cref{app:dd-rank}.

\section{Results and Discussion} 
\label{sec:results}

Our first main result is shown in \cref{fig:main-plot}, which plots the NTS as a function of problem size $n$ for the \wSimon{4}{29} problem run on both the $127$-qubit Sherbrooke and Brisbane devices. Additional results using the older $27$-qubit Cairo and Kolkata devices are provided in \cref{app:extra}. We used bootstrapping to compute all means and standard deviations, as documented in \cref{app:bootstrap}. As explained in \cref{sec:speedup}, we extracted all $n$ values from the largest-$n$ circuits we implemented, of $n=63$ ($126$ qubits). However, we found that for $n>29$, the results were no better than a random guess. We refer to $n\in [2,29]$ as the \emph{quantum range} (of problem sizes), since outside this range, the results are entirely classical. Thus, the results we present below were all extracted from the circuits of width $126$ qubits we implemented, after a partial trace that left us with circuits of $\le 58$ qubits, corresponding to the quantum range.

An immediate observation from \cref{fig:main-plot} is that \emph{the NTS scaling with DD is significantly better than without DD}. The measurement results were further post-processed using MEM before attempting to solve for $b$; see \cref{app:MEM} for details. This results in a slight additional improvement beyond DD. 

For a given \wSimon{w}{n} problem, we aim to determine whether a quantum speedup is observed and, if so, whether it is exponential, subexponential, or polynomial. First, we fit the $\NTS_Q$ data to determine whether it follows the polylog or poly model. Recall that $\NTS_C$ follows the poly model. Thus, if $\NTS_Q$ follows the polylog model, it exhibits an exponential speedup compared to $\NTS_C$. If $\NTS_Q$ follows the mixed model, it exhibits a subexponential/superpolynomial speedup compared to $\NTS_C$. 
Otherwise, we compare the scaling exponents of the poly models; if $\beta_Q < \beta_C$, the speedup is polynomial. 

It is visually clear from \cref{fig:main-plot} that, while far from exhibiting the ideal quantum scaling, in both the Sherbrooke and Brisbane cases the polylog model provides a better fit than the poly model for the \wSimon{4}{29} problem (we do not show the mixed model to avoid overcrowding the plot). This suggests an exponential quantum speedup over the quantum range of problem sizes for $w=4$. We devote the rest of our analysis to more rigorously establishing the validity of this statement, and to determining the range of $w$ values for which it holds in each case, i.e., Sherbrooke and Brisbane, with or without DD.

To determine whether the polylog, mixed, or poly model is a better fit for $\NTS_Q$, we use the coefficient of determination $R^2$ and the Akaike Information Criterion (AIC)~\cite{Akaike:1974aa} as statistical measures. These are independent measures of the quality of model fit. Namely, $0\le R^2\le 1$ measures the proportion of variance explained by the model and is based on the ratio of the residual sum of squares and the total sum of squares; when $R^2=1$, all the variance in the dependent variable can be explained by the independent variables. The AIC, on the other hand, is based on the maximum likelihood estimate and penalizes models with more parameters (since we only consider two-parameter models, the latter aspect does not affect our analysis). The numerical AIC value has no intrinsic significance as such, but lower values indicate a better model. 

In \cref{app:stat-analysis}, we present a more comprehensive statistical analysis that accounts for additional statistical measures beyond $R^2$ and the AIC; we reach identical conclusions regarding the suitability of the polylog and poly models. 

\begin{figure*}
    \centering
    \includegraphics[width=\columnwidth]{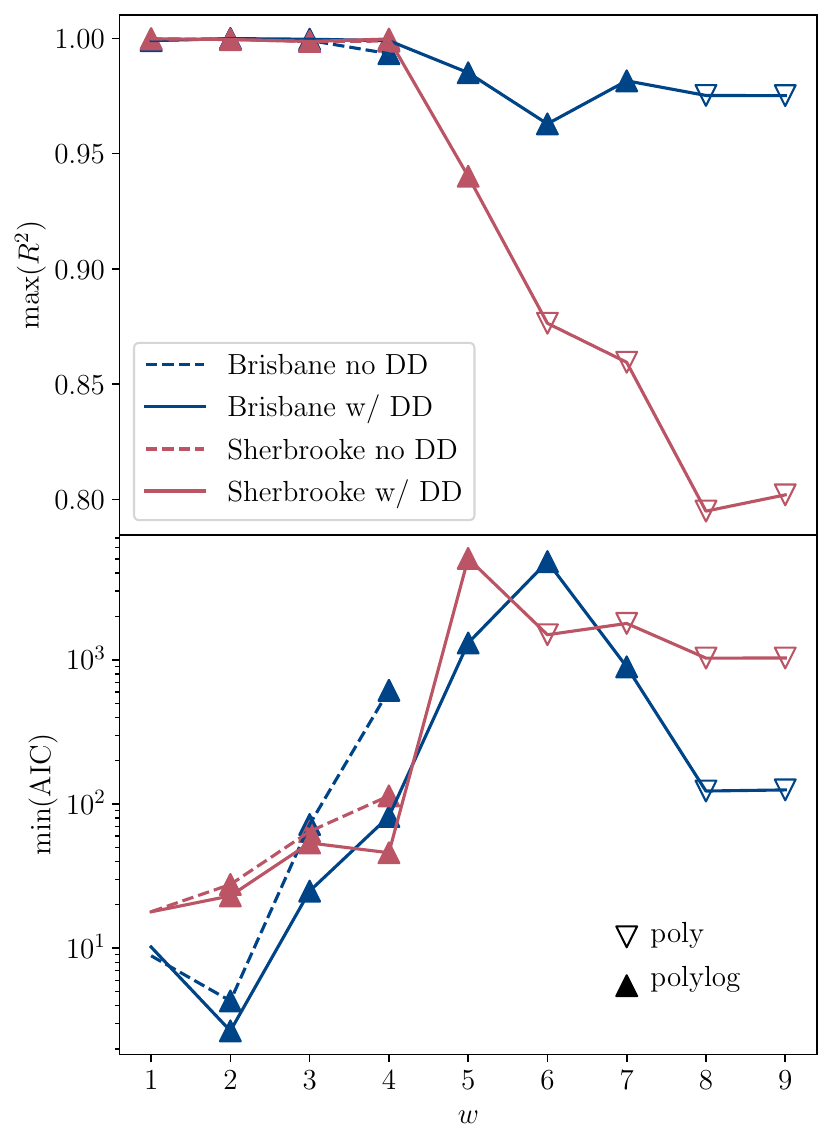}
    \includegraphics[width=\columnwidth]{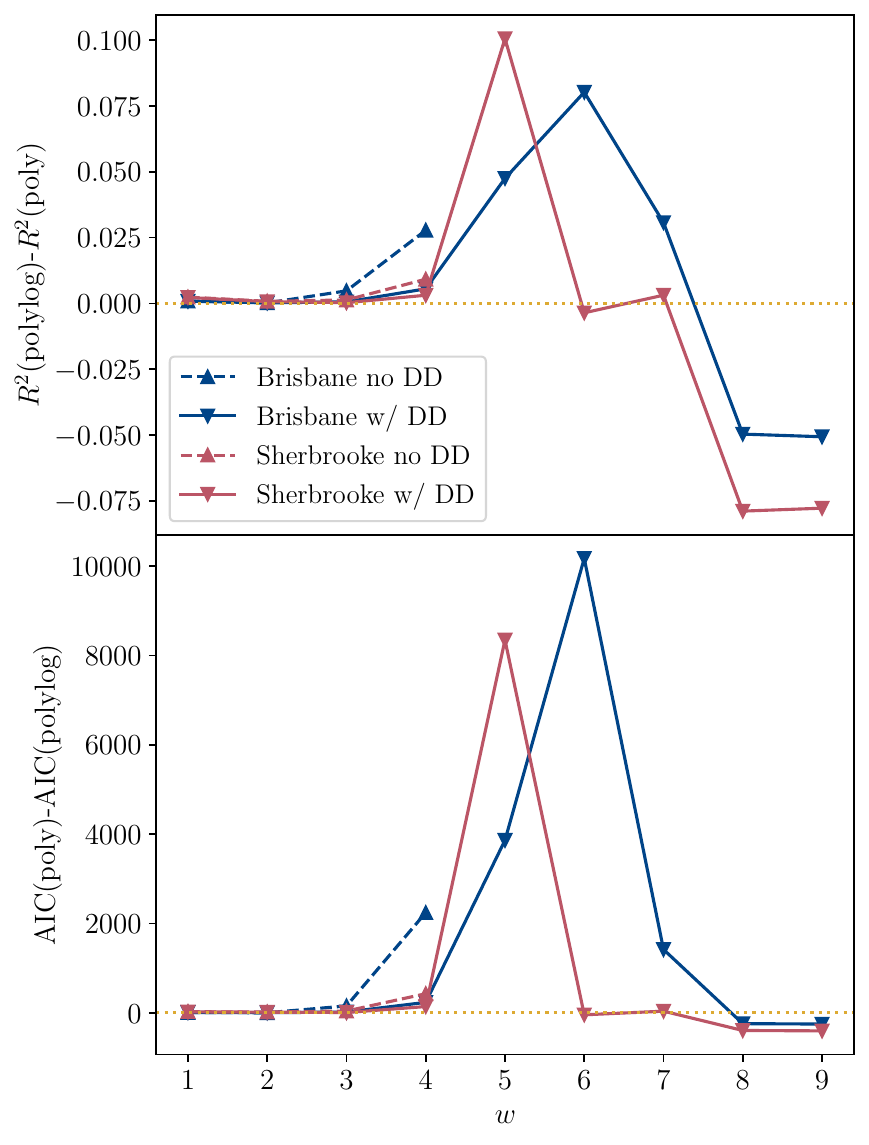}
    \caption{Test of whether the polylog or poly model is the better fit for \wSimon{w}{n} as a function of HW $w$, for both Sherbrooke and Brisbane, either with or without DD, and both with MEM. Top left: at each $w$ value, the symbol shows the model with the highest $R^2$ value. Bottom left: at each $w$ value, the symbol shows the model with the lowest AIC value. The $R^2$ and AIC criteria are in complete agreement. Top right: difference of the $R^2$ values of the polylog and poly models. Bottom right: difference of the AIC values of the poly and polylog models. 
The transition HW value, $w_t$, is the value of $w$ at which the poly model becomes a better fit than the polylog model; as seen from both the left and right panels, $w_t$ is identical for the $R^2$ and the AIC measures.}
    \label{fig:mixed_R^2+AIC}
\end{figure*}

We fit all three models as a function of $N_w$ (or $n$) for each \wSimon{w}{n} problem, for $1\le w \le 9$; the full set of fits is provided in \cref{app:QCTRL,app:all_results}. For each such fit, we compute the $R^2$ and AIC. For a given value of $w$, the model with the higher $R^2$ and lower AIC is the better fit. 
\cref{fig:mixed_R^2+AIC} shows the results for Brisbane and Sherbrooke, both with and without DD. For each HW (value of $w$), we plot the model with the highest $R^2$ (top left) and the lowest AIC (bottom left). Up triangles correspond to the polylog model and down triangles to the poly model. We do not show the results for the mixed model since, as discussed in detail in \cref{app:all_results}, we find that with the exception of two datapoints, it always reduces to either the polylog or the poly model. This means that \emph{we have no evidence of a subexponential/superpolynomial quantum speedup}. Comparing the $R^2$ and AIC criteria in \cref{fig:mixed_R^2+AIC}, we observe that they are in complete agreement on which model (polylog or poly) is the best fit in each case. 

We observe that for Brisbane with DD (blue, solid), the polylog model is the best fit for $w\in [1,7]$, corresponding to an exponential speedup over this range. Then, for $w\in [8,9]$, the poly model is the best fit; we address whether this corresponds to a polynomial speedup below. Similarly, for Sherbrooke with DD (red, solid), the polylog model is the best fit for $w\in [1,7]$ except at $w=6$, corresponding to an exponential speedup over this slightly smaller range. The poly model is the best fit for $w= \{6,8,9\}$. The more comprehensive statistical analysis we report in \cref{app:stat-analysis}, involving additional statistical measures, supports the same conclusions, including the poly model being a better fit for Sherbrooke with DD at $w=6$. Finally, we observe that for $w\in [1,4]$ the polylog model is the best fit for both Brisbane and Sherbrooke without DD, corresponding to an exponential speedup. We do not display results for $w>4$ due to the failure of Brisbane and Sherbrooke without DD to solve the \wSimon{w>4}{n} problem; consequently, we obtained $\le 3$ data points in the $\NTS_Q(n)$ plots, which is insufficient data for reliable model fitting. 

To provide a quantitive sense of the fit quality difference between different models, the right panels of \cref{fig:mixed_R^2+AIC} show the difference in model values between the polylog and poly models. The top-right panel is the difference in $R^2$, and the bottom-right panel is the difference in AIC; positive values indicate that the polylog model is better. We define $w_t$ as the HW value where the transition from the polylog model being better ($w<w_t$) to the poly model being better ($w\ge w_t$) occurs. We can then interpret $w_t-1$ as the largest HW such that an exponential quantum speedup is observed. For $w\ge w_t$ (except for Sherbrooke with DD at $w=7$), if a quantum speedup exists, it is polynomial.

\begin{figure}[t]
    \centering
    \includegraphics[width=0.49\textwidth]{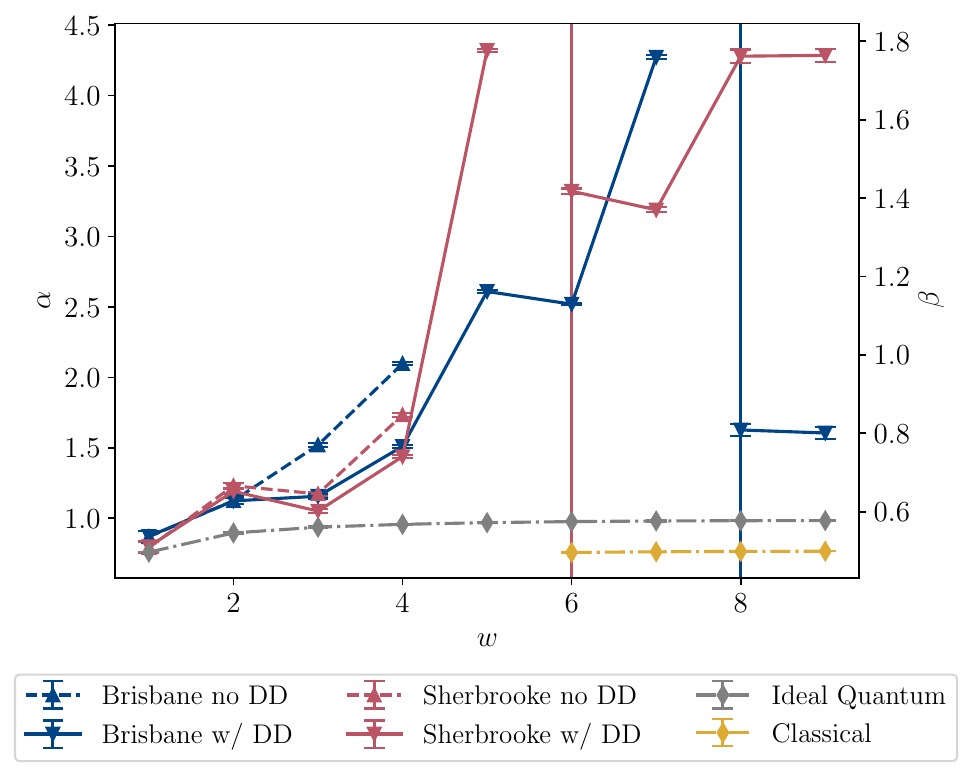}
    \caption{The fitted scaling parameters, $\alpha$ (left axis) and $\beta$ (right axis) of the polylog and poly models, respectively, as a function of $w$ for \wSimon{w}{n} on Sherbrooke and Brisbane, both with MEM. The vertical lines indicate the HW where the transition between the two models occurs, denoted $w_t$ in the text. The polylog (poly) model provides a better fit to the left (right) of the vertical line corresponding to Brisbane with DD (blue, solid) and Sherbrooke with DD (red, solid). For both Brisbane and Sherbrooke without DD the polylog model is always better, but only for $w\in[1,4]$, beyond which we do not have enough data (see text). When $\NTS_Q$ follows the polylog model, an exponential quantum speedup holds with a scaling exponent given by $\alpha$ [\cref{eq:polylog-NTS}]. For the poly model [\cref{eq:poly-NTS}], we find that the quantum slope is always above the classical slope indicated by the dashed-dotted yellow line, corresponding to a polynomial quantum slowdown. 
    The grey curve corresponds to the $\alpha$ value of an ideal (noise-free) quantum implementation. The error bars, representing the standard deviation of the fitted parameters on the bootstrapped data, extend $1\sigma$ in each direction from each data point. The Brisbane results are generally better than Sherbrooke's. The exponential speedup ``quality'', as quantified by the value of $\alpha$, generally deteriorates as $w$ increases.}
    \label{fig:allslopes}
\end{figure}

Having established the HW transition point for each model, we proceed to extract the corresponding scaling exponents. The results are shown in \cref{fig:allslopes}. The vertical lines in this figure indicate $w_t$, and the polylog model provides a better fit to the left side of each such line. Thus, for Brisbane with DD, the solid blue curve to the left of the vertical line at $w_t=8$ gives the scaling exponent $\alpha$ associated with the exponential speedup, and similarly for the red solid curve to the left of the vertical line at $w_t=6$ for Sherbrooke with DD. At or to the right of these two vertical lines, we observe the absence of a polynomial speedup for Brisbane or Sherbrooke since in both cases $\beta_Q > \beta_C$. We summarize the speedup results in \cref{tab:speedups}. 

\begin{table}[h]
    \centering
    \begin{tabular}{|c|c|c|c|}
        \hline
        & exponential speedup &  polynomial slowdown \\
        \hline
        Brisbane no DD & \( w\in[1,4] \) & insufficient data \\
        \hline
        Sherbrooke no DD &  \( w\in[1,4] \)  & insufficient data \\
        \hline
        Brisbane w/ DD & \( w\in[1,7] \) & \( w\in[8,9] \) \\
        \hline
        Sherbrooke w/ DD & \( w\in[1,7]\setminus\{6\} \)  & \( w\in\{6,8,9\} \) \\
                \hline
    \end{tabular}
    \caption{Summary of speedup results.}
    \label{tab:speedups}
\end{table}

These results largely agree with the prediction of an exponential speedup we made in \cref{sec:NISQ-speedup-expectation} based on a simple error model. The main deviation is the observation of a polynomial slowdown shown in \cref{tab:speedups}. In \cref{app:stat-analysis2} we report an analysis that drops the first four data points from our fits, as these points might exhibit small-size effects. When we repeat our model fitting under these conditions, we find that the polylog model (i.e., an exponential speedup) is always the better fit.

These speedups involve $127$ physical qubits, or $58$ after partial trace at the largest problem sizes for which the problems are solved in the quantum range. The previously reported algorithmic quantum speedup result for the single-shot Bernstein-Vazirani problem extended to $27$ physical qubits and used the $27$-qubit Montreal processor~\cite{pokharel2022demonstration}. Our present results thus increase both the number of physical qubits for which an algorithmic quantum speedup has been reported and the complexity of the quantum algorithm for which the result holds.

\section{Summary and Conclusions}
\label{sec:conc}

The goal of demonstrating an algorithmic quantum speedup, i.e., a quantum speedup that scales favorably as the problem size grows, is central to establishing the utility of quantum computers.
Simon's problem is an early example of the Abelian hidden subgroup problem and a precursor to Shor's factoring algorithm. It requires exponential time to solve on a classical computer but only linear time on a noiseless quantum computer, assuming we count oracle queries but do not account for the actual resources spent on executing the oracle. Here, we studied a modified version of Simon's problem, which restricts the allowed Hamming weight of the hidden bitstring to $w\le n$. The classical solution of this version scales as $n^{w/2}$. Our goal was to determine whether NISQ devices are capable of providing an algorithmic quantum speedup in solving this version of Simon's problem.

We ran restricted-HW Simon's algorithm experiments on the IBM Quantum platform and demonstrated that two $127$-qubit devices, Sherbrooke and Brisbane, exhibit an \emph{exponential algorithmic quantum speedup}, which extends to larger HW values when we incorporate suitably optimized DD protection. MEM slightly enhances the speedup. Additional optimization and further improvements are certainly possible; see \cref{app:QCTRL}.

These results significantly extend the scope of 
quantum speedups for oracular algorithms. More generally, our work expands the frontier of empirical quantum advantage results and hints that algorithmic quantum speedup results involving practically relevant algorithms may also be within reach. 

An interesting open question is whether NISQ-device quantum speedups could be achieved for algorithms with deeper and more complex circuits, e.g., Shor's algorithm. Our current implementation of Simon's problem requires roughly $400$ two-qubit gates (after compilation) and $60$ qubits, while a novel construction of Shor's algorithm requires on the order of thousands of two-qubit gates and tens of qubits \cite{rines2018high,10262370}. Whether entering the quantum range with such circuit complexity is within the capacities of error suppression and mitigation methods without requiring full-scale quantum error correction, is a key motivating question for the entire field.

\section{Acknowledgements}
We thank Bibek Pokharel for useful discussions. The research of PS, VK, and DAL was supported by the ARO MURI grant W911NF-22-S-0007 and is based in part upon work supported by the National Science Foundation the Quantum Leap Big Idea under Grant No. OMA-1936388. This material is also based upon work supported by the Defense Advanced Research Projects Agency (DARPA) under Contract No. HR001122C0063. ZZ and GQ acknowledge funding from the U.S. Department of Energy (DOE), Office
of Science, Office of Advanced Scientific Computing Research (ASCR), Accelerated Research in Quantum Computing program under Award Number DE-SC0020316. This research was conducted iusing IBM Quantum Systems provided through USC's IBM Quantum Innovation Center. This research also used resources of the Oak Ridge Leadership Computing Facility, which is a DOE Office of Science User Facility supported under Contract DE-AC05-00OR22725. 

\appendix

\section{Compiler}
\label{app:rules:oracle}

We explained in \cref{sec:rules:f} that even for a fixed $b$ the choice of $f$ requires
at least $(n-1)2^{n-1}$ encoding bits, leading to infeasibly deep circuits for NISQ devices. To circumvent this
problem, we introduce the notion of a compiler that performs the following functions:
(i) it hides the implementation details of $f$ from the player;
(ii) it takes a circuit $C$ with $0$ or more boxes labeled ``$\mcO$'' and produces
a circuit $C'$ obtained from $C$ by replacing each $\mcO$ box with
a circuit implementing it for the current $f$;
(iii) further compiles $C'$ to ensure that it is compatible with the
gate set and connectivity of the NISQ device,
to reduce the number of gates and the circuit depth,
and to select the best layout of the qubits on the device
(i.e., avoiding the noisiest qubits and couplings),
yielding a new circuit $C''$ and classical post-processing instructions;
(iv) sends $C''$ to the NISQ device for execution,
obtains the result, performs the post-processing, and returns the final result to the player. Note that in (ii), ``$\mcO$'' is just a box labeling a place to insert the oracle (which is unknown to the player), as opposed to $\mcO_f$, which is the unitary implementing the actual oracle.

\begin{figure}[ht]
  \centering
\subfigure[]{
\begin{tikzpicture}
    \begin{yquant}
      qubit {$\ket{0}^{\otimes n}$} x;
      qubit {$\ket{0}^{\otimes n}$} a;
      setstyle {very thick} x;
      setstyle {very thick} a;
      box {$H^{\otimes n}$} (x);
      box {$\mcO$} (x, a);
      box {$H^{\otimes n}$} (x);
      measure x, a;
      output {$z$} x;
      output {$a$} a;
    \end{yquant}
  \end{tikzpicture}
 }
\subfigure[]{ 
\begin{tikzpicture}
     \begin{yquant}
      qubit {$x$} x;
      qubit {$a$} a;
      setstyle {very thick} x;
      setstyle {very thick} a;
      hspace {9mm} a;
      [after=a]
      qubit {$\ket{0}^{\otimes n}$} b;
      setstyle {very thick} b;
      box {$\oplus f_0(x)$} b | x;
      box {$f_1$} b;
      cnot a | b;
      box {$f_1^{-1}$} b;
      box {$\oplus f_0(x)$} b | x;
      text {$\ket{0}^{\otimes n}$} b;
      discard b;
      output {$x$} x;
      output {$a$} a;
    \end{yquant}
  \end{tikzpicture}
 }
 \subfigure[]{  
  \begin{tikzpicture}
    \begin{yquant}
      qubit {$\ket{0}^{\otimes n}$} x;
      qubit {$\ket{0}^{\otimes n}$} a;
      setstyle {very thick} x;
      setstyle {very thick} a;
      box {$H^{\otimes n}$} (x);
      box {$\oplus f_0(x)$} a | x;
      box {$H^{\otimes n}$} (x);
      output {$z$} x;
      output {$a$} a;
      measure x, a;
    \end{yquant}
  \end{tikzpicture}
}
  \caption{(a) The intended circuit $C$. The thick wires indicate that each wire represents $n>1$ qubits. $H$ denotes a Hadamard gate and $\mcO$ denotes the oracle. (b) Circuit implementing the oracle $\mcOf$ for arbitrary $f$
    satisfying the conditions (i)--(iii). 
    Any such $f$ can be decomposed as $f(x) = f_1(f_0(x))$, where $f_0$ is any 2-to-1 function satisfying condition (ii) with the same $b$, and $f_1$ is a permutation of bitstrings. 
    All gates are classical in the sense that if the input
    $\ket{x, a}$ consists of a single computational basis state
    (with amplitude $1$), so does the output. The circuit for $f_1$
    may need 3-qubit gates and ancillary qubits. The unitary controlled-($\oplus f_0(x)$) maps a computational basis element $\ket{x}\ket{a}$ to $\ket{x}\ket{a \oplus f_0(x)}$. The gate sequence shown is one possible implementation of $\mcOf$: the circuit prepares $f(x)$ on a third $n$-qubit register, performs a CNOT with the $a$ register as the target, and uncomputes the third register to return it to the $\ket{0}$ state. (c) Compiled circuit. This circuit is equivalent to the unoptimized circuit shown in (a) and (b), assuming $f_1$ is applied to $a$ in post-processing (such optimization is allowed by the rules).}
  \label{fig:circuitC}
\end{figure}
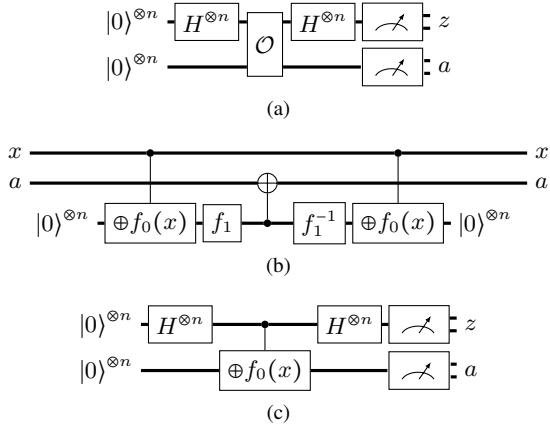

When defining these rules, we need to be careful about the optimizations allowed
in step (iii): if we allow too few, the circuit $C''$ may be too large to be executed
with reasonable fidelity; if we allow too many, $C''$ may become a no-op and
all computations will be performed by classical post-processing.

To understand the rules we choose, consider the intended circuit $C$ and the circuit $\mcOf$ implementing the oracle, shown in \cref{fig:circuitC}(a) and
\cref{fig:circuitC}(b), respectively. To see why \cref{fig:circuitC}(b) implements the oracle, note first that any $f$ satisfying conditions (i)-(iii) can be written as $f(x) = f_1(f_0(x))$. To verify the circuit, it suffices to ensure that it works as intended on computational basis states, i.e., it maps $\ket{x}\ket{a}$ to $\ket{x}\ket{a\oplus f_1(f_0(x))}$. As is easily verified, the list of states along the computation is $\ket{x}\ket{a}\mapsto \ket{x}\ket{a}\ket{0}^{\otimes n} \mapsto \ket{x}\ket{a}\ket{f_0(x)} \mapsto \ket{x}\ket{a}\ket{f_1(f_0(x))} \mapsto \ket{x}\ket{a\oplus f_1(f_0(x))}\ket{f_1(f_0(x))} \mapsto \ket{x}\ket{a\oplus f_1(f_0(x))}\ket{f_0(x)} \mapsto \ket{x}\ket{a\oplus f_1(f_0(x))}\ket{0}^{\otimes n} \mapsto \ket{x}\ket{a\oplus f_1(f_0(x))}$, as required.

Before formulating the rules below, we introduce the notion of classical gates.
We say that a gate is classical if it maps computational basis states
(with amplitude $1$) to computational basis states. In particular,
$X$, CNOT, and the Toffoli (CCNOT) are classical, but $Y$, $Z$, and $H$ are not.
Performing a classical gate immediately before the measurement in the computational
basis is equivalent to performing the measurement first and then applying
the equivalent classical post-processing.

Based on this, we state the following rules for the compiler: (1) it can perform a joint optimization of $C'$, that is, it does not have to preserve the separation between the gates inside and outside $\mcOf$;
  (2) it can perform any optimization that does not change the output distribution of $C'$ if it were executed on an ideal quantum computer;
  (3) it can replace classical gates before the measurement with the equivalent classical post-processing;
  (4) rules \#2 and \#3 can be applied any number of times in any order.

With these rules, we expect the compiler to be able to compile the circuit in \cref{fig:circuitC}(a) with $\mcOf$ from \cref{fig:circuitC}(b) into the circuit $C''$ as shown in \cref{fig:circuitC}(c). Moreover, we expect that the qubits of the $x$ and $a$ registers will be reordered such that circuits with the same $\HW(b)$ will be compiled into the same circuit. The application of $f_1$ and the inverse of the qubit reordering will be done by classical post-processing. Note that \cref{fig:simon-circ} is an expanded version of \cref{fig:circuitC}.

\section{Simon's oracle construction}
\label{app:simon-oracle}

Before presenting it, we remark that our oracle construction is not unique; it only needs to satisfy the condition specified by Simon's problem, i.e., $\forall x,y \in \{0,1\}^n$, $f_b(x) = f_b(y)$ if and only if $x = y$ or $x = y \oplus b$ for a hidden bitstring $b\in \{0,1\}^n$. 

To aid in visualization, we will interchangeably use quantum circuits and directed graphs to represent the oracle. This is useful because the oracle circuit consists entirely of CNOT gates, which can be represented as arrows from the controlled qubit to the target qubit. \cref{fig:circ-graph} shows an example of how to transform between the two representations. The data ($d_j$) and ancilla ($a_j$) qubits are drawn in the first and second columns in the graph representation, respectively.

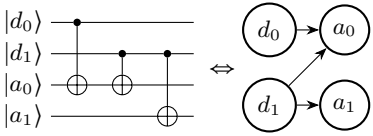
\begin{figure}[h!]
    \centering
    \begin{tikzpicture}
        \begin{yquantgroup}
            \registers{
                qubit {$\ket{d_0}$} x0;
                qubit {$\ket{d_1}$} x1;
                qubit {$\ket{a_0}$} a0;
                qubit {$\ket{a_1}$} a1;
            }
            \circuit{
                cnot a0 | x0;
                cnot a0 | x1;
                cnot a1 | x1;
            }
            \equals[$\Leftrightarrow$]
        \end{yquantgroup}
        \begin{scope}[every node/.style={circle,thick,draw}]
            \node (x0) at (3.6,-0.25) {$d_0$};
            \node (x1) at (3.6,-1.25) {$d_1$};
            \node (a0) at (4.6,-0.25) {$a_0$};
            \node (a1) at (4.6,-1.25) {$a_1$};
        \end{scope}
        \begin{scope}[>={Stealth[black]}]
            \draw[->] (x0) -- (a0);
            \draw[->] (x1) -- (a0);
            \draw[->] (x1) -- (a1);
        \end{scope}
    \end{tikzpicture}
    \caption{Conversion between a pure-CNOT circuit and its graph representation. A CNOT gate is an arrow from the controlled to the target qubit.}
    \label{fig:circ-graph}
\end{figure}

Two operations are used to construct the oracle: classical copy and classical addition modulo 2. Both are classical operations executed by quantum devices. 

\begin{itemize}
\item Classical copy:
This operation copies a classical state from a 1-qubit register into another 1-qubit register. \cref{fig:classical-cp} shows an example of copying $\ket{d_0 d_1}$ into the ancilla qubits $\ket{a_0a_1}$, which are initialized as $\ket{00}$.
\begin{figure}[h!]
    \centering
    \begin{tikzpicture}
        \begin{yquantgroup}
            \registers{
                qubit {$\ket{d_0}$} x0;
                qubit {$\ket{d_1}$} x1;
                qubit {$\ket{0}$} a0;
                qubit {$\ket{0}$} a1;
            }
            \circuit{
                cnot a0 | x0;
                cnot a1 | x1;
            }
            \equals
        \end{yquantgroup}
        \begin{scope}[every node/.style={circle,thick,draw}]
            \node (x0) at (3,-0.25) {$d_0$};
            \node (x1) at (3,-1.25) {$d_1$};
            \node (a0) at (4,-0.25) {$a_0$};
            \node (a1) at (4,-1.25) {$a_1$};
        \end{scope}
        \begin{scope}[>={Stealth[black]}]
            \draw[->] (x0) -- (a0);
            \draw[->] (x1) -- (a1);
        \end{scope}
    \end{tikzpicture}
    \caption{A quantum circuit for the operation $\ket{d_0d_1}\ket{00}\mapsto\ket{d_0d_1}\ket{d_0d_1}$. $d_0, d_1,a_0,a_1\in\{0,1\}$.}
    \label{fig:classical-cp}
\end{figure}
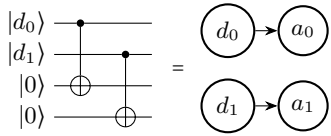

\item Classical XOR:
This operation performs the XOR (single-digit addition modulo 2) between two 1-qubit registers and stores the result in another 1-qubit register. \cref{fig:classical-add} shows an example of adding $\ket{d_0 d_1}$ and storing the result in the first ancilla qubit $\ket{a_0}$, which is initialized as $\ket{0}$.
\begin{figure}[h!]
    \centering
    \begin{tikzpicture}
        \begin{yquantgroup}
            \registers{
                qubit {$\ket{d_0}$} x0;
                qubit {$\ket{d_1}$} x1;
                qubit {$\ket{0}$} a0;
                qubit {$\ket{0}$} a1;
            }
            \circuit{
                cnot a0 | x0;
                cnot a0 | x1;
            }
            \equals
        \end{yquantgroup}
        \begin{scope}[every node/.style={circle,thick,draw}]
            \node (x0) at (3,-0.25) {$d_0$};
            \node (x1) at (3,-1.25) {$d_1$};
            \node (a0) at (4,-0.25) {$a_0$};
            \node (a1) at (4,-1.25) {$a_1$};
        \end{scope}
        \begin{scope}[>={Stealth[black]}]
            \draw[->] (x0) -- (a0);
            \draw[->] (x1) -- (a0);
        \end{scope}
    \end{tikzpicture}
    \caption{A quantum circuit that applies the operation $\ket{d_0d_1}\ket{00}\mapsto\ket{d_0d_1}\ket{d_0\oplus d_1,0}$. $d_0, d_1,a_0,a_1\in\{0,1\}$.}
    \label{fig:classical-add}
\end{figure}
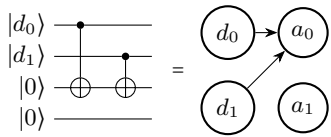

\end{itemize}

Next, we show how to construct a 2-to-1 oracle. 
We represent a length-$n$ bitstring $x=x_0\cdots x_{n-1}$, where $x_j\in\{0,1\}$, and use an oracle that transforms $\ket{x}\ket{0}\rightarrow \ket{x}\ket{f_b(x)}$ when $b=0^{n-i}1^{i}$. For each $b$ of this form, 
$f_b:\{0,1\}^n\mapsto \{0,1\}^n$ is defined as follows:
\begin{equation}
    \begin{aligned}
        f_b(x) = (x_0,\cdots,x_{n-i-1},0,x_{n-i+1}\oplus x_{n-i},\cdots,x_{n-1}\oplus x_{n-i})
            \end{aligned}
    \label{eq:fbx}
\end{equation}
for $1\le i\le n$. 

Equivalently, we can write this transformation as
\begin{equation}
    \begin{aligned}
        &\ket{x_0\cdots x_{n-1}}_d\ket{0^n}_a \mapsto \\
        &\ket{x_0\cdots x_{n-1}}_d\ket{x_0\cdots x_{n-i-1}0 (x_{n-i}\oplus x_{n-i+1})\cdots (x_{n-i}\oplus x_{n-1})}_a.
    \end{aligned}
    \label{eq:fbx-state}
\end{equation}

The following Lemmas show that $f_b(x)$ is a valid family of oracles for Simon's problem. 

\begin{mylemma}
\label{lem:2to1}
$f_b(x)=f_b(y)$ if and only if $x= y$ or $y =  x \oplus b$, where $b=0^{n-i}1^{i}$, $1\le i\le n$. I.e., $f_b(x)$ is a 2-to-1 function.
\end{mylemma}

\begin{proof}
Consider two ancilla registers containing the output from two different inputs $x$ and $y$. According to \cref{eq:fbx-state}, we can write $f_b(x)$ and $f_b(y)$ as
\begin{equation}
    \begin{aligned}
        \ket{f_b(x)}&=\ket{x_0\cdots x_{n-i-1}0 (x_{n-i}\oplus x_{n-i+1})\cdots (x_{n-i}\oplus x_{n-1})}\\
        \ket{f_b(y)}&=\ket{y_0\cdots y_{n-i-1}0 (y_{n-i}\oplus y_{n-i+1})\cdots (y_{n-i}\oplus y_{n-1})}.
    \end{aligned}
    \label{eq:fx=fy}
\end{equation}

($\Rightarrow$) If $x= y$ or $y = x\oplus b$, then $f_b(x)=f_b(y)$: 
The statement is trivial when $x= y$.  We explicitly write $b = 0_0\cdots0_{n-i-1} 1_{n-i}\cdots1_{n-1}$, so that $\{y_0\cdots y_{n-i-1}\}=\{x_0\cdots x_{n-i-1}\}$ and $\{y_{n-i+1}\cdots y_{n-1}\}=\{x_{n-i+1}\oplus 1\cdots x_{n-1}\oplus 1\}$.
It follows from \cref{eq:fbx} that trivially $f_b(y_j) = f_b(x_j)$ for $j\le n-i$ and that 
for $j\ge n-i+1$, $f_b(y_j)=y_{n-i}\oplus y_j=(x_{n-i}\oplus 1)\oplus (x_j \oplus 1)=x_{n-i}\oplus x_j=f_b(x_j)$.

($\Leftarrow$) If $f_b(x)=f_b(y)$, then $x=y$ or $y=x\oplus b$:
The calculation in the `$\Rightarrow$' part of the proof shows that $f_b(y) = f_b(x)$ if $y =x\oplus b$; therefore,  $f_b^{-1}(x)$ outputs $x$ and $x\oplus b$ and likewise, $f_b^{-1}(y)$ outputs $y$ and $y\oplus b$. Hence,
$f_b(x)=f_b(y)$ implies that either (1) $x=y$ and $x\oplus b=y\oplus b$, or (2) $x=y\oplus b$ and $x\oplus b=y$. Case (1) corresponds to $x=y$ and case (2) corresponds to $y=x\oplus b$.

The conclusion that $f_b(x)$ is a 2-to-1 function now follows directly.
\end{proof}

Having shown that $f_b$ as defined in \cref{eq:fbx} and written explicitly in \cref{eq:fbx-state} is a valid oracle for Simon's problem, we proceed to construct its quantum version using classical copy (for the first $n-i$ positions) and classical XOR (for the last $i$ positions). The $a_{n-i}$ qubit remains in the $\ket{0}$ state, hence no operation is required.

To construct the first $n-i$ positions in \cref{eq:fbx-state}, one needs to copy the state in $d_0,\dots,d_{n-i-1}$ into $a_0,\dots,a_{n-i-1}$. The quantum operation for this is a CNOT from each qubit $d_j$ of the data register (controlled qubit) to the corresponding qubit $a_j$ in the ancilla register (target qubit), as shown in \cref{fig:classical-cp}.

To construct the last $i$ positions in \cref{eq:fbx-state}, one needs to perform XOR on $d_{n-i}$ and $d_j$ from the data register and store the result in the qubit $a_j$ in the ancilla register. Two CNOTs are required from two controlled qubits, $d_{n-i}$ and $d_j$, onto the same target qubit $a_j$, similarly to what is shown in \cref{fig:classical-add}.

We can construct an oracle for a Simon-$n$ problem using this method. Where to apply classical copy or classical XOR depends on the hidden bitstring $b=0^{n-i}1^{i}$: a $0$ is an instruction to copy, while a $1$ is an instruction to perform an XOR. \cref{fig:combine} shows an example of combining these two types of operators to construct an oracle for Simon-3 with $b=011$. \cref{fig:simon-n5-graph} shows an example of Simon-5 in the graph representation. 
An example of the full quantum circuit to solve a Simon-3 problem with $b=011$ is given in \cref{fig:simon-011}.

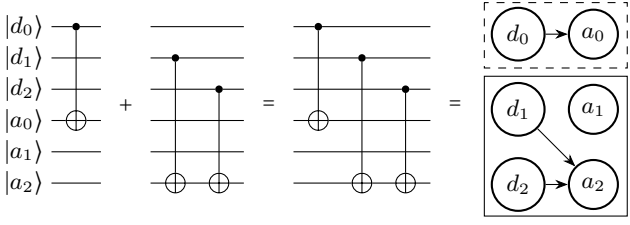
\begin{figure}[t]
    \centering
    \begin{tikzpicture}
       \begin{yquantgroup}[operator/separation=0.8mm]
            \registers{
                qubit {} x0;
                qubit {} x1;
                qubit {} x2;
                qubit {} a0;
                qubit {} a1;
                qubit {} a2;
            }
            \circuit{
                init {$\ket{d_0}$} x0;
                init {$\ket{d_1}$} x1;
                init {$\ket{d_2}$} x2;
                init {$\ket{a_0}$} a0;
                init {$\ket{a_1}$} a1;
                init {$\ket{a_2}$} a2;
                cnot a0 | x0;
            }
            \equals[$+$]
            \circuit{
                cnot a2 | x1;
                cnot a2 | x2;
            }
            \equals
            \circuit{
                cnot a0 | x0;
                cnot a2 | x1;
                cnot a2 | x2;
            }
            \equals
        \end{yquantgroup}
        \begin{scope}[every node/.style={circle,thick,draw}]
            \node (x0) at (6.9,-0.25) {$d_0$};
            \node (x1) at (6.9,-1.25) {$d_1$};
            \node (x2) at (6.9,-2.25) {$d_2$};
            \node (a0) at (7.9,-0.25) {$a_0$};
            \node (a1) at (7.9,-1.25) {$a_1$};
            \node (a2) at (7.9,-2.25) {$a_2$};
        \end{scope}
        \begin{scope}[>={Stealth[black]}]
            \draw[->] (x0) -- (a0);
            \draw[->] (x1) -- (a2);
            \draw[->] (x2) -- (a2);
        \end{scope}
        \draw[dashed] (6.45,0.17) rectangle (8.35,-0.70);
        \draw[-] (6.45,-0.81) rectangle (8.35,-2.67);
    \end{tikzpicture}
    \caption{Construction of an oracle for solving the Simon-3 problem with $b=011$ according to \cref{eq:fbx-state}. The first (second) part in the dashed (solid) box applies classical copy (XOR), corresponding to the $0$'s ($1$'s) in $b$.}
    \label{fig:combine}
\end{figure}

\begin{figure*}[t]
    \begin{tikzpicture}
        \begin{scope}[every node/.style={circle,thick,draw}]
            \node (4x0) at (0,1.7) {$d_0$};
            \node (4x1) at (0,0.7) {$d_1$};
            \node (4x2) at (0,-0.3) {$d_2$};
            \node (4x3) at (0,-1.3) {$d_3$};
            \node (4x4) at (0,-2.3) {$d_4$};
            \node (4a0) at (1,1.7) {$a_0$};
            \node (4a1) at (1,0.7) {$a_1$} ;
            \node (4a2) at (1,-0.3) {$a_2$};
            \node (4a3) at (1,-1.3) {$a_3$};
            \node (4a4) at (1,-2.3) {$a_4$};

            \node (3x0) at (3,1.7) {$d_0$};
            \node (3x1) at (3,0.7) {$d_1$};
            \node (3x2) at (3,-0.3) {$d_2$};
            \node (3x3) at (3,-1.3) {$d_3$};
            \node (3x4) at (3,-2.3) {$d_4$};
            \node (3a0) at (4,1.7) {$a_0$};
            \node (3a1) at (4,0.7) {$a_1$} ;
            \node (3a2) at (4,-0.3) {$a_2$};
            \node (3a3) at (4,-1.3) {$a_3$};
            \node (3a4) at (4,-2.3) {$a_4$};

            \node (2x0) at (6,1.7) {$d_0$};
            \node (2x1) at (6,0.7) {$d_1$};
            \node (2x2) at (6,-0.3) {$d_2$};
            \node (2x3) at (6,-1.3) {$d_3$};
            \node (2x4) at (6,-2.3) {$d_4$};
            \node (2a0) at (7,1.7) {$a_0$};
            \node (2a1) at (7,0.7) {$a_1$} ;
            \node (2a2) at (7,-0.3) {$a_2$};
            \node (2a3) at (7,-1.3) {$a_3$};
            \node (2a4) at (7,-2.3) {$a_4$};

            \node (1x0) at (9,1.7) {$d_0$};
            \node (1x1) at (9,0.7) {$d_1$};
            \node (1x2) at (9,-0.3) {$d_2$};
            \node (1x3) at (9,-1.3) {$d_3$};
            \node (1x4) at (9,-2.3) {$d_4$};
            \node (1a0) at (10,1.7) {$a_0$};
            \node (1a1) at (10,0.7) {$a_1$} ;
            \node (1a2) at (10,-0.3) {$a_2$};
            \node (1a3) at (10,-1.3) {$a_3$};
            \node (1a4) at (10,-2.3) {$a_4$};

            \node (x0) at (12,1.7) {$d_0$};
            \node (x1) at (12,0.7) {$d_1$};
            \node (x2) at (12,-0.3) {$d_2$};
            \node (x3) at (12,-1.3) {$d_3$};
            \node (x4) at (12,-2.3) {$d_4$};
            \node (a0) at (13,1.7) {$a_0$};
            \node (a1) at (13,0.7) {$a_1$} ;
            \node (a2) at (13,-0.3) {$a_2$};
            \node (a3) at (13,-1.3) {$a_3$};
            \node (a4) at (13,-2.3) {$a_4$};
        \end{scope}
        \node (eq) at (0.5,-3) {$b=0^41^1$} ;
        \node (eq) at (3.5,-3) {$b=0^31^2$} ;
        \node (eq) at (6.5,-3) {$b=0^21^3$} ;
        \node (eq) at (9.5,-3) {$b=0^11^4$} ;
        \node (eq) at (12.5,-3) {$b=0^01^5$} ;
        \begin{scope}[>={Stealth[black]}]
            \draw[->] (x0) -- (a1);
            \draw[->] (x0) -- (a2);
            \draw[->] (x0) -- (a3);
            \draw[->] (x0) -- (a4);
            \draw[->] (x1) -- (a1);
            \draw[->] (x2) -- (a2);
            \draw[->] (x3) -- (a3);
            \draw[->] (x4) -- (a4);

            \draw[->] (1x0) -- (1a0);
            \draw[->] (1x1) -- (1a2);
            \draw[->] (1x1) -- (1a3);
            \draw[->] (1x1) -- (1a4);
            \draw[->] (1x2) -- (1a2);
            \draw[->] (1x3) -- (1a3);
            \draw[->] (1x4) -- (1a4);

            \draw[->] (2x0) -- (2a0);
            \draw[->] (2x1) -- (2a1);
            \draw[->] (2x2) -- (2a3);
            \draw[->] (2x2) -- (2a4);
            \draw[->] (2x3) -- (2a3);
            \draw[->] (2x4) -- (2a4);

            \draw[->] (3x0) -- (3a0);
            \draw[->] (3x1) -- (3a1);
            \draw[->] (3x2) -- (3a2);
            \draw[->] (3x3) -- (3a4);
            \draw[->] (3x4) -- (3a4);

            \draw[->] (4x0) -- (4a0);
            \draw[->] (4x1) -- (4a1);
            \draw[->] (4x2) -- (4a2);
            \draw[->] (4x3) -- (4a3);
        \end{scope}

        \draw[dashed] (-0.5,2.15) rectangle (1.5,1.25);
        \draw[dashed] (-0.5,1.15) rectangle (1.5,0.25);
        \draw[dashed] (-0.5,0.15) rectangle (1.5,-0.75);
        \draw[dashed] (-0.5,-0.85) rectangle (1.5,-1.75);
        \draw[-] (-0.5,-1.85) rectangle (1.5,-2.75);

        \draw[dashed] (2.5,2.15) rectangle (4.5,1.25);
        \draw[dashed] (2.5,1.15) rectangle (4.5,0.25);
        \draw[dashed] (2.5,0.15) rectangle (4.5,-0.75);
        \draw[-] (2.5,-0.85) rectangle (4.5,-2.75);

        \draw[dashed] (5.5,2.15) rectangle (7.5,1.25);
        \draw[dashed] (5.5,1.15) rectangle (7.5,0.25);
        \draw[-] (5.5,0.15) rectangle (7.5,-2.75);

        \draw[dashed] (8.5,2.15) rectangle (10.5,1.25);
        \draw[-] (8.5,1.15) rectangle (10.5,-2.75);

        \draw[-] (11.5,2.15) rectangle (13.5,-2.75);
    \end{tikzpicture}
    \caption{Oracle construction for Simon-5 in the graph representation. Using the reduction procedure, we can remove the dashed boxes to arrive at Simon-$m$, where $m<5$.}
    \label{fig:simon-n5-graph}
\end{figure*}
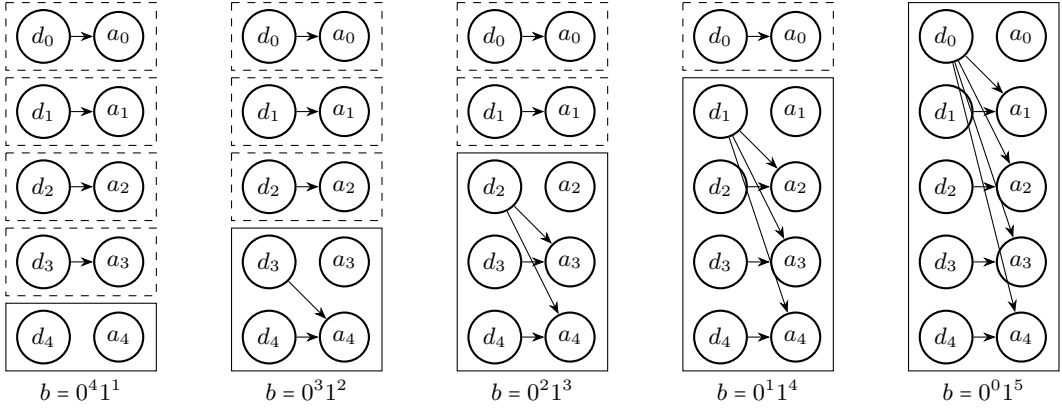

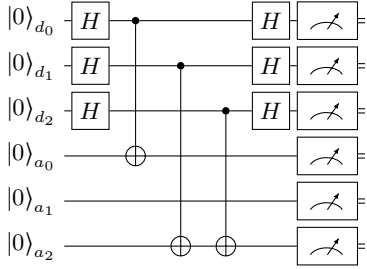
\begin{figure}[h!]
    \centering
    \begin{tikzpicture}
        \begin{yquantgroup}
            \registers{
                qubit {$\ket{0}_{d_0}$} x0;
                qubit {$\ket{0}_{d_1}$} x1;
                qubit {$\ket{0}_{d_2}$} x2;
                qubit {$\ket{0}_{a_0}$} a0;
                qubit {$\ket{0}_{a_1}$} a1;
                qubit {$\ket{0}_{a_2}$} a2;
            }
            \circuit{
                h x0, x1, x2;
                cnot a0 | x0;
                cnot a2 | x1;
                cnot a2 | x2;
                h x0, x1, x2;
                measure x0, x1, x2, a0, a1, a2;
            }
        \end{yquantgroup}
    \end{tikzpicture}
    \caption{The quantum circuit to solve the Simon-3 problem with $b=011$, which can be extracted from the dashed ($d_2,a_2$) box and the lowermost solid box in Simon-5 with $b=0^31^2$.}
    \label{fig:simon-011}
\end{figure}

\section{Two derivations and proofs of the classical complexity of Simon's problem}
\label{app:alternative-deriv}

We first present a simple derivation of the classical complexity of Simon's problem, in the spirit of the birthday paradox. Given any input bitstring pair $x_i\ne x_j$ such that $f_b(x_i)=f_b(x_j)$ (a collision), we have the solution to Simon's problem: $b=x_i\oplus x_j$. In the worst case, excluding the $0^n$ string, it takes $k_{\max} = 2^{n-1}$ strings to find a collision. But suppose we pick $k<k_{\max}$ strings at random. There are ${k\choose 2} = k(k-1)/2$ pairs. Given a random string $x_i$, we can pick a second random string $x_j$, and $x_j\ne 0^n,x_i$ with probability $1/(N_n-1)$, where $N_n=2^n-1$. This is also the probability of a collision between $x_i$ and $x_j$. Therefore
\begin{equation}
    \Pr[f_b(x_i)=f_b(x_j)]=\frac{{k\choose 2}}{N_n-1}, \quad i\neq j.
    \label{eq:prob-classical2}
\end{equation}
Setting $\Pr[f_b(x_i)=f_b(x_j)]= 1$ guarantees a collision.
Solving for $k$, the required number of queries $k$ in \cref{eq:prob-classical2} is $O(2^{n/2})$.

The classical complexity is affected by the change of the total number of possible $b$'s from $N_n$ to $N_w=\sum_{j=1}^{w}{n \choose j}$ in \wSimon{w}{n}. Hence, \cref{eq:prob-classical2} becomes:
\begin{equation}
    \Pr[f_b(x_i)=f_b(x_j)]=\frac{{k\choose 2}}{N_w-1} .
\end{equation}
By setting $\Pr[f_b(x_i)=f_b(x_j)]\ge 1$, we assume that a collision is not found until all bitstring pairs have been tested, which yields the worst-case number of classical queries (for the best possible classical algorithm), and we recover \cref{eq:cl-worst}.\\

\noindent Next, we present a formal proof of \cref{th:Simon-classical}. 
\begin{proof}
Consider first the original Simon's problem and a classical player
querying bitstrings $x_1, x_2, \dots, x_k$
to get values $y_1, \dots, y_k$, where $y_j = f(x_j)\in\{0,1\}$. Having made $m$ queries, the player forms the set $R_m = \{x_j \oplus x_l\colon 1 \leq j < l \leq m\}$.
If $b \in R_m$, then $b = x_j \oplus x_l$ for some $j, l \leq m$
and $y_j = y_l$. In this case, the player knows $b$ and their best strategy is to output it without further queries to the oracle.
Under the assumption that $f$ is a uniformly random function satisfying conditions (i)--(iii) discussed in \cref{s:rules}, if $b\notin R_m$, then
the only information learned after $m$ queries is that $b$ is not in $R_m$. Thus, the player cannot benefit from adjusting the sequence of queried bitstrings based on the learned information and can assume that the classical algorithm is fully described by the sequence $x_1, x_2, \dots, x_k$. Finally, one can check that the $\NTS$ cannot be improved by guessing prematurely. Hence, for an optimal algorithm, the player is guaranteed to know $b$ after all $k$ queries, i.e., $S \setminus R_k$ should contain at most one element.
On the other hand, $\abs{R_k} \leq k (k-1)/2$. I.e., $k$ should satisfy
\begin{equation}
  \label{eq:prob-classical}
  \frac{k(k-1)}{2} \geq N_n - 1.
\end{equation}
The minimal $k$ satisfying this inequality is the lower bound on the worst case complexity of solving the Simon's problem for a classical player: if the player is lucky they would be able to figure out $b$ earlier when they find $y_j = y_l$ for $j, l < k$.

Let $S_k = S \cap R_k$.
Similar to the above, if $b \notin S_k$, then this is the only information about $b$ available to the player, i.e., their posterior distribution of all possible values of $b$ is uniform in $S \setminus S_k$. 
The player is guaranteed to know the correct $b$ iff $S \setminus S_k$ contains at most one element, i.e., $\abs{S_k} \geq N_w - 1$, but $\abs{S_k} \leq \abs{R_k} \leq k (k-1)/2$. Hence,
\begin{equation}
  \label{eq:prob-classical.2}
  \frac{k(k-1)}{2} \geq N_w - 1.
\end{equation}
Solving \cref{eq:prob-classical.2} for $k$, we obtain \cref{eq:cl-worst}. 
\end{proof}

The lower bound \cref{eq:cl-expect} can now be obtained by summing the lower bounds on the probability
$b \notin S_i$. Since there are at most $i(i-1)/2$ elements in $S_i$, and the total number of possible $b$ values is $N_w$, we have $\Pr(b \in S_i)\le \frac{i(i-1)}{2N_w}$. The probability that the next query \#$i+1$ is needed (i.e., that $i$ queries are not enough: $i=0, \dots, k-1$) is $1-\Pr(b \in S_i)$, yielding
\begin{equation}
     \expv{Q_C} \geq 
     \sum_{i=0}^{k-1} \left(1-\frac{i(i-1)}{2N_w}\right) ,
    \label{eq:cl-expect-app}
\end{equation}
and after evaluating the sum we obtain \cref{eq:cl-expect}.

\section{Upper bound on \texorpdfstring{$\NTS_C$}{NTS C}}
\label{app:bounds-gap}

As remarked in the main text, there is no guarantee that the player will find a sequence of bitstrings $x = \{x_1,x_2,\dots,x_k\}$ that achieves the lower bound on $\NTS_C$ given by \cref{eq:cl-expect2}. For \cref{eq:cl-expect} to be an equality, all values of $x_l \oplus x_m$ should be different and satisfy $\HW(x_l \oplus x_m) \leq w$ for $1 \leq l < m \leq k$. To find an upper bound on $\NTS_C$, we consider an arbitrary sequence $x_1,\dots,x_k$ of bitstrings that a classical player intends to submit to the oracle. Without loss of generality, we can assume that this sequence is independent of the oracle replies unless a match $f(x_j) = f(x_l)$ is found for $j < l \leq k$ (in which case the classical player stops and returns $b = x_j \oplus x_l$). We can also assume that
\begin{equation}
\label{eq:k-def}
\abs{S\setminus S_j} \leq 1 \Leftrightarrow j = k.
\end{equation}
If \cref{eq:k-def} is not satisfied for $j=k$, the classical algorithm fails to find $b$ for $b \in S \setminus S_k$, and if \cref{eq:k-def} is satisfied for $j < k$, we can set $k$ to the smallest such $j$. \Cref{eq:k-def} can be seen as an exact version of \cref{eq:cl-worst}: it ensures that $k$ is the worst-case classical number of queries of an algorithm described by the sequence $x$. For $x$ satisfying \cref{eq:k-def}, we can compute the NTS exactly:
\begin{equation}
    \NTS_C(x)=\sum_{i=0}^{k-1}\left(\frac{|S\setminus S_i|}{|S|}\right)=\sum_{i=0}^{k-1}\left(1-\frac{|S_i|}{|S|}\right),
    \label{eq:nts-c-x}
\end{equation}
which is the exact version of \cref{eq:cl-expect}. Then
\begin{equation}
    \label{eq:nts-c-exact}
    \NTS_C = \min_{x:\ \mathrm{\cref{eq:k-def}}} \NTS_C(x).
\end{equation}
We do not know a method to find this minimum exactly for $n \geq 10$ because, as $n$ grows, the search space involves sequences of $k$ bitstrings, where $k$ grows exponentially with $n$, i.e., the search space size grows as $2^{n2^{n/2}}$ for the original Simon's problem and as $2^{n\sqrt{\binom{n}{w}}}$ for \wSimon{w}{n} with constant $w$, which is infeasible already for $n=10$, $w=6$. We circumvent this problem using a heuristic search for larger $n$. If we could show that the sequence $x$ we found using a heuristic search minimizes \cref{eq:nts-c-exact}, then we would be able to claim that we found the exact value of $\NTS_C$. In practice, we cannot show this for all but the few smallest values of $n$, hence can only claim that $\NTS_C \leq \NTS_C(x)$.

The solution we compute for \cref{eq:nts-c-exact} from a heuristic search serves as the upper bound for $\NTS_C$, while \cref{eq:cl-expect2} represents its lower bound. We selected the latter when contrasting the slope parameters $a_Q$ and $a_C$. This ensures that the quantum algorithm's performance surpasses the lower threshold of its classical counterpart, guaranteeing a quantum scaling advantage if achieved.

In larger problem sizes where the upper and lower bounds do not align, the actual $\NTS_C$ lies somewhere between the two. The discrepancy between the theoretical lower bound in \cref{eq:cl-expect2} and the heuristic upper bound in \cref{eq:nts-c-exact} is depicted in \cref{fig:UB-LB} for \wSimon{7}{n} from $n=2$ to $n=14$.

\begin{figure}[t]
    \centering
    \includegraphics[width=0.5\textwidth]{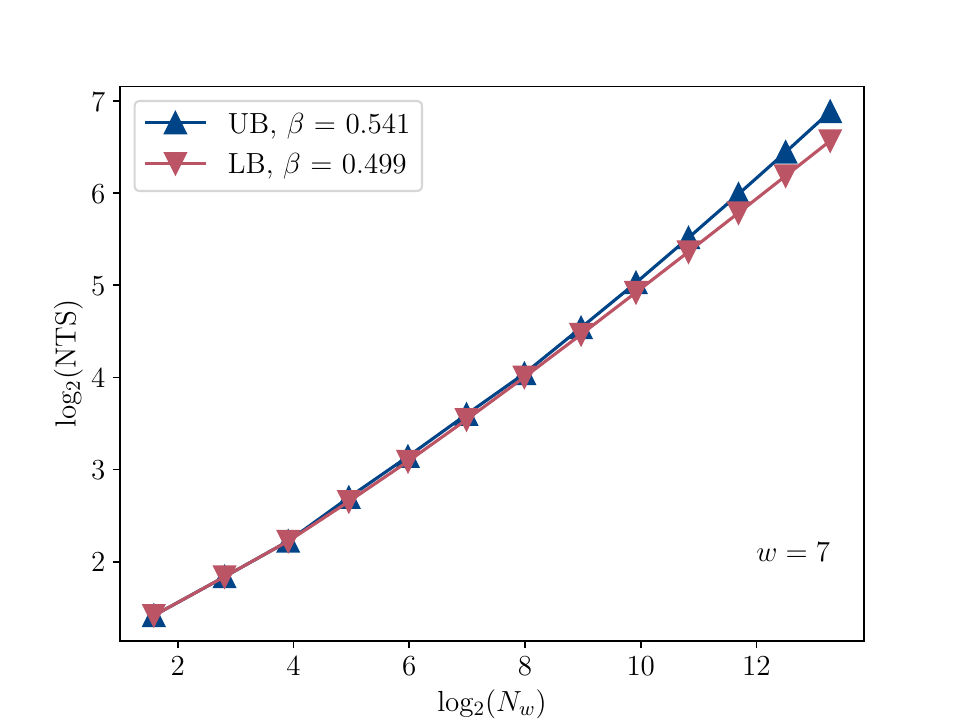}
    \caption{Difference between the upper bound (UB) and lower bound (LB) of $\NTS_C$ and the corresponding scaling parameter $\beta$ for $w=7$. We use the LB slope in the quantum speedup analysis.}
    \label{fig:UB-LB}
\end{figure}

\section{Solving \texorpdfstring{Simon-$n$ in $O(n)$}{Simon-n in O(n)} oracle queries on a noiseless QC}
\label{app:SimonQ}

For completeness, we provide the proof that the circuit in \cref{fig:simon-circ} solves Simon's problem in $O(n)$ oracle queries, assuming that it is run on a noiseless QC. Starting from the left of \cref{fig:simon-circ}, every qubit is initialized in the all-zero state $\ket{\psi_0}=\ket{0}^{\otimes n}\ket{0}^{\otimes n}$, and the first $n$ Hadamard gates put every data qubit into a uniform superposition:
\begin{equation}
    \ket{\psi_1}= H^{\otimes n}\ket{\psi_0}=\frac{1}{\sqrt{2^n}}\sum_{x\in\{0,1\}^n}\ket{x}\ket{0}^{\otimes n}.
\end{equation}
Next, for an input $x$, the oracle $\mc{O}_b$ outputs $f_b(x)$ stored in the ancilla register, where $f_b(x)=f_b(y)$ if and only if $y=x$ or $y=x\oplus b$,
\begin{equation}
    \ket{\psi_2}=\mc{O}_b\ket{\psi_1}=\frac{1}{\sqrt{2^n}}\sum_{x\in\{0,1\}^n}\ket{x}\ket{f_b(x)}.
\end{equation}
The ancilla qubits are then measured in the computational basis, an operation denoted by $M_a^{\otimes n}$. The result could be any $f_b(x)$ where $x\in\{0,1\}^n$ with equal probability, which is then discarded. The remaining state in the data register is
\begin{equation}
    \ket{\psi_3}=M_a^{\otimes n}\ket{\psi_2}=\frac{1}{\sqrt{2}}(\ket{x}+\ket{x\oplus b}).
\end{equation}
Applying the last set of Hadamard gates on the data qubits, we have
\bes
\begin{align}
        \ket{\psi_b}&= H^{\otimes n}\ket{\psi_3}\\
    &=\frac{1}{\sqrt{2^{n+1}}}\sum_{z\in \{0,1\}^n}[(-1)^{x\cdot z}+(-1)^{(x\oplus b)\cdot z}]\ket{z}\\
    &=\frac{1}{\sqrt{2^{n-1}}}\sum_{\{z|z\cdot b=0\}}(-1)^{x\cdot z}\ket{z},
\end{align}
\ees
where the last line arises from the fact that the element in the sum vanishes when $x\cdot z \neq (x\oplus b)\cdot z$. The remaining terms must have $x\cdot z = (x\oplus b)\cdot z$, which reduces to $z\cdot b = 0, \forall x$.

\section{Proof of \cref{th:1}}
\label{app:th1-proof}

\begin{proof}
The $\abs{S} = 1$ case is trivial.
  To prove the lower bound in \cref{eq:NTSIQ-bounds},
  consider the situation after $k$ executions of the circuit.
  It is known that $b$ is uniformly distributed in some subset $S_k$ of $S$ (e.g., $S_0 = S$).
  Let $B_k$ be the random variable representing the value of $b$
  and $Z_k$ be the random variable representing the value of $z$ obtained from the $k+1$'th execution of the circuit. Then the information gain about $B_k$ from $Z_k$ is given by their mutual information
   $I(B_k;Z_k) = H(B_k) + H(Z_k) - H(B_k,Z_k)$,
  where $H$ is the Shannon entropy. We have $H(B_k) = \log_2 \abs{S_k}$,
  $H(Z_k) \leq n$, and $H(B_k,Z_k) = (n-1) + \log_2 \abs{S_{k}}$ [because $(B_k,Z_k)$ is uniformly distributed across $\abs{S_{k}} 2^{n-1}$ pairs $(b, z)$: for each $b$ there are $2^{n-1}$ options for $z$ (the ones satisfying $b \cdot z = 0$)].
  Hence $I(B_k; Z_k) \leq 1$, and it follows that $\NTSIQ \geq \log_2 \abs{S}$:
  we learn at most $1$ bit of information in every step and we need to learn $\log_2 \abs{S}$ bits.

  To prove the upper bound in \cref{eq:NTSIQ-bounds}, let $b^*$ be the true (hidden) value of $b$
  and let $S^* = S \setminus \{b^*\}$. For each $b \in S^*$ let $X_b$ be the random variable
  representing the number of circuit executions needed to learn that $b^{*} \neq b$.
  Since upon each circuit execution
  $z$ is uniformly distributed among all $z$ s.t. $b \cdot z^{*} = 0$,
  we know that $X_b$ is geometrically distributed: $\Pr(X_b = k) = 2^{-k}$ for $k \geq 1$.
  Let $X_{\textnormal{max}} = \max_{b \in S^*} X_b$.
  Then
  \begin{equation}
    \label{eq:Xmax}
   \Pr(X_{\textnormal{max}} = k) \leq \sum_{b \in S^*} \Pr(X_b = k) = (\abs{S} - 1) 2^{-k},
  \end{equation}
which implies $\Pr(X_{\textnormal{max}} \geq k) \leq (\abs{S} - 1) 2^{1-k}$.
Now, let $a = \floor{\log_2(\abs{S}-1)}$ and let $b$ be the fractional part of $\log_2(\abs{S}-1)$ so that $\abs{S} - 1 = 2^{a + b}$. We have:
  \begin{multline}
    \label{eq:NTSIQ.Xmax}
    \NTSIQ = \mathbb{E} [X_{\textnormal{max}}] =
    \sum_{k=1}^{\infty} \Pr(X_{\textnormal{max}} \geq k) \\
    \leq \sum_{k=1}^{\infty} \min(1, 2^{a + b + 1 - k})
    = \sum_{k=1}^{a+1} 1 + \sum_{k=a+2}^{\infty} 2^{a+b+1-k} \\
    = a + 1 + 2^{b} = \log_2(\abs{S}-1) + 1 + 2^{b} - b \\
    \leq \log_2(\abs{S}-1) + 2 ,
  \end{multline}
  where in the second line we used $\sum^\infty_{k' = k} 2^{-k'} = 2^{1-k}$.

  \Cref{eq:NTSIQ-winf} follows from the observation that in the first step we either
  obtain $z=0^n$ with probability $2^{1-n}$, in which case we learn nothing about $b^{*}$,
  or we reduce the problem to Simon-($n-1$). Each term in the sum in \cref{eq:NTSIQ-winf}
  represents the expected number of steps needed to reduce the problem Simon-$(k+1)$ to Simon-$k$. The Erd\H{o}s
-Borwein constant is defined as $\EEB \equiv \sum_{k=1}^\infty 1/(2^k - 1)$, so $\sum_{k=1}^{n-1} \frac{1}{1-2^{-k}} - (n + \EEB - 1) = -\sum_{k=n}^{\infty} 1 / (2^k - 1) = O(2^{-n})$.

  \Cref{eq:NTSIQ-w1} follows from the observation that for $w=1$ all $X_b$ are independent,
  hence the task of computing $\NTSIQ$ reduces to the task of computing the expected value
  of the maximum of $n-1$ independent geometrically distributed random variables, which
  was solved in \cite{poblete2006binomial}.
\end{proof}

\section{Ideal quantum $\NTSIQ$ calculation}
\label{app:NTSIQ}

\Cref{th:1} states the upper and lower bounds for $\NTSIQ$ and that it can be calculated exactly for the two limiting cases $w=1$ and $w=\infty$. In the main text, we constructed an interpolation between the two known values, \cref{eq:NTSIQ-winf,eq:NTSIQ-w1}, which we can rewrite as:
\begin{equation}
    \NTSIQ = \log_2(N_w) + t \Gamma_{w=\infty} + (1 - t) \Gamma_{w=1},
    \label{eq:NTSIQ-rewrite}
\end{equation}
where the weighted $\Gamma_{w=\infty}=E_{\text{EB}} - 1$ and $\Gamma_{w=1}=0.5 + \gamma/ \ln(2)$ are constants from \cref{eq:NTSIQ-winf,eq:NTSIQ-w1}, respectively. 

An alternative approach is to use a Monte Carlo simulation to compute $\NTSIQ$. For a $\wSimon{w}{n}$ problem, we randomly picked 10 $b^*$'s from the possible set $S_0$ such that $|S_0|=N_w$. For each $b^*$, a series of valid $z_k$'s such that $b^*\cdot z_k=0$ is drawn. After each $z_k$ was drawn, we eliminated any $b$ in $S_k$ that did not satisfy $b\cdot z_k=0$ and stopped when $|S_k|=1$. At this point, we expected to only have $b^*$ as the only element left in $S_k$. The simulated $\NTSIQ$ was calculated 10,000 times for each $b^*$, and the average was taken. After performing the simulation for 10 $b^*$'s, we found the mean and standard deviation and reported this as the Monte Carlo version of $\NTSIQ$.

\cref{fig:NTSIQ} compares the upper bound and lower bound from \cref{eq:NTSIQ-bounds}, the interpolated version of $\NTSIQ$, and the Monte Carlo version of $\NTSIQ$ for $w=7$.

\begin{figure}[t]
    \centering
    \includegraphics[width=0.5\textwidth]{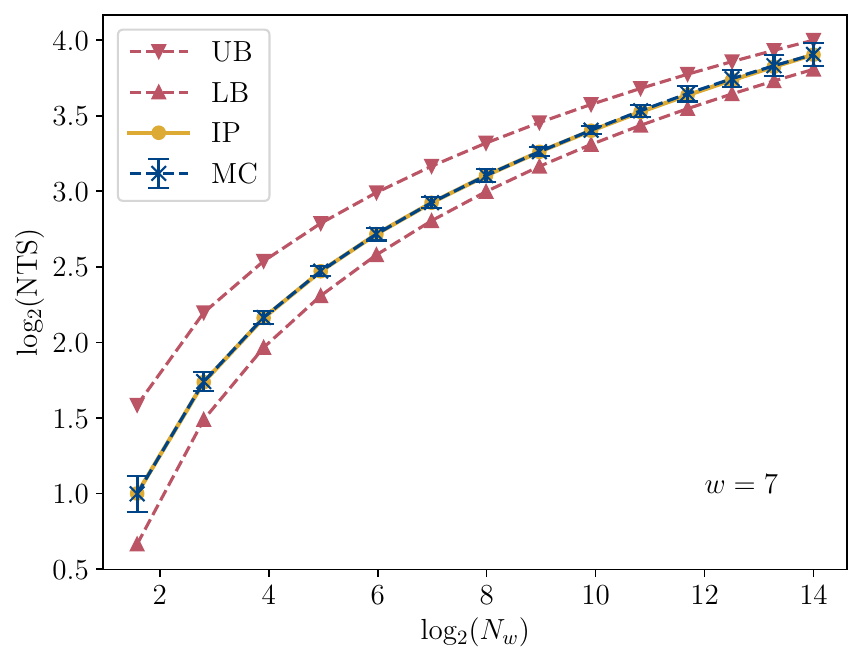}
    \caption{Upper bound (UB) and lower bound (LB) from \cref{eq:NTSIQ-bounds}, the interpolated version of $\NTSIQ$ (IP), and the Monte Carlo version of $\NTSIQ$ (MC) for $w=7$. The error bars on the blue lines are $5\sigma$ in each direction.}
    \label{fig:NTSIQ}
\end{figure}

\section{Error Mitigation and Mutual Information}
\label{app:MEM}

We first discuss our error mitigation strategy, and then delve more deeply into the underlying data using a mutual information analysis.

\subsection{Error Mitigation}
\label{subapp:MEM}

DD suppresses noise during circuit execution. To complement DD, we used measurement error mitigation (MEM)~\cite{Temme:2017aa} to reduce measurement-related errors in the NISQ algorithm setting. We applied two MEM methods, PyIBU~\cite{Srinivasan:22} and M3~\cite{Nation:2021aa}, in separate experiments. 

The results reported in the main plots for 127-qubit devices are mitigated using M3. The mitigation process is built-in and runs through Qiskit. It performs extra experiments alongside the main experiments to obtain information used to estimate the measurement error. This information is then used to correct the results from the main experiments.

We used PyIBU with measurement error information from the device properties on the days of the experiments. It runs MEM based on iterative Bayesian unfolding. Most MEM methods assume a response matrix $R$ that modifies the ideal measurement probability $\vect{\theta}$ to a noisy probability $\vect{p}=R\vect{\theta}$. 
MEM via iterative Bayesian unfolding \cite{Srinivasan:22} uses the Expectation-Maximization algorithm from machine learning and iteratively applies Bayes' rule to find the mitigated probability distribution $\theta_j^{k+1}=\sum_{i=1}^{2^n}p_i  R_{ij}\theta^k_j/(\sum_m R_{im}\theta^k_m)$, starting from the initial guess $\vect{\theta}^0$ until $\vect{\theta}^{k+1}$ converges at the $(k+1)$'th step.

\cref{fig:p_127} shows the probability of obtaining $z\cdot b=0$ as a function of $n$, the length of the hidden bitstring $b$. Note that the data is obtained from two separate experiments:
\begin{itemize}
\item Experiment \#1 (up to 30 qubits): DD is compared to DD+MEM in the same calibration cycle. Here, the DD sequence is performed once and MEM is performed offline by applying PyIBU on the results obtained from the DD experiment. This experiment shows that adding MEM to DD (blue curves) marginally improves the results over using DD only (red curves).
\item Experiment \#2 (up to 126 qubits): we ran DD+MEM in the same calibration cycle (green curves). This time, MEM is performed online via the built-in M3 on Qiskit. Since MEM is performed along with DD, there are no DD-only results in experiment \#2. Since this experiment was conducted separately from the PyIBU and no-MEM experiments, it should not be compared to the data from experiment \#1.
\end{itemize}

\subsection{Mutual information} 
\label{subapp:I}

Solving Simon's problem requires a series of ``good'' outcomes of the form $z\cdot b=0$ and $z\neq 0$. Let $p$ be the probability of $z\cdot b=0$ and let $q$ be the probability of obtaining $z=0$. The result of each measurement round then falls into one of the following three categories:
\begin{equation}
    \begin{aligned}
    \Pr(z= 0) &= q,\\
    \Pr(z\cdot b = 0 | z\neq 0) &= (1-q)p,\\
    \Pr(z\cdot b = 1 | z\neq 0) &= (1-q)(1-p).
    \end{aligned}
    \label{eq:prob-pq}
\end{equation}

We can extract $p$ and $q$ from the counts data of the experiment. 
The probability $\Pr(z\cdot b = 0) = q + (1-q)p$ for both $127$-qubit devices is plotted in \cref{fig:p_127}. Note that we set $p=0$ when $n=1$ by definition. This is because the only possible non-zero hidden bitstring is $b=1$ when $n=1$, hence, every $z$ such that $z\cdot b=0$ is $z=0$, which contributes only to $q$. With this choice, we start the plot in \cref{fig:p_127} from $n=2$.

We calculate the estimated information gained per circuit round, with the optimistic assumption that we gain the same amount of information in every round. Letting $S(X)$ denote the Shannon entropy of a random variable $X$, the classical mutual information $I(Z;B)$ is defined as
\begin{equation}
    \begin{aligned}
    I(Z;B)&=S(Z)+S(B)-S(Z|B)\\
    &=\sum_{z,b} \Pr(z,b) \log_2 \frac{\Pr(z,b)}{\Pr(z)\Pr(b)}\\
     &=\sum_{z,b} \Pr(z|b)\Pr(b) \log_2 \frac{\Pr(z|b)}{\Pr(z)}\\
     &= \left( \sum_b\Pr(b) \right)\left( \sum_{z} \Pr(z|b)\log_2 \frac{\Pr(z|b)}{\Pr(z)}\right)\\
     &=\sum_{z} \Pr(z|b) \log_2 \frac{\Pr(z|b)}{\Pr(z)}.
    \end{aligned}
\end{equation}
To go from the second to the third line, we used the fact that $\Pr(z|b)\Pr(b)=\Pr(z,b)$. In the fourth line, the summation over $b$ is separated from the summation over $z$ because the quantity in the second bracket is independent of $b$ by permutation symmetry. We used $\sum_b \Pr(b)=1$ to arrive at the last line and remove the dependence on $b$.

\begin{figure}[t]
    \centering
    \includegraphics[width=0.5\textwidth]{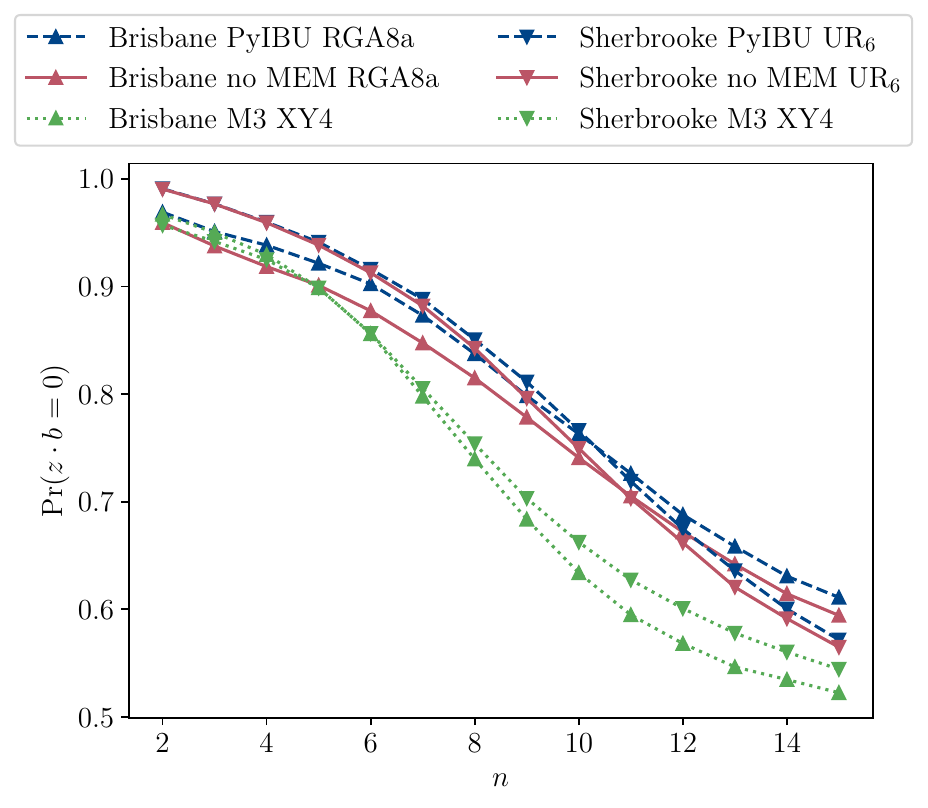}
    \caption{Probability of obtaining $z\cdot b=0$ with MEM (blue, green) and without MEM (red) on Brisbane and Sherbrooke, as a function of the bitstring length $n$. Each data point is a weighted average over HW from 1 to $n$. The data shown uses the best-performing DD sequences on each device.}
    \label{fig:p_127}
\end{figure}

\begin{figure}
    \centering
    \includegraphics[width=0.5\textwidth]{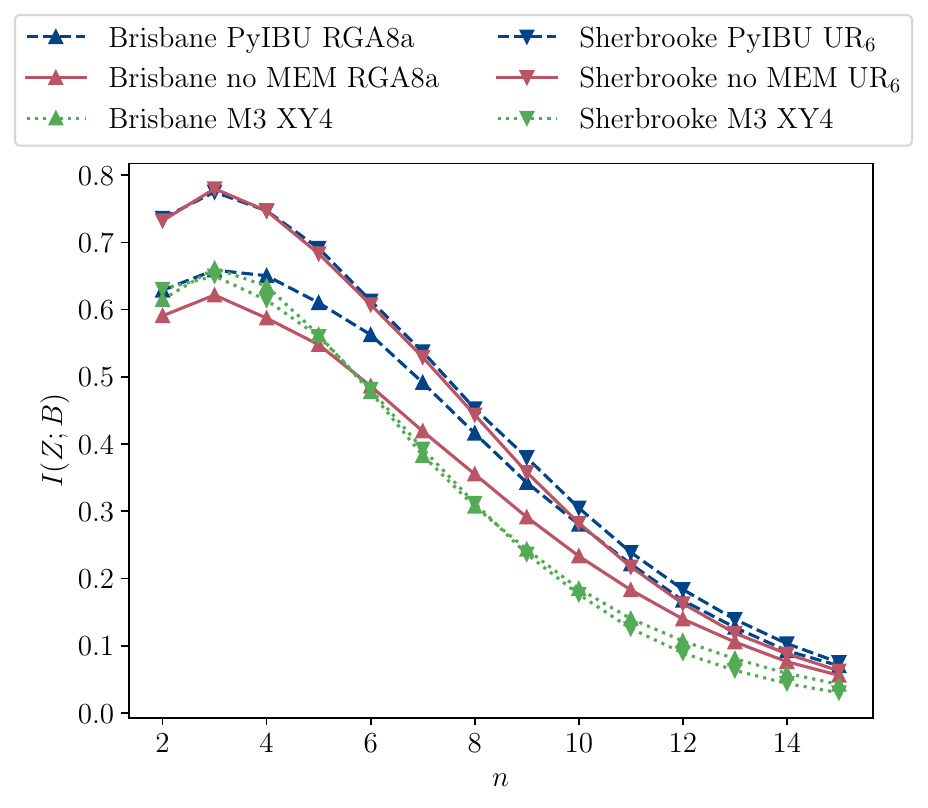}
    \caption{Information learned per circuit as a function of problem size $n$, with and without MEM from the best sequence on 127-qubit devices, Brisbane and Sherbrooke.}
    \label{fig:I_127}
\end{figure}

Consider dividing the sum into three cases according to \cref{eq:prob-pq}, we have $I(Z;B)=\sum_{i=1,2,3}I^i(Z;B)$, where $I^i(Z;B)$ is given as follows.
\begin{itemize}
    \item Case I: $z=0$.
    \begin{equation}
        \begin{aligned}
            I^1(Z;B) &= \Pr(z=0|b)\log_2\frac{\Pr(z=0|b)}{\Pr(z=0)}\\
            &= q\log_2\frac{q}{q} =0 .
        \end{aligned}
        \label{eq:I1}
    \end{equation}
    \item Case II: $z\cdot b=0$, $z\neq 0$. There are $N_2 = 2^{n-1}-1$ possibilities of $z$'s (half of $2^n$ and excluding 0).
    \begin{equation}
        \begin{aligned}
            I^2(Z;B) &= \sum_{z}\Pr(z\cdot b = 0\land z\neq 0|b)\log_2\frac{\Pr(z\cdot b = 0\land z\neq 0|b)}{\Pr(z\cdot b = 0\land z\neq 0)}\\
            &= \sum_{z}\frac{(1-q)p}{N_2}\log_2\frac{(1-q)p/N_2}{(1-q)/(2^n-1)}\\
            &=(1-q)p\log_2\frac{p(2^n-1)}{2^{n-1}-1}
        \end{aligned}
        \label{eq:I2}
    \end{equation}
    \item Case III: $z\cdot b=1$, $z\neq 0$. There are $N_3 = 2^{n-1}$ possibilities of $z$'s (half of $2^n$).
    \begin{equation}
        \begin{aligned}
            I^3(Z;B) &= \sum_{z}\Pr(z\cdot b = 1\land z\neq 0|b)\log_2\frac{\Pr(z\cdot b = 1\land z\neq 0|b)}{\Pr(z\cdot b = 1\land z\neq 0)}\\
            &= \sum_{z}\frac{(1-q)(1-p)}{N_3}\log_2\frac{(1-q)(1-p)/N_3}{(1-q)/(2^n-1)}\\
            &=(1-q)(1-p)\log_2\frac{(1-p)(2^n-1)}{2^{n-1}}
        \end{aligned}
        \label{eq:I3}
    \end{equation}
\end{itemize}

We can calculate the classical mutual information $I(Z;B)$ in terms of $p$ and $q$ as the sum of \cref{eq:I1}, \cref{eq:I2}, and \cref{eq:I3}:
\begin{equation}
    \begin{aligned}
    I(Z;B)&=(1-q)\left(p\log_2\frac{p(2^n-1)}{2^{n-1}-1}\right.\\
    &\qquad\qquad\left.+(1-p)\log_2\frac{(1-p)(2^n-1)}{2^{n-1}}\right).
    \end{aligned}
    \label{eq:IZB}
\end{equation}

The expected information learned per circuit, $I(Z;B)$, is plotted as a function of problem size $n$ in \cref{fig:I_127}. As explained above, the data is obtained from two separate experiments; (1) PyIBU and no-MEM are performed together, and (2) M3 is performed separately. Therefore, the data should not be compared across experiments. Thus, comparing only the PyIBU (blue) and no MEM (red) data, \cref{fig:I_127} shows that MEM results in a significant benefit in the Brisbane case, but only a very slight benefit in the Sherbrooke case. Interestingly, \cref{tab:127devices-limit15} shows that Brisbane has, on average, lower readout errors than Sherbrooke, so one might have expected the former to benefit less from MEM than the latter, in contrast to our results in \cref{fig:I_127}.

\section{Post-processing for the `NISQ' variant of the quantum algorithm}
\label{app:post-processing}

The original noiseless solver requires sampling $n-1$ independent output bitstrings $z$'s and solves for the unique $b$ that satisfies $z\cdot b=0, \forall z$. 
This means that if one of the $z$'s is altered by noise, the resulting $b$ is guaranteed to be wrong. To reduce the effect of noise, we employ \cref{algo:new}, a post-processing algorithm that applies in the `NISQ' setting described in the main text.

\cref{algo:new} assumes that the player has access to the distribution $\Pr(z\cdot b=0,z\neq 0)$ generated by the NISQ device, just as in \cref{fig:p_127}. This probability distribution must be obtained in an experiment separate from the experiments that obtain the stream $Z=\{z_1, z_2,\dots z_m\}$. The player starts by initializing all possible $b_i$ to have an equal prior probability of being the correct $b$. After each call to the oracle to obtain $z_j$, the player computes $\Pr_{\text{post}}$ according to whether $b_i\cdot z_j=0$, as denoted in \cref{algo:new}. The player has a choice to guess $b$ according to their probability distribution or to make more oracle calls, in which the prior probability of the next round is reset to be the posterior probability of the current round.

\begin{algorithm}[H]
\caption{Solving for $b$ given $\Pr(z\cdot b|\HW(b),z\neq 0)=0$.}
\label{algo:new}
\begin{algorithmic}
\Require input $n,w$, stream $Z=\{z_1, z_2,\dots z_m\}$ of bitstrings, \qquad $\Pr(z\cdot b|\HW(b),z\neq 0)=0$
\Ensure a list $B$ of all possible $b_i$: $\{b_1,b_2,\dots,b_{N_w}\}$
\Ensure a dictionary $\Pr_{\text{pre}}(B)$ of each $\Pr_{\text{pre}}(b_i)=1/N_w$\;
\Ensure a dictionary $\Pr_{\text{post}}(B)$\;
  \While{Termination condition not met}
    \For{$z_i$ in $Z$}
        \For{$b_j$ in $B$}
          \If{$z_i\cdot b_j =0$}
            \State $\Pr_{\text{post}}(b_j)=\Pr_{\text{pre}}(b_j)*f(\text{HW}(b_j))$\;
          \Else
            \State $\Pr_{\text{post}}(b_j)=\Pr_{\text{pre}}(b_j)*(1-f(\text{HW}(b_j)))$\;
          \EndIf
        \EndFor
        \State $\Pr_{\text{pre}}(b_j)=\Pr_{\text{post}}(b_j)$
    \EndFor
  \EndWhile
  \State \Return $\text{argmax}(\Pr_{\text{post}}(b_j))$\;
\end{algorithmic}
\end{algorithm}

\section{Circuit Complexity Reduction}
\label{app:circ-reduction}

To avoid the impractical task of running $2^n-1$ different oracles for every possible hidden bitstring for Simon-$n$ problems, we use permutation symmetry and only run $n$ circuits for a problem size $n$. Our chosen bitstrings are $b^i=0^{n-i}1^{i}$ for $1\le i\le n$. The number of queries $Q_i$ for each $b^i$ is weighted in the NTS calculation in \cref{eq:NTS-wn} according to the corresponding weight $h_i = {n \choose i}$, as explained in \cref{sec:speedup}. We assume that the permutation symmetry is processed within the compiler, as detailed in \cref{app:rules:oracle}. 

We aim to quantify the quantum speedup for a fixed $w\le n$. To further reduce the number of circuits, instead of running $\wSimon{w}{n}$ for every $n$ and $w$, we only run $\wSimon{w}{n_{\max}}$ and extract the $\wSimon{w}{m}$ cases where $m<n_{\max}$ during post-processing.

Furthermore, we reduce Simon-$n$ to Simon-$m$, where $m< n$. For example, we can use the result of the experiment at problem size $n=29$ to extract the results for smaller problem sizes, $m=2$ to $m=28$. \Cref{fig:simon-m-n} shows an example of the reduction from $n=3$ ($b=011$) to $m=2$ ($b=11$). 
\begin{figure}[h!]
    \centering
    \begin{tikzpicture}
        \begin{yquantgroup}
            \registers{
                qubit {$\ket{0}_{d_0}$} x0;
                qubit {$\ket{0}_{a_0}$} a0;
                qubit {$\ket{0}_{d_1}$} x1;
                qubit {$\ket{0}_{a_1}$} a1;
                qubit {$\ket{0}_{d_2}$} x2;
                qubit {$\ket{0}_{a_2}$} a2;
            }
            \circuit{
                h x0, x1, x2;
                cnot a0 | x0;
                cnot a2 | x1;
                cnot a2 | x2;
                h x0, x1, x2;
            }
            \equals[$\mapsto$]
            \circuit{
                h x1, x2;
                cnot a2 | x1;
                cnot a2 | x2;
                h x1, x2;
            }
        \end{yquantgroup}
    \end{tikzpicture}
    \caption{Circuit reduction from Simon-3 to Simon-2. The right circuit $m=2$ ($b=11$) can be extracted from the left circuit $n=3$ ($b=011$) by tracing out qubits $d_0$ and $a_0$.}
    \label{fig:simon-m-n}
\end{figure}
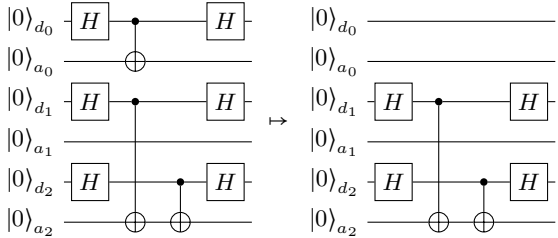

This reduction is possible under the assumption of a CPTP map governing the circuit in the open-system setting. We give the proof of this fact next. 

Since qubits $(d_j,a_j)$ are unentangled in our oracle construction from the rest of the circuit if the $j$th bit of the hidden bitstring $b$ is zero, these two qubits are separated from the rest of the oracle, as indicated by the individual dashed boxes in \cref{fig:simon-n5-graph}.
Therefore, the only difference between $b_n=0^{n-i}1^{i}$ and $b_m=0^{m-i}1^{i}$ is that the circuit for $|b_n|=n$ applies extra operations on extra qubits compared to the circuit for $|b_m|=m$. Let $m\in [i,n-1]$; then all the circuits for $b_m=0^{m-i}1^{i}$ have the identical output as the circuit for $b_n=0^{n-i}1^{i}$ if we consider only the overlapping set of qubits, as illustrated in \cref{fig:simon-m-n}. Intuitively, we may thus extract the Simon-$m$ results from the Simon-$n$ results by running only the Simon-$n$ circuits and tracing over the first $2(n-m)$ data qubits, a practice we implemented in our experiments and subsequent analysis. Note that the left circuit in \cref{fig:simon-m-n} is the same circuit as in \cref{fig:simon-011} but drawn with a different qubit order.

In the context of mapping the circuits to physical qubits in our experiments, the initial set of the first $2(n-m)$ qubits that we trace out remains the same for each experiment. For example, consider a bitstring from an $n=5$ experiment, with partial trace to $m=4$ and $m=3$. The five bits of $b=00011$ could be measured from physical qubits $q_5, q_1, q_{120}, q_{30}$, and $q_{100}$. The partial trace to obtain $b=0011$ is derived from the experiment involving $q_1, q_{120}, q_{30}$, and $q_{100}$. Similarly, the partial trace to obtain $b=011$ is derived from the experiment involving $q_{120}, q_{30}$, and $q_{100}$. The mapping from the precompiled circuit to the physical layout on the IBM machines was performed using the Qiskit compiler with optimization level 3. The compiler heuristically identifies an optimized circuit compatible with the heavy-hex connectivity. The measurement results were returned as bitstrings, maintaining the same bit order as in the precompiled circuit. Consequently, partial tracing can be done without concern for the device's physical connectivity.

Let us now prove the equivalence of our procedure to actually running the Simon-$m$ circuits, as long as the completely positive, trace preserving (CPTP) map governing the circuit in the open system case factors into a product over the ``copied'' and ``XOR-ed'' qubits, i.e., those corresponding to a 0 (copied) or 1 (XOR-ed) in the bitstring $b$ that defines the given oracle. 

In preparation for our more general discussion below, let us equivalently represent the action of the closed-system Simon-$n$ circuit with hidden bitstring $b$ on some initial state $\rho$ of the $n$ data and $n$ ancilla qubits as
\beq
\mathrm{Simon}_n(b)[\rho] = \Tr_{a}\left[ \left( H_d^{\otimes n} \circ \mathcal{O}_{b} \circ H_d^{\otimes n} \right) [\rho] \right] ,
\label{eq:Simon}
\eeq
where $\mathcal{O}_{b}$ represents the Simon oracle, and $\Tr_a$ means that the state of the ancilla qubits (labeled $a_0,\dots,a_{n-1}$) is discarded at the end, so that $\mathrm{Simon}_n(b)[\rho]$ is the state of the $n$ data qubits at the end of one run of the algorithm. Thus, if we write $\Pi_b = \ketbra{\psi_b}$ and $\ketbra{\psi_0} = \Pi_{0^n} \otimes \Pi_{0^n}$, then it follows from \cref{eq:Simon} that $\mathrm{Simon}_n(b)[\ketbra{\psi_0}] = \Pi_b$. 

\begin{figure*}[t]
    \centering
    \includegraphics[width=0.49\linewidth]{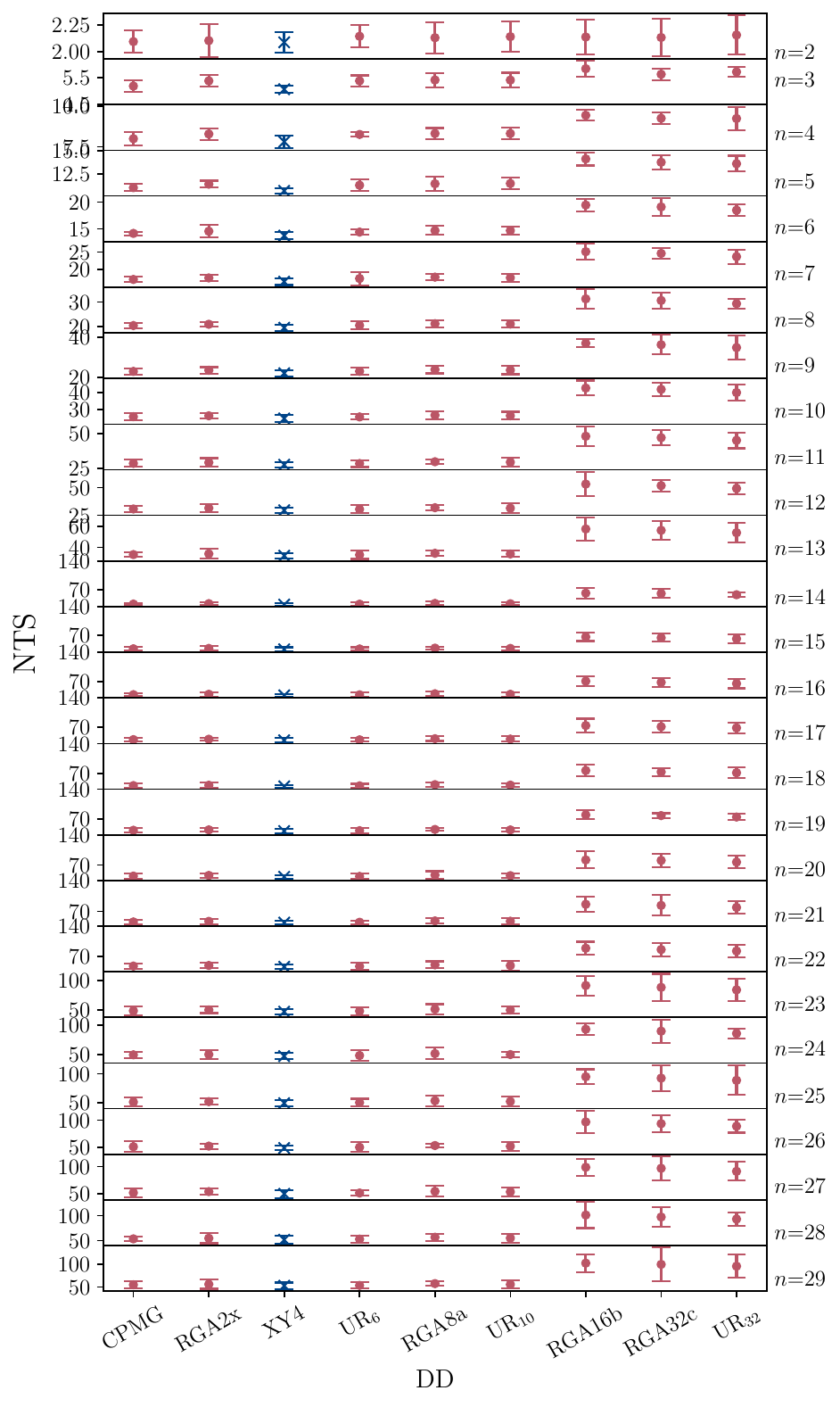}
    \includegraphics[width=0.49\linewidth]{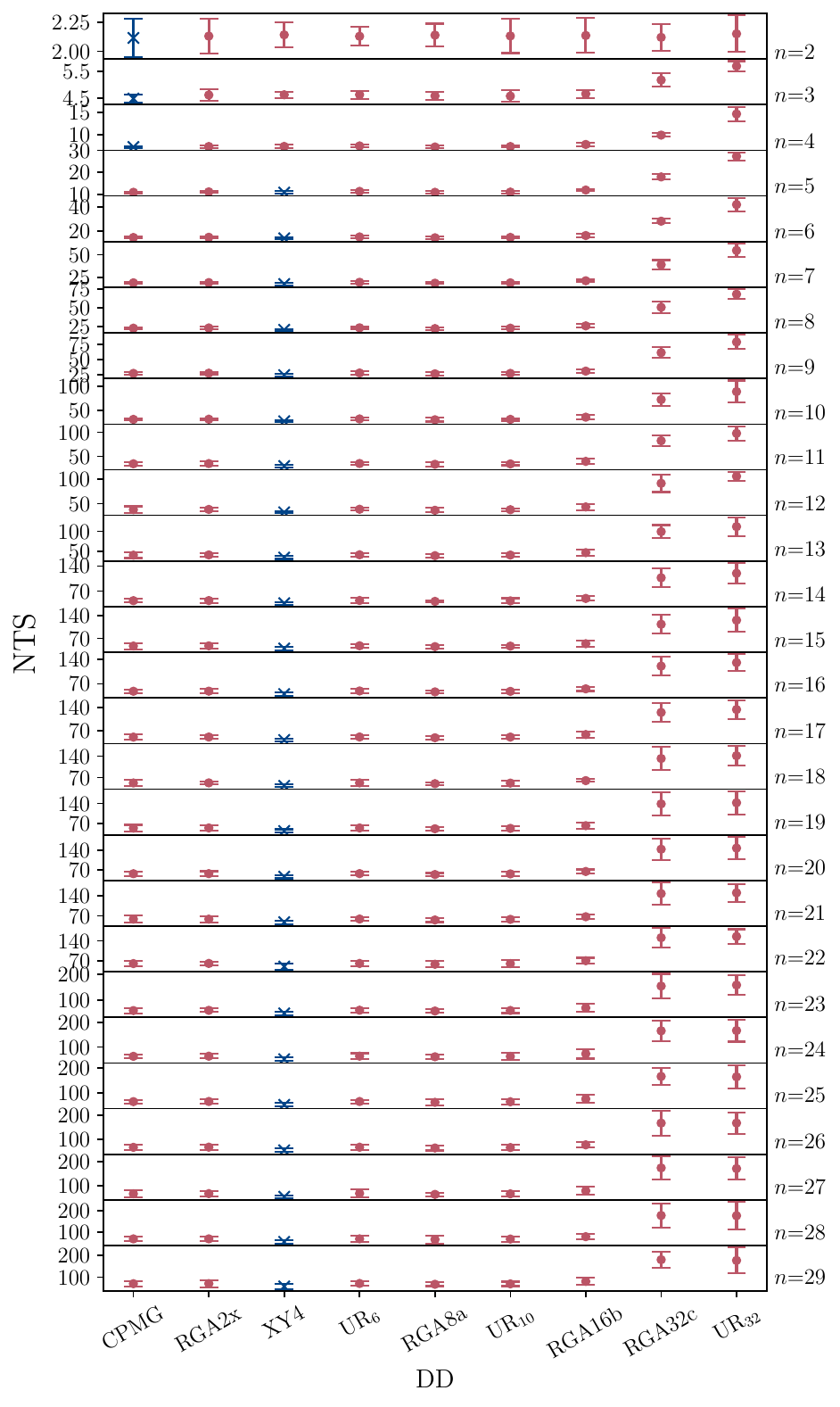}
    \caption{DD performance for \wSimon{4}{n} according to the NTS metric on Sherbrooke (left) and Brisbane (right) over the nine sequences in the set $\mathcal{D}=\{$CPMG, RGA$_{2x}$, XY4, UR$_6$, RGA$_{8a}$, UR$_{10}$, RGA$_{16b}$, RGA$_{32c}$, UR$_{32}\}$, and using Experiment \#2 as described in \cref{app:MEM}. Sequences colored in blue ($\times$) denote those identified as optimal performers on their respective machines at varying problem sizes $n$. The arrangement of sequences in the plot is sorted by the number of pulses in each sequence. XY4 emerges as the best-performing sequence overall on both Sherbrooke and Brisbane and becomes the top-performing sequence for all $n\ge 5$ on both devices. The error bars, representing confidence intervals derived through bootstrapping, extend 1$\sigma$ in both directions from each data point.}
    \label{fig:dd-perf}
\end{figure*}

Let us denote the qubit sets as follows: $D_\nu=\{d_0,\dots,d_{n-m-1}\}$ ($A_\nu=\{a_0,\dots,a_{n-m-1}\}$) being the first $n-m$ data (ancilla) qubits and $D_\mu=\{d_{n-m},\dots,d_{n-1}\}$ ($A_\mu=\{a_{n-m},\dots,a_{n-1}\}$) being the remaining $m$ data (ancilla) qubits in the circuit. Let $\cup_\nu\equiv D_\nu\cup A_\nu=\{ d_0,a_0, \dots, d_{n-m-1},a_{n-m-1} \}$. Let us show that
if $b_n = 0^{n-i}1^{i}$, then 
\beq
\mathrm{Simon}_m(b_m)[\Pi_{0^m} \otimes \Pi_{0^m}]
= \Tr_{\cup_\nu} \left[  \mathrm{Simon}_n(b_n)[\ketbra{\psi_0}] \right] ,
\label{eq:BV-equiv}
\eeq
where the trace means that the states of qubits in $\cup_\nu$ are discarded from the result of running $\mathrm{Simon}_n(b_n)$. To prove this claim in the absence of any noise, note that:
\bes
    \begin{align}
        \Tr_{\cup_\nu} & \left[  \mathrm{Simon}_n(b_n)[\ket{\psi_0}\! \bra{\psi_0}] \right] \\
        &=\Tr_{\cup_\nu}\left[\Pi_{b_n}\right]\\
        &=\Tr_{\cup_\nu}\left[\Pi_{\nu}\otimes \Pi_{b_m}\right]\\
         &=\Pi_{b_m}\\
         &=\mathrm{Simon}_m(b_m)[\Pi_{0^m}\otimes\Pi_{0^m}],
    \end{align}
\ees
as claimed. Here, $b_n=0^{n-i}1^{i}$, $b_m=0^{m-i}1^i$ ($m<n$) and in the third line, $\Pi_{\nu}=\ketbra{\psi_\nu}$, where $\ket{\psi_\nu}=\ket{+}^{n-m}$ is derived from the last line of \cref{eq:Simon}:
\bes
\begin{align}
        &\frac{1}{\sqrt{2^{n-1}}}\sum_{\{z|z\cdot b_n=0\}}(-1)^{x\cdot z}\ket{z}_d\\
        &=\frac{1}{\sqrt{2^{n-1}}}\sum_{y\in\{0,1\}^{n-m}}\ket{y}_{D_\nu}\otimes\sum_{\{z|z\cdot b_m=0\}}(-1)^{x\cdot z}\ket{z}_{D_\mu}\\
        &=\frac{1}{\sqrt{2^{n-m}}}\sum_{y\in\{0,1\}^{n-m}}\ket{y}_{D_\nu}\otimes\notag\\
        &\qquad \frac{1}{\sqrt{2^{m-1}}}\sum_{\{z|z\cdot b_m=0\}}(-1)^{x\cdot z}\ket{z}_{D_\mu}\\
        &=\ket{\psi_\nu} \otimes \ket{\psi_{b_m}}.
\end{align}
\ees

Now consider the case where each gate is represented not by a unitary but by a CPTP map. We can rewrite the initial state as
\begin{equation}
    \ketbra{\psi_0}=\Pi_{0^n}\otimes\Pi_{0^n}=\rho_\nu\otimes \rho_\mu,
\end{equation}
where $\rho_\nu=\Pi_{0^{n-m}}\otimes\Pi_{0^{n-m}}$ and $\rho_\mu=\Pi_{0^{m}}\otimes\Pi_{0^{m}}$.
The Simon oracle does not introduce any two-qubit gates between qubit sectors $\cup_\mu=D_\mu\cup A_\mu$ and $\cup_\nu = D_{\nu}\cup A_\nu$, so it is reasonable to assume that under the coupling to the environment, they remain uncoupled (as long as there is no unintended crosstalk between the two sectors). Therefore, the CPTP map for the noisy Simon-$n$ algorithm will be
\begin{align}
\mathcal{S}_n(b_n)[\rho_\nu \otimes \rho_\mu]  &= \Tr_{\cup_\nu}\left[ (\mathcal{H}_\nu \otimes \mathcal{H}_\mu) \circ  (\mathcal{O}_\nu \otimes \mathcal{O}_\mu) \right. \notag \\
&\qquad\qquad \left. \circ  (\mathcal{H}_\nu \otimes \mathcal{H}_\mu)[\rho_\nu \otimes \rho_\mu] \right], 
\end{align}
where $\mathcal{H}_{\nu,\mu}$ and $\mathcal{O}_{\nu,\mu}$ represent the CPTP maps corresponding to the experimental implementation of the unitaries $H$ (multi-qubit Hadamard) and $\mcOf$ (oracle) acting on qubit sectors $\nu,\mu$. Recall that for arbitrary CPTP maps $\mathcal{U}$ and $\mathcal{V}$ acting on a tensor-product space 
\begin{equation}
\mathcal{U} \otimes \mathcal{V}[\rho \otimes \sigma] = \mathcal{U}[\rho] \otimes \mathcal{V}[\sigma]. 
\end{equation}

Therefore, 
\begin{subequations}
\begin{align} 
&\Tr_{\cup_\nu} \left[  \mathcal{S}_n(b_n)[ \ket{\psi_0}\! \bra{\psi_0} ] \right] \\
&\quad = \Tr_{\cup_\nu} \left[  (\mathcal{H}_\nu \otimes \mathcal{H}_\mu) \circ  (\mathcal{O}_\nu \otimes \mathcal{O}_\mu) \right. \notag \\
&\left. \qquad \qquad\qquad \circ  (\mathcal{H}_\nu \otimes \mathcal{H}_\mu) [\rho_\nu \otimes \rho_\mu] \right] \\
&\quad = \Tr_{\cup_\nu} \left[  (\mathcal{H}_\nu \circ \mathcal{O}_\nu \circ \mathcal{H}_\nu) [\rho_\nu]\notag \right. \\
&\left. \qquad\qquad\qquad \otimes (\mathcal{H}_\mu \circ \mathcal{O}_\mu \circ \mathcal{H}_\mu) [ \rho_\mu] \right] \\
&\quad = \mathcal{S}_m(b_m)[\rho_\mu] \\
&\quad = \mathcal{S}_m(b_m)[\Pi_{0^m} \otimes \Pi_{0^m}]. 
\end{align}
\end{subequations}
This is the CPTP map generalization of the closed-system result~\cref{eq:BV-equiv}, and it shows that the reduction from Simon-$n$ to Simon-$m$ holds rigorously also in the open system setting, as long as the CPTP map factors according to the qubit sectors $\nu$ and $\mu$.

\section{Information-theoretic properties of the distribution of output bitstrings $z$'s}
\label{app:dkl}

In this section, which supports \cref{sec:NISQ-speedup-expectation}, information-related quantities are measured in bits, i.e., defined using $\log_2$. E.g., the Shannon entropy of the uniform distribution $U$ on $N$ elements is $H(U) = \log_2(N)$.

\begin{mylemma}
  \label{lm:dkl.vs.chi2}
  Let $P, Q$ be two probability distributions, let $\chi^2$ denote the chi-squared divergence, and let $dP / dQ$ denote the Radon-Nikodym derivative of $P$ relative to $Q$~\cite{guntuboyina2013sharp}. Then
  \begin{enumerate}
  \item If $dP / dQ \geq 1-p$ almost everywhere for some $p \in [0, 1]$ then
  \begin{equation}
    \label{eq:dkl.chi2.ub}
    D_{\text{KL}}(P \parallel Q) \leq \chi^2(P \parallel Q) C(-p) / \ln(2).
  \end{equation}
  \item If $dP / dQ \leq 1 + p$ almost everywhere for some $p \in [1, \infty]$ then
  \begin{equation}
    \label{eq:dkl.chi2.lb}
    D_{\text{KL}}(P \parallel Q) \geq \chi^2(P \parallel Q) C(p) / \ln(2).
  \end{equation}
  \end{enumerate} 
  Here $C(p) = \frac{1}{p^2}[(1+p) \ln(1+p) - p]$ for $p \in (-1, 0) \cup (0, \infty)$
  and extended by continuity to $[-1, \infty]$ as $C(-1) = 1$, $C(0) = 1/2$, $C(\infty) = 0$.
  $C(p)$ is monotonically decreasing on $[-1, \infty]$.
\end{mylemma}
We note that the series expansion of $C(p)$ at $p=0$ is
\begin{equation}
  \label{eq:dkl.chi2.cp-series}
  C(p) = \frac{1}{2} - \frac{p}{6} + O(p^2).
\end{equation}
We also note that the bounds are tight, i.e., there are distributions $P$
and $Q$ with the ratio $D_{\text{KL}} (P \parallel Q) / \chi^2(P \parallel Q)$
arbitrarily close to $C(-p) / \ln(2)$ and $C(p) / \ln(2)$ for parts 1 and 2
respectively. An example is the following: for $p \in [-1, \infty)$ consider
distributions on $\{0, 1\}$ with $\Pr(P = 0) = \epsilon (1 + p)$, $\Pr(Q = 0) = \epsilon$,
$\Pr(P = 1) = 1 - \epsilon (1 + p)$, and $\Pr(Q = 1) = 1 - \epsilon$.
Then one can calculate that
$D_{\text{KL}}(P \parallel Q) = \chi^2(P \parallel Q) C(p) / \ln(2) (1 + O(\epsilon))$.
For a more extensive and systematic study of bounds on $f$-divergences, see Ref.~\cite{guntuboyina2013sharp}.
\begin{proof}[Proof of \cref{lm:dkl.vs.chi2}]
  If $P$ is not absolutely continuous with respect to $Q$ then \cref{eq:dkl.chi2.ub,eq:dkl.chi2.lb} hold trivially ($\infty \leq \infty$).
  Otherwise, we have
  \bes
  \begin{align}
    \label{eq:dkl.chi2.2}
    D_{\text{KL}}(P \parallel Q) &= \int \frac{dP}{dQ} \log_2\left(\frac{dP}{dQ}\right) dQ\\
    &= \frac{1}{\ln(2)} \int f\left(\frac{dP}{dQ} - 1\right) dQ,
  \end{align}
  \ees
  where $f(x) = (x+1) \log_2(x+1) - x$.
  Similarly,
  \begin{equation}
    \label{eq:dkl.chi2.3}
    \chi^2(P \parallel Q) = \int g(x) dQ,
    \textrm{ where } g(x) = x^2.
  \end{equation}
  Consider $C(x) = f(x) / g(x)$ defined for $x \in [-1, \infty]$
  (values at $x \in \{-1, 0, \infty\}$ are defined by continuity as $C(-1) = 1$, $C(0) = 1/2$, $C(\infty) = 0$).
  Note that
  \begin{equation}
    \frac{\partial C(x)}{\partial x} = - (2x^{-3} + x^{-2}) \ln(1+x) < 0 \textrm{ for } x \in (-1, \infty).
  \end{equation}
  Therefore, $C(x)$ is monotonically decreasing on $[-1, \infty]$.
  In the context of the first claim, for $x = dP/dQ - 1$ we have $x \geq -p$. Therefore, $C(x) \leq C(-p)$, i.e., $f(x) \leq g(x) C(-p)$. Applying this to \cref{eq:dkl.chi2.2} and comparing the result with \cref{eq:dkl.chi2.3} we get \cref{eq:dkl.chi2.ub}.
  The second claim is shown similarly: we have $x \leq p$, which implies $C(x) \geq C(p)$, which implies \cref{eq:dkl.chi2.lb}.
\end{proof}

\begin{mylemma}
  \label{lm:dklz}
  Let $n$ be a positive integer, let
  $S \subset \mathbb{F}_2^n \setminus \{0\}$ be nonempty,
  and let $U$ be a random variable distributed uniformly on~$\mathbb{F}_2^n$.
  The following statements hold:

  \begin{enumerate}
    \item \label{item:dklz.any}
      For any distribution $P_Z$ on $\mathbb{F}_2^n$, we have
      \begin{equation}
        \label{eq:dklz.any}
        \DKL(P_Z \parallel P_U) = H(U) - H(Z)
        \leq \frac{\chi^2(P_Z \parallel P_U)}{\ln(2)}.
      \end{equation}

    \item \label{item:dklz.S}
      Suppose $Z$ is defined by the following two-stage sampling:
      \begin{enumerate}
        \item Pick $b$ uniformly at random from $S$.
        \item Given $b$, pick $z$ uniformly at random from
          $\{b\}^\perp  = \{\, z \in \mathbb{F}_2^n : b \cdot z = 0 \}$.
      \end{enumerate}
      Let $P_Z$ be the resulting distribution of $Z$.  Then
      \begin{equation}
        \label{eq:dklz.S}
        \chi^2(P_Z \parallel P_U) = \frac{1}{\abs{S}}.
      \end{equation}

    \item \label{item:dklz.PB}
      In the same setup as above, let $b$ instead be drawn according to
      some distribution $P_B$ supported on $S$.  That is,
      $\Pr[b = b_0] = P_B(b_0)$ for $b_0 \in S$,
      and given $b$, pick $z$ uniformly at random from
      $\{b\}^\perp$.  Then
      \begin{equation}
        \label{eq:dklz.PB}
        \chi^2(P_Z \parallel P_U) = \sum_{b \in S} P_B(b)^2.
      \end{equation}
  \end{enumerate}
  \smallskip \noindent \textbf{Note:}
  From numerical evidence, the bound in
  \cref{item:dklz.any} on $\DKL(P_Z \parallel P_U)$
  appears tight \emph{without} the factor $\ln(2)$
  for the specific distributions of parts~\ref{item:dklz.S} and~\ref{item:dklz.PB}.
  However, our current proof requires the $\ln(2)$ factor.
\end{mylemma}

\begin{proof}
  First, note that $D_{\text{KL}}(P_Z \parallel P_U) = H(U) - H(Z)$ holds for any distribution $P_Z$:
  by definition we have
  \bes
  \begin{align}
    &D_{\text{KL}}(P_Z \parallel P_U) = \sum_{z} P_Z(z) \log_2\left(\frac{P_Z(z)}{P_U(z)}\right) \\
    &\qquad = -H(P_Z) - \sum_{z} P_Z(z) \log_2(P_U(z)) \\
    &\qquad = H(P_U) - H(P_Z) .
  \end{align}
  \ees
  Further, the inequality
  \begin{equation}
    D_{\text{KL}}(P_Z \parallel P_U) \leq \chi^2(P_Z \parallel P_U) / \ln(2)
  \end{equation}
  holds for any two distributions $P_Z$ and $P_U$ and follows from the inequality $\ln(x) \leq x - 1$ for $x \geq 0$
  applied to $x = P_Z(z) / P_U(z)$. It also follows from \cref{lm:dkl.vs.chi2}, \cref{eq:dkl.chi2.ub} with $p = 1$.

  Note that part 2 of the lemma is a special case of part 3 with $P_B(b) = 1 / \abs{S}$.

  It remains to show that $\chi^2(P_Z \parallel P_U) = \sum_{b\in S} P_B(b)^2$.
  For $z \in \mathbb{F}_2^n$ let $M_z = \sum_{b \in S} P_B(b) (-1)^{b \cdot z}$.
  Then $P_Z(z) = 2^{-n} (1 + M_z)$ and
  \begin{multline}
    \chi^2(P_Z \parallel P_U) = \sum_{z} P_U(z) \left(\frac{P_Z(z)}{P_U(z)} - 1\right)^2 \\
    = 2^{-n} \sum_{z} \left(\frac{2^{-n} (1 + M_z)}{2^{-n}} - 1\right)^2
    = 2^{-n} \sum_z M_z^2 \\
    = 2^{-n} \sum_z \sum_{b,b'\in S} P_B(b) P_B(b') (-1)^{b\cdot z} (-1)^{b'\cdot z} \\
    = 2^{-n} \sum_{b,b' \in S} P_B(b) P_B(b') \sum_z (-1)^{(b+b')\cdot z} = \sum_{b\in S} P_B(b)^2.
  \end{multline}
  In the last step we observed that $\sum_z (-1)^{c \cdot z} = 2^n$ if $c = 0$ and $0$ otherwise,
  where $c = b + b'$.
\end{proof}

\begin{mylemma}
  \label{lm:dkl.zp}
  Let $n$, $S$, $P_B$ be as in \cref{lm:dklz}, and $p \in [0, 1]$.
  Let $K = \sum_{b\in S} P_B(b)^2$ as in \cref{eq:dklz.PB}.
  Let $B$ be a random variable distributed according to $P_B$
  on $S$ and $Z_p$ be a random variable defined as follows: given $B = b$, with probability $p$ we have $Z_p$ equal to
  an element of $\{b\}^\perp$ picked uniformly at random. Otherwise, $Z_p$ is an element of $\mathbb{F}_2^n$
  picked uniformly at random.
  Then
  \begin{equation}
    \label{eq:I.B.Zp}
    I(B; Z_p) = 1 - H((1+p)/2) - r ,
  \end{equation}
  where
  $H(p') = -p' \log_2(p') - (1-p') \log_2(1-p')$ is the Shannon entropy of a Bernoulli distribution and
  the remainder $r$ is equal to $D_{\text{KL}}(P_{Z_p} \parallel P_U)$ and satisfies
  \begin{equation}
    \label{eq:dkl.zp.r-bound}
    r \in \left[\frac{K p^2 C(p)}{\ln(2)}, \frac{K p^2 C(-p)}{\ln(2)}\right], \quad r = \frac{K p^2}{2\ln(2)} (1 + O(p)).
  \end{equation}
  The remainder $r$ may depend on $p, n, S, P_B$ and $C(x)$ is the same as in \cref{lm:dkl.vs.chi2}.
\end{mylemma}

We are also interested in the limit where $p \to 0$ and $K \in (0, 1)$ might have arbitrary
behavior (i.e., it might or might not go to zero). In this limit \cref{eq:I.B.Zp} is
\begin{equation}
  \label{eq:I.B.Zp.asymptotic}
  I(B; Z_p) = \frac{(1 - K) p^2}{2 \ln(2)} + O(p^4 + K p^3).
\end{equation}

\begin{proof}
  We let $r = D_{\text{KL}}(P_{Z_p} \parallel P_U)$.
  From the definition of $Z_p$, we compute
  \begin{equation}
    \label{eq:HZqB}
    H(Z_p | B) = n - 1 + H(Z_p \in \{b\}^\perp| B) = n - 1 + H((1+p)/2).
  \end{equation}
  Also, we know that $H(U) = n$ and $r = D_{\text{KL}}(P_{Z_p} \parallel P_U) = H(U) - H(Z_p)$, so $H(Z_p) = n - r$.
  Using these we compute
  \bes
  \begin{align}
    I(B; Z_p) &= H(Z_p) - H(Z_p | B) =\\
    &= n - r - (n - 1 + H((1+p)/2)) \\
    &= 1 - H((1+p)/2) - r.
  \end{align}
  \ees
  To estimate the remainder $r$ we
  let $Z$ be the same as in \cref{lm:dklz} and
  note that from definition of $Z_p$ and from \cref{lm:dklz} we have
  \begin{equation}
    \chi^2(P_{Z_p} \parallel P_U) = p^2 \chi^2(P_{Z} \parallel P_U) = K p^2.
  \end{equation}
  Thus, from the first statement of \cref{lm:dkl.vs.chi2} [i.e., from \cref{eq:dkl.chi2.ub}] we have
  \begin{equation}
    r \leq \frac{C(-p)}{\ln(2)} \chi^2(P_{Z_p} \parallel P_U) = \frac{K p^2 C(-p)}{\ln(2)}.
  \end{equation}
  The lower bound on $r$ similarly follows from the second statement in \cref{lm:dkl.vs.chi2}, i.e., \cref{eq:dkl.chi2.lb}. The second part of \cref{eq:dkl.zp.r-bound} follows from the first and the series expansion \cref{eq:dkl.chi2.cp-series} of $C(p)$.
\end{proof}

We now consider the whole game.
\begin{mylemma}
  \label{lm:Qmax}
  Suppose an agent is trying to solve Simon's problem with hidden bitstring $b$ having a prior distribution $P_B$
  on $S$ with entropy $H(P_B)$. Let $p \in (0, 1]$.
  Assume the agent has access to a stream of bitstrings $z$ generated as in \cref{lm:dkl.zp}
  until a certain stopping criterion is met. Further, assume that the agent always stops when
  $\sum_{b \in S} \Pr(B = b | Z_1 = z_1, \dots, Z_m = z_m)^2 \geq K_1$ for some $K_1 \in (0, 1)$.
  Then the expected number of queries used by the agent satisfies
  \begin{equation}
    \label{eq:Qmax}
    \mathbb{E}(Q) \leq \frac{H(P_B)}{1 - H[(1 + p)/2] - K_1 p^2 C(-p) / \ln(2)}.
  \end{equation}

\end{mylemma}

Note that asymptotically [in the same limit as in \cref{eq:I.B.Zp.asymptotic}] \cref{eq:Qmax} becomes
\begin{equation}
  \label{eq:Qmax.asymptotic}
  \mathbb{E}(Q) \leq 2 \ln(2) p^{-2} H(P_B) \frac{1}{1 - K_1 + O(K_1 p)}.
\end{equation}
Here, we replaced $O(p^2 + K_1 p)$ with $O(K_1 p)$ in the denominator because the $O(p^2)$ term is always non-negative and can thus be safely discarded while preserving the inequality.

We also note that an example of a stopping criterion is to stop whenever $\max_b \Pr(B = b | Z_1 = z_1, \dots, Z_m = z_m) \geq 1/2$.
This criterion has expected probability of guessing $b$ correctly at least $1/2$ and satisfies the condition of the lemma
with $K_1 = 1/2$.

\begin{proof}
  Let $I_{{\min}}$ be the minimal amount of information the agent learns about $b$
  from a single query assuming the agent knows some prior distribution $P_{B,m-1}$ of $B$
  and still decides to execute another query. From the condition on the stopping criterion
  assumed by this lemma and from \cref{lm:dkl.zp} we have
  \begin{equation}
    I_{{\min}} \geq 1 - H[(1 + p)/2] - K_1 p^2 C(-p) / \ln(2),
  \end{equation}
  and note that the RHS is the denominator of \cref{eq:Qmax}.
  Denote the number of queries used by the agent by $Q$, and let $Q_m = \min(Q, m)$.
  Denote by $H_m$ the random variable equal to the entropy
  of $P_{B, m}(b) = \Pr(B = b | Z_1 = z_1, \dots, Z_{Q_m} = z_{Q_m})$. By definition $Q_0 = 0$ and $H_0 = H(P_B)$.
  It is sufficient to show that $A_m = \mathbb{E}(Q_m I_{{\min}} + H_m)$ is a non-increasing function of $m$,
  i.e., that for every $m \geq 1$ we have $A_m \leq A_{m-1}$.
  This is trivially true if the expectation is taken only over events with $Q_m < m$,
  since in this case $Q_m = Q_{m-1}$ and $H_m = H_{m-1}$.
  Consider a case where $Q_m = m$ and fix $P_{B, m-1}$.
  Then
  \begin{multline}
    \mathbb{E}[(Q_m - Q_{m-1}) I_{{\min}} + (H_m - H_{m-1}) | Q_m = m, P_{B, m-1}] \\
    = I_{{\min}} - I(B; Z_p | P_{B, m-1}) \leq 0.
  \end{multline}
  Taking the expectation over $P_{B, m-1}$ we obtain $A_m \leq A_{m-1}$.
\end{proof}

\section{Device specifications}
\label{app:device}

We performed our experiments using the $127$-qubit devices Sherbrooke (ibm\_sherbrooke) and Brisbane (ibm\_brisbane), as well as the $27$-qubit devices Cairo (ibm\_cairo) and Kolkata (ibmq\_kolkata),
\cref{tab:127devices-limit15,tab:127devices,tab:27devices} give the device specifications on the day of the experiments.

We used $100,000$ shots per experiment parameterized by the Hamming weight of the oracle $b$ for each device.

\begin{table}[h!]
\centering
\begin{tabular}{|l|ccc|ccc|}
\hline
\multicolumn{1}{|c|}{} & \multicolumn{3}{c|}{Brisbane} & \multicolumn{3}{c|}{Sherbrooke} \\ \hline
\multicolumn{1}{|c|}{} & \multicolumn{1}{c|}{Min} & \multicolumn{1}{c|}{Mean} & Max & \multicolumn{1}{c|}{Min} & \multicolumn{1}{c|}{Mean} & Max \\ \hline
$T_1 (\mu s)$ & \multicolumn{1}{c|}{27.59} & \multicolumn{1}{c|}{290.03} & 519.75 & \multicolumn{1}{c|}{45.60} & \multicolumn{1}{c|}{233.38} & 404.29 \\ \hline
$T_2 (\mu s)$ & \multicolumn{1}{c|}{11.21} & \multicolumn{1}{c|}{173.18} & 439.99 & \multicolumn{1}{c|}{9.45} & \multicolumn{1}{c|}{144.76} & 366.50 \\ \hline
1QG Error (\%) & \multicolumn{1}{c|}{0.026} & \multicolumn{1}{c|}{0.080} & 0.424 & \multicolumn{1}{c|}{0.028} & \multicolumn{1}{c|}{0.118} & 3.324 \\ \hline
2QG Error (\%) & \multicolumn{1}{c|}{0.278} & \multicolumn{1}{c|}{0.376} & 100 & \multicolumn{1}{c|}{0.373} & \multicolumn{1}{c|}{2.437} & 100 \\ \hline
1QG Duration $(\mu s)$ & \multicolumn{1}{c|}{0.171} & \multicolumn{1}{c|}{0.171} & 0.171 & \multicolumn{1}{c|}{0.180} & \multicolumn{1}{c|}{0.180} & 0.180 \\ \hline
2QG Duration $(\mu s)$ & \multicolumn{1}{c|}{0.448} & \multicolumn{1}{c|}{0.539} & 0.882 & \multicolumn{1}{c|}{0.660} & \multicolumn{1}{c|}{0.665} & 0.860 \\ \hline
RO Error (\%) & \multicolumn{1}{c|}{0.180} & \multicolumn{1}{c|}{1.813} & 27.670 & \multicolumn{1}{c|}{0.430} & \multicolumn{1}{c|}{2.063} & 11.90 \\ \hline
RO Duration $(\mu s)$ & \multicolumn{1}{c|}{1.244} & \multicolumn{1}{c|}{1.244} & 1.244 & \multicolumn{1}{c|}{4.000} & \multicolumn{1}{c|}{4.000} & 4.000 \\ \hline
\end{tabular}
\caption{Device specifications for the 127-qubit devices Brisbane and Sherbrooke on May 5, 2023 (Experiment \#1). 1QG and 2QG denote 1-qubit gate and 2-qubit gate, respectively. RO denotes readout. The numbers shown are averages over the qubits used in our experiments.}
\label{tab:127devices-limit15}
\end{table}

\begin{table}[h!]
\centering
\begin{tabular}{|l|ccc|ccc|}
\hline
\multicolumn{1}{|c|}{} & \multicolumn{3}{c|}{Brisbane} & \multicolumn{3}{c|}{Sherbrooke} \\ \hline
\multicolumn{1}{|c|}{} & \multicolumn{1}{c|}{Min} & \multicolumn{1}{c|}{Mean} & Max & \multicolumn{1}{c|}{Min} & \multicolumn{1}{c|}{Mean} & Max \\ \hline
$T_1 (\mu s)$ & \multicolumn{1}{c|}{47.36} & \multicolumn{1}{c|}{238.44} & 386.83 & \multicolumn{1}{c|}{17.08} & \multicolumn{1}{c|}{257.73} & 466.35 \\ \hline
$T_2 (\mu s)$ & \multicolumn{1}{c|}{16.14} & \multicolumn{1}{c|}{156.72} & 446.73 & \multicolumn{1}{c|}{14.55} & \multicolumn{1}{c|}{188.32} & 571.77 \\ \hline
1QG Error (\%) & \multicolumn{1}{c|}{0.021} & \multicolumn{1}{c|}{0.098} & 0.974 & \multicolumn{1}{c|}{0.025} & \multicolumn{1}{c|}{0.116} & 1.856 \\ \hline
2QG Error (\%) & \multicolumn{1}{c|}{0.355} & \multicolumn{1}{c|}{0.901} & 2.945 & \multicolumn{1}{c|}{0.283} & \multicolumn{1}{c|}{5.122} & 100 \\ \hline
1QG Duration $(\mu s)$ & \multicolumn{1}{c|}{0.180} & \multicolumn{1}{c|}{0.180} & 0.180 & \multicolumn{1}{c|}{0.171} & \multicolumn{1}{c|}{0.171} & 0.171 \\ \hline
2QG Duration $(\mu s)$ & \multicolumn{1}{c|}{0.660} & \multicolumn{1}{c|}{0.666} & 1.100 & \multicolumn{1}{c|}{0.341} & \multicolumn{1}{c|}{0.541} & 0.882 \\ \hline
RO Error (\%) & \multicolumn{1}{c|}{0.470} & \multicolumn{1}{c|}{2.350} & 27.25 & \multicolumn{1}{c|}{0.300} & \multicolumn{1}{c|}{2.784} & 32.83 \\ \hline
RO Duration $(\mu s)$ & \multicolumn{1}{c|}{4.000} & \multicolumn{1}{c|}{4.000} & 4.000 & \multicolumn{1}{c|}{1.244} & \multicolumn{1}{c|}{1.244} & 1.244 \\ \hline
\end{tabular}
\caption{Device specifications for the 127-qubit devices Brisbane and Sherbrooke on May 17, 2024 (Experiment \#2; the main experiment reported in main text). Otherwise as in \cref{tab:127devices-limit15}.}
\label{tab:127devices}
\end{table}

\begin{table}[h!]
\centering
\begin{tabular}{|l|ccc|ccc|}
\hline
\multicolumn{1}{|c|}{} & \multicolumn{3}{c|}{Cairo} & \multicolumn{3}{c|}{Kolkata} \\ \hline
\multicolumn{1}{|c|}{} & \multicolumn{1}{c|}{Min} & \multicolumn{1}{c|}{Mean} & Max & \multicolumn{1}{c|}{Min} & \multicolumn{1}{c|}{Mean} & Max \\ \hline
$T_1 (\mu s)$ & \multicolumn{1}{c|}{63.79} & \multicolumn{1}{c|}{116.39} & 195.17 & \multicolumn{1}{c|}{37.13} & \multicolumn{1}{c|}{121.03} & 235.27 \\ \hline
$T_2 (\mu s)$ & \multicolumn{1}{c|}{15.59} & \multicolumn{1}{c|}{97.63} & 212.26 & \multicolumn{1}{c|}{16.35} & \multicolumn{1}{c|}{74.51} & 192.35 \\ \hline
1QG Error (\%) & \multicolumn{1}{c|}{0.043} & \multicolumn{1}{c|}{0.093} & 0.356 & \multicolumn{1}{c|}{0.047} & \multicolumn{1}{c|}{0.090} & 0.204 \\ \hline
2QG Error (\%) & \multicolumn{1}{c|}{0.491} & \multicolumn{1}{c|}{4.680} & 100 & \multicolumn{1}{c|}{0.591} & \multicolumn{1}{c|}{11.98} & 100 \\ \hline
1QG Duration $(\mu s)$ & \multicolumn{1}{c|}{0.075} & \multicolumn{1}{c|}{0.075} & 0.075 & \multicolumn{1}{c|}{0.107} & \multicolumn{1}{c|}{0.107} & 0.107 \\ \hline
2QG Duration $(\mu s)$ & \multicolumn{1}{c|}{0.160} & \multicolumn{1}{c|}{0.351} & 0.946 & \multicolumn{1}{c|}{0.341} & \multicolumn{1}{c|}{0.458} & 0.626 \\ \hline
RO Error (\%) & \multicolumn{1}{c|}{0.750} & \multicolumn{1}{c|}{1.679} & 6.120 & \multicolumn{1}{c|}{0.340} & \multicolumn{1}{c|}{2.607} & 11.76 \\ \hline
RO Duration $(\mu s)$ & \multicolumn{1}{c|}{0.732} & \multicolumn{1}{c|}{0.732} & 0.732 & \multicolumn{1}{c|}{0.640} & \multicolumn{1}{c|}{0.640} & 0.640 \\ \hline
\end{tabular}
\caption{Device specifications for the 27-qubit devices Cairo and Kolkata on April 20 and 22, 2023, respectively. Otherwise as in \cref{tab:127devices-limit15}.}
\label{tab:27devices}
\end{table}

\section{DD performance}
\label{app:dd-rank}

This section presents the performance evaluation of all DD sequences investigated in the experiments on Sherbrooke and Brisbane, specifically the nine sequences from the set $\mathcal{D}$ as introduced in \cref{sec:expt}. The performance assessed by the NTS metric is illustrated in \cref{fig:dd-perf}. The NTS values are presented for each DD sequence across the problem size $n$, in the range $[2,29]$. Sequences colored in blue ($\times$) indicate optimal performers identified on their corresponding machines at each problem size. The analysis reveals that XY4 exhibits the best performance among sequences on both Sherbrooke and Brisbane, consistently achieving the lowest NTS across a majority of $n$ values ranging from $5$ to $29$.

\begin{figure}
    \centering
    \includegraphics[width=0.49\textwidth]{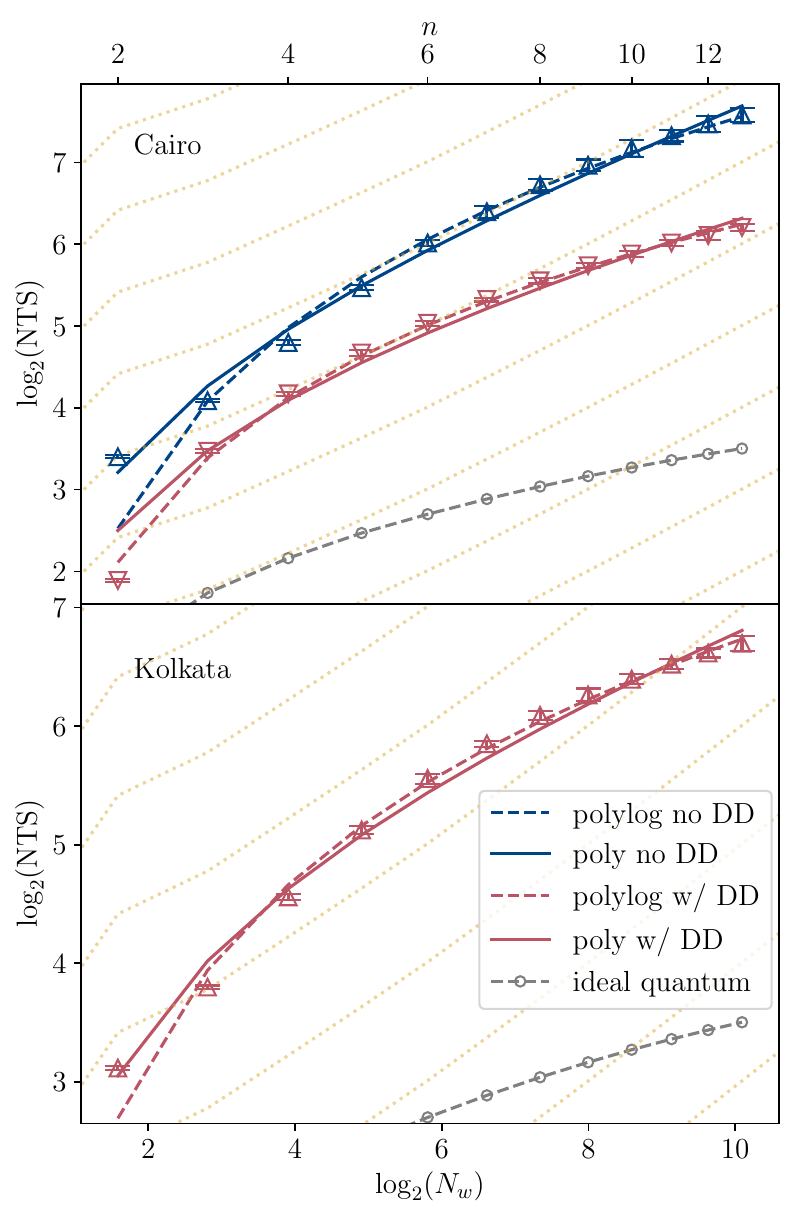}
    \caption{As in \cref{fig:main-plot} but for \wSimon{4}{13} run on Cairo (top) and Kolkata (bottom). The no-DD result is missing from the Kolkata plot since performance in this case was as poor as making a random guess.}
    \label{fig:27q-devices-speedup}
\end{figure}

\begin{figure}[t]
    \centering
    \includegraphics[width=0.49\textwidth]{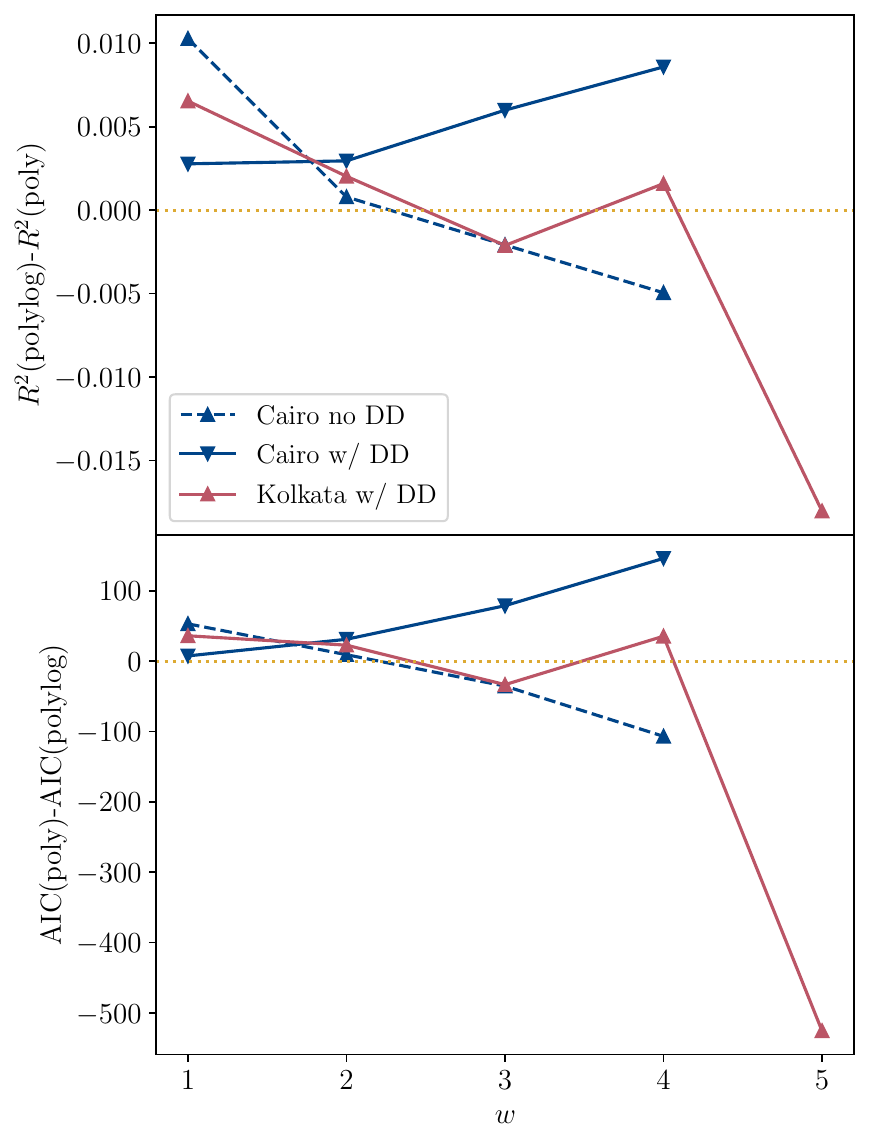}
    \caption{As in \cref{fig:mixed_R^2+AIC} but for \wSimon{w}{13} run on Cairo and Kolkata.}
    \label{fig:AIC-R2-27}
\end{figure}

\begin{figure}
    \centering
    \includegraphics[width=0.49\textwidth]{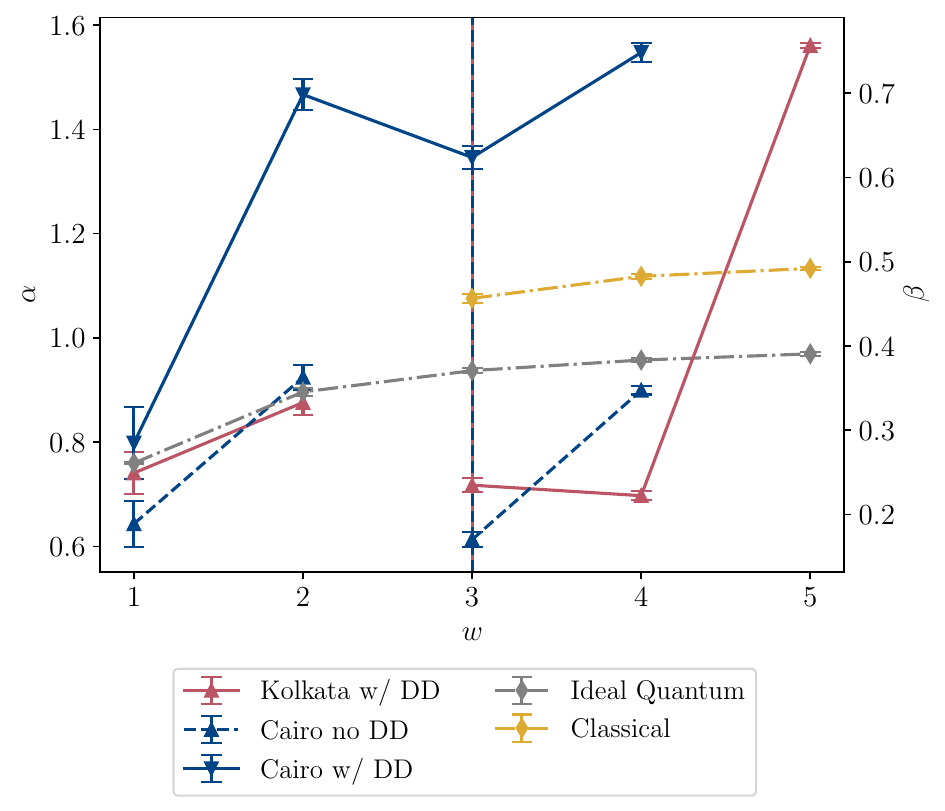}
    \caption{As in \cref{fig:allslopes} but for \wSimon{w}{13} run on Cairo and Kolkata.}
    \label{fig:27q-devices-scaling}
\end{figure}

\section{Results for 27-qubit devices: algorithmic quantum speedup at low Hamming weight}
\label{app:extra}

We performed the same set of experiments on two 27-qubit IBM Quantum devices, Cairo (ibm\_cairo) and Kolkata (ibmq\_kolkata), each with 100,000 shots. Both devices feature the heavy-hex lattice 
and the Falcon r5.11 processor. Their specifications at the time of our experiments are given in \cref{tab:27devices}. 

For these experiments, we set $n_{\max}=13$, the maximum possible value compatible with $27$ qubits. All the $27$-qubit experiments we report are without MEM: we tried MEM offline but the result was worse than without MEM. 

Since $27$-qubit devices are noisier than the $127$-qubit ones, we are restricted to the lower HW regime. For Cairo, we consider both DD and no-DD experiments, since the latter still perform better than random guessing. In the Kolkata case, the no-DD experiments underperform random guessing. We use a similar procedure for determining quantum speedup as described in \cref{sec:results} of the main text, except that we only consider the polylog and poly models in this section. A representative example is shown in \cref{fig:27q-devices-speedup}, illustrating the \wSimon{4}{13} case. The polylog model is the best fit in this case for both Cairo (with and without DD) and Kolkata (with DD).

As shown in \cref{fig:AIC-R2-27}, no model transition occurs for the with-DD  experiments on Cairo, and polylog consistently provides the best fit. Consequently, we declare an exponential speedup for $w\le 4$. Note that we exclude $w\ge 5$ from consideration in this particular case because there are only three data points when $w\ge 5$ [see \cref{app:all_results}], hence, any speedup does not extend to \wSimon{w}{13}.

In the no-DD experiments on Cairo, a model transition is observed at $w_t=3$, as depicted in \cref{fig:27q-devices-scaling} (dashed blue). At $w=3,4$, we observe that $\beta_Q < \beta_C$, indicating a polynomial speedup. Hence, we declare an exponential quantum speedup for $w\in[1,2]$ and a polynomial speedup for $w\in[3,4]$.

A similar pattern is observed for the with-DD experiments on Kolkata, with the same model transition point at $w_t=3$ in \cref{fig:27q-devices-scaling} (dashed blue). Likewise, we declare an exponential speedup for $w\in[1,2]$ and a polynomial speedup for $w\in[3,4]$. 

Our findings indicate that an exponential quantum speedup is observed at $w\le 4$ on the 27-qubit devices. This speedup is evident in experiments both with DD and without DD on Cairo and only in experiments with DD on Kolkata. The no-DD experiments on Kolkata are excluded because their performance is as poor as performing a random guess.

\section{Bootstrapping}
\label{app:bootstrap}

We conducted a robust statistical analysis employing a $100,000$-sample bootstrapping technique on the NTS values subjected to a $10$-fold cross-validation. Each sample, consisting of $9$ uniformly selected data points from the cross-validation set, underwent resampling to derive the mean and standard deviation. The resulting bootstrapped means and standard deviations are presented as individual data points in \cref{fig:main-plot,fig:27q-devices-speedup}, 
where the mean values are plotted as data points, with error bars extending $1\sigma$ in both directions. These statistics were then used to fit for the parameters $\alpha, \beta$, and $\gamma$ using Mathematica, taking into account the bootstrapped mean as the $\NTS$ value and one standard deviation as its error. In \cref{fig:allslopes,fig:27q-devices-scaling}, the parameters $\alpha$ and $\beta$ were fitted using the bootstrapped data in the same way as in \cref{fig:main-plot} but with different $w$ values. The standard deviation from the fitting is again reported as $1\sigma$ in each direction from each data point.

\begin{figure}
    \centering
    \includegraphics[width=0.47\textwidth]{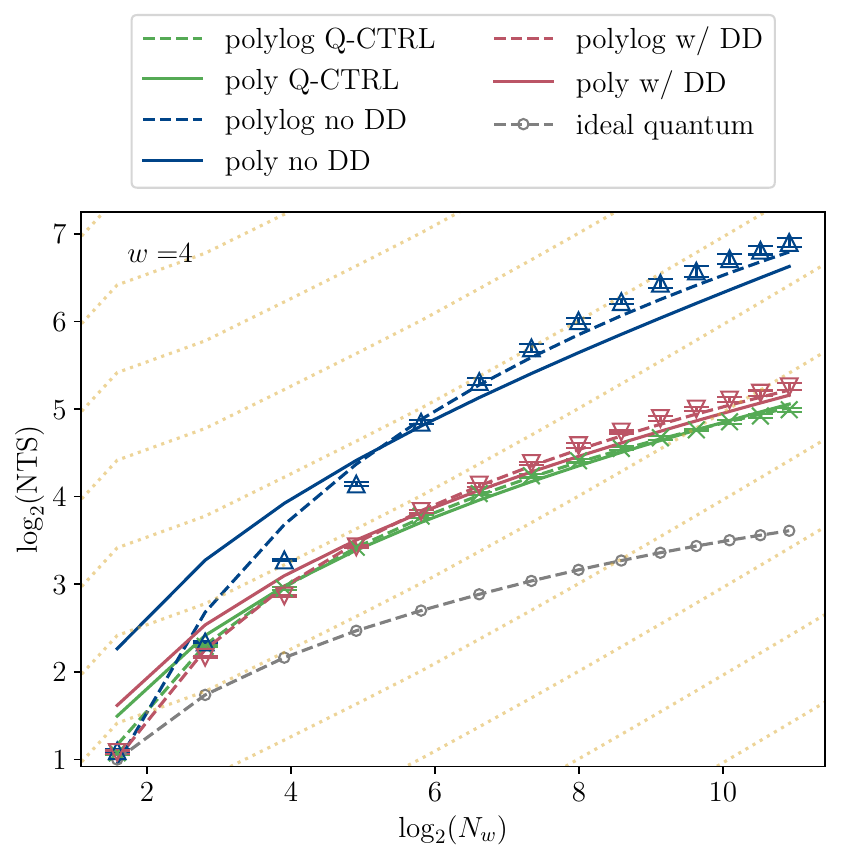}
    \caption{Brisbane results for \wSimon{4}{15}, exactly as in \cref{fig:main-plot} but with Q-CTRL data added (green). The latter exhibits somewhat better scaling relative to the polylog model than our manually optimized results with DD.}
    \label{fig:qctrl-main}
\end{figure}

\begin{figure}
    \centering
    \includegraphics[width=.47\textwidth]{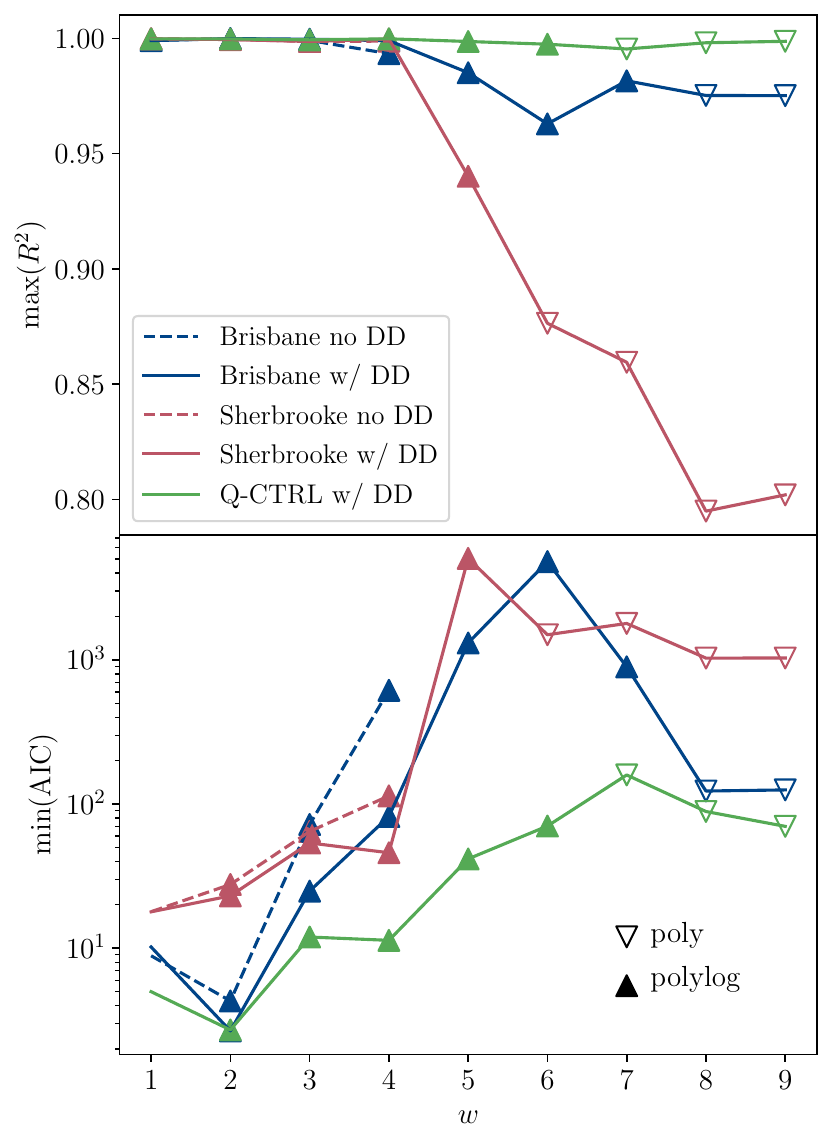}
    \caption{As in \cref{fig:mixed_R^2+AIC} but with results from Q-CTRL on Brisbane added.}
    \label{fig:mixed_R^2+AIC_qctrl}
\end{figure}

\begin{figure}
    \centering
    \includegraphics[width=0.49\textwidth]{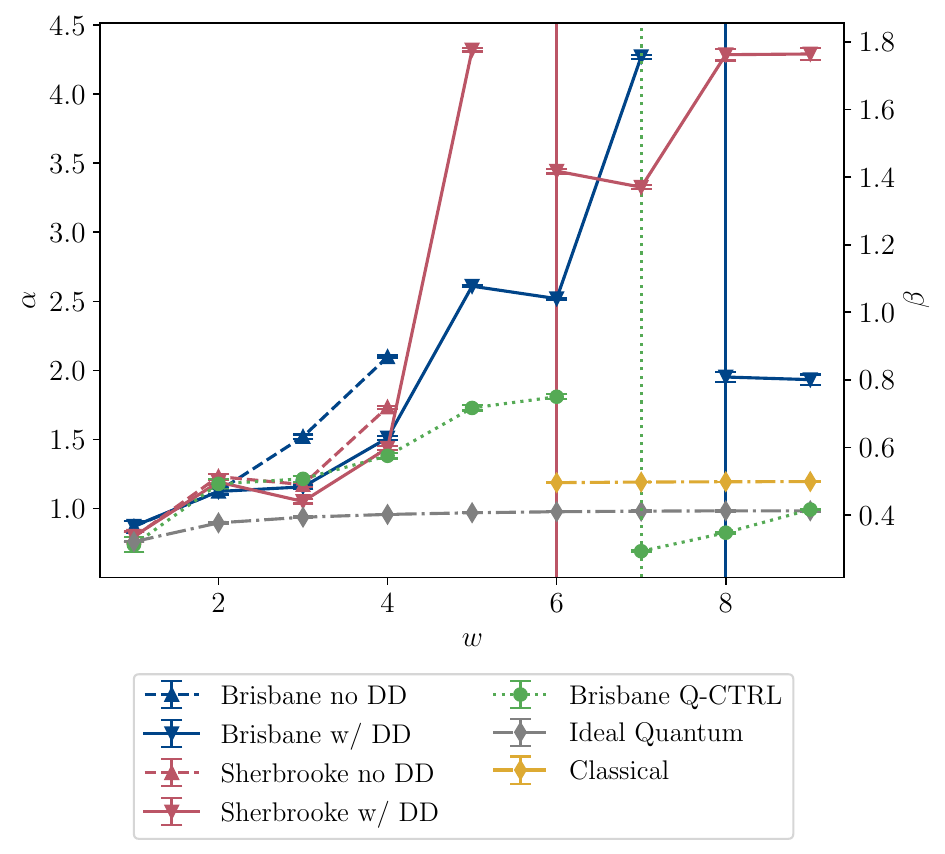}
    \caption{As in \cref{fig:allslopes}, but including the Brisbane results for Q-CTRL, for which the polylog model transitions to the poly model at $w_t=7$ (dotted green vertical line). The scaling exponent $\alpha$ in the exponential speedup regime of $w\in [1,4]$ is similar for Sherbrooke with DD, Brisbane with DD, and Q-CTRL. The latter's scaling exponent ($\alpha$) is significantly smaller for $w\in[5,6]$, but transitions to a polynomial speedup for $w\in[7,9]$, whereas Brisbane with DD retains an exponential speedup up to $w=7$.}
    \label{fig:qctrl-params}
\end{figure}

\begin{figure}
    \centering
    \includegraphics[width=0.49\textwidth]{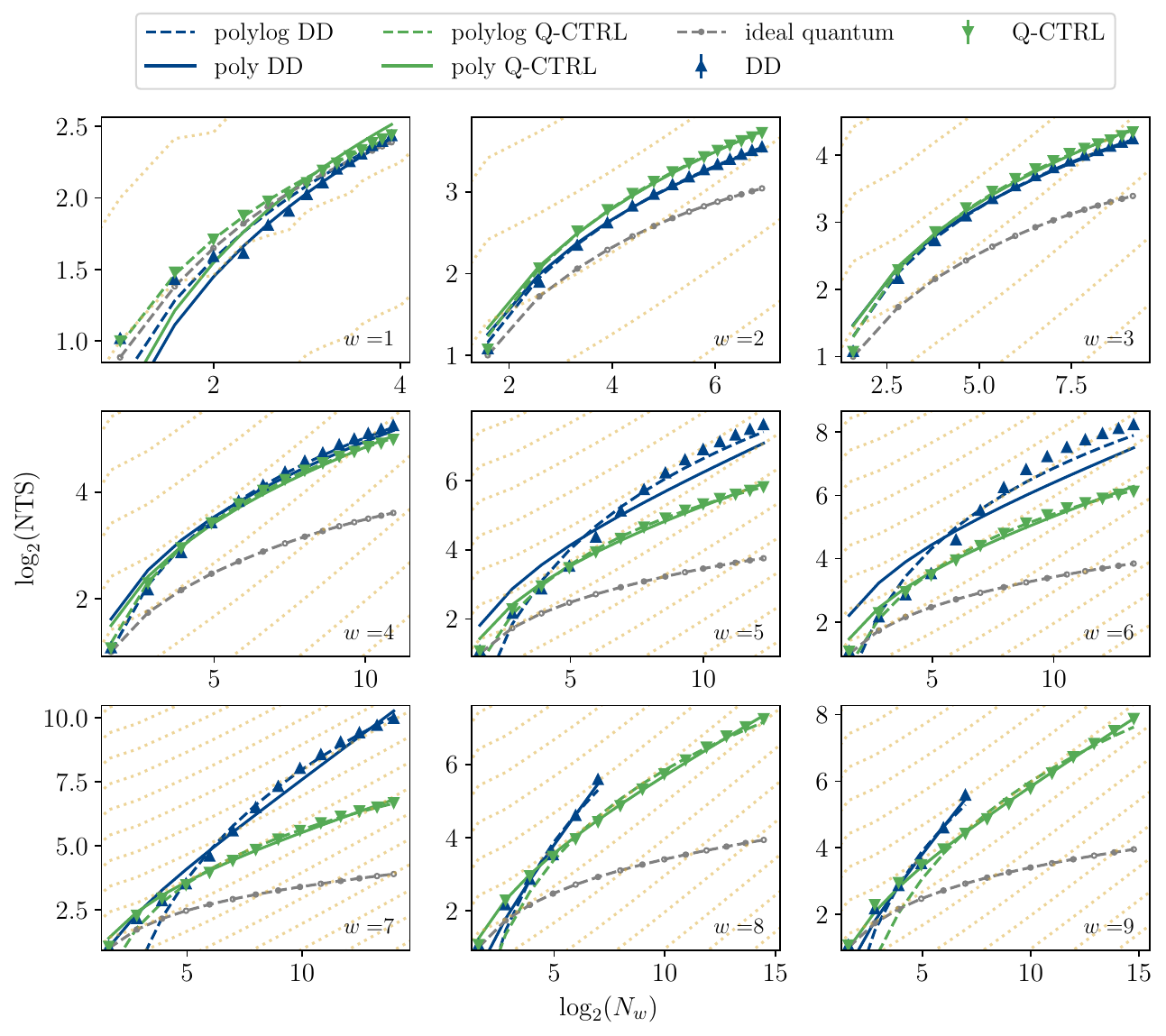}
    \caption{NTS scaling plots of Q-CTRL \textit{vs} manual DD on Brisbane from HW $w=1$ to $w=9$.}
    \label{fig:grid-qctrl}
\end{figure}

\section{Statistical comparison of the poly and polylog models}
\label{app:stat-analysis}

In the main text, we focused on the AIC and $R^2$ as statistical measures for deciding between the polylog and poly models. This section presents a more comprehensive statistical analysis comparing the performance of the poly and polylog models, across the four different cases and HW values. 
For each with-DD case, we use the Hamming weight $w\in\{1,\dots,9\}$ as a parameter, evaluating which model provides a better fit based on various statistical measures. We limit $w\in\{1,2,3,4\}$ for the without-DD cases because there are only three data points for $w>4$. 

The statistical measures include the $p$-value, $t$-statistic, Akaike Information Criterion (AIC), corrected AIC (AICc), Bayesian Information Criterion (BIC), R-squared ($R^2$), and adjusted R-squared ($\text{Adj }R^2$), which are reported without limiting the parameter range, i.e., for $w\in\{1,\dots,9\}$. These seven statistical measures are all computed using Mathematica's NonlinearModelFit while accounting for the confidence intervals we computed using bootstrapping, as documented in \cref{app:bootstrap}. By combining all seven statistical measures we obtain a comprehensive statistical verdict. 

We also compute the probability that each model is better based on Akaike weights.

\subsection{Tables of statistical measures}
In the tables below we abbreviate the four cases as:
\begin{itemize}
    \item Br w/ : Brisbane with DD
    \item Br no : Brisbane without DD
    \item Sh w/ : Sherbrooke with DD
    \item Sh no : Sherbrooke without DD
\end{itemize}
We present results for both the polylog and poly models, which are explicitly labeled in the table for $p$-values, \cref{tab:p-vals}. The remaining tables display the results for the other statistical measures and follow the same format of polylog and poly having white and gray backgrounds, respectively. 

The tables are: \cref{tab:t-stat} ($t$-statistic), \cref{tab:AIC} (Akaike Information Criterion), \cref{tab:AICc} (corrected Akaike Information Criterion), \cref{tab:BIC} (Akaike Information Criterion), \cref{tab:R^2} ($R^2$), and \cref{tab:aR^2} (adjusted $R^2$).

\begin{table*}[t]
\centering
\caption{$p$-values (lower values in boldface)}
\begin{tabular}{|c|c|c|c|c|c|c|c|c|c|c|}
\hline
 $w$ &  & 1 & 2 & 3 & 4 & 5 & 6 & 7 & 8 & 9 \\
\hline
\multirow{2}{*}{Br no} & polylog & $4.5 {e}{-18}$ & $\mathbf{9.2 {e}{-27}}$ & $\mathbf{9.7 {e}{-35}}$ & $\mathbf{3.6 {e}{-43}}$ & &&&& \\
\cline{2-11}
 & poly & $\mathbf{9.6 {e}{-36}}$ & $2.2 {e}{-05}$ & $2.6 {e}{-24}$ & $2.6 {e}{-39}$ &&&&& \\
\hline
\multirow{2}{*}{Sh no} & polylog & $4.7 {e}{-18}$ & $\mathbf{1.3 {e}{-28}}$ & $\mathbf{4.6 {e}{-33}}$ & $\mathbf{1.0 {e}{-39}}$ & &&&& \\
\cline{2-11}
 & poly & $\mathbf{2.6 {e}{-36}}$ & $5.3 {e}{-11}$ & $7.2 {e}{-12}$ & $1.9 {e}{-32}$ & &&&& \\
\hline
\multirow{2}{*}{Br w/} & polylog & $1.1 {e}{-18}$ & $\mathbf{4.0 {e}{-27}}$ & $\mathbf{1.6 {e}{-31}}$ & $\mathbf{2.6 {e}{-37}}$ & $\mathbf{6.0 {e}{-47}}$ & $\mathbf{8.8 {e}{-50}}$ & $2.0 {e}{-24}$ & $\mathbf{6.8 {e}{-07}}$ & $\mathbf{6.3 {e}{-07}}$ \\
\cline{2-11}
 & poly & $\mathbf{2.5 {e}{-36}}$ & $1.7 {e}{-05}$ & $6.3 {e}{-10}$ & $1.0 {e}{-26}$ & $2.4 {e}{-45}$ & $8.3 {e}{-48}$ & $\mathbf{4.1 {e}{-25}}$ & $8.2 {e}{-07}$ & $7.8 {e}{-07}$ \\
\hline
\multirow{2}{*}{Sh w/} & polylog & $9.5 {e}{-18}$ & $\mathbf{4.6 {e}{-29}}$ & $\mathbf{3.4 {e}{-31}}$ & $\mathbf{1.7 {e}{-37}}$ & $1.1 {e}{-34}$ & $2.0 {e}{-12}$ & $1.5 {e}{-12}$ & $6.2 {e}{-08}$ & $5.8 {e}{-08}$ \\
\cline{2-11}
 & poly & $\mathbf{5.8 {e}{-36}}$ & $9.9 {e}{-10}$ & $7.0 {e}{-02}$ & $1.0 {e}{-25}$ & $\mathbf{1.5 {e}{-35}}$ & $\mathbf{9.9 {e}{-13}}$ & $\mathbf{8.1 {e}{-13}}$ & $\mathbf{5.4 {e}{-08}}$ & $\mathbf{4.9 {e}{-08}}$ \\
\hline
\end{tabular}
\label{tab:p-vals}
\end{table*}

\begin{table}[h!]
\centering
\caption{$t$-statistic (higher is better). Background color: polylog = white, poly = gray.}
\begin{tabular}{|c|c|c|c|c|c|c|c|c|c|}
\hline
 $w$& 1 & 2 & 3 & 4 & 5 & 6 & 7 & 8 & 9 \\
\hline
\multirow{2}{*}{Br no} & 21.5 & \textbf{47.4} & \textbf{96.4} & \textbf{204} & &&&& \\
 & \cellcolor{gray!20}\textbf{105} & \cellcolor{gray!20}5.16 & \cellcolor{gray!20}38.0 & \cellcolor{gray!20}145 & &&&& \\
\hline
\multirow{2}{*}{Sh no} & 21.4 & \textbf{55.9} & \textbf{83.0} & \textbf{150} & & & & & \\
 & \cellcolor{gray!20}\textbf{111} & \cellcolor{gray!20}10.7 & \cellcolor{gray!20}11.7 & \cellcolor{gray!20}78.6 & & & & & \\
\hline
\multirow{2}{*}{Br w/} & 22.7 & \textbf{48.9} & \textbf{72.3} & \textbf{121} & \textbf{284} & \textbf{366} & 288 & \textbf{54.5} & \textbf{55.5} \\
 & \cellcolor{gray!20}\textbf{111} & \cellcolor{gray!20}5.25 & \cellcolor{gray!20}9.48 & \cellcolor{gray!20}47.2 & \cellcolor{gray!20}247 & \cellcolor{gray!20}307 & \cellcolor{gray!20}\textbf{329} & \cellcolor{gray!20}52.0 & \cellcolor{gray!20}52.6 \\
\hline
\multirow{2}{*}{Sh w/} & 20.8 & \textbf{58.2} & \textbf{70.3} & \textbf{123} & 373 & 180 & 188 & 99.0 & 101 \\
 & \cellcolor{gray!20}\textbf{107} & \cellcolor{gray!20}9.28 & \cellcolor{gray!20}1.89 & \cellcolor{gray!20}43.2 & \cellcolor{gray!20}\textbf{418} & \cellcolor{gray!20}\textbf{202} & \cellcolor{gray!20}\textbf{209} & \cellcolor{gray!20}\textbf{103} & \cellcolor{gray!20}\textbf{105} \\
\hline
\end{tabular}
\label{tab:t-stat}
\end{table}

\begin{table}[h!]
\centering
\caption{Akaike Information Criterion (AIC) values (lower is better).}
\begin{tabular}{|c|c|c|c|c|c|c|c|c|c|}
\hline
 $w$ & 1 & 2 & 3 & 4 & 5 & 6 & 7 & 8 & 9 \\
\hline
\multirow{2}{*}{Br no} & \textbf{-8.9} & \textbf{4.3} & \textbf{72.2} & \textbf{615} & &&&& \\
 & \cellcolor{gray!20}0.5 & \cellcolor{gray!20}8.4 & \cellcolor{gray!20}228 & \cellcolor{gray!20}2859 & &&&&\\
\hline
\multirow{2}{*}{Sh no} & \textbf{-17.9} & \textbf{27.7} & \textbf{64.6} & \textbf{114} & &&&& \\
 & \cellcolor{gray!20}9.4 & \cellcolor{gray!20}47.0 & \cellcolor{gray!20}107 & \cellcolor{gray!20}543 & &&&&\\
\hline
\multirow{2}{*}{Br w/} & \textbf{-10.2} & \textbf{2.7} & \textbf{24.9} & \textbf{82.1} & \textbf{1310} & \textbf{4806} & \textbf{895} & 363 & 374 \\
& \cellcolor{gray!20}-0.3 & \cellcolor{gray!20}7.1 & \cellcolor{gray!20}44.8 & \cellcolor{gray!20}318 & \cellcolor{gray!20}5174 & \cellcolor{gray!20}14979 & \cellcolor{gray!20}2315 & \cellcolor{gray!20}\textbf{123} & \cellcolor{gray!20}\textbf{125}\\
\hline
\multirow{2}{*}{Sh w/} & \textbf{-17.9} & \textbf{23.1} & \textbf{53.8} & \textbf{46.0} & \textbf{5060} & 1537 & \textbf{1750} & 1421 & 1432 \\
 & \cellcolor{gray!20}6.3 & \cellcolor{gray!20}40.3 & \cellcolor{gray!20}63.0 & \cellcolor{gray!20}188 & \cellcolor{gray!20}13406 & \cellcolor{gray!20}\textbf{1494} & \cellcolor{gray!20}1790 & \cellcolor{gray!20}\textbf{1029} & \cellcolor{gray!20}\textbf{1030} \\

\hline
\end{tabular}
\label{tab:AIC}
\end{table}

\begin{table}[h!]
\centering
\caption{Corrected AIC (AICc) values (lower is better).}
\begin{tabular}{|c|c|c|c|c|c|c|c|c|c|}
\hline
 $w$& 1 & 2 & 3 & 4 & 5 & 6 & 7 & 8 & 9 \\
\hline
\multirow{2}{*}{Br no} & \textbf{-7.9} & \textbf{5.3} & \textbf{73.2} & \textbf{616}&  &  &  &  &  \\
 & \cellcolor{gray!20}1.5 & \cellcolor{gray!20}9.4 & \cellcolor{gray!20}229 & \cellcolor{gray!20}2860 & &&&& \\
\hline
\multirow{2}{*}{Sh no} & \textbf{-16.9} & \textbf{28.7} & \textbf{65.6} & \textbf{115} &  &  &  &  &  \\
 & \cellcolor{gray!20}10.4 & \cellcolor{gray!20}48.0 & \cellcolor{gray!20}108 & \cellcolor{gray!20}544 & &&&& \\
\hline
\multirow{2}{*}{Br w/} & \textbf{-9.2} & \textbf{3.7} & \textbf{25.9} & \textbf{83.1} & \textbf{1311} & \textbf{4807} & \textbf{897} & 375 & 386 \\
& \cellcolor{gray!20}0.7 & \cellcolor{gray!20}8.1 & \cellcolor{gray!20}45.8 & \cellcolor{gray!20}319 & \cellcolor{gray!20}5175 & \cellcolor{gray!20}14980 & \cellcolor{gray!20}{2317} & \cellcolor{gray!20}\textbf{135} & \cellcolor{gray!20}\textbf{137} \\

\hline
\multirow{2}{*}{Sh w/} & \textbf{-16.9} & \textbf{24.1} & \textbf{54.8} & \textbf{47.0} & \textbf{5061} & 1543 & \textbf{1756} & 1433 & 1444 \\
 & \cellcolor{gray!20}7.3 & \cellcolor{gray!20}41.3 & \cellcolor{gray!20}64.0 & \cellcolor{gray!20}189 & \cellcolor{gray!20}13408 & \cellcolor{gray!20}\textbf{1500} & \cellcolor{gray!20}1796 & \cellcolor{gray!20}\textbf{1041} & \cellcolor{gray!20}\textbf{1042} \\
\hline
\end{tabular}
\label{tab:AICc}
\end{table}

\begin{table}[h!]
\centering
\caption{Bayesian Information Criterion (BIC) values (lower is better).}
\begin{tabular}{|c|c|c|c|c|c|c|c|c|c|}
\hline
 $w$& 1 & 2 & 3 & 4 & 5 & 6 & 7 & 8 & 9 \\
\hline
\multirow{2}{*}{Br no} & \textbf{-4.9} & \textbf{8.3} & \textbf{76.2} & \textbf{619} & &&&& \\
 & \cellcolor{gray!20}4.5 & \cellcolor{gray!20}12.4 & \cellcolor{gray!20}232 & \cellcolor{gray!20}2863 & &&&& \\
 \hline
\multirow{2}{*}{Sh no} & \textbf{-13.9} & \textbf{31.7} & \textbf{68.6} & \textbf{118}&&&&& \\
 & \cellcolor{gray!20}13.4 & \cellcolor{gray!20}51.0 & \cellcolor{gray!20}111 & \cellcolor{gray!20}547 & &&&& \\
 \hline
\multirow{2}{*}{Br w/} & \textbf{-6.2} & \textbf{6.7} & \textbf{28.9} & \textbf{86.1} & \textbf{1314} & \textbf{4810} & \textbf{897} & 362 & 373 \\
& \cellcolor{gray!20}3.7 & \cellcolor{gray!20}11.1 & \cellcolor{gray!20}48.8 & \cellcolor{gray!20}322 & \cellcolor{gray!20}5178 & \cellcolor{gray!20}14983 & \cellcolor{gray!20}2317 & \cellcolor{gray!20}\textbf{122} & \cellcolor{gray!20}\textbf{124}\\

\hline
\multirow{2}{*}{Sh w/} & \textbf{-13.9} & \textbf{27.1} & \textbf{57.8} & \textbf{50.0} & \textbf{5062} & 1537 & \textbf{1750} & 1421 & 1431 \\
 & \cellcolor{gray!20}10.3 & \cellcolor{gray!20}44.3 & \cellcolor{gray!20}67.0 & \cellcolor{gray!20}192 & \cellcolor{gray!20}13409 & \cellcolor{gray!20}\textbf{1494} & \cellcolor{gray!20}1790 & \cellcolor{gray!20}\textbf{1028} & \cellcolor{gray!20}\textbf{1030} \\
\hline
\end{tabular}
\label{tab:BIC}
\end{table}

\begin{table}[h!]
\centering
\caption{$R^2$ values (higher is better).}
\begin{tabular}{|c|c|c|c|c|c|c|c|c|c|}
\hline
 $w$& 1 & 2 & 3 & 4 & 5 & 6 & 7 & 8 & 9 \\
\hline
\multirow{2}{*}{Br no} & \textbf{0.999} & \textbf{1.000} & \textbf{0.999} & \textbf{0.994}
 & &&&& \\
 & \cellcolor{gray!20}0.998 & \cellcolor{gray!20}1.000 & \cellcolor{gray!20}0.995 & \cellcolor{gray!20}0.968 & &&&& \\
\hline
\multirow{2}{*}{Sh no} & \textbf{1.000} & \textbf{1.000} & \textbf{0.999} & \textbf{0.999} & &&&& \\
 & \cellcolor{gray!20}0.998 & \cellcolor{gray!20}0.999 & \cellcolor{gray!20}0.998 & \cellcolor{gray!20}0.990 & &&&& \\
\hline
\multirow{2}{*}{Br w/} & \textbf{0.999} & \textbf{1.000} & \textbf{1.000} & \textbf{0.999} & \textbf{0.986} & \textbf{0.966} & \textbf{0.984}
 & 0.950 & 0.950 \\
 & \cellcolor{gray!20}0.998 & \cellcolor{gray!20}1.000 & \cellcolor{gray!20}0.999 & \cellcolor{gray!20}0.994 & \cellcolor{gray!20}0.942 & \cellcolor{gray!20}0.891 & \cellcolor{gray!20}0.958 & \cellcolor{gray!20}\textbf{0.983} & \cellcolor{gray!20}\textbf{0.983} \\
\hline
\multirow{2}{*}{Sh w/} & \textbf{1.000} & \textbf{1.000} & \textbf{0.999} & \textbf{1.000} & \textbf{0.946} & 0.905 & \textbf{0.897}
 & 0.811 & 0.816 \\
 & \cellcolor{gray!20}0.998 & \cellcolor{gray!20}0.999 & \cellcolor{gray!20}0.998 & \cellcolor{gray!20}0.997 & \cellcolor{gray!20}0.856 & \cellcolor{gray!20}\textbf{0.907} & \cellcolor{gray!20}0.895 & \cellcolor{gray!20}\textbf{0.863} & \cellcolor{gray!20}\textbf{0.868} \\
\hline
\end{tabular}
\label{tab:R^2}
\end{table}

\begin{table}[h!]
\centering
\caption{Adjusted $R^2$ values (higher is better).}
\begin{tabular}{|c|c|c|c|c|c|c|c|c|c|}
\hline
 $w$& 1 & 2 & 3 & 4 & 5 & 6 & 7 & 8 & 9 \\
\hline
\multirow{2}{*}{Br no} & \textbf{0.999} & \textbf{1.000} & \textbf{0.999} & \textbf{0.993} & &&&&\\
 & \cellcolor{gray!20}0.998 & \cellcolor{gray!20}1.000 & \cellcolor{gray!20}0.994 & \cellcolor{gray!20}0.966 & &&&& \\
\hline
\multirow{2}{*}{Sh no} & \textbf{1.000} & \textbf{1.000} & \textbf{0.999} & \textbf{0.999} & &&&& \\
 & \cellcolor{gray!20}0.997 & \cellcolor{gray!20}0.999 & \cellcolor{gray!20}0.997 & \cellcolor{gray!20}0.990 & &&&& \\
\hline
\multirow{2}{*}{Br w/} & \textbf{0.999} & \textbf{1.000} & \textbf{1.000} & \textbf{0.999} & \textbf{0.985} & \textbf{0.963} & \textbf{0.982} & 0.926 & 0.924 \\
 & \cellcolor{gray!20}0.998 & \cellcolor{gray!20}1.000 & \cellcolor{gray!20}0.999 & \cellcolor{gray!20}0.993 & \cellcolor{gray!20}0.938 & \cellcolor{gray!20}0.883 & \cellcolor{gray!20}0.951 & \cellcolor{gray!20}\textbf{0.975} & \cellcolor{gray!20}\textbf{0.975} \\
\hline
\multirow{2}{*}{Sh w/} & \textbf{1.000} & \textbf{0.999} & \textbf{0.999} & \textbf{1.000} & \textbf{0.940} & 0.873 & \textbf{0.863} & 0.716 & 0.724 \\
 & \cellcolor{gray!20}0.998 & \cellcolor{gray!20}0.999 & \cellcolor{gray!20}0.998 & \cellcolor{gray!20}0.997 & \cellcolor{gray!20}0.840 & \cellcolor{gray!20}\textbf{0.876} & \cellcolor{gray!20}0.859 & \cellcolor{gray!20}\textbf{0.795} & \cellcolor{gray!20}\textbf{0.802} \\
\hline
\end{tabular}
\label{tab:aR^2}
\end{table}

\subsection{Results for each of the four cases}

Since there is no standard way to combine all the different measures we consider into a single measure, we instead report the number of measures favoring each model and use a majority vote to determine which model is preferred. For the two cases without DD, we only include $w\in\{1,2,3,4\}$ for the reasons discussed above. The results are displayed in \cref{tab-majority-Brn} (Brisbane without DD), \cref{tab-majority-Shn} (Sherbrooke without DD), \cref{tab-majority-Brw} (Brisbane with DD), and \cref{tab-majority-Shw} (Sherbrooke with DD). They are in full agreement with the summary of speedup results reported in \cref{tab:speedups} of the main text.

\begin{table}[h!]
\caption{Model preference by majority vote for Br no}
\begin{tabular}{c c c c}
\hline\hline
$w$ & better & {measures favoring polylog} & {measures favoring poly} \\
\hline
1 & polylog & 5 & 2 \\
2 & polylog & 7 & 0 \\
3 & polylog & 7 & 0 \\
4 & polylog & 7 & 0 \\
\hline\hline
\end{tabular}
\label{tab-majority-Brn}
\end{table}

\begin{table}[h!]
\caption{Model preference by majority vote for Sh no}
\begin{tabular}{c c c c}
\hline\hline
$w$ & better & {measures favoring polylog} & {measures favoring poly} \\
1 & polylog & 5 & 2 \\
2 & polylog & 7 & 0 \\
3 & polylog & 7 & 0 \\
4 & polylog & 7 & 0 \\
\hline\hline
\end{tabular}
\label{tab-majority-Shn}
\end{table}

\begin{table}[h!]
\caption{Model preference by majority vote for Br w/}
\begin{tabular}{c c c c}
\hline\hline
$w$ & better & {measures favoring polylog} & {measures favoring poly} \\
\hline
1 & {polylog} & 5 & 2 \\
2 & {polylog} & 7 & 0 \\
3 & {polylog} & 7 & 0 \\
4 & {polylog} & 7 & 0 \\
5 & {polylog} & 7 & 0 \\
6 & {polylog} & 7 & 0 \\
7 & {polylog} & 5 & 2 \\
8 & {poly} & 2 & 5 \\
9 & {poly} & 2 & 5 \\
\hline\hline
\end{tabular}
\label{tab-majority-Brw}
\end{table}

\begin{table}[h!]
\caption{Model preference by majority vote for Sh w/}
\begin{tabular}{c c c c}
\hline\hline
$w$ & better & {measures favoring polylog} & {measures favoring poly} \\
\hline
1 & {polylog} & 5 & 2 \\
2 & {polylog} & 7 & 0 \\
3 & {polylog} & 7 & 0 \\
4 & {polylog} & 7 & 0 \\
5 & {polylog} & 5 & 2 \\
6 & {poly} & 0 & 7 \\
7 & {polylog} & 5 & 2 \\
8 & {poly} & 0 & 7 \\
9 & {poly} & 0 & 7 \\
\hline\hline
\end{tabular}
\label{tab-majority-Shw}
\end{table}

\subsection{Model selection via model probabilities: Akaike weights}

To further boost confidence in our model selection, we compute probabilities for each model based on Akaike weights via the widely used methodology developed in Ref.~\cite{Wagenmakers:2004aa}. Akaike weights provide an interpretation as the probabilities of each model being the best model in an AIC sense, i.e., the model that has the smallest Kullback-Leibler distance~\cite{Burnham:book}.

The Akaike weight ($w_i$ for model $i$) for each case and parameter is calculated as:
\beq
w_i = \frac{\exp\left(-\frac{1}{2} \Delta \text{AIC}_i\right)}{\sum_{i=1}^{R} \exp\left(-\frac{1}{2} \Delta \text{AIC}_i\right)}
\eeq
where $R$ is the number of models compared (in our case, 2), $\Delta \text{AIC}_i= \text{AIC}_i - \text{AIC}_{\text{min}}$. The numerator is the relative likelihood, and the Akaike weight is the probability that a model is the best in the set.

As an example, here is the calculation for Brisbane with DD for $w=1$. In this case, 
$\text{AIC}_{\text{polylog}} = -10.2$ and $\text{AIC}_{\text{poly}} = -0.3$.

We compute $\Delta$AIC:
\begin{equation}
\begin{aligned}
&\Delta \text{AIC}_{\text{polylog}} = -10.2 - (-10.2) = 0 \\
&\Delta \text{AIC}_{\text{poly}} = -0.3 - (-10.2) = 9.9 ,
\end{aligned}
\end{equation}
relative likelihoods:
\begin{equation}
\begin{aligned}
&L_{\text{polylog}} = \exp\left(-\frac{1}{2} \times 0\right) = 1 \\
&L_{\text{poly}} = \exp\left(-\frac{1}{2} \times 9.9\right)  \approx 0.0071 ,
\end{aligned}
\end{equation}
and compute Akaike weights:
\begin{equation}
\begin{aligned}
&w_{\text{polylog}} = \frac{1}{1 + 0.0071} \approx 0.9929 \\
&w_{\text{poly}} = \frac{0.0071}{1 + 0.0071} \approx 0.0071
\end{aligned}
\end{equation}

The interpretation is that there is a 99.29\% probability that the polylog model is better.

\subsection{Akaike weights: results}

The results are presented in \cref{tab:AIC-weights-Brn} (Brisbane without DD), \cref{tab:AIC-weights-Shn} (Sherbrooke without DD), \cref{tab:AIC-weights-Brw} (Brisbane with DD), and \cref{tab:AIC-weights-Shw} (Sherbrooke with DD). The Akaike weights overwhelmingly support the conclusions in \cref{tab:speedups} of the main text, including the anomaly of the poly model being a better fit for Sherbrooke with DD at $w=6$.

\begin{table}[h!]
\centering
\caption{AIC and Akaike Weights for Brisbane without DD}
\begin{tabular}{|c|c|c|c|c|c|}
\hline
$w$ & $\text{AIC}_\text{polylog}$ & $\text{AIC}_\text{poly}$& $\Delta$AIC & $w_{\text{polylog}}$ & favored model \\
\hline
1 & -8.9 & 0.5 & 9.4 & 99.12\% & polylog \\
2 & 4.3 & 8.4 & 4.1 & 88.42\% & polylog \\
3 & 72.2 & 228 & 156 & $\sim$100\% & polylog \\
4 & 615 & 2859 & 2244 & $\sim$100\% & polylog \\
\hline
\end{tabular}
\label{tab:AIC-weights-Brn}
\end{table}

\begin{table}[h!]
\centering
\caption{AIC and Akaike Weights for Sherbrooke without DD}
\begin{tabular}{|c|c|c|c|c|c|}
\hline
$w$ & $\text{AIC}_\text{polylog}$ & $\text{AIC}_\text{poly}$& $\Delta$AIC & $w_{\text{polylog}}$ & favored model \\
\hline
1 & -17.9 & 10.4 & 28.3 & $\sim$100\% & polylog \\
2 & 27.7 & 48.0 & 20.3 & $\sim$100\% & polylog \\
3 & 64.6 & 108 & 43.7 & $\sim$100\% & polylog \\
4 & 114 & 544 & 430 & $\sim$100\% & polylog \\
\hline
\end{tabular}
\label{tab:AIC-weights-Shn}
\end{table}

\begin{table}[h!]
\centering
\caption{AIC and Akaike Weights for Brisbane with DD}
\begin{tabular}{|c|c|c|c|c|c|}
\hline
$w$ & $\text{AIC}_\text{polylog}$ & $\text{AIC}_\text{poly}$& $\Delta$AIC & $w_{\text{polylog}}$ & favored model \\
\hline
1 & -10.2 & -0.3 & 9.9 & 99.29\% & polylog \\
2 & 2.7 & 7.1 & 4.4 & 89.32\% & polylog \\
3 & 24.9 & 44.8 & 19.9 & $\sim$100\% & polylog \\
4 & 82.1 & 318.6 & 236 & $\sim$100\% & polylog \\
5 & 1310 & 5174 & 3864 & $\sim$100\% & polylog \\
6 & 4806 & 14979 & 10172 & $\sim$100\% & polylog \\
7 & 895 & 2315 & 1420 & $\sim$100\% & polylog \\
8 & 363 & 123 & -239 & $\sim$0\% & poly \\
9 & 374 & 125 & -248 & $\sim$0\% & poly \\
\hline
\end{tabular}
\label{tab:AIC-weights-Brw}
\end{table}

\begin{table}[h!]
\centering
\caption{AIC and Akaike Weights for Sherbrooke with DD}
\begin{tabular}{|c|c|c|c|c|c|}
\hline
$w$ & $\text{AIC}_\text{polylog}$ & $\text{AIC}_\text{poly}$& $\Delta$AIC & $w_{\text{polylog}}$ & favored model \\
\hline
1 & -17.9 & 6.3 & 24.2 & $\sim$100\% & polylog \\
2 & 23.1 & 40.3 & 17.2 & $\sim$100\% & polylog \\
3 & 53.8 & 63.0 & 9.2 & 99.01\% & polylog \\
4 & 46.0 & 188 & 142 & $\sim$100\% & polylog \\
5 & 5060 & 13406 & 8346 & $\sim$100\% & polylog \\
6 & 1537 & 1494 & -42.5 & $\sim$0\% & poly \\
7 & 1750 & 1790 & 39.3 & $\sim$ 100\% & polylog \\
8 & 1421 & 1029 & -392 & $\sim$0\% & poly \\
9 & 1432 & 1030 & -401 & $\sim$0\% & poly \\
\hline
\end{tabular}
\label{tab:AIC-weights-Shw}
\end{table}

\section{Results with the help of Q-CTRL}
\label{app:QCTRL}

The company Q-CTRL provides an IBM Qiskit interface via the \texttt{channel$\_$strategy="q-ctrl"}~\cite{QCTRL-IBM} command. Q-CTRL employs a sophisticated array of optimizations, some of which are AI-driven, including circuit depth reduction and logical transpilation, error-aware hardware mapping, DD for crosstalk reduction, optimized gate replacement, MEM, and others~\cite{Mundada:2023aa}. As we report below, we find that the Q-CTRL results improve upon our manually optimized DD approach in most cases. This suggests that there is room for additional improvement of our results with further, more aggressive optimization, such as implemented by Q-CTRL. However, it is important to note that Q-CTRL provides a black-box approach in the sense that the details of the optimizations are invisible to the user. Consequently, \emph{we are unable to verify that the experiments using Q-CTRL adhere to the restrictions defined for our compiler as outlined in \cref{app:rules:oracle}}. With this caveat, we next report our Q-CTRL results.

We conducted the same experiments on Brisbane using Q-CTRL  up to $n=15$ (30 qubits). \cref{fig:qctrl-main,fig:mixed_R^2+AIC_qctrl,fig:qctrl-params} are identical to \cref{fig:main-plot,fig:mixed_R^2+AIC,fig:allslopes}, except that we also include the data generated using Q-CTRL.
\cref{fig:qctrl-main} shows that with Q-CTRL the scaling of the \wSimon{4}{15} is slightly better than with our manually optimized DD.

In \cref{sec:results} of the main text, we showed that Brisbane exhibits an exponential speedup up to $w=7$ with the application of manually optimized DD and MEM. \cref{fig:mixed_R^2+AIC_qctrl}(left) shows that using Q-CTRL, the polylog model is the better fit up to $w=6$. For $w\in[7,9]$ the poly model is the better fit. 

Thus, with Q-CTRL's optimization, the evidence for the exponential speedup is weakened (due to the $w=7$ point) compared to our manual DD optimization. However, with Q-CTRL we do observe evidence for a polynomial quantum speedup at $w\in [8,9]$, where manual DD optimization does not yield even a polynomial quantum advantage. The corresponding scaling exponents are all shown in \cref{fig:qctrl-params}.

NTS scaling plots for Brisbane comparing manually optimized DD and Q-CTRL from $w=1$ to $w=9$ are shown in \cref{fig:grid-qctrl}. These plots reveal that manually optimized DD and Q-CTRL exhibit a similar performance at lower HW values, with DD outperforming Q-CTRL at $w=2,3$, but with a performance gap growing as a function of the HW $w$ in Q-CTRL's favor for $w\ge 4$. 

\section{Complete scaling results for Sherbrooke, Brisbane, Cairo, and Kolkata, for all values of $w$ and all three models}
\label{app:all_results}

\Cref{fig:mixed-SB,fig:mixed-BB} illustrate the transition from the polylog model (dashed blue) to the poly model (solid red) through the mixed model (dashed-dotted green) on Sherbrooke and Brisbane, respectively, both with DD. The mixed model is only shown when $\gamma\ne 0,1$, i.e., for Sherbrooke with DD at $w=6$ and for Brisbane with DD at $w=7$, and as expected, the green lines tend to lie between the blue and red lines. 

The full set of fitting parameters of the mixed model is given in \cref{tab:mixed-table}. In almost all cases (with just two exceptions) the mixed model yields $\gamma=0$ or $\gamma=1$, i.e., the same scaling as the polylog or poly model, respectively. This is why we claimed in the main text that we have no evidence for a subexponential/superpolynomial speedup.

Finally, \cref{fig:mixed-CR,fig:mixed-KK} illustrate the same for the Cairo and Kolkata results.

\begin{figure*}
    \centering
    \includegraphics[width=0.8\textwidth]{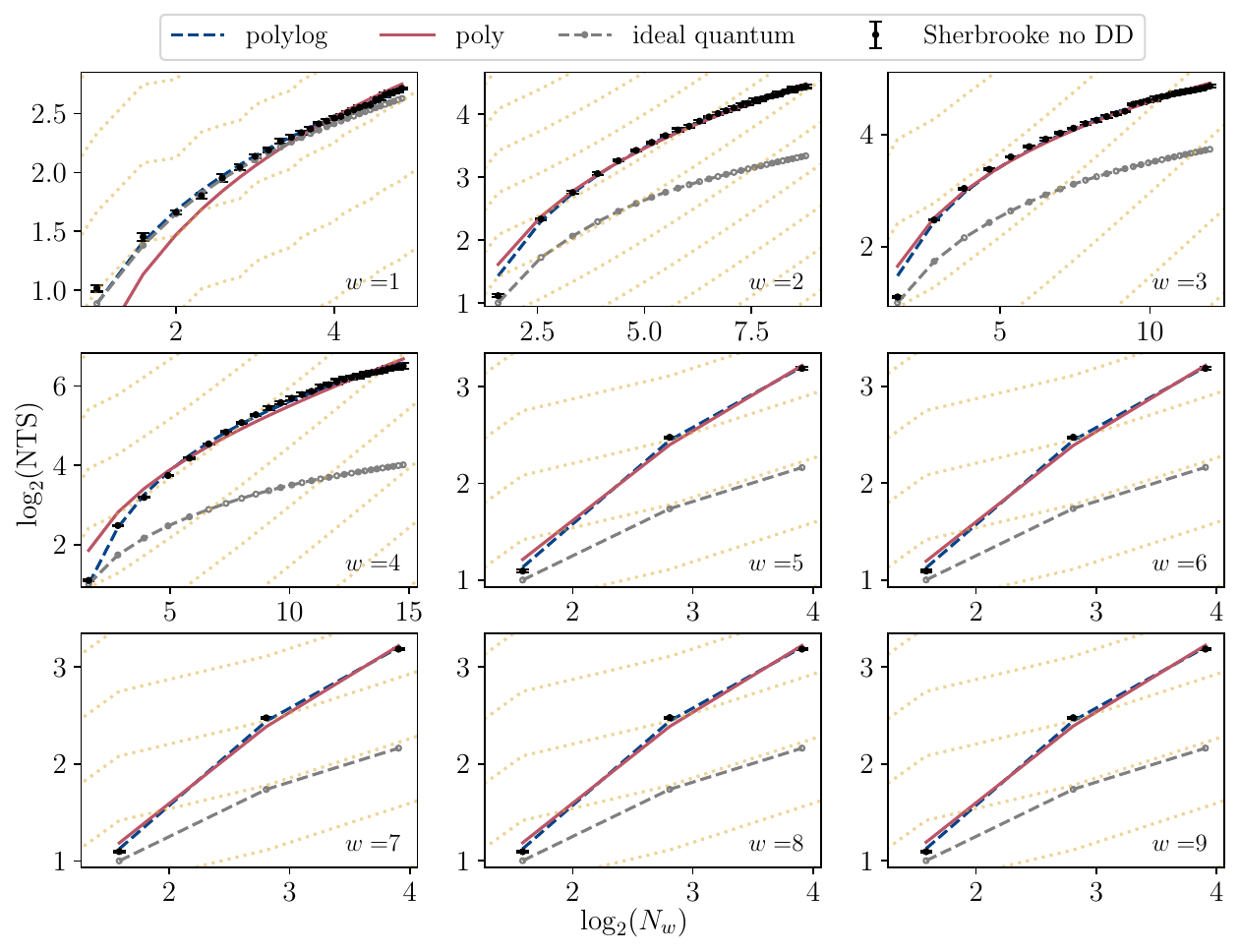}
    \includegraphics[width=0.8\textwidth]{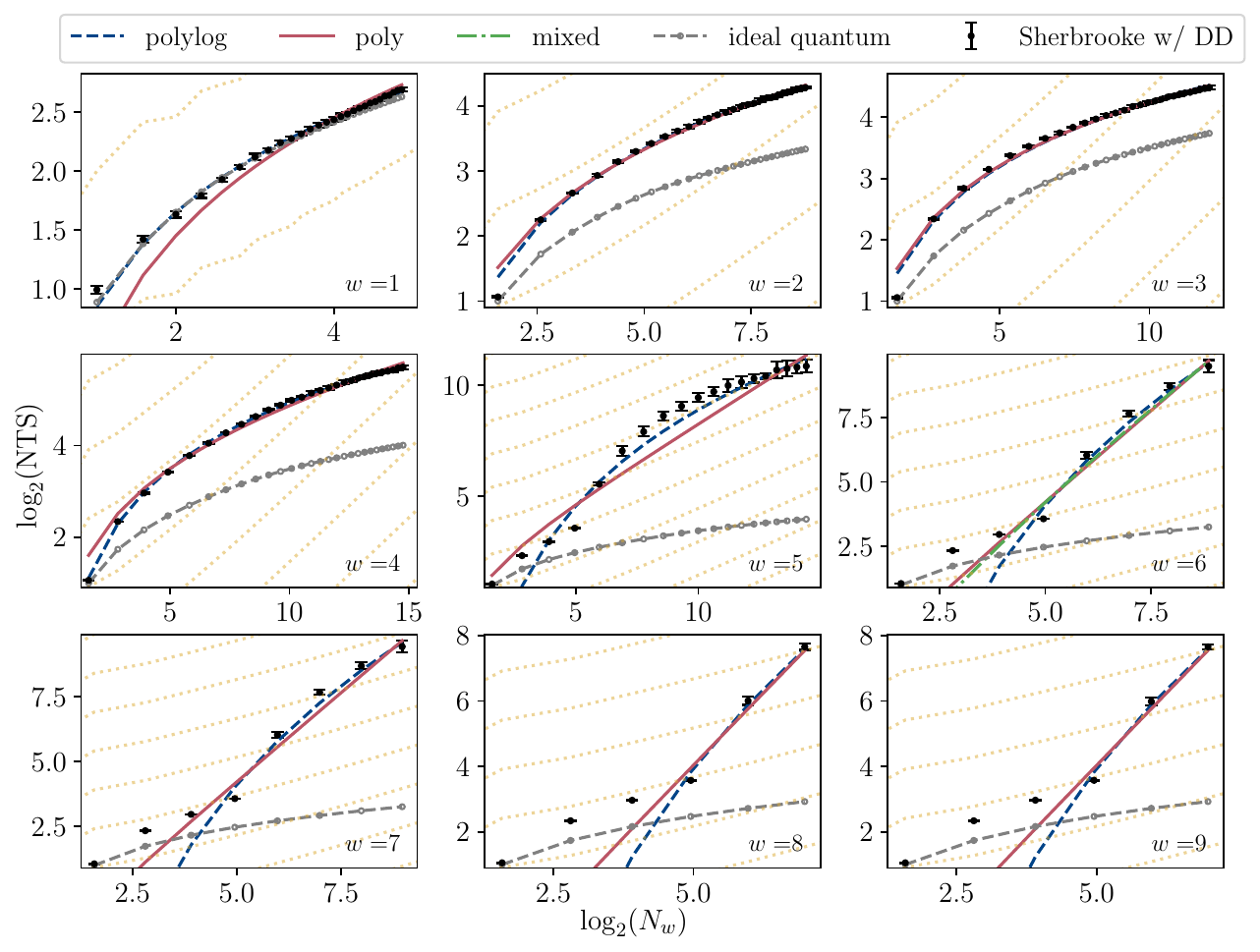}
    \caption{Transition from the polylog model (dashed blue) to the poly model (solid red) through the mixed model (dashed-dotted green for $w=6$) on Sherbrooke with and without DD from HW $w=1$ to $w=9$. Note the kink in the Sherbrooke and Brisbane w/ DD data at $\log_2(N_w)=5$ for $w\ge 5$; fits that account for this are addressed in \cref{app:stat-analysis2}.}
    \label{fig:mixed-SB}
\end{figure*}

\begin{figure*}
    \centering
    \includegraphics[width=0.8\textwidth]{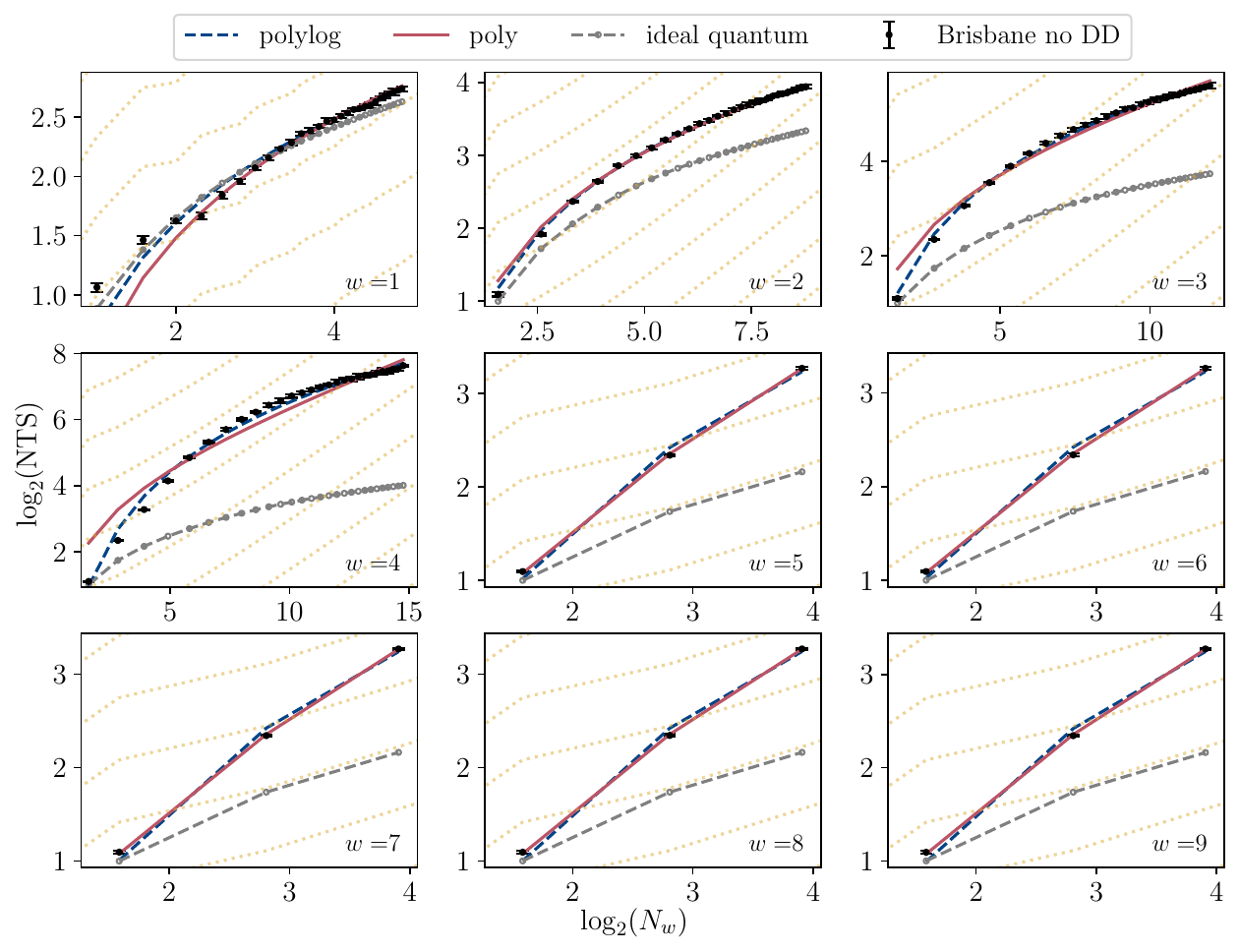}
    \includegraphics[width=0.8\textwidth]{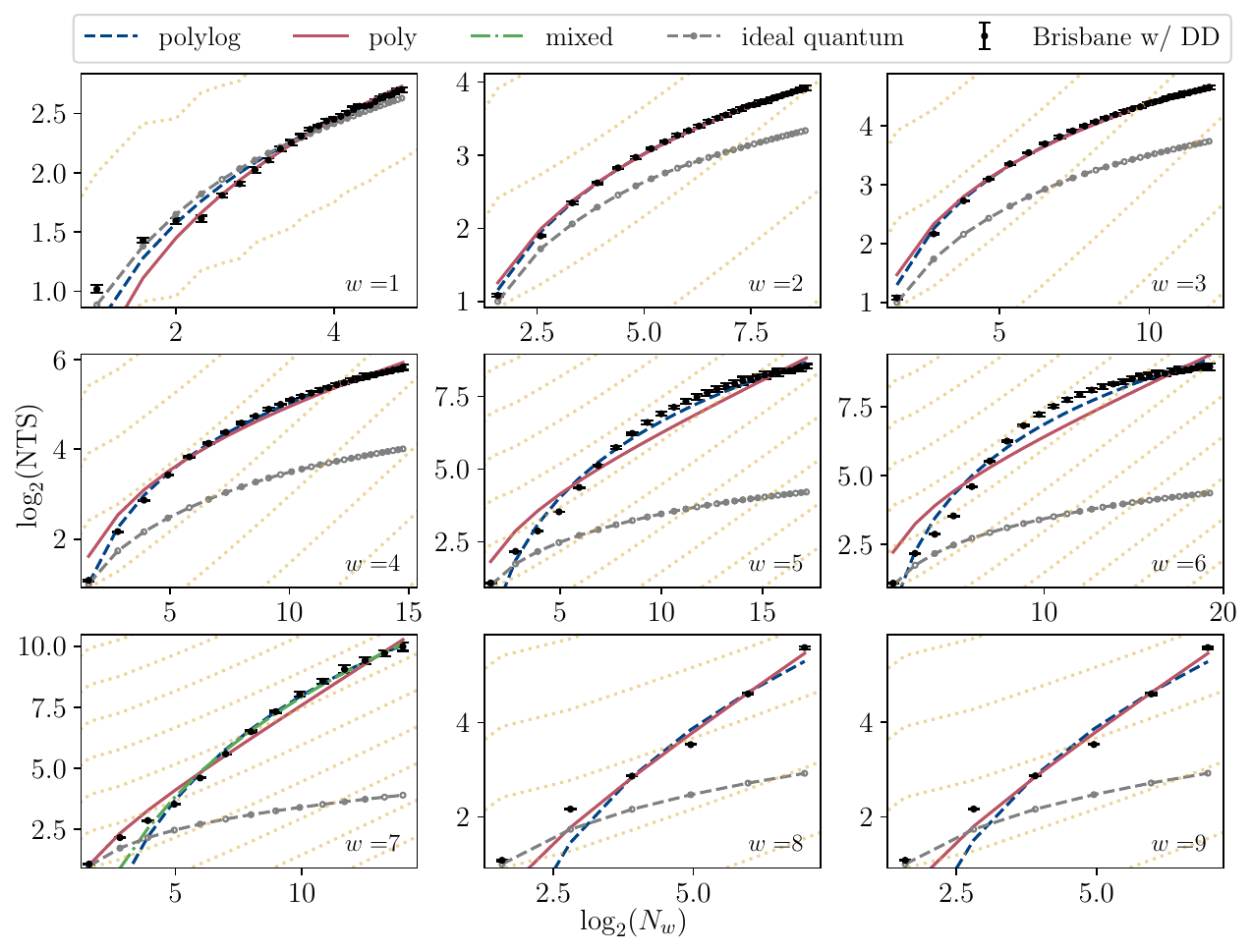}
    \caption{Transition from the polylog model (dashed blue) to the poly model (solid red) through the mixed model (dashed-dotted green for $w=7$) on Brisbane with and without DD from HW $w=1$ to $w=9$. Note the kink in the Sherbrooke and Brisbane w/ DD data at $\log_2(N_w)=5$ for $w\ge 5$; fits that account for this are addressed in \cref{app:stat-analysis2}.}
    \label{fig:mixed-BB}
\end{figure*}

\begin{table*}[h!]
\centering
\begin{tabular}{|c|c|c|c|c|c|c|c|c|c|}
\hline
\multicolumn{2}{|c|}{$w\rightarrow$}&2&3&4&5&6&7&8&9\\
\hline
\multirow{ 3}{*}{Brisbane w/ DD}&$C$&
$1.18 , 0.43$ &
$1.22 , 0.26$ &
$1.67 , 0.15$ &
$2.91 , 0.1$ &
$2.8 , 0.07$ &
$4.01 , 0.17$ &
$0.56 , 0.32$ &
$0.56 , 0.31$ \\
\cline{2-10}
&$\gamma$&
$0.0 , 0.13$ &
$0.0 , 0.07$ &
$0.0 , 0.03$ &
$0.0 , 0.01$ &
$0.0 , 0.01$ &
$\mathit{0.06 , 0.02}$ &
$1.0 , 0.33$ &
$1.0 , 0.33$ \\
\cline{2-10}
&$c$&
$1.1 , 0.61$ &
$1.18 , 0.41$ &
$0.6 , 0.11$ &
$0.09 , 0.01$ &
$0.14 , 0.01$ &
$0.01 , 0.0$ &
$0.9 , 0.6$ &
$0.93 , 0.61$ \\
\hline
\multirow{ 3}{*}{Brisbane no DD}&$C$&
$1.19 , 0.46$ &
$1.7 , 0.2$ &
$2.37 , 0.1$ &
- &
- &
- &
- &
- \\
\cline{2-10}
&$\gamma$&
$0.0 , 0.14$ &
$0.0 , 0.04$ &
$0.0 , 0.02$ &
- &
- &
- &
- &
- \\
\cline{2-10}
&$c$&
$1.12 , 0.66$ &
$0.65 , 0.15$ &
$0.31 , 0.04$ &
- &
- &
- &
- &
- \\
\hline
\multirow{ 3}{*}{Sherbrooke w/ DD}&$C$&
$1.28 , 0.35$ &
$1.06 , 0.23$ &
$0.09 , 0.04$ &
$4.85 , 0.16$ &
$1.55 , 0.2$ &
$7.49 , 0.86$ &
$1.22 , 0.36$ &
$1.22 , 0.36$ \\
\cline{2-10}
&$\gamma$&
$0.0 , 0.1$ &
$0.0 , 0.07$ &
$1.0 , 0.15$ &
$0.0 , 0.01$ &
$\mathit{0.75 , 0.07}$ &
$0.0 , 0.05$ &
$1.0 , 0.18$ &
$1.0 , 0.18$ \\
\cline{2-10}
&$c$&
$1.13 , 0.47$ &
$1.6 , 0.5$ &
$19.54 , 7.85$ &
$0.0 , 0.0$ &
$0.05 , 0.02$ &
$0.0 , 0.0$ &
$0.04 , 0.03$ &
$0.04 , 0.02$ \\
\hline
\multirow{ 3}{*}{Sherbrooke no DD}&$C$&
$1.34 , 0.36$ &
$1.24 , 0.2$ &
$1.94 , 0.13$ &
- &
- &
- &
- &
- \\
\cline{2-10}
&$\gamma$&
$0.0 , 0.1$ &
$0.0 , 0.06$ &
$0.0 , 0.02$ &
- &
- &
- &
- &
- \\
\cline{2-10}
&$c$&
$1.1 , 0.48$ &
$1.29 , 0.32$ &
$0.46 , 0.07$ &
- &
- &
- &
- &
- \\
\hline
\end{tabular}
\caption{Fitted parameters for $\NTS_\text{mixed}$ in \cref{eq:mixed-hw-NTS}, reported as ($p,u$), where $p$ is the fitted parameter and $u$ is the standard error. For $C$, $c$, and $\gamma\ne 0,1$, the value is $p\pm u$. For $\gamma=0$ or $1$, we have $\gamma\in [0,u]$ or $\gamma\in[1-u,1]$, respectively. The table starts from $w=2$ because the confidence interval of this constrained model could not be calculated at $w=1$. The no-DD models in both machines stop at $w=5$ because they only have three data points for $w>5$. Italic fonts are used for Sherbrooke with DD at $w=6$ and Brisbane with DD at $w=7$ to highlight the only two cases where $\gamma\ne 0,1$ as depicted in \cref{fig:mixed-SB} and \cref{fig:mixed-BB}.}
\label{tab:mixed-table}
\end{table*}

\begin{figure*}
    \centering
    \includegraphics[width=0.8\textwidth]{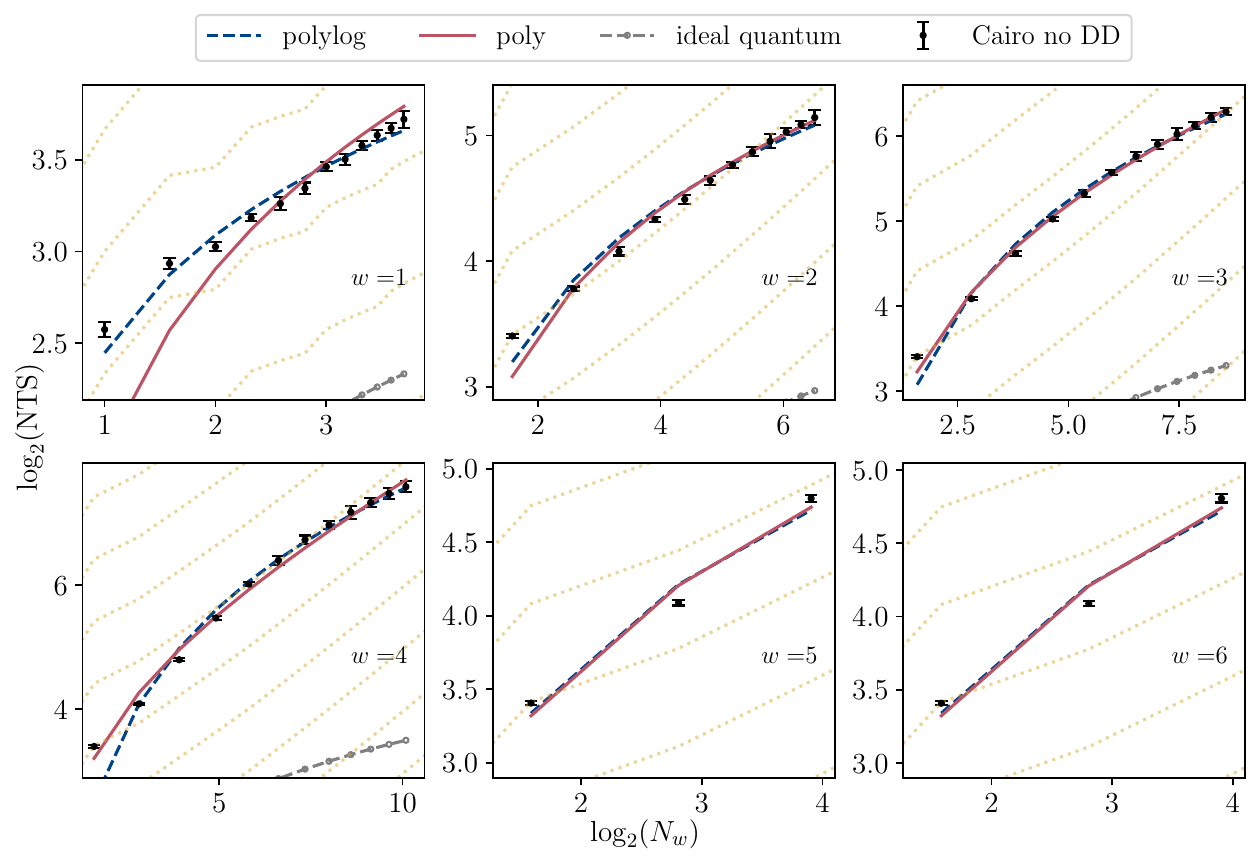}
    \includegraphics[width=0.8\textwidth]{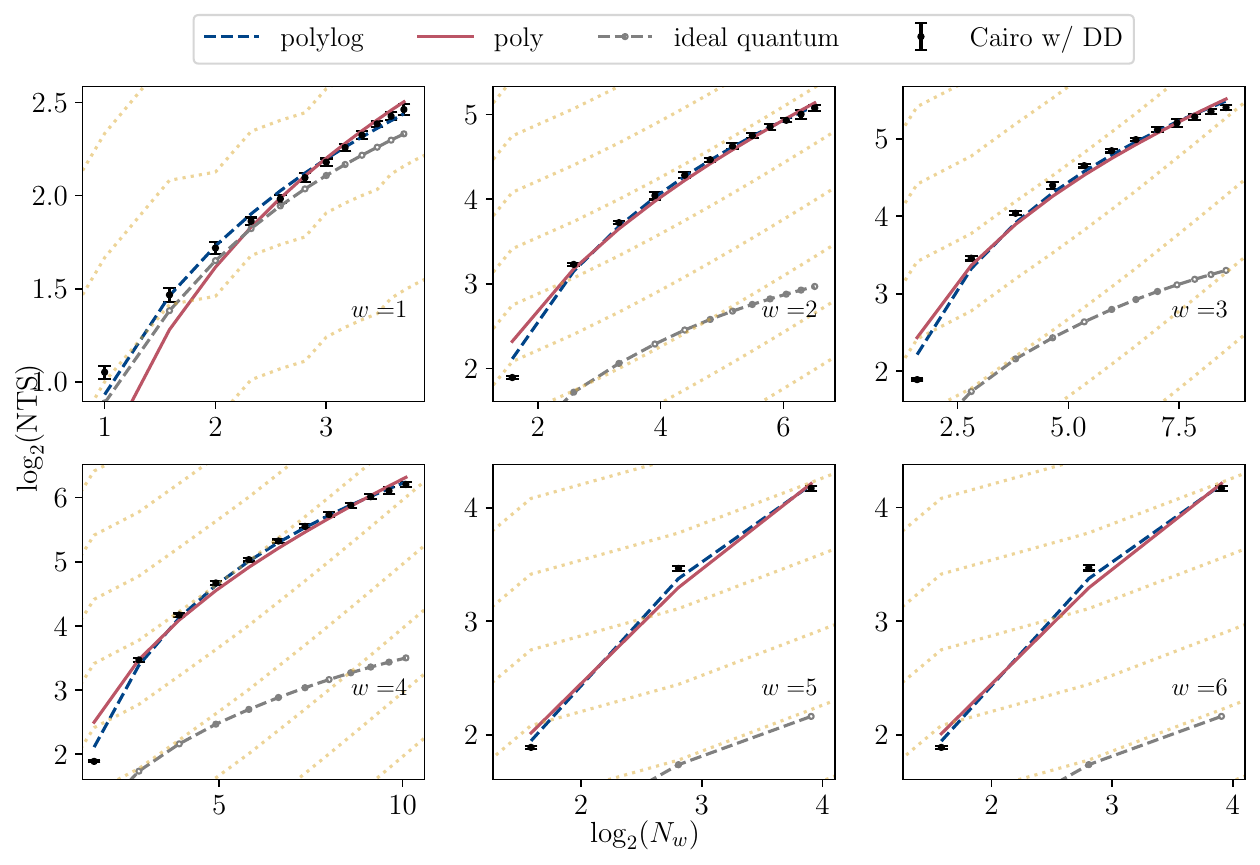}
    \caption{Transition from the polylog model (dashed blue) to the poly model (solid red) on Cairo with and without DD from HW $w=1$ to $w=6$.}
    \label{fig:mixed-CR}
\end{figure*}

\begin{figure*}
    \centering
    \includegraphics[width=0.8\textwidth]{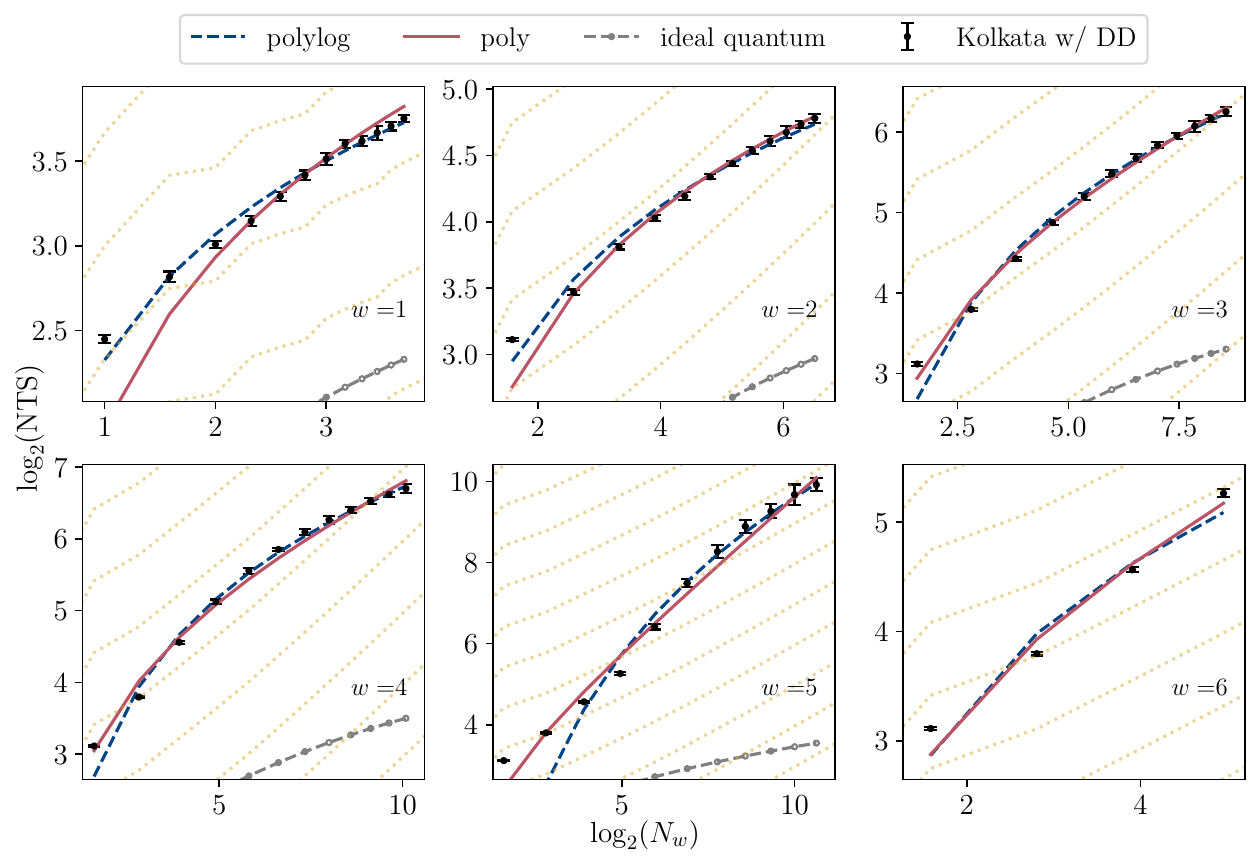}
    \caption{Transition from the polylog model (dashed blue) to the poly model (solid red) on Kolkata with DD from HW $w=1$ to $w=6$.}
    \label{fig:mixed-KK}
\end{figure*}

\section{Statistical comparison of the poly and polylog models after removing small-problem-size data points}
\label{app:stat-analysis2}

In \cref{fig:mixed-SB} and \cref{fig:mixed-BB}, we observe that a kink develops in the experimental Sherbrooke and Brisbane with DD data at problem size $\log_2(N_w)=5$, for $w\ge 5$. This could be attributed to a small-size effect that affects the first four data points in each of the plots. To account for the impact this might have on our model fits, in this section we conduct a statistical analysis similar to \cref{app:stat-analysis} after excluding these first four points.

\subsection{Tables of statistical measures}

The tables are: \cref{tab:t-stat-5plus} ($t$-statistic), \cref{tab:AIC-5plus} (Akaike Information Criterion), \cref{tab:AICc-5plus} (corrected Akaike Information Criterion), \cref{tab:BIC-5plus} (Akaike Information Criterion), \cref{tab:R^2-5plus} ($R^2$), and \cref{tab:aR^2-5plus} (adjusted $R^2$).

The tables in this section are presented for $w\le 7$ because after removing the first four points, we are left with only two data points for $w=8$ and $w=9$ [as can be seen from \cref{fig:mixed-SB} and \cref{fig:mixed-BB}], which is insufficient for curve fitting.

\begin{table*}[t]
\centering
\caption{$p$-values after removing the first 4 points (lower values in boldface)}
\begin{tabular}{|c|c|c|c|c|c|c|c|c|}
\hline
 & w & 1 & 2 & 3 & 4 & 5 & 6 & 7 \\
\hline
\multirow[t]{2}{*}{Br no} & plog & 3.7e-13 & \textbf{2.11e-19} & \textbf{1.09e-24} & \textbf{2.62e-33} &  &  &  \\
 & poly & \textbf{5.47e-31} & 0.0187 & 3.54e-14 & 7.44e-28 &  &  &  \\
\cline{1-9}
\multirow[t]{2}{*}{Sb no} & plog & 5.71e-12 & \textbf{7.95e-21} & \textbf{9.73e-23} & \textbf{1.35e-30} &  &  &  \\
 & poly & \textbf{1.89e-31} & 1.16e-05 & 0.000479 & 1.3e-23 &  &  &  \\
\cline{1-9}
\multirow[t]{2}{*}{Br w/} & plog & 1.14e-13 & \textbf{9.44e-20} & \textbf{4.69e-23} & \textbf{8.93e-28} & \textbf{1.6e-37} & \textbf{1.38e-39} & 1.35e-16 \\
 & poly & \textbf{1.71e-31} & 0.0153 & 0.000518 & 1.83e-17 & 5.22e-35 & 4.28e-36 & \textbf{1.14e-16} \\
\cline{1-9}
\multirow[t]{2}{*}{Sb w/} & plog & 1.65e-12 & \textbf{2.08e-20} & 6.61e-21 & \textbf{6.73e-28} & \textbf{1.07e-24} & 0.000158 & 0.000157 \\
 & poly & \textbf{3.59e-31} & 0.000419 & \textbf{2.45e-36} & 1.94e-17 & 1.45e-24 & \textbf{0.000153} & \textbf{0.000155} \\
\cline{1-9}
\hline
\end{tabular}
\label{tab:p-vals-5plus}
\end{table*}

\begin{table}[h!]
\centering
\caption{$t$-statistic after removing the first 4 points (higher is better).}
\begin{tabular}{|c|c|c|c|c|c|c|c|c|}
\hline
 & w & 1 & 2 & 3 & 4 & 5 & 6 & 7 \\
\hline
\multirow[t]{2}{*}{Br no} & plog & 15.21 & \textbf{30.2} & \textbf{52.92} & \textbf{130.9} &  &  &  \\
 & poly & \textbf{102.6} & 2.538 & 17.07 & 73.84 &  &  &  \\
\cline{1-9}
\multirow[t]{2}{*}{Sb no} & plog & 13.26 & \textbf{35.15} & \textbf{43.06} & \textbf{98.46} &  &  &  \\
 & poly & \textbf{107.7} & 5.63 & 4.095 & 47.24 &  &  &  \\
\cline{1-9}
\multirow[t]{2}{*}{Br w/} & plog & 16.12 & \textbf{31.35} & \textbf{44.53} & \textbf{73.23} & \textbf{203.4} & \textbf{252.5} & 231.8 \\
 & poly & \textbf{108.2} & 2.63 & 4.062 & 24.51 & 156.4 & 175.2 & \textbf{236.7} \\
\cline{1-9}
\multirow[t]{2}{*}{Sb w/} & plog & 14.12 & \textbf{33.62} & 35.46 & \textbf{74.18} & \textbf{223.7} & 79.49 & 79.74 \\
 & poly & \textbf{104.6} & 4.149 & \textbf{179.7} & 24.44 & 218.6 & \textbf{80.88} & \textbf{80.35} \\
\cline{1-9}
\hline
\end{tabular}
\label{tab:t-stat-5plus}
\end{table}

\begin{table}[h!]
\centering
\caption{Akaike Information Criterion (AIC) values after removing the first 4 points (lower is better).}
\begin{tabular}{|c|c|c|c|c|c|c|c|c|}
\hline
 & w & 1 & 2 & 3 & 4 & 5 & 6 & 7 \\
\hline
\multirow[t]{2}{*}{Br no} & plog & \textbf{-12.12} & \textbf{6.949} & \textbf{53.85} & \textbf{272.6} &  &  &  \\
 & poly & -11.48 & 7.202 & 66.91 & 606.5 &  &  &  \\
\cline{1-9}
\multirow[t]{2}{*}{Sb no} & plog & \textbf{-14.56} & \textbf{20.49} & \textbf{32.87} & \textbf{107.3} &  &  &  \\
 & poly & -7.17 & 21.45 & 33.16 & 212 &  &  &  \\
\cline{1-9}
\multirow[t]{2}{*}{Br w/} & plog & \textbf{-14.87} & \textbf{5.559} & \textbf{19.12} & \textbf{62.42} & \textbf{930.4} & \textbf{2922} & \textbf{216.7} \\
 & poly & -14.38 & 5.82 & 20.36 & 95.9 & 2557 & 6289 & 1431 \\
\cline{1-9}
\multirow[t]{2}{*}{Sb w/} & plog & \textbf{-14.66} & \textbf{15.34} & \textbf{15.74} & \textbf{48.72} & \textbf{1858} & \textbf{81.99} & \textbf{115.4} \\
 & poly & -6.222 & 15.96 & 18.64 & 74.56 & 4800 & 177.8 & 231.4 \\
\cline{1-9}
\hline
\end{tabular}
\label{tab:AIC-5plus}
\end{table}

\begin{table}[h!]
\centering
\caption{Corrected AIC (AICc) values after removing the first 4 points (lower is better).}
\begin{tabular}{|c|c|c|c|c|c|c|c|c|}
\hline
 & w & 1 & 2 & 3 & 4 & 5 & 6 & 7 \\
\hline
\multirow[t]{2}{*}{Br no} & plog & \textbf{-10.92} & \textbf{8.149} & \textbf{55.05} & \textbf{273.8} &  &  &  \\
 & poly & -10.28 & 8.402 & 68.11 & 607.7 &  &  &  \\
\cline{1-9}
\multirow[t]{2}{*}{Sb no} & plog & \textbf{-13.36} & \textbf{21.69} & \textbf{34.07} & \textbf{108.5} &  &  &  \\
 & poly & -5.97 & 22.65 & 34.36 & 213.2 &  &  &  \\
\cline{1-9}
\multirow[t]{2}{*}{Br w/} & plog & \textbf{-13.67} & \textbf{6.759} & \textbf{20.32} & \textbf{63.62} & \textbf{931.6} & \textbf{2923} & \textbf{220.7} \\
 & poly & -13.18 & 7.02 & 21.56 & 97.1 & 2558 & 6290 & 1435 \\
\cline{1-9}
\multirow[t]{2}{*}{Sb w/} & plog & \textbf{-13.46} & \textbf{16.54} & \textbf{16.94} & \textbf{49.92} & \textbf{1860} & inf & inf \\
 & poly & -5.022 & 17.16 & 19.84 & 75.76 & 4802 & inf & inf \\
\cline{1-9}
\hline
\end{tabular}
\label{tab:AICc-5plus}
\end{table}

\begin{table}[h!]
\centering
\caption{Bayesian Information Criterion (BIC) values after removing the first 4 points (lower is better).}
\begin{tabular}{|c|c|c|c|c|c|c|c|c|}
\hline
 & w & 1 & 2 & 3 & 4 & 5 & 6 & 7 \\
\hline
\multirow[t]{2}{*}{Br no} & plog & \textbf{-8.591} & \textbf{10.48} & \textbf{57.39} & \textbf{276.1} &  &  &  \\
 & poly & -7.949 & 10.74 & 70.44 & 610 &  &  &  \\
\cline{1-9}
\multirow[t]{2}{*}{Sb no} & plog & \textbf{-11.03} & \textbf{24.02} & \textbf{36.4} & \textbf{110.8} &  &  &  \\
 & poly & -3.636 & 24.99 & 36.7 & 215.5 &  &  &  \\
\cline{1-9}
\multirow[t]{2}{*}{Br w/} & plog & \textbf{-11.34} & \textbf{9.093} & \textbf{22.66} & \textbf{65.96} & \textbf{933.9} & \textbf{2926} & \textbf{217.6} \\
 & poly & -10.84 & 9.354 & 23.89 & 99.43 & 2560 & 6292 & 1432 \\
\cline{1-9}
\multirow[t]{2}{*}{Sb w/} & plog & \textbf{-11.13} & \textbf{18.88} & \textbf{19.27} & \textbf{52.25} & \textbf{1860} & \textbf{80.15} & \textbf{113.6} \\
 & poly & -2.688 & 19.49 & 22.18 & 78.09 & 4802 & 176 & 229.5 \\
\cline{1-9}
\hline
\end{tabular}
\label{tab:BIC-5plus}
\end{table}

\begin{table}[h!]
\centering
\caption{$R^2$ values after removing the first 4 points (higher is better).}
\begin{tabular}{|c|c|c|c|c|c|c|c|c|}
\hline
 & w & 1 & 2 & 3 & 4 & 5 & 6 & 7 \\
\hline
\multirow[t]{2}{*}{Br no} & plog & \textbf{0.9999} & \textbf{1.000} & \textbf{0.9998} & \textbf{0.9978} &  &  &  \\
 & poly & 0.9998 & 1.000 & 0.9994 & 0.9938 &  &  &  \\
\cline{1-9}
\multirow[t]{2}{*}{Sb no} & plog & \textbf{0.9999} & \textbf{1.000} & \textbf{0.9998} & \textbf{0.9991} &  &  &  \\
 & poly & 0.9993 & 0.9999 & 0.9998 & 0.9968 &  &  &  \\
\cline{1-9}
\multirow[t]{2}{*}{Br w/} & plog & \textbf{0.9999} & \textbf{1.000} & \textbf{1.000} & \textbf{0.9996} & \textbf{0.9902} & \textbf{0.9789} & \textbf{0.9967} \\
 & poly & 0.9998 & 1.000 & 0.9999 & 0.9988 & 0.9709 & 0.9536 & 0.9727 \\
\cline{1-9}
\multirow[t]{2}{*}{Sb w/} & plog & \textbf{1.000} & \textbf{1.000} & \textbf{1.000} & \textbf{0.9998} & \textbf{0.9804} & \textbf{0.9956} & \textbf{0.9932} \\
 & poly & 0.9992 & 1.000 & 0.9999 & 0.9992 & 0.9476 & 0.988 & 0.9845 \\
\cline{1-9}
\hline
\end{tabular}
\label{tab:R^2-5plus}
\end{table}

\begin{table}[h!]
\centering
\caption{Adjusted $R^2$ values after removing the first 4 points (higher is better).}
\begin{tabular}{|c|c|c|c|c|c|c|c|c|}
\hline
 & w & 1 & 2 & 3 & 4 & 5 & 6 & 7 \\
\hline
\multirow[t]{2}{*}{Br no} & plog & \textbf{0.9999} & \textbf{1.000} & \textbf{0.9998} & \textbf{0.9976} &  &  &  \\
 & poly & 0.9998 & 1.000 & 0.9993 & 0.9932 &  &  &  \\
\cline{1-9}
\multirow[t]{2}{*}{Sb no} & plog & \textbf{0.9999} & \textbf{1.000} & \textbf{0.9998} & \textbf{0.999} &  &  &  \\
 & poly & 0.9992 & 0.9999 & 0.9998 & 0.9965 &  &  &  \\
\cline{1-9}
\multirow[t]{2}{*}{Br w/} & plog & \textbf{0.9999} & \textbf{1.000} & \textbf{1.000} & \textbf{0.9996} & \textbf{0.9893} & \textbf{0.977} & \textbf{0.9959} \\
 & poly & 0.9998 & 1.000 & 0.9999 & 0.9987 & 0.9682 & 0.9494 & 0.9659 \\
\cline{1-9}
\multirow[t]{2}{*}{Sb w/} & plog & \textbf{1.000} & \textbf{1.000} & \textbf{1.000} & \textbf{0.9998} & \textbf{0.9774} & \textbf{0.9912} & \textbf{0.9865} \\
 & poly & 0.9991 & 0.9999 & 0.9999 & 0.9992 & 0.9395 & 0.9761 & 0.969 \\
\cline{1-9}
\hline
\end{tabular}
\label{tab:aR^2-5plus}
\end{table}

\subsection{Results for each of the four cases}

Similar to the previous section, we report the number of measures favoring each model and use a majority vote to determine which model is preferred.

The results are displayed in \cref{tab-majority-Brn-5plus} (Brisbane without DD), \cref{tab-majority-Shn-5plus} (Sherbrooke without DD), \cref{tab-majority-Brw-5plus} (Brisbane with DD), and \cref{tab-majority-Shw-5plus} (Sherbrooke with DD). The majority of measures suggest that the data follow the polylog model for all cases.

\begin{table}[h!]
\caption{Model preference by majority vote for Br no}
\begin{tabular}{c c c c}
\hline\hline
$w$ & better & {measures favoring polylog} & {measures favoring poly} \\
\hline
1 & polylog & 5 & 2 \\
2 & polylog & 7 & 0 \\
3 & polylog & 7 & 0 \\
4 & polylog & 7 & 0 \\
\hline\hline
\end{tabular}
\label{tab-majority-Brn-5plus}
\end{table}

\begin{table}[h!]
\caption{Model preference by majority vote for Sh no}
\begin{tabular}{c c c c}
\hline\hline
$w$ & better & {measures favoring polylog} & {measures favoring poly} \\
1 & polylog & 5 & 2 \\
2 & polylog & 7 & 0 \\
3 & polylog & 7 & 0 \\
4 & polylog & 7 & 0 \\
\hline\hline
\end{tabular}
\label{tab-majority-Shn-5plus}
\end{table}

\begin{table}[h!]
\caption{Model preference by majority vote for Br w/}
\begin{tabular}{c c c c}
\hline\hline
$w$ & better & {measures favoring polylog} & {measures favoring poly} \\
\hline
1 & {polylog} & 5 & 2 \\
2 & {polylog} & 7 & 0 \\
3 & {polylog} & 7 & 0 \\
4 & {polylog} & 7 & 0 \\
5 & {polylog} & 7 & 0 \\
6 & {polylog} & 7 & 0 \\
7 & {polylog} & 5 & 2 \\
\hline\hline
\end{tabular}
\label{tab-majority-Brw-5plus}
\end{table}

\begin{table}[h!]
\caption{Model preference by majority vote for Sh w/}
\begin{tabular}{c c c c}
\hline\hline
$w$ & better & {measures favoring polylog} & {measures favoring poly} \\
\hline
1 & {polylog} & 5 & 2 \\
2 & {polylog} & 7 & 0 \\
3 & {polylog} & 5 & 2 \\
4 & {polylog} & 7 & 0 \\
5 & {polylog} & 7 & 0 \\
6 & {polylog} & 4 & 2 \\
7 & {polylog} & 4 & 2 \\
\hline\hline
\end{tabular}
\label{tab-majority-Shw-5plus}
\end{table}

\subsection{Akaike weights: results}

The Akaike weights after removing the small problem-size data points indicate that the data follow the polylog model for all cases, in agreement with the majority vote results.

The results are presented in \cref{tab:AIC-weights-Brn-5plus} (Brisbane without DD), \cref{tab:AIC-weights-Shn-5plus} (Sherbrooke without DD), \cref{tab:AIC-weights-Brw-5plus} (Brisbane with DD), and \cref{tab:AIC-weights-Shw-5plus} (Sherbrooke with DD).

\begin{table}[h!]
\centering
\caption{AIC and Akaike Weights for Brisbane without DD}
\begin{tabular}{|c|c|c|c|c|c|}
\hline
 $w$ & AIC$_\text{polylog}$ & AIC$_\text{poly}$ & $\Delta$AIC & $w_\text{polylog}$ & favored model\\
\hline
1 & -12.12 & -11.48 & 0.64 & 57.95\% & polylog\\
2 & 6.95 & 7.20 & 0.25 & 53.16\% & polylog\\
3 & 53.85 & 66.91 & 13.05 & 99.85\% & polylog\\
4 & 272.58 & 606.46 & 333.88 & $\sim$ 100\% & polylog\\
\hline
\end{tabular}
\label{tab:AIC-weights-Brn-5plus}
\end{table}

\begin{table}[h!]
\centering
\caption{AIC and Akaike Weights for Sherbrooke without DD}
\begin{tabular}{|c|c|c|c|c|c|}
\hline
 $w$ & AIC$_\text{polylog}$ & AIC$_\text{poly}$ & $\Delta$AIC & $w_\text{polylog}$ & favored model\\
\hline
1 & -14.56 & -7.17 & 7.39 & 97.58\% & polylog\\
2 & 20.49 & 21.45 & 0.97 & 61.85\% & polylog\\
3 & 32.87 & 33.16 & 0.30 & 53.70\% & polylog\\
4 & 107.30 & 212.01 & 104.72 & $\sim$ 100\% & polylog\\
\hline
\end{tabular}
\label{tab:AIC-weights-Shn-5plus}
\end{table}

\begin{table}[h!]
\centering
\caption{AIC and Akaike Weights for Brisbane with DD}
\begin{tabular}{|c|c|c|c|c|c|}
\hline
 $w$ & AIC$_\text{polylog}$ & AIC$_\text{poly}$ & $\Delta$AIC & $w_\text{polylog}$ & favored model\\
\hline
1 & -14.87 & -14.38 & 0.50 & 56.18\% & polylog\\
2 & 5.56 & 5.82 & 0.26 & 53.26\% & polylog\\
3 & 19.12 & 20.36 & 1.23 & 64.95\% & polylog\\
4 & 62.42 & 95.90 & 33.48 & $\sim$ 100\% & polylog\\
5 & 930.37 & 2556.69 & 1626.32 & $\sim$ 100\% & polylog\\
6 & 2922.02 & 6288.71 & 3366.69 & $\sim$ 100\% & polylog\\
7 & 216.71 & 1431.37 & 1214.67 & $\sim$ 100\% & polylog\\
\hline
\end{tabular}
\label{tab:AIC-weights-Brw-5plus}
\end{table}

\begin{table}[h!]
\centering
\caption{AIC and Akaike Weights for Sherbrooke with DD}
\begin{tabular}{|c|c|c|c|c|c|}
\hline
 $w$ & AIC$_\text{polylog}$ & AIC$_\text{poly}$ & $\Delta$AIC & $w_\text{polylog}$ & favored model\\
\hline
1 & -14.66 & -6.22 & 8.44 & 98.55\% & polylog\\
2 & 15.34 & 15.96 & 0.61 & 57.63\% & polylog\\
3 & 15.74 & 18.64 & 2.90 & 81.03\% & polylog\\
4 & 48.72 & 74.56 & 25.84 & $\sim$ 100\% & polylog\\
5 & 1857.87 & 4800.09 & 2942.22 &$\sim$ 100\% & polylog\\
6 & 81.99 & 177.80 & 95.82 & $\sim$ 100\% & polylog\\
7 & 115.44 & 231.36 & 115.92 & $\sim$ 100\% & polylog\\
\hline
\end{tabular}
\label{tab:AIC-weights-Shw-5plus}
\end{table}

\clearpage


\begin{thebibliography}{79}%
\makeatletter
\providecommand \@ifxundefined [1]{%
 \@ifx{#1\undefined}
}%
\providecommand \@ifnum [1]{%
 \ifnum #1\expandafter \@firstoftwo
 \else \expandafter \@secondoftwo
 \fi
}%
\providecommand \@ifx [1]{%
 \ifx #1\expandafter \@firstoftwo
 \else \expandafter \@secondoftwo
 \fi
}%
\providecommand \natexlab [1]{#1}%
\providecommand \enquote  [1]{``#1''}%
\providecommand \bibnamefont  [1]{#1}%
\providecommand \bibfnamefont [1]{#1}%
\providecommand \citenamefont [1]{#1}%
\providecommand \href@noop [0]{\@secondoftwo}%
\providecommand \href [0]{\begingroup \@sanitize@url \@href}%
\providecommand \@href[1]{\@@startlink{#1}\@@href}%
\providecommand \@@href[1]{\endgroup#1\@@endlink}%
\providecommand \@sanitize@url [0]{\catcode `\\12\catcode `\$12\catcode
  `\&12\catcode `\#12\catcode `\^12\catcode `\_12\catcode `\%12\relax}%
\providecommand \@@startlink[1]{}%
\providecommand \@@endlink[0]{}%
\providecommand \url  [0]{\begingroup\@sanitize@url \@url }%
\providecommand \@url [1]{\endgroup\@href {#1}{\urlprefix }}%
\providecommand \urlprefix  [0]{URL }%
\providecommand \Eprint [0]{\href }%
\providecommand \doibase [0]{http://dx.doi.org/}%
\providecommand \selectlanguage [0]{\@gobble}%
\providecommand \bibinfo  [0]{\@secondoftwo}%
\providecommand \bibfield  [0]{\@secondoftwo}%
\providecommand \translation [1]{[#1]}%
\providecommand \BibitemOpen [0]{}%
\providecommand \bibitemStop [0]{}%
\providecommand \bibitemNoStop [0]{.\EOS\space}%
\providecommand \EOS [0]{\spacefactor3000\relax}%
\providecommand \BibitemShut  [1]{\csname bibitem#1\endcsname}%
\let\auto@bib@innerbib\@empty
\bibitem [{\citenamefont {Deutsch}\ and\ \citenamefont
  {Jozsa}(1992)}]{Deutsch:92}%
  \BibitemOpen
  \bibfield  {author} {\bibinfo {author} {\bibfnamefont {David}\ \bibnamefont
  {Deutsch}}\ and\ \bibinfo {author} {\bibfnamefont {Richard}\ \bibnamefont
  {Jozsa}},\ }\bibfield  {title} {\enquote {\bibinfo {title} {Rapid solution of
  problems by quantum computation},}\ }\bibfield  {booktitle} {\emph {\bibinfo
  {booktitle} {Proceedings of the Royal Society of London. Series A:
  Mathematical and Physical Sciences}},\ }\href {\doibase
  10.1098/rspa.1992.0167} {\bibfield  {journal} {\bibinfo  {journal}
  {Proceedings of the Royal Society of London. Series A: Mathematical and
  Physical Sciences}\ }\textbf {\bibinfo {volume} {439}},\ \bibinfo {pages}
  {553--558} (\bibinfo {year} {1992})}\BibitemShut {NoStop}%
\bibitem [{\citenamefont {Bernstein}\ and\ \citenamefont
  {Vazirani}(1997)}]{bernsteinQuantumComplexityTheory1997}%
  \BibitemOpen
  \bibfield  {author} {\bibinfo {author} {\bibfnamefont {Ethan}\ \bibnamefont
  {Bernstein}}\ and\ \bibinfo {author} {\bibfnamefont {Umesh}\ \bibnamefont
  {Vazirani}},\ }\bibfield  {title} {\enquote {\bibinfo {title} {{Quantum
  Complexity Theory}},}\ }\href {https://doi.org/10.1137/S0097539796300921}
  {\bibfield  {journal} {\bibinfo  {journal} {SIAM J. Comput.}\ }\textbf
  {\bibinfo {volume} {26}},\ \bibinfo {pages} {1411--1473} (\bibinfo {year}
  {1997})}\BibitemShut {NoStop}%
\bibitem [{\citenamefont {Simon}(1997)}]{Simon:94}%
  \BibitemOpen
  \bibfield  {author} {\bibinfo {author} {\bibfnamefont {Daniel~R.}\
  \bibnamefont {Simon}},\ }\bibfield  {title} {\enquote {\bibinfo {title} {On
  the power of quantum computation},}\ }\href
  {https://doi.org/10.1137/S0097539796298637} {\bibfield  {journal} {\bibinfo
  {journal} {SIAM Journal on Computing}\ }\textbf {\bibinfo {volume} {26}},\
  \bibinfo {pages} {1474--1483} (\bibinfo {year} {1997})}\BibitemShut {NoStop}%
\bibitem [{\citenamefont {Grover}(1997)}]{Grover:97a}%
  \BibitemOpen
  \bibfield  {author} {\bibinfo {author} {\bibfnamefont {Lov~K.}\ \bibnamefont
  {Grover}},\ }\bibfield  {title} {\enquote {\bibinfo {title} {Quantum
  mechanics helps in searching for a needle in a haystack},}\ }\href
  {http://link.aps.org/doi/10.1103/PhysRevLett.79.325} {\bibfield  {journal}
  {\bibinfo  {journal} {Phys. Rev. Lett.}\ }\textbf {\bibinfo {volume} {79}},\
  \bibinfo {pages} {325--328} (\bibinfo {year} {1997})}\BibitemShut {NoStop}%
\bibitem [{\citenamefont {Shor}(1997)}]{Shor:97}%
  \BibitemOpen
  \bibfield  {author} {\bibinfo {author} {\bibfnamefont {P.}~\bibnamefont
  {Shor}},\ }\bibfield  {title} {\enquote {\bibinfo {title} {Polynomial-time
  algorithms for prime factorization and discrete logarithms on a quantum
  computer},}\ }\href {\doibase 10.1137/S0097539795293172} {\bibfield
  {journal} {\bibinfo  {journal} {SIAM Journal on Computing}\ }\textbf
  {\bibinfo {volume} {26}},\ \bibinfo {pages} {1484--1509} (\bibinfo {year}
  {1997})}\BibitemShut {NoStop}%
\bibitem [{\citenamefont {Childs}\ \emph {et~al.}(2003)\citenamefont {Childs},
  \citenamefont {Cleve}, \citenamefont {Deotto}, \citenamefont {Farhi},
  \citenamefont {Gutmann},\ and\ \citenamefont
  {Spielman}}]{childs2003exponential}%
  \BibitemOpen
  \bibfield  {author} {\bibinfo {author} {\bibfnamefont {Andrew~M}\
  \bibnamefont {Childs}}, \bibinfo {author} {\bibfnamefont {Richard}\
  \bibnamefont {Cleve}}, \bibinfo {author} {\bibfnamefont {Enrico}\
  \bibnamefont {Deotto}}, \bibinfo {author} {\bibfnamefont {Edward}\
  \bibnamefont {Farhi}}, \bibinfo {author} {\bibfnamefont {Sam}\ \bibnamefont
  {Gutmann}}, \ and\ \bibinfo {author} {\bibfnamefont {Daniel~A}\ \bibnamefont
  {Spielman}},\ }\bibfield  {title} {\enquote {\bibinfo {title} {Exponential
  algorithmic speedup by a quantum walk},}\ }in\ \href {\doibase
  10.1145/780542.780552} {\emph {\bibinfo {booktitle} {Proceedings of the
  thirty-fifth annual ACM symposium on Theory of computing}}}\ (\bibinfo
  {organization} {ACM},\ \bibinfo {year} {2003})\ pp.\ \bibinfo {pages}
  {59--68}\BibitemShut {NoStop}%
\bibitem [{\citenamefont {Van~Dam}\ \emph {et~al.}(2006)\citenamefont
  {Van~Dam}, \citenamefont {Hallgren},\ and\ \citenamefont
  {Ip}}]{Van-Dam:2006aa}%
  \BibitemOpen
  \bibfield  {author} {\bibinfo {author} {\bibfnamefont {Wim}\ \bibnamefont
  {Van~Dam}}, \bibinfo {author} {\bibfnamefont {Sean}\ \bibnamefont
  {Hallgren}}, \ and\ \bibinfo {author} {\bibfnamefont {Lawrence}\ \bibnamefont
  {Ip}},\ }\bibfield  {title} {\enquote {\bibinfo {title} {Quantum algorithms
  for some hidden shift problems},}\ }\href
  {https://epubs.siam.org/doi/pdf/10.1137/S009753970343141X} {\bibfield
  {journal} {\bibinfo  {journal} {SIAM Journal on Computing}\ }\textbf
  {\bibinfo {volume} {36}},\ \bibinfo {pages} {763--778} (\bibinfo {year}
  {2006})}\BibitemShut {NoStop}%
\bibitem [{\citenamefont {Harrow}\ \emph {et~al.}(2009)\citenamefont {Harrow},
  \citenamefont {Hassidim},\ and\ \citenamefont {Lloyd}}]{Harrow:2009aa}%
  \BibitemOpen
  \bibfield  {author} {\bibinfo {author} {\bibfnamefont {Aram~W.}\ \bibnamefont
  {Harrow}}, \bibinfo {author} {\bibfnamefont {Avinatan}\ \bibnamefont
  {Hassidim}}, \ and\ \bibinfo {author} {\bibfnamefont {Seth}\ \bibnamefont
  {Lloyd}},\ }\bibfield  {title} {\enquote {\bibinfo {title} {Quantum algorithm
  for linear systems of equations},}\ }\href {\doibase
  10.1103/PhysRevLett.103.150502} {\bibfield  {journal} {\bibinfo  {journal}
  {Physical Review Letters}\ }\textbf {\bibinfo {volume} {103}},\ \bibinfo
  {pages} {150502--} (\bibinfo {year} {2009})}\BibitemShut {NoStop}%
\bibitem [{\citenamefont
  {Montanaro}(2016)}]{montanaroQuantumAlgorithmsOverview2016}%
  \BibitemOpen
  \bibfield  {author} {\bibinfo {author} {\bibfnamefont {Ashley}\ \bibnamefont
  {Montanaro}},\ }\bibfield  {title} {\enquote {\bibinfo {title} {Quantum
  algorithms: An overview},}\ }\href
  {https://www.nature.com/articles/npjqi201523} {\bibfield  {journal} {\bibinfo
   {journal} {NPJ Quantum Inf.}\ }\textbf {\bibinfo {volume} {2}},\ \bibinfo
  {pages} {15023} (\bibinfo {year} {2016})}\BibitemShut {NoStop}%
\bibitem [{\citenamefont {Bravyi}\ \emph {et~al.}(2018)\citenamefont {Bravyi},
  \citenamefont {Gosset},\ and\ \citenamefont {K{\"o}nig}}]{Bravyi:2017aa}%
  \BibitemOpen
  \bibfield  {author} {\bibinfo {author} {\bibfnamefont {Sergey}\ \bibnamefont
  {Bravyi}}, \bibinfo {author} {\bibfnamefont {David}\ \bibnamefont {Gosset}},
  \ and\ \bibinfo {author} {\bibfnamefont {Robert}\ \bibnamefont {K{\"o}nig}},\
  }\bibfield  {title} {\enquote {\bibinfo {title} {Quantum advantage with
  shallow circuits},}\ }\href {\doibase 10.1126/science.aar3106} {\bibfield
  {journal} {\bibinfo  {journal} {Science}\ }\textbf {\bibinfo {volume}
  {362}},\ \bibinfo {pages} {308} (\bibinfo {year} {2018})}\BibitemShut
  {NoStop}%
\bibitem [{\citenamefont {Bravyi}\ \emph {et~al.}(2020)\citenamefont {Bravyi},
  \citenamefont {Gosset}, \citenamefont {K{\"o}nig},\ and\ \citenamefont
  {Tomamichel}}]{Bravyi:2020aa}%
  \BibitemOpen
  \bibfield  {author} {\bibinfo {author} {\bibfnamefont {Sergey}\ \bibnamefont
  {Bravyi}}, \bibinfo {author} {\bibfnamefont {David}\ \bibnamefont {Gosset}},
  \bibinfo {author} {\bibfnamefont {Robert}\ \bibnamefont {K{\"o}nig}}, \ and\
  \bibinfo {author} {\bibfnamefont {Marco}\ \bibnamefont {Tomamichel}},\
  }\bibfield  {title} {\enquote {\bibinfo {title} {Quantum advantage with noisy
  shallow circuits},}\ }\href {\doibase 10.1038/s41567-020-0948-z} {\bibfield
  {journal} {\bibinfo  {journal} {Nature Physics}\ }\textbf {\bibinfo {volume}
  {16}},\ \bibinfo {pages} {1040--1045} (\bibinfo {year} {2020})}\BibitemShut
  {NoStop}%
\bibitem [{\citenamefont {Bharti}\ \emph {et~al.}(2022)\citenamefont {Bharti},
  \citenamefont {Cervera-Lierta}, \citenamefont {Kyaw}, \citenamefont {Haug},
  \citenamefont {Alperin-Lea}, \citenamefont {Anand}, \citenamefont {Degroote},
  \citenamefont {Heimonen}, \citenamefont {Kottmann}, \citenamefont {Menke},
  \citenamefont {Mok}, \citenamefont {Sim}, \citenamefont {Kwek},\ and\
  \citenamefont {Aspuru-Guzik}}]{Bharti:2022aa}%
  \BibitemOpen
  \bibfield  {author} {\bibinfo {author} {\bibfnamefont {Kishor}\ \bibnamefont
  {Bharti}}, \bibinfo {author} {\bibfnamefont {Alba}\ \bibnamefont
  {Cervera-Lierta}}, \bibinfo {author} {\bibfnamefont {Thi~Ha}\ \bibnamefont
  {Kyaw}}, \bibinfo {author} {\bibfnamefont {Tobias}\ \bibnamefont {Haug}},
  \bibinfo {author} {\bibfnamefont {Sumner}\ \bibnamefont {Alperin-Lea}},
  \bibinfo {author} {\bibfnamefont {Abhinav}\ \bibnamefont {Anand}}, \bibinfo
  {author} {\bibfnamefont {Matthias}\ \bibnamefont {Degroote}}, \bibinfo
  {author} {\bibfnamefont {Hermanni}\ \bibnamefont {Heimonen}}, \bibinfo
  {author} {\bibfnamefont {Jakob~S.}\ \bibnamefont {Kottmann}}, \bibinfo
  {author} {\bibfnamefont {Tim}\ \bibnamefont {Menke}}, \bibinfo {author}
  {\bibfnamefont {Wai-Keong}\ \bibnamefont {Mok}}, \bibinfo {author}
  {\bibfnamefont {Sukin}\ \bibnamefont {Sim}}, \bibinfo {author} {\bibfnamefont
  {Leong-Chuan}\ \bibnamefont {Kwek}}, \ and\ \bibinfo {author} {\bibfnamefont
  {Al{\'a}n}\ \bibnamefont {Aspuru-Guzik}},\ }\bibfield  {title} {\enquote
  {\bibinfo {title} {Noisy intermediate-scale quantum algorithms},}\ }\href
  {\doibase 10.1103/RevModPhys.94.015004} {\bibfield  {journal} {\bibinfo
  {journal} {Reviews of Modern Physics}\ }\textbf {\bibinfo {volume} {94}},\
  \bibinfo {pages} {015004--} (\bibinfo {year} {2022})}\BibitemShut {NoStop}%
\bibitem [{\citenamefont {Daley}\ \emph {et~al.}(2022)\citenamefont {Daley},
  \citenamefont {Bloch}, \citenamefont {Kokail}, \citenamefont {Flannigan},
  \citenamefont {Pearson}, \citenamefont {Troyer},\ and\ \citenamefont
  {Zoller}}]{Daley:2022vu}%
  \BibitemOpen
  \bibfield  {author} {\bibinfo {author} {\bibfnamefont {Andrew~J.}\
  \bibnamefont {Daley}}, \bibinfo {author} {\bibfnamefont {Immanuel}\
  \bibnamefont {Bloch}}, \bibinfo {author} {\bibfnamefont {Christian}\
  \bibnamefont {Kokail}}, \bibinfo {author} {\bibfnamefont {Stuart}\
  \bibnamefont {Flannigan}}, \bibinfo {author} {\bibfnamefont {Natalie}\
  \bibnamefont {Pearson}}, \bibinfo {author} {\bibfnamefont {Matthias}\
  \bibnamefont {Troyer}}, \ and\ \bibinfo {author} {\bibfnamefont {Peter}\
  \bibnamefont {Zoller}},\ }\bibfield  {title} {\enquote {\bibinfo {title}
  {Practical quantum advantage in quantum simulation},}\ }\href {\doibase
  10.1038/s41586-022-04940-6} {\bibfield  {journal} {\bibinfo  {journal}
  {Nature}\ }\textbf {\bibinfo {volume} {607}},\ \bibinfo {pages} {667--676}
  (\bibinfo {year} {2022})}\BibitemShut {NoStop}%
\bibitem [{\citenamefont {Preskill}(2018)}]{Preskill2018}%
  \BibitemOpen
  \bibfield  {author} {\bibinfo {author} {\bibfnamefont {John}\ \bibnamefont
  {Preskill}},\ }\bibfield  {title} {\enquote {\bibinfo {title} {Quantum
  {C}omputing in the {NISQ} era and beyond},}\ }\href {\doibase
  10.22331/q-2018-08-06-79} {\bibfield  {journal} {\bibinfo  {journal}
  {{Quantum}}\ }\textbf {\bibinfo {volume} {2}},\ \bibinfo {pages} {79}
  (\bibinfo {year} {2018})}\BibitemShut {NoStop}%
\bibitem [{\citenamefont {Albash}\ and\ \citenamefont
  {Lidar}(2018)}]{Albash:2017aa}%
  \BibitemOpen
  \bibfield  {author} {\bibinfo {author} {\bibfnamefont {Tameem}\ \bibnamefont
  {Albash}}\ and\ \bibinfo {author} {\bibfnamefont {Daniel~A.}\ \bibnamefont
  {Lidar}},\ }\bibfield  {title} {\enquote {\bibinfo {title} {Demonstration of
  a scaling advantage for a quantum annealer over simulated annealing},}\
  }\href {\doibase 10.1103/PhysRevX.8.031016} {\bibfield  {journal} {\bibinfo
  {journal} {Physical Review X}\ }\textbf {\bibinfo {volume} {8}},\ \bibinfo
  {pages} {031016--} (\bibinfo {year} {2018})}\BibitemShut {NoStop}%
\bibitem [{\citenamefont {King}\ \emph {et~al.}(2021)\citenamefont {King},
  \citenamefont {Raymond}, \citenamefont {Lanting}, \citenamefont {Isakov},
  \citenamefont {Mohseni}, \citenamefont {Poulin-Lamarre}, \citenamefont
  {Ejtemaee}, \citenamefont {Bernoudy}, \citenamefont {Ozfidan}, \citenamefont
  {Smirnov}, \citenamefont {Reis}, \citenamefont {Altomare}, \citenamefont
  {Babcock}, \citenamefont {Baron}, \citenamefont {Berkley}, \citenamefont
  {Boothby}, \citenamefont {Bunyk}, \citenamefont {Christiani}, \citenamefont
  {Enderud}, \citenamefont {Evert}, \citenamefont {Harris}, \citenamefont
  {Hoskinson}, \citenamefont {Huang}, \citenamefont {Jooya}, \citenamefont
  {Khodabandelou}, \citenamefont {Ladizinsky}, \citenamefont {Li},
  \citenamefont {Lott}, \citenamefont {MacDonald}, \citenamefont {Marsden},
  \citenamefont {Marsden}, \citenamefont {Medina}, \citenamefont {Molavi},
  \citenamefont {Neufeld}, \citenamefont {Norouzpour}, \citenamefont {Oh},
  \citenamefont {Pavlov}, \citenamefont {Perminov}, \citenamefont {Prescott},
  \citenamefont {Rich}, \citenamefont {Sato}, \citenamefont {Sheldan},
  \citenamefont {Sterling}, \citenamefont {Swenson}, \citenamefont {Tsai},
  \citenamefont {Volkmann}, \citenamefont {Whittaker}, \citenamefont
  {Wilkinson}, \citenamefont {Yao}, \citenamefont {Neven}, \citenamefont
  {Hilton}, \citenamefont {Ladizinsky}, \citenamefont {Johnson},\ and\
  \citenamefont {Amin}}]{King:2019aa}%
  \BibitemOpen
  \bibfield  {author} {\bibinfo {author} {\bibfnamefont {Andrew~D.}\
  \bibnamefont {King}}, \bibinfo {author} {\bibfnamefont {Jack}\ \bibnamefont
  {Raymond}}, \bibinfo {author} {\bibfnamefont {Trevor}\ \bibnamefont
  {Lanting}}, \bibinfo {author} {\bibfnamefont {Sergei~V.}\ \bibnamefont
  {Isakov}}, \bibinfo {author} {\bibfnamefont {Masoud}\ \bibnamefont
  {Mohseni}}, \bibinfo {author} {\bibfnamefont {Gabriel}\ \bibnamefont
  {Poulin-Lamarre}}, \bibinfo {author} {\bibfnamefont {Sara}\ \bibnamefont
  {Ejtemaee}}, \bibinfo {author} {\bibfnamefont {William}\ \bibnamefont
  {Bernoudy}}, \bibinfo {author} {\bibfnamefont {Isil}\ \bibnamefont
  {Ozfidan}}, \bibinfo {author} {\bibfnamefont {Anatoly~Yu.}\ \bibnamefont
  {Smirnov}}, \bibinfo {author} {\bibfnamefont {Mauricio}\ \bibnamefont
  {Reis}}, \bibinfo {author} {\bibfnamefont {Fabio}\ \bibnamefont {Altomare}},
  \bibinfo {author} {\bibfnamefont {Michael}\ \bibnamefont {Babcock}}, \bibinfo
  {author} {\bibfnamefont {Catia}\ \bibnamefont {Baron}}, \bibinfo {author}
  {\bibfnamefont {Andrew~J.}\ \bibnamefont {Berkley}}, \bibinfo {author}
  {\bibfnamefont {Kelly}\ \bibnamefont {Boothby}}, \bibinfo {author}
  {\bibfnamefont {Paul~I.}\ \bibnamefont {Bunyk}}, \bibinfo {author}
  {\bibfnamefont {Holly}\ \bibnamefont {Christiani}}, \bibinfo {author}
  {\bibfnamefont {Colin}\ \bibnamefont {Enderud}}, \bibinfo {author}
  {\bibfnamefont {Bram}\ \bibnamefont {Evert}}, \bibinfo {author}
  {\bibfnamefont {Richard}\ \bibnamefont {Harris}}, \bibinfo {author}
  {\bibfnamefont {Emile}\ \bibnamefont {Hoskinson}}, \bibinfo {author}
  {\bibfnamefont {Shuiyuan}\ \bibnamefont {Huang}}, \bibinfo {author}
  {\bibfnamefont {Kais}\ \bibnamefont {Jooya}}, \bibinfo {author}
  {\bibfnamefont {Ali}\ \bibnamefont {Khodabandelou}}, \bibinfo {author}
  {\bibfnamefont {Nicolas}\ \bibnamefont {Ladizinsky}}, \bibinfo {author}
  {\bibfnamefont {Ryan}\ \bibnamefont {Li}}, \bibinfo {author} {\bibfnamefont
  {P.~Aaron}\ \bibnamefont {Lott}}, \bibinfo {author} {\bibfnamefont {Allison
  J.~R.}\ \bibnamefont {MacDonald}}, \bibinfo {author} {\bibfnamefont {Danica}\
  \bibnamefont {Marsden}}, \bibinfo {author} {\bibfnamefont {Gaelen}\
  \bibnamefont {Marsden}}, \bibinfo {author} {\bibfnamefont {Teresa}\
  \bibnamefont {Medina}}, \bibinfo {author} {\bibfnamefont {Reza}\ \bibnamefont
  {Molavi}}, \bibinfo {author} {\bibfnamefont {Richard}\ \bibnamefont
  {Neufeld}}, \bibinfo {author} {\bibfnamefont {Mana}\ \bibnamefont
  {Norouzpour}}, \bibinfo {author} {\bibfnamefont {Travis}\ \bibnamefont {Oh}},
  \bibinfo {author} {\bibfnamefont {Igor}\ \bibnamefont {Pavlov}}, \bibinfo
  {author} {\bibfnamefont {Ilya}\ \bibnamefont {Perminov}}, \bibinfo {author}
  {\bibfnamefont {Thomas}\ \bibnamefont {Prescott}}, \bibinfo {author}
  {\bibfnamefont {Chris}\ \bibnamefont {Rich}}, \bibinfo {author}
  {\bibfnamefont {Yuki}\ \bibnamefont {Sato}}, \bibinfo {author} {\bibfnamefont
  {Benjamin}\ \bibnamefont {Sheldan}}, \bibinfo {author} {\bibfnamefont
  {George}\ \bibnamefont {Sterling}}, \bibinfo {author} {\bibfnamefont
  {Loren~J.}\ \bibnamefont {Swenson}}, \bibinfo {author} {\bibfnamefont
  {Nicholas}\ \bibnamefont {Tsai}}, \bibinfo {author} {\bibfnamefont {Mark~H.}\
  \bibnamefont {Volkmann}}, \bibinfo {author} {\bibfnamefont {Jed~D.}\
  \bibnamefont {Whittaker}}, \bibinfo {author} {\bibfnamefont {Warren}\
  \bibnamefont {Wilkinson}}, \bibinfo {author} {\bibfnamefont {Jason}\
  \bibnamefont {Yao}}, \bibinfo {author} {\bibfnamefont {Hartmut}\ \bibnamefont
  {Neven}}, \bibinfo {author} {\bibfnamefont {Jeremy~P.}\ \bibnamefont
  {Hilton}}, \bibinfo {author} {\bibfnamefont {Eric}\ \bibnamefont
  {Ladizinsky}}, \bibinfo {author} {\bibfnamefont {Mark~W.}\ \bibnamefont
  {Johnson}}, \ and\ \bibinfo {author} {\bibfnamefont {Mohammad~H.}\
  \bibnamefont {Amin}},\ }\bibfield  {title} {\enquote {\bibinfo {title}
  {Scaling advantage over path-integral monte carlo in quantum simulation of
  geometrically frustrated magnets},}\ }\href {\doibase
  10.1038/s41467-021-20901-5} {\bibfield  {journal} {\bibinfo  {journal}
  {Nature Communications}\ }\textbf {\bibinfo {volume} {12}},\ \bibinfo {pages}
  {1113} (\bibinfo {year} {2021})}\BibitemShut {NoStop}%
\bibitem [{\citenamefont {Saggio}\ \emph {et~al.}(2021)\citenamefont {Saggio},
  \citenamefont {Asenbeck}, \citenamefont {Hamann}, \citenamefont
  {Str{\"o}mberg}, \citenamefont {Schiansky}, \citenamefont {Dunjko},
  \citenamefont {Friis}, \citenamefont {Harris}, \citenamefont {Hochberg},
  \citenamefont {Englund}, \citenamefont {W{\"o}lk}, \citenamefont {Briegel},\
  and\ \citenamefont {Walther}}]{Saggio:2021vh}%
  \BibitemOpen
  \bibfield  {author} {\bibinfo {author} {\bibfnamefont {V.}~\bibnamefont
  {Saggio}}, \bibinfo {author} {\bibfnamefont {B.~E.}\ \bibnamefont
  {Asenbeck}}, \bibinfo {author} {\bibfnamefont {A.}~\bibnamefont {Hamann}},
  \bibinfo {author} {\bibfnamefont {T.}~\bibnamefont {Str{\"o}mberg}}, \bibinfo
  {author} {\bibfnamefont {P.}~\bibnamefont {Schiansky}}, \bibinfo {author}
  {\bibfnamefont {V.}~\bibnamefont {Dunjko}}, \bibinfo {author} {\bibfnamefont
  {N.}~\bibnamefont {Friis}}, \bibinfo {author} {\bibfnamefont {N.~C.}\
  \bibnamefont {Harris}}, \bibinfo {author} {\bibfnamefont {M.}~\bibnamefont
  {Hochberg}}, \bibinfo {author} {\bibfnamefont {D.}~\bibnamefont {Englund}},
  \bibinfo {author} {\bibfnamefont {S.}~\bibnamefont {W{\"o}lk}}, \bibinfo
  {author} {\bibfnamefont {H.~J.}\ \bibnamefont {Briegel}}, \ and\ \bibinfo
  {author} {\bibfnamefont {P.}~\bibnamefont {Walther}},\ }\bibfield  {title}
  {\enquote {\bibinfo {title} {Experimental quantum speed-up in reinforcement
  learning agents},}\ }\href {\doibase 10.1038/s41586-021-03242-7} {\bibfield
  {journal} {\bibinfo  {journal} {Nature}\ }\textbf {\bibinfo {volume} {591}},\
  \bibinfo {pages} {229--233} (\bibinfo {year} {2021})}\BibitemShut {NoStop}%
\bibitem [{\citenamefont {Centrone}\ \emph {et~al.}(2021)\citenamefont
  {Centrone}, \citenamefont {Kumar}, \citenamefont {Diamanti},\ and\
  \citenamefont {Kerenidis}}]{Centrone:2021tq}%
  \BibitemOpen
  \bibfield  {author} {\bibinfo {author} {\bibfnamefont {Federico}\
  \bibnamefont {Centrone}}, \bibinfo {author} {\bibfnamefont {Niraj}\
  \bibnamefont {Kumar}}, \bibinfo {author} {\bibfnamefont {Eleni}\ \bibnamefont
  {Diamanti}}, \ and\ \bibinfo {author} {\bibfnamefont {Iordanis}\ \bibnamefont
  {Kerenidis}},\ }\bibfield  {title} {\enquote {\bibinfo {title} {Experimental
  demonstration of quantum advantage for np verification with limited
  information},}\ }\href {https://doi.org/10.1038/s41467-021-21119-1}
  {\bibfield  {journal} {\bibinfo  {journal} {Nature Communications}\ }\textbf
  {\bibinfo {volume} {12}},\ \bibinfo {pages} {850} (\bibinfo {year}
  {2021})}\BibitemShut {NoStop}%
\bibitem [{\citenamefont {Maslov}\ \emph {et~al.}(2021)\citenamefont {Maslov},
  \citenamefont {Kim}, \citenamefont {Bravyi}, \citenamefont {Yoder},\ and\
  \citenamefont {Sheldon}}]{Maslov:2021aa}%
  \BibitemOpen
  \bibfield  {author} {\bibinfo {author} {\bibfnamefont {Dmitri}\ \bibnamefont
  {Maslov}}, \bibinfo {author} {\bibfnamefont {Jin-Sung}\ \bibnamefont {Kim}},
  \bibinfo {author} {\bibfnamefont {Sergey}\ \bibnamefont {Bravyi}}, \bibinfo
  {author} {\bibfnamefont {Theodore~J.}\ \bibnamefont {Yoder}}, \ and\ \bibinfo
  {author} {\bibfnamefont {Sarah}\ \bibnamefont {Sheldon}},\ }\bibfield
  {title} {\enquote {\bibinfo {title} {Quantum advantage for computations with
  limited space},}\ }\href {\doibase 10.1038/s41567-021-01271-7} {\bibfield
  {journal} {\bibinfo  {journal} {Nature Physics}\ }\textbf {\bibinfo {volume}
  {17}},\ \bibinfo {pages} {894--897} (\bibinfo {year} {2021})}\BibitemShut
  {NoStop}%
\bibitem [{\citenamefont {Xia}\ \emph {et~al.}(2021)\citenamefont {Xia},
  \citenamefont {Li}, \citenamefont {Zhuang},\ and\ \citenamefont
  {Zhang}}]{Xia:2021ux}%
  \BibitemOpen
  \bibfield  {author} {\bibinfo {author} {\bibfnamefont {Yi}~\bibnamefont
  {Xia}}, \bibinfo {author} {\bibfnamefont {Wei}\ \bibnamefont {Li}}, \bibinfo
  {author} {\bibfnamefont {Quntao}\ \bibnamefont {Zhuang}}, \ and\ \bibinfo
  {author} {\bibfnamefont {Zheshen}\ \bibnamefont {Zhang}},\ }\bibfield
  {title} {\enquote {\bibinfo {title} {Quantum-enhanced data classification
  with a variational entangled sensor network},}\ }\href
  {https://link.aps.org/doi/10.1103/PhysRevX.11.021047} {\bibfield  {journal}
  {\bibinfo  {journal} {Physical Review X}\ }\textbf {\bibinfo {volume} {11}},\
  \bibinfo {pages} {021047--} (\bibinfo {year} {2021})}\BibitemShut {NoStop}%
\bibitem [{\citenamefont {Huang}\ \emph {et~al.}(2022)\citenamefont {Huang},
  \citenamefont {Broughton}, \citenamefont {Cotler}, \citenamefont {Chen},
  \citenamefont {Li}, \citenamefont {Mohseni}, \citenamefont {Neven},
  \citenamefont {Babbush}, \citenamefont {Kueng}, \citenamefont {Preskill},\
  and\ \citenamefont {McClean}}]{Huang:2021}%
  \BibitemOpen
  \bibfield  {author} {\bibinfo {author} {\bibfnamefont {Hsin-Yuan}\
  \bibnamefont {Huang}}, \bibinfo {author} {\bibfnamefont {Michael}\
  \bibnamefont {Broughton}}, \bibinfo {author} {\bibfnamefont {Jordan}\
  \bibnamefont {Cotler}}, \bibinfo {author} {\bibfnamefont {Sitan}\
  \bibnamefont {Chen}}, \bibinfo {author} {\bibfnamefont {Jerry}\ \bibnamefont
  {Li}}, \bibinfo {author} {\bibfnamefont {Masoud}\ \bibnamefont {Mohseni}},
  \bibinfo {author} {\bibfnamefont {Hartmut}\ \bibnamefont {Neven}}, \bibinfo
  {author} {\bibfnamefont {Ryan}\ \bibnamefont {Babbush}}, \bibinfo {author}
  {\bibfnamefont {Richard}\ \bibnamefont {Kueng}}, \bibinfo {author}
  {\bibfnamefont {John}\ \bibnamefont {Preskill}}, \ and\ \bibinfo {author}
  {\bibfnamefont {Jarrod~R.}\ \bibnamefont {McClean}},\ }\bibfield  {title}
  {\enquote {\bibinfo {title} {Quantum advantage in learning from
  experiments},}\ }\href {https://doi.org/10.1126/science.abn7293} {\bibfield
  {journal} {\bibinfo  {journal} {Science}\ }\textbf {\bibinfo {volume}
  {376}},\ \bibinfo {pages} {1182--1186} (\bibinfo {year} {2022})}\BibitemShut
  {NoStop}%
\bibitem [{\citenamefont {Ebadi}\ \emph {et~al.}(2022)\citenamefont {Ebadi},
  \citenamefont {Keesling}, \citenamefont {Cain}, \citenamefont {Wang},
  \citenamefont {Levine}, \citenamefont {Bluvstein}, \citenamefont {Semeghini},
  \citenamefont {Omran}, \citenamefont {Liu}, \citenamefont {Samajdar},
  \citenamefont {Luo}, \citenamefont {Nash}, \citenamefont {Gao}, \citenamefont
  {Barak}, \citenamefont {Farhi}, \citenamefont {Sachdev}, \citenamefont
  {Gemelke}, \citenamefont {Zhou}, \citenamefont {Choi}, \citenamefont
  {Pichler}, \citenamefont {Wang}, \citenamefont {Greiner}, \citenamefont
  {Vuletic},\ and\ \citenamefont {Lukin}}]{Ebadi:22}%
  \BibitemOpen
  \bibfield  {author} {\bibinfo {author} {\bibfnamefont {S.}~\bibnamefont
  {Ebadi}}, \bibinfo {author} {\bibfnamefont {A.}~\bibnamefont {Keesling}},
  \bibinfo {author} {\bibfnamefont {M.}~\bibnamefont {Cain}}, \bibinfo {author}
  {\bibfnamefont {T.~T.}\ \bibnamefont {Wang}}, \bibinfo {author}
  {\bibfnamefont {H.}~\bibnamefont {Levine}}, \bibinfo {author} {\bibfnamefont
  {D.}~\bibnamefont {Bluvstein}}, \bibinfo {author} {\bibfnamefont
  {G.}~\bibnamefont {Semeghini}}, \bibinfo {author} {\bibfnamefont
  {A.}~\bibnamefont {Omran}}, \bibinfo {author} {\bibfnamefont {J.-G.}\
  \bibnamefont {Liu}}, \bibinfo {author} {\bibfnamefont {R.}~\bibnamefont
  {Samajdar}}, \bibinfo {author} {\bibfnamefont {X.-Z.}\ \bibnamefont {Luo}},
  \bibinfo {author} {\bibfnamefont {B.}~\bibnamefont {Nash}}, \bibinfo {author}
  {\bibfnamefont {X.}~\bibnamefont {Gao}}, \bibinfo {author} {\bibfnamefont
  {B.}~\bibnamefont {Barak}}, \bibinfo {author} {\bibfnamefont
  {E.}~\bibnamefont {Farhi}}, \bibinfo {author} {\bibfnamefont
  {S.}~\bibnamefont {Sachdev}}, \bibinfo {author} {\bibfnamefont
  {N.}~\bibnamefont {Gemelke}}, \bibinfo {author} {\bibfnamefont
  {L.}~\bibnamefont {Zhou}}, \bibinfo {author} {\bibfnamefont {S.}~\bibnamefont
  {Choi}}, \bibinfo {author} {\bibfnamefont {H.}~\bibnamefont {Pichler}},
  \bibinfo {author} {\bibfnamefont {S.-T.}\ \bibnamefont {Wang}}, \bibinfo
  {author} {\bibfnamefont {M.}~\bibnamefont {Greiner}}, \bibinfo {author}
  {\bibfnamefont {V.}~\bibnamefont {Vuletic}}, \ and\ \bibinfo {author}
  {\bibfnamefont {M.~D.}\ \bibnamefont {Lukin}},\ }\bibfield  {title} {\enquote
  {\bibinfo {title} {Quantum optimization of maximum independent set using
  rydberg atom arrays},}\ }\href {https://doi.org/10.1126/science.abo6587}
  {\bibfield  {journal} {\bibinfo  {journal} {Science}\ }\textbf {\bibinfo
  {volume} {376}},\ \bibinfo {pages} {1209--1215} (\bibinfo {year}
  {2022})}\BibitemShut {NoStop}%
\bibitem [{\citenamefont {Zhou}\ \emph {et~al.}(2022)\citenamefont {Zhou},
  \citenamefont {Cao}, \citenamefont {Lu}, \citenamefont {Wang}, \citenamefont
  {Bao}, \citenamefont {Jia}, \citenamefont {Fu}, \citenamefont {Yin},\ and\
  \citenamefont {Chen}}]{zhouExperimentalQuantumAdvantage2022}%
  \BibitemOpen
  \bibfield  {author} {\bibinfo {author} {\bibfnamefont {Min-Gang}\
  \bibnamefont {Zhou}}, \bibinfo {author} {\bibfnamefont {Xiao-Yu}\
  \bibnamefont {Cao}}, \bibinfo {author} {\bibfnamefont {Yu-Shuo}\ \bibnamefont
  {Lu}}, \bibinfo {author} {\bibfnamefont {Yang}\ \bibnamefont {Wang}},
  \bibinfo {author} {\bibfnamefont {Yu}~\bibnamefont {Bao}}, \bibinfo {author}
  {\bibfnamefont {Zhao-Ying}\ \bibnamefont {Jia}}, \bibinfo {author}
  {\bibfnamefont {Yao}\ \bibnamefont {Fu}}, \bibinfo {author} {\bibfnamefont
  {Hua-Lei}\ \bibnamefont {Yin}}, \ and\ \bibinfo {author} {\bibfnamefont
  {Zeng-Bing}\ \bibnamefont {Chen}},\ }\bibfield  {title} {\enquote {\bibinfo
  {title} {Experimental {{Quantum Advantage}} with {{Quantum Coupon
  Collector}}},}\ }\href
  {https://spj.sciencemag.org/journals/research/2022/9798679/} {\bibfield
  {journal} {\bibinfo  {journal} {Research}\ }\textbf {\bibinfo {volume}
  {2022}} (\bibinfo {year} {2022})}\BibitemShut {NoStop}%
\bibitem [{\citenamefont {King}\ \emph {et~al.}(2023)\citenamefont {King},
  \citenamefont {Raymond}, \citenamefont {Lanting}, \citenamefont {Harris},
  \citenamefont {Zucca}, \citenamefont {Altomare}, \citenamefont {Berkley},
  \citenamefont {Boothby}, \citenamefont {Ejtemaee}, \citenamefont {Enderud},
  \citenamefont {Hoskinson}, \citenamefont {Huang}, \citenamefont {Ladizinsky},
  \citenamefont {MacDonald}, \citenamefont {Marsden}, \citenamefont {Molavi},
  \citenamefont {Oh}, \citenamefont {Poulin-Lamarre}, \citenamefont {Reis},
  \citenamefont {Rich}, \citenamefont {Sato}, \citenamefont {Tsai},
  \citenamefont {Volkmann}, \citenamefont {Whittaker}, \citenamefont {Yao},
  \citenamefont {Sandvik},\ and\ \citenamefont {Amin}}]{King:22}%
  \BibitemOpen
  \bibfield  {author} {\bibinfo {author} {\bibfnamefont {Andrew~D.}\
  \bibnamefont {King}}, \bibinfo {author} {\bibfnamefont {Jack}\ \bibnamefont
  {Raymond}}, \bibinfo {author} {\bibfnamefont {Trevor}\ \bibnamefont
  {Lanting}}, \bibinfo {author} {\bibfnamefont {Richard}\ \bibnamefont
  {Harris}}, \bibinfo {author} {\bibfnamefont {Alex}\ \bibnamefont {Zucca}},
  \bibinfo {author} {\bibfnamefont {Fabio}\ \bibnamefont {Altomare}}, \bibinfo
  {author} {\bibfnamefont {Andrew~J.}\ \bibnamefont {Berkley}}, \bibinfo
  {author} {\bibfnamefont {Kelly}\ \bibnamefont {Boothby}}, \bibinfo {author}
  {\bibfnamefont {Sara}\ \bibnamefont {Ejtemaee}}, \bibinfo {author}
  {\bibfnamefont {Colin}\ \bibnamefont {Enderud}}, \bibinfo {author}
  {\bibfnamefont {Emile}\ \bibnamefont {Hoskinson}}, \bibinfo {author}
  {\bibfnamefont {Shuiyuan}\ \bibnamefont {Huang}}, \bibinfo {author}
  {\bibfnamefont {Eric}\ \bibnamefont {Ladizinsky}}, \bibinfo {author}
  {\bibfnamefont {Allison J.~R.}\ \bibnamefont {MacDonald}}, \bibinfo {author}
  {\bibfnamefont {Gaelen}\ \bibnamefont {Marsden}}, \bibinfo {author}
  {\bibfnamefont {Reza}\ \bibnamefont {Molavi}}, \bibinfo {author}
  {\bibfnamefont {Travis}\ \bibnamefont {Oh}}, \bibinfo {author} {\bibfnamefont
  {Gabriel}\ \bibnamefont {Poulin-Lamarre}}, \bibinfo {author} {\bibfnamefont
  {Mauricio}\ \bibnamefont {Reis}}, \bibinfo {author} {\bibfnamefont {Chris}\
  \bibnamefont {Rich}}, \bibinfo {author} {\bibfnamefont {Yuki}\ \bibnamefont
  {Sato}}, \bibinfo {author} {\bibfnamefont {Nicholas}\ \bibnamefont {Tsai}},
  \bibinfo {author} {\bibfnamefont {Mark}\ \bibnamefont {Volkmann}}, \bibinfo
  {author} {\bibfnamefont {Jed~D.}\ \bibnamefont {Whittaker}}, \bibinfo
  {author} {\bibfnamefont {Jason}\ \bibnamefont {Yao}}, \bibinfo {author}
  {\bibfnamefont {Anders~W.}\ \bibnamefont {Sandvik}}, \ and\ \bibinfo {author}
  {\bibfnamefont {Mohammad~H.}\ \bibnamefont {Amin}},\ }\bibfield  {title}
  {\enquote {\bibinfo {title} {Quantum critical dynamics in a 5,000-qubit
  programmable spin glass},}\ }\href {\doibase 10.1038/s41586-023-05867-2}
  {\bibfield  {journal} {\bibinfo  {journal} {Nature}\ }\textbf {\bibinfo
  {volume} {617}},\ \bibinfo {pages} {61--66} (\bibinfo {year}
  {2023})}\BibitemShut {NoStop}%
\bibitem [{\citenamefont {Kim}\ \emph {et~al.}(2023)\citenamefont {Kim},
  \citenamefont {Eddins}, \citenamefont {Anand}, \citenamefont {Wei},
  \citenamefont {van~den Berg}, \citenamefont {Rosenblatt}, \citenamefont
  {Nayfeh}, \citenamefont {Wu}, \citenamefont {Zaletel}, \citenamefont
  {Temme},\ and\ \citenamefont {Kandala}}]{Kim:2023aa}%
  \BibitemOpen
  \bibfield  {author} {\bibinfo {author} {\bibfnamefont {Youngseok}\
  \bibnamefont {Kim}}, \bibinfo {author} {\bibfnamefont {Andrew}\ \bibnamefont
  {Eddins}}, \bibinfo {author} {\bibfnamefont {Sajant}\ \bibnamefont {Anand}},
  \bibinfo {author} {\bibfnamefont {Ken~Xuan}\ \bibnamefont {Wei}}, \bibinfo
  {author} {\bibfnamefont {Ewout}\ \bibnamefont {van~den Berg}}, \bibinfo
  {author} {\bibfnamefont {Sami}\ \bibnamefont {Rosenblatt}}, \bibinfo {author}
  {\bibfnamefont {Hasan}\ \bibnamefont {Nayfeh}}, \bibinfo {author}
  {\bibfnamefont {Yantao}\ \bibnamefont {Wu}}, \bibinfo {author} {\bibfnamefont
  {Michael}\ \bibnamefont {Zaletel}}, \bibinfo {author} {\bibfnamefont
  {Kristan}\ \bibnamefont {Temme}}, \ and\ \bibinfo {author} {\bibfnamefont
  {Abhinav}\ \bibnamefont {Kandala}},\ }\bibfield  {title} {\enquote {\bibinfo
  {title} {Evidence for the utility of quantum computing before fault
  tolerance},}\ }\href {\doibase 10.1038/s41586-023-06096-3} {\bibfield
  {journal} {\bibinfo  {journal} {Nature}\ }\textbf {\bibinfo {volume} {618}},\
  \bibinfo {pages} {500--505} (\bibinfo {year} {2023})}\BibitemShut {NoStop}%
\bibitem [{\citenamefont {Aaronson}\ and\ \citenamefont
  {Chen}(2017)}]{aaronson2016}%
  \BibitemOpen
  \bibfield  {author} {\bibinfo {author} {\bibfnamefont {Scott}\ \bibnamefont
  {Aaronson}}\ and\ \bibinfo {author} {\bibfnamefont {Lijie}\ \bibnamefont
  {Chen}},\ }\bibfield  {title} {\enquote {\bibinfo {title}
  {Complexity-theoretic foundations of quantum supremacy experiments},}\ }in\
  \href {https://arxiv.org/abs/1612.05903} {\emph {\bibinfo {booktitle}
  {Proceedings of the 32nd Computational Complexity Conference}}},\ \bibinfo
  {series and number} {CCC '17}\ (\bibinfo  {publisher} {Schloss
  Dagstuhl--Leibniz-Zentrum fuer Informatik},\ \bibinfo {address} {Dagstuhl,
  DEU},\ \bibinfo {year} {2017})\BibitemShut {NoStop}%
\bibitem [{\citenamefont {Arute}\ \emph {et~al.}(2019)\citenamefont {Arute},
  \citenamefont {Arya}, \citenamefont {Babbush}, \citenamefont {Bacon},
  \citenamefont {Bardin}, \citenamefont {Barends}, \citenamefont {Biswas},
  \citenamefont {Boixo}, \citenamefont {Brandao}, \citenamefont {Buell},
  \citenamefont {Burkett}, \citenamefont {Chen}, \citenamefont {Chen},
  \citenamefont {Chiaro}, \citenamefont {Collins}, \citenamefont {Courtney},
  \citenamefont {Dunsworth}, \citenamefont {Farhi}, \citenamefont {Foxen},
  \citenamefont {Fowler}, \citenamefont {Gidney}, \citenamefont {Giustina},
  \citenamefont {Graff}, \citenamefont {Guerin}, \citenamefont {Habegger},
  \citenamefont {Harrigan}, \citenamefont {Hartmann}, \citenamefont {Ho},
  \citenamefont {Hoffmann}, \citenamefont {Huang}, \citenamefont {Humble},
  \citenamefont {Isakov}, \citenamefont {Jeffrey}, \citenamefont {Jiang},
  \citenamefont {Kafri}, \citenamefont {Kechedzhi}, \citenamefont {Kelly},
  \citenamefont {Klimov}, \citenamefont {Knysh}, \citenamefont {Korotkov},
  \citenamefont {Kostritsa}, \citenamefont {Landhuis}, \citenamefont
  {Lindmark}, \citenamefont {Lucero}, \citenamefont {Lyakh}, \citenamefont
  {Mandr{\`a}}, \citenamefont {McClean}, \citenamefont {McEwen}, \citenamefont
  {Megrant}, \citenamefont {Mi}, \citenamefont {Michielsen}, \citenamefont
  {Mohseni}, \citenamefont {Mutus}, \citenamefont {Naaman}, \citenamefont
  {Neeley}, \citenamefont {Neill}, \citenamefont {Niu}, \citenamefont {Ostby},
  \citenamefont {Petukhov}, \citenamefont {Platt}, \citenamefont {Quintana},
  \citenamefont {Rieffel}, \citenamefont {Roushan}, \citenamefont {Rubin},
  \citenamefont {Sank}, \citenamefont {Satzinger}, \citenamefont {Smelyanskiy},
  \citenamefont {Sung}, \citenamefont {Trevithick}, \citenamefont
  {Vainsencher}, \citenamefont {Villalonga}, \citenamefont {White},
  \citenamefont {Yao}, \citenamefont {Yeh}, \citenamefont {Zalcman},
  \citenamefont {Neven},\ and\ \citenamefont {Martinis}}]{Arute:2019aa}%
  \BibitemOpen
  \bibfield  {author} {\bibinfo {author} {\bibfnamefont {Frank}\ \bibnamefont
  {Arute}}, \bibinfo {author} {\bibfnamefont {Kunal}\ \bibnamefont {Arya}},
  \bibinfo {author} {\bibfnamefont {Ryan}\ \bibnamefont {Babbush}}, \bibinfo
  {author} {\bibfnamefont {Dave}\ \bibnamefont {Bacon}}, \bibinfo {author}
  {\bibfnamefont {Joseph~C.}\ \bibnamefont {Bardin}}, \bibinfo {author}
  {\bibfnamefont {Rami}\ \bibnamefont {Barends}}, \bibinfo {author}
  {\bibfnamefont {Rupak}\ \bibnamefont {Biswas}}, \bibinfo {author}
  {\bibfnamefont {Sergio}\ \bibnamefont {Boixo}}, \bibinfo {author}
  {\bibfnamefont {Fernando G. S.~L.}\ \bibnamefont {Brandao}}, \bibinfo
  {author} {\bibfnamefont {David~A.}\ \bibnamefont {Buell}}, \bibinfo {author}
  {\bibfnamefont {Brian}\ \bibnamefont {Burkett}}, \bibinfo {author}
  {\bibfnamefont {Yu}~\bibnamefont {Chen}}, \bibinfo {author} {\bibfnamefont
  {Zijun}\ \bibnamefont {Chen}}, \bibinfo {author} {\bibfnamefont {Ben}\
  \bibnamefont {Chiaro}}, \bibinfo {author} {\bibfnamefont {Roberto}\
  \bibnamefont {Collins}}, \bibinfo {author} {\bibfnamefont {William}\
  \bibnamefont {Courtney}}, \bibinfo {author} {\bibfnamefont {Andrew}\
  \bibnamefont {Dunsworth}}, \bibinfo {author} {\bibfnamefont {Edward}\
  \bibnamefont {Farhi}}, \bibinfo {author} {\bibfnamefont {Brooks}\
  \bibnamefont {Foxen}}, \bibinfo {author} {\bibfnamefont {Austin}\
  \bibnamefont {Fowler}}, \bibinfo {author} {\bibfnamefont {Craig}\
  \bibnamefont {Gidney}}, \bibinfo {author} {\bibfnamefont {Marissa}\
  \bibnamefont {Giustina}}, \bibinfo {author} {\bibfnamefont {Rob}\
  \bibnamefont {Graff}}, \bibinfo {author} {\bibfnamefont {Keith}\ \bibnamefont
  {Guerin}}, \bibinfo {author} {\bibfnamefont {Steve}\ \bibnamefont
  {Habegger}}, \bibinfo {author} {\bibfnamefont {Matthew~P.}\ \bibnamefont
  {Harrigan}}, \bibinfo {author} {\bibfnamefont {Michael~J.}\ \bibnamefont
  {Hartmann}}, \bibinfo {author} {\bibfnamefont {Alan}\ \bibnamefont {Ho}},
  \bibinfo {author} {\bibfnamefont {Markus}\ \bibnamefont {Hoffmann}}, \bibinfo
  {author} {\bibfnamefont {Trent}\ \bibnamefont {Huang}}, \bibinfo {author}
  {\bibfnamefont {Travis~S.}\ \bibnamefont {Humble}}, \bibinfo {author}
  {\bibfnamefont {Sergei~V.}\ \bibnamefont {Isakov}}, \bibinfo {author}
  {\bibfnamefont {Evan}\ \bibnamefont {Jeffrey}}, \bibinfo {author}
  {\bibfnamefont {Zhang}\ \bibnamefont {Jiang}}, \bibinfo {author}
  {\bibfnamefont {Dvir}\ \bibnamefont {Kafri}}, \bibinfo {author}
  {\bibfnamefont {Kostyantyn}\ \bibnamefont {Kechedzhi}}, \bibinfo {author}
  {\bibfnamefont {Julian}\ \bibnamefont {Kelly}}, \bibinfo {author}
  {\bibfnamefont {Paul~V.}\ \bibnamefont {Klimov}}, \bibinfo {author}
  {\bibfnamefont {Sergey}\ \bibnamefont {Knysh}}, \bibinfo {author}
  {\bibfnamefont {Alexander}\ \bibnamefont {Korotkov}}, \bibinfo {author}
  {\bibfnamefont {Fedor}\ \bibnamefont {Kostritsa}}, \bibinfo {author}
  {\bibfnamefont {David}\ \bibnamefont {Landhuis}}, \bibinfo {author}
  {\bibfnamefont {Mike}\ \bibnamefont {Lindmark}}, \bibinfo {author}
  {\bibfnamefont {Erik}\ \bibnamefont {Lucero}}, \bibinfo {author}
  {\bibfnamefont {Dmitry}\ \bibnamefont {Lyakh}}, \bibinfo {author}
  {\bibfnamefont {Salvatore}\ \bibnamefont {Mandr{\`a}}}, \bibinfo {author}
  {\bibfnamefont {Jarrod~R.}\ \bibnamefont {McClean}}, \bibinfo {author}
  {\bibfnamefont {Matthew}\ \bibnamefont {McEwen}}, \bibinfo {author}
  {\bibfnamefont {Anthony}\ \bibnamefont {Megrant}}, \bibinfo {author}
  {\bibfnamefont {Xiao}\ \bibnamefont {Mi}}, \bibinfo {author} {\bibfnamefont
  {Kristel}\ \bibnamefont {Michielsen}}, \bibinfo {author} {\bibfnamefont
  {Masoud}\ \bibnamefont {Mohseni}}, \bibinfo {author} {\bibfnamefont {Josh}\
  \bibnamefont {Mutus}}, \bibinfo {author} {\bibfnamefont {Ofer}\ \bibnamefont
  {Naaman}}, \bibinfo {author} {\bibfnamefont {Matthew}\ \bibnamefont
  {Neeley}}, \bibinfo {author} {\bibfnamefont {Charles}\ \bibnamefont {Neill}},
  \bibinfo {author} {\bibfnamefont {Murphy~Yuezhen}\ \bibnamefont {Niu}},
  \bibinfo {author} {\bibfnamefont {Eric}\ \bibnamefont {Ostby}}, \bibinfo
  {author} {\bibfnamefont {Andre}\ \bibnamefont {Petukhov}}, \bibinfo {author}
  {\bibfnamefont {John~C.}\ \bibnamefont {Platt}}, \bibinfo {author}
  {\bibfnamefont {Chris}\ \bibnamefont {Quintana}}, \bibinfo {author}
  {\bibfnamefont {Eleanor~G.}\ \bibnamefont {Rieffel}}, \bibinfo {author}
  {\bibfnamefont {Pedram}\ \bibnamefont {Roushan}}, \bibinfo {author}
  {\bibfnamefont {Nicholas~C.}\ \bibnamefont {Rubin}}, \bibinfo {author}
  {\bibfnamefont {Daniel}\ \bibnamefont {Sank}}, \bibinfo {author}
  {\bibfnamefont {Kevin~J.}\ \bibnamefont {Satzinger}}, \bibinfo {author}
  {\bibfnamefont {Vadim}\ \bibnamefont {Smelyanskiy}}, \bibinfo {author}
  {\bibfnamefont {Kevin~J.}\ \bibnamefont {Sung}}, \bibinfo {author}
  {\bibfnamefont {Matthew~D.}\ \bibnamefont {Trevithick}}, \bibinfo {author}
  {\bibfnamefont {Amit}\ \bibnamefont {Vainsencher}}, \bibinfo {author}
  {\bibfnamefont {Benjamin}\ \bibnamefont {Villalonga}}, \bibinfo {author}
  {\bibfnamefont {Theodore}\ \bibnamefont {White}}, \bibinfo {author}
  {\bibfnamefont {Z.~Jamie}\ \bibnamefont {Yao}}, \bibinfo {author}
  {\bibfnamefont {Ping}\ \bibnamefont {Yeh}}, \bibinfo {author} {\bibfnamefont
  {Adam}\ \bibnamefont {Zalcman}}, \bibinfo {author} {\bibfnamefont {Hartmut}\
  \bibnamefont {Neven}}, \ and\ \bibinfo {author} {\bibfnamefont {John~M.}\
  \bibnamefont {Martinis}},\ }\bibfield  {title} {\enquote {\bibinfo {title}
  {Quantum supremacy using a programmable superconducting processor},}\ }\href
  {https://doi.org/10.1038/s41586-019-1666-5} {\bibfield  {journal} {\bibinfo
  {journal} {Nature}\ }\textbf {\bibinfo {volume} {574}},\ \bibinfo {pages}
  {505--510} (\bibinfo {year} {2019})}\BibitemShut {NoStop}%
\bibitem [{\citenamefont {Wu}\ \emph {et~al.}(2021)\citenamefont {Wu},
  \citenamefont {Bao}, \citenamefont {Cao}, \citenamefont {Chen}, \citenamefont
  {Chen}, \citenamefont {Chen}, \citenamefont {Chung}, \citenamefont {Deng},
  \citenamefont {Du}, \citenamefont {Fan}, \citenamefont {Gong}, \citenamefont
  {Guo}, \citenamefont {Guo}, \citenamefont {Guo}, \citenamefont {Han},
  \citenamefont {Hong}, \citenamefont {Huang}, \citenamefont {Huo},
  \citenamefont {Li}, \citenamefont {Li}, \citenamefont {Li}, \citenamefont
  {Li}, \citenamefont {Liang}, \citenamefont {Lin}, \citenamefont {Lin},
  \citenamefont {Qian}, \citenamefont {Qiao}, \citenamefont {Rong},
  \citenamefont {Su}, \citenamefont {Sun}, \citenamefont {Wang}, \citenamefont
  {Wang}, \citenamefont {Wu}, \citenamefont {Xu}, \citenamefont {Yan},
  \citenamefont {Yang}, \citenamefont {Yang}, \citenamefont {Ye}, \citenamefont
  {Yin}, \citenamefont {Ying}, \citenamefont {Yu}, \citenamefont {Zha},
  \citenamefont {Zhang}, \citenamefont {Zhang}, \citenamefont {Zhang},
  \citenamefont {Zhang}, \citenamefont {Zhao}, \citenamefont {Zhao},
  \citenamefont {Zhou}, \citenamefont {Zhu}, \citenamefont {Lu}, \citenamefont
  {Peng}, \citenamefont {Zhu},\ and\ \citenamefont {Pan}}]{wu2021strong}%
  \BibitemOpen
  \bibfield  {author} {\bibinfo {author} {\bibfnamefont {Yulin}\ \bibnamefont
  {Wu}}, \bibinfo {author} {\bibfnamefont {Wan-Su}\ \bibnamefont {Bao}},
  \bibinfo {author} {\bibfnamefont {Sirui}\ \bibnamefont {Cao}}, \bibinfo
  {author} {\bibfnamefont {Fusheng}\ \bibnamefont {Chen}}, \bibinfo {author}
  {\bibfnamefont {Ming-Cheng}\ \bibnamefont {Chen}}, \bibinfo {author}
  {\bibfnamefont {Xiawei}\ \bibnamefont {Chen}}, \bibinfo {author}
  {\bibfnamefont {Tung-Hsun}\ \bibnamefont {Chung}}, \bibinfo {author}
  {\bibfnamefont {Hui}\ \bibnamefont {Deng}}, \bibinfo {author} {\bibfnamefont
  {Yajie}\ \bibnamefont {Du}}, \bibinfo {author} {\bibfnamefont {Daojin}\
  \bibnamefont {Fan}}, \bibinfo {author} {\bibfnamefont {Ming}\ \bibnamefont
  {Gong}}, \bibinfo {author} {\bibfnamefont {Cheng}\ \bibnamefont {Guo}},
  \bibinfo {author} {\bibfnamefont {Chu}\ \bibnamefont {Guo}}, \bibinfo
  {author} {\bibfnamefont {Shaojun}\ \bibnamefont {Guo}}, \bibinfo {author}
  {\bibfnamefont {Lianchen}\ \bibnamefont {Han}}, \bibinfo {author}
  {\bibfnamefont {Linyin}\ \bibnamefont {Hong}}, \bibinfo {author}
  {\bibfnamefont {He-Liang}\ \bibnamefont {Huang}}, \bibinfo {author}
  {\bibfnamefont {Yong-Heng}\ \bibnamefont {Huo}}, \bibinfo {author}
  {\bibfnamefont {Liping}\ \bibnamefont {Li}}, \bibinfo {author} {\bibfnamefont
  {Na}~\bibnamefont {Li}}, \bibinfo {author} {\bibfnamefont {Shaowei}\
  \bibnamefont {Li}}, \bibinfo {author} {\bibfnamefont {Yuan}\ \bibnamefont
  {Li}}, \bibinfo {author} {\bibfnamefont {Futian}\ \bibnamefont {Liang}},
  \bibinfo {author} {\bibfnamefont {Chun}\ \bibnamefont {Lin}}, \bibinfo
  {author} {\bibfnamefont {Jin}\ \bibnamefont {Lin}}, \bibinfo {author}
  {\bibfnamefont {Haoran}\ \bibnamefont {Qian}}, \bibinfo {author}
  {\bibfnamefont {Dan}\ \bibnamefont {Qiao}}, \bibinfo {author} {\bibfnamefont
  {Hao}\ \bibnamefont {Rong}}, \bibinfo {author} {\bibfnamefont {Hong}\
  \bibnamefont {Su}}, \bibinfo {author} {\bibfnamefont {Lihua}\ \bibnamefont
  {Sun}}, \bibinfo {author} {\bibfnamefont {Liangyuan}\ \bibnamefont {Wang}},
  \bibinfo {author} {\bibfnamefont {Shiyu}\ \bibnamefont {Wang}}, \bibinfo
  {author} {\bibfnamefont {Dachao}\ \bibnamefont {Wu}}, \bibinfo {author}
  {\bibfnamefont {Yu}~\bibnamefont {Xu}}, \bibinfo {author} {\bibfnamefont
  {Kai}\ \bibnamefont {Yan}}, \bibinfo {author} {\bibfnamefont {Weifeng}\
  \bibnamefont {Yang}}, \bibinfo {author} {\bibfnamefont {Yang}\ \bibnamefont
  {Yang}}, \bibinfo {author} {\bibfnamefont {Yangsen}\ \bibnamefont {Ye}},
  \bibinfo {author} {\bibfnamefont {Jianghan}\ \bibnamefont {Yin}}, \bibinfo
  {author} {\bibfnamefont {Chong}\ \bibnamefont {Ying}}, \bibinfo {author}
  {\bibfnamefont {Jiale}\ \bibnamefont {Yu}}, \bibinfo {author} {\bibfnamefont
  {Chen}\ \bibnamefont {Zha}}, \bibinfo {author} {\bibfnamefont {Cha}\
  \bibnamefont {Zhang}}, \bibinfo {author} {\bibfnamefont {Haibin}\
  \bibnamefont {Zhang}}, \bibinfo {author} {\bibfnamefont {Kaili}\ \bibnamefont
  {Zhang}}, \bibinfo {author} {\bibfnamefont {Yiming}\ \bibnamefont {Zhang}},
  \bibinfo {author} {\bibfnamefont {Han}\ \bibnamefont {Zhao}}, \bibinfo
  {author} {\bibfnamefont {Youwei}\ \bibnamefont {Zhao}}, \bibinfo {author}
  {\bibfnamefont {Liang}\ \bibnamefont {Zhou}}, \bibinfo {author}
  {\bibfnamefont {Qingling}\ \bibnamefont {Zhu}}, \bibinfo {author}
  {\bibfnamefont {Chao-Yang}\ \bibnamefont {Lu}}, \bibinfo {author}
  {\bibfnamefont {Cheng-Zhi}\ \bibnamefont {Peng}}, \bibinfo {author}
  {\bibfnamefont {Xiaobo}\ \bibnamefont {Zhu}}, \ and\ \bibinfo {author}
  {\bibfnamefont {Jian-Wei}\ \bibnamefont {Pan}},\ }\bibfield  {title}
  {\enquote {\bibinfo {title} {Strong quantum computational advantage using a
  superconducting quantum processor},}\ }\href
  {https://link.aps.org/doi/10.1103/PhysRevLett.127.180501} {\bibfield
  {journal} {\bibinfo  {journal} {Phys. Rev. Lett.}\ }\textbf {\bibinfo
  {volume} {127}},\ \bibinfo {pages} {180501--} (\bibinfo {year}
  {2021})}\BibitemShut {NoStop}%
\bibitem [{\citenamefont {Zhong}\ \emph {et~al.}(2020)\citenamefont {Zhong},
  \citenamefont {Wang}, \citenamefont {Deng}, \citenamefont {Chen},
  \citenamefont {Peng}, \citenamefont {Luo}, \citenamefont {Qin}, \citenamefont
  {Wu}, \citenamefont {Ding}, \citenamefont {Hu}, \citenamefont {Hu},
  \citenamefont {Yang}, \citenamefont {Zhang}, \citenamefont {Li},
  \citenamefont {Li}, \citenamefont {Jiang}, \citenamefont {Gan}, \citenamefont
  {Yang}, \citenamefont {You}, \citenamefont {Wang}, \citenamefont {Li},
  \citenamefont {Liu}, \citenamefont {Lu},\ and\ \citenamefont
  {Pan}}]{Zhong:2020aa}%
  \BibitemOpen
  \bibfield  {author} {\bibinfo {author} {\bibfnamefont {Han-Sen}\ \bibnamefont
  {Zhong}}, \bibinfo {author} {\bibfnamefont {Hui}\ \bibnamefont {Wang}},
  \bibinfo {author} {\bibfnamefont {Yu-Hao}\ \bibnamefont {Deng}}, \bibinfo
  {author} {\bibfnamefont {Ming-Cheng}\ \bibnamefont {Chen}}, \bibinfo {author}
  {\bibfnamefont {Li-Chao}\ \bibnamefont {Peng}}, \bibinfo {author}
  {\bibfnamefont {Yi-Han}\ \bibnamefont {Luo}}, \bibinfo {author}
  {\bibfnamefont {Jian}\ \bibnamefont {Qin}}, \bibinfo {author} {\bibfnamefont
  {Dian}\ \bibnamefont {Wu}}, \bibinfo {author} {\bibfnamefont {Xing}\
  \bibnamefont {Ding}}, \bibinfo {author} {\bibfnamefont {Yi}~\bibnamefont
  {Hu}}, \bibinfo {author} {\bibfnamefont {Peng}\ \bibnamefont {Hu}}, \bibinfo
  {author} {\bibfnamefont {Xiao-Yan}\ \bibnamefont {Yang}}, \bibinfo {author}
  {\bibfnamefont {Wei-Jun}\ \bibnamefont {Zhang}}, \bibinfo {author}
  {\bibfnamefont {Hao}\ \bibnamefont {Li}}, \bibinfo {author} {\bibfnamefont
  {Yuxuan}\ \bibnamefont {Li}}, \bibinfo {author} {\bibfnamefont {Xiao}\
  \bibnamefont {Jiang}}, \bibinfo {author} {\bibfnamefont {Lin}\ \bibnamefont
  {Gan}}, \bibinfo {author} {\bibfnamefont {Guangwen}\ \bibnamefont {Yang}},
  \bibinfo {author} {\bibfnamefont {Lixing}\ \bibnamefont {You}}, \bibinfo
  {author} {\bibfnamefont {Zhen}\ \bibnamefont {Wang}}, \bibinfo {author}
  {\bibfnamefont {Li}~\bibnamefont {Li}}, \bibinfo {author} {\bibfnamefont
  {Nai-Le}\ \bibnamefont {Liu}}, \bibinfo {author} {\bibfnamefont {Chao-Yang}\
  \bibnamefont {Lu}}, \ and\ \bibinfo {author} {\bibfnamefont {Jian-Wei}\
  \bibnamefont {Pan}},\ }\bibfield  {title} {\enquote {\bibinfo {title}
  {Quantum computational advantage using photons},}\ }\href
  {https://www.science.org/doi/abs/10.1126/science.abe8770} {\bibfield
  {journal} {\bibinfo  {journal} {Science}\ }\textbf {\bibinfo {volume}
  {370}},\ \bibinfo {pages} {1460--1463} (\bibinfo {year} {2020})}\BibitemShut
  {NoStop}%
\bibitem [{\citenamefont {Zhong}\ \emph {et~al.}(2021)\citenamefont {Zhong},
  \citenamefont {Deng}, \citenamefont {Qin}, \citenamefont {Wang},
  \citenamefont {Chen}, \citenamefont {Peng}, \citenamefont {Luo},
  \citenamefont {Wu}, \citenamefont {Gong}, \citenamefont {Su}, \citenamefont
  {Hu}, \citenamefont {Hu}, \citenamefont {Yang}, \citenamefont {Zhang},
  \citenamefont {Li}, \citenamefont {Li}, \citenamefont {Jiang}, \citenamefont
  {Gan}, \citenamefont {Yang}, \citenamefont {You}, \citenamefont {Wang},
  \citenamefont {Li}, \citenamefont {Liu}, \citenamefont {Renema},
  \citenamefont {Lu},\ and\ \citenamefont {Pan}}]{Zhong:2021wv}%
  \BibitemOpen
  \bibfield  {author} {\bibinfo {author} {\bibfnamefont {Han-Sen}\ \bibnamefont
  {Zhong}}, \bibinfo {author} {\bibfnamefont {Yu-Hao}\ \bibnamefont {Deng}},
  \bibinfo {author} {\bibfnamefont {Jian}\ \bibnamefont {Qin}}, \bibinfo
  {author} {\bibfnamefont {Hui}\ \bibnamefont {Wang}}, \bibinfo {author}
  {\bibfnamefont {Ming-Cheng}\ \bibnamefont {Chen}}, \bibinfo {author}
  {\bibfnamefont {Li-Chao}\ \bibnamefont {Peng}}, \bibinfo {author}
  {\bibfnamefont {Yi-Han}\ \bibnamefont {Luo}}, \bibinfo {author}
  {\bibfnamefont {Dian}\ \bibnamefont {Wu}}, \bibinfo {author} {\bibfnamefont
  {Si-Qiu}\ \bibnamefont {Gong}}, \bibinfo {author} {\bibfnamefont {Hao}\
  \bibnamefont {Su}}, \bibinfo {author} {\bibfnamefont {Yi}~\bibnamefont {Hu}},
  \bibinfo {author} {\bibfnamefont {Peng}\ \bibnamefont {Hu}}, \bibinfo
  {author} {\bibfnamefont {Xiao-Yan}\ \bibnamefont {Yang}}, \bibinfo {author}
  {\bibfnamefont {Wei-Jun}\ \bibnamefont {Zhang}}, \bibinfo {author}
  {\bibfnamefont {Hao}\ \bibnamefont {Li}}, \bibinfo {author} {\bibfnamefont
  {Yuxuan}\ \bibnamefont {Li}}, \bibinfo {author} {\bibfnamefont {Xiao}\
  \bibnamefont {Jiang}}, \bibinfo {author} {\bibfnamefont {Lin}\ \bibnamefont
  {Gan}}, \bibinfo {author} {\bibfnamefont {Guangwen}\ \bibnamefont {Yang}},
  \bibinfo {author} {\bibfnamefont {Lixing}\ \bibnamefont {You}}, \bibinfo
  {author} {\bibfnamefont {Zhen}\ \bibnamefont {Wang}}, \bibinfo {author}
  {\bibfnamefont {Li}~\bibnamefont {Li}}, \bibinfo {author} {\bibfnamefont
  {Nai-Le}\ \bibnamefont {Liu}}, \bibinfo {author} {\bibfnamefont {Jelmer~J.}\
  \bibnamefont {Renema}}, \bibinfo {author} {\bibfnamefont {Chao-Yang}\
  \bibnamefont {Lu}}, \ and\ \bibinfo {author} {\bibfnamefont {Jian-Wei}\
  \bibnamefont {Pan}},\ }\bibfield  {title} {\enquote {\bibinfo {title}
  {Phase-programmable gaussian boson sampling using stimulated squeezed
  light},}\ }\href {https://link.aps.org/doi/10.1103/PhysRevLett.127.180502}
  {\bibfield  {journal} {\bibinfo  {journal} {Phys. Rev. Lett.}\ }\textbf
  {\bibinfo {volume} {127}},\ \bibinfo {pages} {180502--} (\bibinfo {year}
  {2021})}\BibitemShut {NoStop}%
\bibitem [{\citenamefont {Morvan}\ \emph {et~al.}(2023)\citenamefont {Morvan},
  \citenamefont {Villalonga}, \citenamefont {Mi}, \citenamefont {Mandr{\`a}},
  \citenamefont {Bengtsson}, \citenamefont {Klimov}, \citenamefont {Chen},
  \citenamefont {Hong}, \citenamefont {Erickson}, \citenamefont {Drozdov},
  \citenamefont {Chau}, \citenamefont {Laun}, \citenamefont {Movassagh},
  \citenamefont {Asfaw}, \citenamefont {Brand{\~a}o}, \citenamefont {Peralta},
  \citenamefont {Abanin}, \citenamefont {Acharya}, \citenamefont {Allen},
  \citenamefont {Andersen}, \citenamefont {Anderson}, \citenamefont {Ansmann},
  \citenamefont {Arute}, \citenamefont {Arya}, \citenamefont {Atalaya},
  \citenamefont {Bardin}, \citenamefont {Bilmes}, \citenamefont {Bortoli},
  \citenamefont {Bourassa}, \citenamefont {Bovaird}, \citenamefont {Brill},
  \citenamefont {Broughton}, \citenamefont {Buckley}, \citenamefont {Buell},
  \citenamefont {Burger}, \citenamefont {Burkett}, \citenamefont {Bushnell},
  \citenamefont {Campero}, \citenamefont {Chang}, \citenamefont {Chiaro},
  \citenamefont {Chik}, \citenamefont {Chou}, \citenamefont {Cogan},
  \citenamefont {Collins}, \citenamefont {Conner}, \citenamefont {Courtney},
  \citenamefont {Crook}, \citenamefont {Curtin}, \citenamefont {Debroy},
  \citenamefont {Barba}, \citenamefont {Demura}, \citenamefont {Paolo},
  \citenamefont {Dunsworth}, \citenamefont {Faoro}, \citenamefont {Farhi},
  \citenamefont {Fatemi}, \citenamefont {Ferreira}, \citenamefont {Burgos},
  \citenamefont {Forati}, \citenamefont {Fowler}, \citenamefont {Foxen},
  \citenamefont {Garcia}, \citenamefont {Genois}, \citenamefont {Giang},
  \citenamefont {Gidney}, \citenamefont {Gilboa}, \citenamefont {Giustina},
  \citenamefont {Gosula}, \citenamefont {Dau}, \citenamefont {Gross},
  \citenamefont {Habegger}, \citenamefont {Hamilton}, \citenamefont {Hansen},
  \citenamefont {Harrigan}, \citenamefont {Harrington}, \citenamefont {Heu},
  \citenamefont {Hoffmann}, \citenamefont {Huang}, \citenamefont {Huff},
  \citenamefont {Huggins}, \citenamefont {Ioffe}, \citenamefont {Isakov},
  \citenamefont {Iveland}, \citenamefont {Jeffrey}, \citenamefont {Jiang},
  \citenamefont {Jones}, \citenamefont {Juhas}, \citenamefont {Kafri},
  \citenamefont {Khattar}, \citenamefont {Khezri}, \citenamefont
  {Kieferov{\'a}}, \citenamefont {Kim}, \citenamefont {Kitaev}, \citenamefont
  {Klots}, \citenamefont {Korotkov}, \citenamefont {Kostritsa}, \citenamefont
  {Kreikebaum}, \citenamefont {Landhuis}, \citenamefont {Laptev}, \citenamefont
  {Lau}, \citenamefont {Laws}, \citenamefont {Lee}, \citenamefont {Lee},
  \citenamefont {Lensky}, \citenamefont {Lester}, \citenamefont {Lill},
  \citenamefont {Liu}, \citenamefont {Locharla}, \citenamefont {Malone},
  \citenamefont {Martin}, \citenamefont {Martin}, \citenamefont {McClean},
  \citenamefont {McEwen}, \citenamefont {Miao}, \citenamefont {Mieszala},
  \citenamefont {Montazeri}, \citenamefont {Mruczkiewicz}, \citenamefont
  {Naaman}, \citenamefont {Neeley}, \citenamefont {Neill}, \citenamefont
  {Nersisyan}, \citenamefont {Newman}, \citenamefont {Ng}, \citenamefont
  {Nguyen}, \citenamefont {Nguyen}, \citenamefont {Niu}, \citenamefont
  {O'Brien}, \citenamefont {Omonije}, \citenamefont {Opremcak}, \citenamefont
  {Petukhov}, \citenamefont {Potter}, \citenamefont {Pryadko}, \citenamefont
  {Quintana}, \citenamefont {Rhodes}, \citenamefont {Rocque}, \citenamefont
  {Roushan}, \citenamefont {Rubin}, \citenamefont {Saei}, \citenamefont {Sank},
  \citenamefont {Sankaragomathi}, \citenamefont {Satzinger}, \citenamefont
  {Schurkus}, \citenamefont {Schuster}, \citenamefont {Shearn}, \citenamefont
  {Shorter}, \citenamefont {Shutty}, \citenamefont {Shvarts}, \citenamefont
  {Sivak}, \citenamefont {Skruzny}, \citenamefont {Smith}, \citenamefont
  {Somma}, \citenamefont {Sterling}, \citenamefont {Strain}, \citenamefont
  {Szalay}, \citenamefont {Thor}, \citenamefont {Torres}, \citenamefont
  {Vidal}, \citenamefont {Heidweiller}, \citenamefont {White}, \citenamefont
  {Woo}, \citenamefont {Xing}, \citenamefont {Yao}, \citenamefont {Yeh},
  \citenamefont {Yoo}, \citenamefont {Young}, \citenamefont {Zalcman},
  \citenamefont {Zhang}, \citenamefont {Zhu}, \citenamefont {Zobrist},
  \citenamefont {Rieffel}, \citenamefont {Biswas}, \citenamefont {Babbush},
  \citenamefont {Bacon}, \citenamefont {Hilton}, \citenamefont {Lucero},
  \citenamefont {Neven}, \citenamefont {Megrant}, \citenamefont {Kelly},
  \citenamefont {Aleiner}, \citenamefont {Smelyanskiy}, \citenamefont
  {Kechedzhi}, \citenamefont {Chen},\ and\ \citenamefont
  {Boixo}}]{morvan2023phase}%
  \BibitemOpen
  \bibfield  {author} {\bibinfo {author} {\bibfnamefont {A.}~\bibnamefont
  {Morvan}}, \bibinfo {author} {\bibfnamefont {B.}~\bibnamefont {Villalonga}},
  \bibinfo {author} {\bibfnamefont {X.}~\bibnamefont {Mi}}, \bibinfo {author}
  {\bibfnamefont {S.}~\bibnamefont {Mandr{\`a}}}, \bibinfo {author}
  {\bibfnamefont {A.}~\bibnamefont {Bengtsson}}, \bibinfo {author}
  {\bibfnamefont {P.~V.}\ \bibnamefont {Klimov}}, \bibinfo {author}
  {\bibfnamefont {Z.}~\bibnamefont {Chen}}, \bibinfo {author} {\bibfnamefont
  {S.}~\bibnamefont {Hong}}, \bibinfo {author} {\bibfnamefont {C.}~\bibnamefont
  {Erickson}}, \bibinfo {author} {\bibfnamefont {I.~K.}\ \bibnamefont
  {Drozdov}}, \bibinfo {author} {\bibfnamefont {J.}~\bibnamefont {Chau}},
  \bibinfo {author} {\bibfnamefont {G.}~\bibnamefont {Laun}}, \bibinfo {author}
  {\bibfnamefont {R.}~\bibnamefont {Movassagh}}, \bibinfo {author}
  {\bibfnamefont {A.}~\bibnamefont {Asfaw}}, \bibinfo {author} {\bibfnamefont
  {L.~T. A.~N.}\ \bibnamefont {Brand{\~a}o}}, \bibinfo {author} {\bibfnamefont
  {R.}~\bibnamefont {Peralta}}, \bibinfo {author} {\bibfnamefont
  {D.}~\bibnamefont {Abanin}}, \bibinfo {author} {\bibfnamefont
  {R.}~\bibnamefont {Acharya}}, \bibinfo {author} {\bibfnamefont
  {R.}~\bibnamefont {Allen}}, \bibinfo {author} {\bibfnamefont {T.~I.}\
  \bibnamefont {Andersen}}, \bibinfo {author} {\bibfnamefont {K.}~\bibnamefont
  {Anderson}}, \bibinfo {author} {\bibfnamefont {M.}~\bibnamefont {Ansmann}},
  \bibinfo {author} {\bibfnamefont {F.}~\bibnamefont {Arute}}, \bibinfo
  {author} {\bibfnamefont {K.}~\bibnamefont {Arya}}, \bibinfo {author}
  {\bibfnamefont {J.}~\bibnamefont {Atalaya}}, \bibinfo {author} {\bibfnamefont
  {J.~C.}\ \bibnamefont {Bardin}}, \bibinfo {author} {\bibfnamefont
  {A.}~\bibnamefont {Bilmes}}, \bibinfo {author} {\bibfnamefont
  {G.}~\bibnamefont {Bortoli}}, \bibinfo {author} {\bibfnamefont
  {A.}~\bibnamefont {Bourassa}}, \bibinfo {author} {\bibfnamefont
  {J.}~\bibnamefont {Bovaird}}, \bibinfo {author} {\bibfnamefont
  {L.}~\bibnamefont {Brill}}, \bibinfo {author} {\bibfnamefont
  {M.}~\bibnamefont {Broughton}}, \bibinfo {author} {\bibfnamefont {B.~B.}\
  \bibnamefont {Buckley}}, \bibinfo {author} {\bibfnamefont {D.~A.}\
  \bibnamefont {Buell}}, \bibinfo {author} {\bibfnamefont {T.}~\bibnamefont
  {Burger}}, \bibinfo {author} {\bibfnamefont {B.}~\bibnamefont {Burkett}},
  \bibinfo {author} {\bibfnamefont {N.}~\bibnamefont {Bushnell}}, \bibinfo
  {author} {\bibfnamefont {J.}~\bibnamefont {Campero}}, \bibinfo {author}
  {\bibfnamefont {H.~S.}\ \bibnamefont {Chang}}, \bibinfo {author}
  {\bibfnamefont {B.}~\bibnamefont {Chiaro}}, \bibinfo {author} {\bibfnamefont
  {D.}~\bibnamefont {Chik}}, \bibinfo {author} {\bibfnamefont {C.}~\bibnamefont
  {Chou}}, \bibinfo {author} {\bibfnamefont {J.}~\bibnamefont {Cogan}},
  \bibinfo {author} {\bibfnamefont {R.}~\bibnamefont {Collins}}, \bibinfo
  {author} {\bibfnamefont {P.}~\bibnamefont {Conner}}, \bibinfo {author}
  {\bibfnamefont {W.}~\bibnamefont {Courtney}}, \bibinfo {author}
  {\bibfnamefont {A.~L.}\ \bibnamefont {Crook}}, \bibinfo {author}
  {\bibfnamefont {B.}~\bibnamefont {Curtin}}, \bibinfo {author} {\bibfnamefont
  {D.~M.}\ \bibnamefont {Debroy}}, \bibinfo {author} {\bibfnamefont
  {A.~Del~Toro}\ \bibnamefont {Barba}}, \bibinfo {author} {\bibfnamefont
  {S.}~\bibnamefont {Demura}}, \bibinfo {author} {\bibfnamefont {A.~Di}\
  \bibnamefont {Paolo}}, \bibinfo {author} {\bibfnamefont {A.}~\bibnamefont
  {Dunsworth}}, \bibinfo {author} {\bibfnamefont {L.}~\bibnamefont {Faoro}},
  \bibinfo {author} {\bibfnamefont {E.}~\bibnamefont {Farhi}}, \bibinfo
  {author} {\bibfnamefont {R.}~\bibnamefont {Fatemi}}, \bibinfo {author}
  {\bibfnamefont {V.~S.}\ \bibnamefont {Ferreira}}, \bibinfo {author}
  {\bibfnamefont {L.~Flores}\ \bibnamefont {Burgos}}, \bibinfo {author}
  {\bibfnamefont {E.}~\bibnamefont {Forati}}, \bibinfo {author} {\bibfnamefont
  {A.~G.}\ \bibnamefont {Fowler}}, \bibinfo {author} {\bibfnamefont
  {B.}~\bibnamefont {Foxen}}, \bibinfo {author} {\bibfnamefont
  {G.}~\bibnamefont {Garcia}}, \bibinfo {author} {\bibfnamefont
  {E.}~\bibnamefont {Genois}}, \bibinfo {author} {\bibfnamefont
  {W.}~\bibnamefont {Giang}}, \bibinfo {author} {\bibfnamefont
  {C.}~\bibnamefont {Gidney}}, \bibinfo {author} {\bibfnamefont
  {D.}~\bibnamefont {Gilboa}}, \bibinfo {author} {\bibfnamefont
  {M.}~\bibnamefont {Giustina}}, \bibinfo {author} {\bibfnamefont
  {R.}~\bibnamefont {Gosula}}, \bibinfo {author} {\bibfnamefont {A.~Grajales}\
  \bibnamefont {Dau}}, \bibinfo {author} {\bibfnamefont {J.~A.}\ \bibnamefont
  {Gross}}, \bibinfo {author} {\bibfnamefont {S.}~\bibnamefont {Habegger}},
  \bibinfo {author} {\bibfnamefont {M.~C.}\ \bibnamefont {Hamilton}}, \bibinfo
  {author} {\bibfnamefont {M.}~\bibnamefont {Hansen}}, \bibinfo {author}
  {\bibfnamefont {M.~P.}\ \bibnamefont {Harrigan}}, \bibinfo {author}
  {\bibfnamefont {S.~D.}\ \bibnamefont {Harrington}}, \bibinfo {author}
  {\bibfnamefont {P.}~\bibnamefont {Heu}}, \bibinfo {author} {\bibfnamefont
  {M.~R.}\ \bibnamefont {Hoffmann}}, \bibinfo {author} {\bibfnamefont
  {T.}~\bibnamefont {Huang}}, \bibinfo {author} {\bibfnamefont
  {A.}~\bibnamefont {Huff}}, \bibinfo {author} {\bibfnamefont {W.~J.}\
  \bibnamefont {Huggins}}, \bibinfo {author} {\bibfnamefont {L.~B.}\
  \bibnamefont {Ioffe}}, \bibinfo {author} {\bibfnamefont {S.~V.}\ \bibnamefont
  {Isakov}}, \bibinfo {author} {\bibfnamefont {J.}~\bibnamefont {Iveland}},
  \bibinfo {author} {\bibfnamefont {E.}~\bibnamefont {Jeffrey}}, \bibinfo
  {author} {\bibfnamefont {Z.}~\bibnamefont {Jiang}}, \bibinfo {author}
  {\bibfnamefont {C.}~\bibnamefont {Jones}}, \bibinfo {author} {\bibfnamefont
  {P.}~\bibnamefont {Juhas}}, \bibinfo {author} {\bibfnamefont
  {D.}~\bibnamefont {Kafri}}, \bibinfo {author} {\bibfnamefont
  {T.}~\bibnamefont {Khattar}}, \bibinfo {author} {\bibfnamefont
  {M.}~\bibnamefont {Khezri}}, \bibinfo {author} {\bibfnamefont
  {M.}~\bibnamefont {Kieferov{\'a}}}, \bibinfo {author} {\bibfnamefont
  {S.}~\bibnamefont {Kim}}, \bibinfo {author} {\bibfnamefont {A.}~\bibnamefont
  {Kitaev}}, \bibinfo {author} {\bibfnamefont {A.~R.}\ \bibnamefont {Klots}},
  \bibinfo {author} {\bibfnamefont {A.~N.}\ \bibnamefont {Korotkov}}, \bibinfo
  {author} {\bibfnamefont {F.}~\bibnamefont {Kostritsa}}, \bibinfo {author}
  {\bibfnamefont {J.~M.}\ \bibnamefont {Kreikebaum}}, \bibinfo {author}
  {\bibfnamefont {D.}~\bibnamefont {Landhuis}}, \bibinfo {author}
  {\bibfnamefont {P.}~\bibnamefont {Laptev}}, \bibinfo {author} {\bibfnamefont
  {K.~M.}\ \bibnamefont {Lau}}, \bibinfo {author} {\bibfnamefont
  {L.}~\bibnamefont {Laws}}, \bibinfo {author} {\bibfnamefont {J.}~\bibnamefont
  {Lee}}, \bibinfo {author} {\bibfnamefont {K.~W.}\ \bibnamefont {Lee}},
  \bibinfo {author} {\bibfnamefont {Y.~D.}\ \bibnamefont {Lensky}}, \bibinfo
  {author} {\bibfnamefont {B.~J.}\ \bibnamefont {Lester}}, \bibinfo {author}
  {\bibfnamefont {A.~T.}\ \bibnamefont {Lill}}, \bibinfo {author}
  {\bibfnamefont {W.}~\bibnamefont {Liu}}, \bibinfo {author} {\bibfnamefont
  {A.}~\bibnamefont {Locharla}}, \bibinfo {author} {\bibfnamefont {F.~D.}\
  \bibnamefont {Malone}}, \bibinfo {author} {\bibfnamefont {O.}~\bibnamefont
  {Martin}}, \bibinfo {author} {\bibfnamefont {S.}~\bibnamefont {Martin}},
  \bibinfo {author} {\bibfnamefont {J.~R.}\ \bibnamefont {McClean}}, \bibinfo
  {author} {\bibfnamefont {M.}~\bibnamefont {McEwen}}, \bibinfo {author}
  {\bibfnamefont {K.~C.}\ \bibnamefont {Miao}}, \bibinfo {author}
  {\bibfnamefont {A.}~\bibnamefont {Mieszala}}, \bibinfo {author}
  {\bibfnamefont {S.}~\bibnamefont {Montazeri}}, \bibinfo {author}
  {\bibfnamefont {W.}~\bibnamefont {Mruczkiewicz}}, \bibinfo {author}
  {\bibfnamefont {O.}~\bibnamefont {Naaman}}, \bibinfo {author} {\bibfnamefont
  {M.}~\bibnamefont {Neeley}}, \bibinfo {author} {\bibfnamefont
  {C.}~\bibnamefont {Neill}}, \bibinfo {author} {\bibfnamefont
  {A.}~\bibnamefont {Nersisyan}}, \bibinfo {author} {\bibfnamefont
  {M.}~\bibnamefont {Newman}}, \bibinfo {author} {\bibfnamefont {J.~H.}\
  \bibnamefont {Ng}}, \bibinfo {author} {\bibfnamefont {A.}~\bibnamefont
  {Nguyen}}, \bibinfo {author} {\bibfnamefont {M.}~\bibnamefont {Nguyen}},
  \bibinfo {author} {\bibfnamefont {M.~Yuezhen}\ \bibnamefont {Niu}}, \bibinfo
  {author} {\bibfnamefont {T.~E.}\ \bibnamefont {O'Brien}}, \bibinfo {author}
  {\bibfnamefont {S.}~\bibnamefont {Omonije}}, \bibinfo {author} {\bibfnamefont
  {A.}~\bibnamefont {Opremcak}}, \bibinfo {author} {\bibfnamefont
  {A.}~\bibnamefont {Petukhov}}, \bibinfo {author} {\bibfnamefont
  {R.}~\bibnamefont {Potter}}, \bibinfo {author} {\bibfnamefont {L.~P.}\
  \bibnamefont {Pryadko}}, \bibinfo {author} {\bibfnamefont {C.}~\bibnamefont
  {Quintana}}, \bibinfo {author} {\bibfnamefont {D.~M.}\ \bibnamefont
  {Rhodes}}, \bibinfo {author} {\bibfnamefont {C.}~\bibnamefont {Rocque}},
  \bibinfo {author} {\bibfnamefont {P.}~\bibnamefont {Roushan}}, \bibinfo
  {author} {\bibfnamefont {N.~C.}\ \bibnamefont {Rubin}}, \bibinfo {author}
  {\bibfnamefont {N.}~\bibnamefont {Saei}}, \bibinfo {author} {\bibfnamefont
  {D.}~\bibnamefont {Sank}}, \bibinfo {author} {\bibfnamefont {K.}~\bibnamefont
  {Sankaragomathi}}, \bibinfo {author} {\bibfnamefont {K.~J.}\ \bibnamefont
  {Satzinger}}, \bibinfo {author} {\bibfnamefont {H.~F.}\ \bibnamefont
  {Schurkus}}, \bibinfo {author} {\bibfnamefont {C.}~\bibnamefont {Schuster}},
  \bibinfo {author} {\bibfnamefont {M.~J.}\ \bibnamefont {Shearn}}, \bibinfo
  {author} {\bibfnamefont {A.}~\bibnamefont {Shorter}}, \bibinfo {author}
  {\bibfnamefont {N.}~\bibnamefont {Shutty}}, \bibinfo {author} {\bibfnamefont
  {V.}~\bibnamefont {Shvarts}}, \bibinfo {author} {\bibfnamefont
  {V.}~\bibnamefont {Sivak}}, \bibinfo {author} {\bibfnamefont
  {J.}~\bibnamefont {Skruzny}}, \bibinfo {author} {\bibfnamefont {W.~C.}\
  \bibnamefont {Smith}}, \bibinfo {author} {\bibfnamefont {R.~D.}\ \bibnamefont
  {Somma}}, \bibinfo {author} {\bibfnamefont {G.}~\bibnamefont {Sterling}},
  \bibinfo {author} {\bibfnamefont {D.}~\bibnamefont {Strain}}, \bibinfo
  {author} {\bibfnamefont {M.}~\bibnamefont {Szalay}}, \bibinfo {author}
  {\bibfnamefont {D.}~\bibnamefont {Thor}}, \bibinfo {author} {\bibfnamefont
  {A.}~\bibnamefont {Torres}}, \bibinfo {author} {\bibfnamefont
  {G.}~\bibnamefont {Vidal}}, \bibinfo {author} {\bibfnamefont {C.~Vollgraff}\
  \bibnamefont {Heidweiller}}, \bibinfo {author} {\bibfnamefont
  {T.}~\bibnamefont {White}}, \bibinfo {author} {\bibfnamefont {B.~W.~K.}\
  \bibnamefont {Woo}}, \bibinfo {author} {\bibfnamefont {C.}~\bibnamefont
  {Xing}}, \bibinfo {author} {\bibfnamefont {Z.~J.}\ \bibnamefont {Yao}},
  \bibinfo {author} {\bibfnamefont {P.}~\bibnamefont {Yeh}}, \bibinfo {author}
  {\bibfnamefont {J.}~\bibnamefont {Yoo}}, \bibinfo {author} {\bibfnamefont
  {G.}~\bibnamefont {Young}}, \bibinfo {author} {\bibfnamefont
  {A.}~\bibnamefont {Zalcman}}, \bibinfo {author} {\bibfnamefont
  {Y.}~\bibnamefont {Zhang}}, \bibinfo {author} {\bibfnamefont
  {N.}~\bibnamefont {Zhu}}, \bibinfo {author} {\bibfnamefont {N.}~\bibnamefont
  {Zobrist}}, \bibinfo {author} {\bibfnamefont {E.~G.}\ \bibnamefont
  {Rieffel}}, \bibinfo {author} {\bibfnamefont {R.}~\bibnamefont {Biswas}},
  \bibinfo {author} {\bibfnamefont {R.}~\bibnamefont {Babbush}}, \bibinfo
  {author} {\bibfnamefont {D.}~\bibnamefont {Bacon}}, \bibinfo {author}
  {\bibfnamefont {J.}~\bibnamefont {Hilton}}, \bibinfo {author} {\bibfnamefont
  {E.}~\bibnamefont {Lucero}}, \bibinfo {author} {\bibfnamefont
  {H.}~\bibnamefont {Neven}}, \bibinfo {author} {\bibfnamefont
  {A.}~\bibnamefont {Megrant}}, \bibinfo {author} {\bibfnamefont
  {J.}~\bibnamefont {Kelly}}, \bibinfo {author} {\bibfnamefont
  {I.}~\bibnamefont {Aleiner}}, \bibinfo {author} {\bibfnamefont
  {V.}~\bibnamefont {Smelyanskiy}}, \bibinfo {author} {\bibfnamefont
  {K.}~\bibnamefont {Kechedzhi}}, \bibinfo {author} {\bibfnamefont
  {Y.}~\bibnamefont {Chen}}, \ and\ \bibinfo {author} {\bibfnamefont
  {S.}~\bibnamefont {Boixo}},\ }\bibfield  {title} {\enquote {\bibinfo {title}
  {Phase transition in random circuit sampling},}\ }\href
  {https://arxiv.org/abs/2304.11119} {\bibfield  {journal} {\bibinfo  {journal}
  {arXiv e-prints}\ } (\bibinfo {year} {2023})},\ \Eprint
  {http://arxiv.org/abs/2304.11119} {arXiv:2304.11119 [quant-ph]} \BibitemShut
  {NoStop}%
\bibitem [{\citenamefont {Figgatt}\ \emph {et~al.}(2017)\citenamefont
  {Figgatt}, \citenamefont {Maslov}, \citenamefont {Landsman}, \citenamefont
  {Linke}, \citenamefont {Debnath},\ and\ \citenamefont
  {Monroe}}]{figgattComplete3QubitGrover2017}%
  \BibitemOpen
  \bibfield  {author} {\bibinfo {author} {\bibfnamefont {C.}~\bibnamefont
  {Figgatt}}, \bibinfo {author} {\bibfnamefont {D.}~\bibnamefont {Maslov}},
  \bibinfo {author} {\bibfnamefont {K.~A.}\ \bibnamefont {Landsman}}, \bibinfo
  {author} {\bibfnamefont {N.~M.}\ \bibnamefont {Linke}}, \bibinfo {author}
  {\bibfnamefont {S.}~\bibnamefont {Debnath}}, \ and\ \bibinfo {author}
  {\bibfnamefont {C.}~\bibnamefont {Monroe}},\ }\bibfield  {title} {\enquote
  {\bibinfo {title} {Complete 3-{{Qubit Grover}} search on a programmable
  quantum computer},}\ }\href
  {https://www.nature.com/articles/s41467-017-01904-7} {\bibfield  {journal}
  {\bibinfo  {journal} {Nat. Commun.}\ }\textbf {\bibinfo {volume} {8}},\
  \bibinfo {pages} {1--9} (\bibinfo {year} {2017})}\BibitemShut {NoStop}%
\bibitem [{\citenamefont {Wright}\ \emph {et~al.}(2019)\citenamefont {Wright},
  \citenamefont {Beck}, \citenamefont {Debnath}, \citenamefont {Amini},
  \citenamefont {Nam}, \citenamefont {Grzesiak}, \citenamefont {Chen},
  \citenamefont {Pisenti}, \citenamefont {Chmielewski}, \citenamefont
  {Collins}, \citenamefont {Hudek}, \citenamefont {Mizrahi}, \citenamefont
  {{Wong-Campos}}, \citenamefont {Allen}, \citenamefont {Apisdorf},
  \citenamefont {Solomon}, \citenamefont {Williams}, \citenamefont {Ducore},
  \citenamefont {Blinov}, \citenamefont {Kreikemeier}, \citenamefont {Chaplin},
  \citenamefont {Keesan}, \citenamefont {Monroe},\ and\ \citenamefont
  {Kim}}]{wrightBenchmarking11qubitQuantum2019}%
  \BibitemOpen
  \bibfield  {author} {\bibinfo {author} {\bibfnamefont {K.}~\bibnamefont
  {Wright}}, \bibinfo {author} {\bibfnamefont {K.~M.}\ \bibnamefont {Beck}},
  \bibinfo {author} {\bibfnamefont {S.}~\bibnamefont {Debnath}}, \bibinfo
  {author} {\bibfnamefont {J.~M.}\ \bibnamefont {Amini}}, \bibinfo {author}
  {\bibfnamefont {Y.}~\bibnamefont {Nam}}, \bibinfo {author} {\bibfnamefont
  {N.}~\bibnamefont {Grzesiak}}, \bibinfo {author} {\bibfnamefont {J.-S.}\
  \bibnamefont {Chen}}, \bibinfo {author} {\bibfnamefont {N.~C.}\ \bibnamefont
  {Pisenti}}, \bibinfo {author} {\bibfnamefont {M.}~\bibnamefont
  {Chmielewski}}, \bibinfo {author} {\bibfnamefont {C.}~\bibnamefont
  {Collins}}, \bibinfo {author} {\bibfnamefont {K.~M.}\ \bibnamefont {Hudek}},
  \bibinfo {author} {\bibfnamefont {J.}~\bibnamefont {Mizrahi}}, \bibinfo
  {author} {\bibfnamefont {J.~D.}\ \bibnamefont {{Wong-Campos}}}, \bibinfo
  {author} {\bibfnamefont {S.}~\bibnamefont {Allen}}, \bibinfo {author}
  {\bibfnamefont {J.}~\bibnamefont {Apisdorf}}, \bibinfo {author}
  {\bibfnamefont {P.}~\bibnamefont {Solomon}}, \bibinfo {author} {\bibfnamefont
  {M.}~\bibnamefont {Williams}}, \bibinfo {author} {\bibfnamefont {A.~M.}\
  \bibnamefont {Ducore}}, \bibinfo {author} {\bibfnamefont {A.}~\bibnamefont
  {Blinov}}, \bibinfo {author} {\bibfnamefont {S.~M.}\ \bibnamefont
  {Kreikemeier}}, \bibinfo {author} {\bibfnamefont {V.}~\bibnamefont
  {Chaplin}}, \bibinfo {author} {\bibfnamefont {M.}~\bibnamefont {Keesan}},
  \bibinfo {author} {\bibfnamefont {C.}~\bibnamefont {Monroe}}, \ and\ \bibinfo
  {author} {\bibfnamefont {J.}~\bibnamefont {Kim}},\ }\bibfield  {title}
  {\enquote {\bibinfo {title} {Benchmarking an 11-qubit quantum computer},}\
  }\href {https://www.nature.com/articles/s41467-019-13534-2} {\bibfield
  {journal} {\bibinfo  {journal} {Nat. Commun.}\ }\textbf {\bibinfo {volume}
  {10}},\ \bibinfo {pages} {5464} (\bibinfo {year} {2019})}\BibitemShut
  {NoStop}%
\bibitem [{\citenamefont {Roy}\ \emph {et~al.}(2020)\citenamefont {Roy},
  \citenamefont {Hazra}, \citenamefont {Kundu}, \citenamefont {Chand},
  \citenamefont {Patankar},\ and\ \citenamefont
  {Vijay}}]{royProgrammableSuperconductingProcessor2020}%
  \BibitemOpen
  \bibfield  {author} {\bibinfo {author} {\bibfnamefont {Tanay}\ \bibnamefont
  {Roy}}, \bibinfo {author} {\bibfnamefont {Sumeru}\ \bibnamefont {Hazra}},
  \bibinfo {author} {\bibfnamefont {Suman}\ \bibnamefont {Kundu}}, \bibinfo
  {author} {\bibfnamefont {Madhavi}\ \bibnamefont {Chand}}, \bibinfo {author}
  {\bibfnamefont {Meghan~P.}\ \bibnamefont {Patankar}}, \ and\ \bibinfo
  {author} {\bibfnamefont {R.}~\bibnamefont {Vijay}},\ }\bibfield  {title}
  {\enquote {\bibinfo {title} {Programmable {{Superconducting Processor}} with
  {{Native Three-Qubit Gates}}},}\ }\href
  {https://journals.aps.org/prapplied/abstract/10.1103/PhysRevApplied.14.014072}
  {\bibfield  {journal} {\bibinfo  {journal} {Phys. Rev. Applied}\ }\textbf
  {\bibinfo {volume} {14}},\ \bibinfo {pages} {014072} (\bibinfo {year}
  {2020})}\BibitemShut {NoStop}%
\bibitem [{\citenamefont {Pelofske}\ \emph {et~al.}(2022)\citenamefont
  {Pelofske}, \citenamefont {Bartschi},\ and\ \citenamefont
  {Eidenbenz}}]{Pelofske2022}%
  \BibitemOpen
  \bibfield  {author} {\bibinfo {author} {\bibfnamefont {Elijah}\ \bibnamefont
  {Pelofske}}, \bibinfo {author} {\bibfnamefont {Andreas}\ \bibnamefont
  {Bartschi}}, \ and\ \bibinfo {author} {\bibfnamefont {Stephan}\ \bibnamefont
  {Eidenbenz}},\ }\bibfield  {title} {\enquote {\bibinfo {title} {Quantum
  volume in practice: What users can expect from nisq devices},}\ }\href
  {\doibase 10.1109/TQE.2022.3184764} {\bibfield  {journal} {\bibinfo
  {journal} {IEEE Transactions on Quantum Engineering}\ }\textbf {\bibinfo
  {volume} {3}},\ \bibinfo {pages} {1--19} (\bibinfo {year}
  {2022})}\BibitemShut {NoStop}%
\bibitem [{\citenamefont {Lubinski}\ \emph {et~al.}(2023)\citenamefont
  {Lubinski}, \citenamefont {Johri}, \citenamefont {Varosy}, \citenamefont
  {Coleman}, \citenamefont {Zhao}, \citenamefont {Necaise}, \citenamefont
  {Baldwin}, \citenamefont {Mayer},\ and\ \citenamefont
  {Proctor}}]{Lubinski2021}%
  \BibitemOpen
  \bibfield  {author} {\bibinfo {author} {\bibfnamefont {T.}~\bibnamefont
  {Lubinski}}, \bibinfo {author} {\bibfnamefont {S.}~\bibnamefont {Johri}},
  \bibinfo {author} {\bibfnamefont {P.}~\bibnamefont {Varosy}}, \bibinfo
  {author} {\bibfnamefont {J.}~\bibnamefont {Coleman}}, \bibinfo {author}
  {\bibfnamefont {L.}~\bibnamefont {Zhao}}, \bibinfo {author} {\bibfnamefont
  {J.}~\bibnamefont {Necaise}}, \bibinfo {author} {\bibfnamefont {C.~H.}\
  \bibnamefont {Baldwin}}, \bibinfo {author} {\bibfnamefont {K.}~\bibnamefont
  {Mayer}}, \ and\ \bibinfo {author} {\bibfnamefont {T.}~\bibnamefont
  {Proctor}},\ }\bibfield  {title} {\enquote {\bibinfo {title}
  {Application-oriented performance benchmarks for quantum computing},}\ }\href
  {\doibase 10.1109/TQE.2023.3253761} {\bibfield  {journal} {\bibinfo
  {journal} {IEEE Transactions on Quantum Engineering}\ }\textbf {\bibinfo
  {volume} {4}},\ \bibinfo {pages} {1--32} (\bibinfo {year}
  {2023})}\BibitemShut {NoStop}%
\bibitem [{\citenamefont {Barak}\ \emph {et~al.}(2021)\citenamefont {Barak},
  \citenamefont {Chou},\ and\ \citenamefont {Gao}}]{Barak:spoofing}%
  \BibitemOpen
  \bibfield  {author} {\bibinfo {author} {\bibfnamefont {Boaz}\ \bibnamefont
  {Barak}}, \bibinfo {author} {\bibfnamefont {Chi-Ning}\ \bibnamefont {Chou}},
  \ and\ \bibinfo {author} {\bibfnamefont {Xun}\ \bibnamefont {Gao}},\
  }\bibfield  {title} {\enquote {\bibinfo {title} {{Spoofing Linear
  Cross-Entropy Benchmarking in Shallow Quantum Circuits}},}\ }in\ \href
  {https://drops.dagstuhl.de/opus/volltexte/2021/13569} {\emph {\bibinfo
  {booktitle} {12th Innovations in Theoretical Computer Science Conference
  (ITCS 2021)}}},\ \bibinfo {series} {Leibniz International Proceedings in
  Informatics (LIPIcs)}, Vol.\ \bibinfo {volume} {185},\ \bibinfo {editor}
  {edited by\ \bibinfo {editor} {\bibfnamefont {James~R.}\ \bibnamefont
  {Lee}}}\ (\bibinfo  {publisher} {Schloss Dagstuhl--Leibniz-Zentrum f{\"u}r
  Informatik},\ \bibinfo {address} {Dagstuhl, Germany},\ \bibinfo {year}
  {2021})\ pp.\ \bibinfo {pages} {30:1--30:20}\BibitemShut {NoStop}%
\bibitem [{\citenamefont {Zlokapa}\ \emph {et~al.}(2023)\citenamefont
  {Zlokapa}, \citenamefont {Villalonga}, \citenamefont {Boixo},\ and\
  \citenamefont {Lidar}}]{Zlokapa:2023aa}%
  \BibitemOpen
  \bibfield  {author} {\bibinfo {author} {\bibfnamefont {Alexander}\
  \bibnamefont {Zlokapa}}, \bibinfo {author} {\bibfnamefont {Benjamin}\
  \bibnamefont {Villalonga}}, \bibinfo {author} {\bibfnamefont {Sergio}\
  \bibnamefont {Boixo}}, \ and\ \bibinfo {author} {\bibfnamefont {Daniel~A.}\
  \bibnamefont {Lidar}},\ }\bibfield  {title} {\enquote {\bibinfo {title}
  {Boundaries of quantum supremacy via random circuit sampling},}\ }\href
  {\doibase 10.1038/s41534-023-00703-x} {\bibfield  {journal} {\bibinfo
  {journal} {npj Quantum Information}\ }\textbf {\bibinfo {volume} {9}},\
  \bibinfo {pages} {36} (\bibinfo {year} {2023})}\BibitemShut {NoStop}%
\bibitem [{\citenamefont {Aharonov}\ \emph {et~al.}(2023)\citenamefont
  {Aharonov}, \citenamefont {Gao}, \citenamefont {Landau}, \citenamefont
  {Liu},\ and\ \citenamefont {Vazirani}}]{Aharonov:22}%
  \BibitemOpen
  \bibfield  {author} {\bibinfo {author} {\bibfnamefont {Dorit}\ \bibnamefont
  {Aharonov}}, \bibinfo {author} {\bibfnamefont {Xun}\ \bibnamefont {Gao}},
  \bibinfo {author} {\bibfnamefont {Zeph}\ \bibnamefont {Landau}}, \bibinfo
  {author} {\bibfnamefont {Yunchao}\ \bibnamefont {Liu}}, \ and\ \bibinfo
  {author} {\bibfnamefont {Umesh}\ \bibnamefont {Vazirani}},\ }\bibfield
  {title} {\enquote {\bibinfo {title} {A polynomial-time classical algorithm
  for noisy random circuit sampling},}\ }in\ \href {\doibase
  10.1145/3564246.3585234} {\emph {\bibinfo {booktitle} {Proceedings of the
  55th Annual ACM Symposium on Theory of Computing}}},\ \bibinfo {series and
  number} {STOC 2023}\ (\bibinfo  {publisher} {Association for Computing
  Machinery},\ \bibinfo {address} {New York, NY, USA},\ \bibinfo {year}
  {2023})\ pp.\ \bibinfo {pages} {945--957}\BibitemShut {NoStop}%
\bibitem [{\citenamefont {Ronnow}\ \emph {et~al.}(2014)\citenamefont {Ronnow},
  \citenamefont {Wang}, \citenamefont {Job}, \citenamefont {Boixo},
  \citenamefont {Isakov}, \citenamefont {Wecker}, \citenamefont {Martinis},
  \citenamefont {Lidar},\ and\ \citenamefont {Troyer}}]{speedup}%
  \BibitemOpen
  \bibfield  {author} {\bibinfo {author} {\bibfnamefont {Troels~F.}\
  \bibnamefont {Ronnow}}, \bibinfo {author} {\bibfnamefont {Zhihui}\
  \bibnamefont {Wang}}, \bibinfo {author} {\bibfnamefont {Joshua}\ \bibnamefont
  {Job}}, \bibinfo {author} {\bibfnamefont {Sergio}\ \bibnamefont {Boixo}},
  \bibinfo {author} {\bibfnamefont {Sergei~V.}\ \bibnamefont {Isakov}},
  \bibinfo {author} {\bibfnamefont {David}\ \bibnamefont {Wecker}}, \bibinfo
  {author} {\bibfnamefont {John~M.}\ \bibnamefont {Martinis}}, \bibinfo
  {author} {\bibfnamefont {Daniel~A.}\ \bibnamefont {Lidar}}, \ and\ \bibinfo
  {author} {\bibfnamefont {Matthias}\ \bibnamefont {Troyer}},\ }\bibfield
  {title} {\enquote {\bibinfo {title} {{Defining and detecting quantum
  speedup}},}\ }\href {http://science.sciencemag.org/content/345/6195/420}
  {\bibfield  {journal} {\bibinfo  {journal} {Science}\ }\textbf {\bibinfo
  {volume} {345}},\ \bibinfo {pages} {420--424} (\bibinfo {year}
  {2014})}\BibitemShut {NoStop}%
\bibitem [{\citenamefont {Pokharel}\ and\ \citenamefont
  {Lidar}(2023)}]{pokharel2022demonstration}%
  \BibitemOpen
  \bibfield  {author} {\bibinfo {author} {\bibfnamefont {Bibek}\ \bibnamefont
  {Pokharel}}\ and\ \bibinfo {author} {\bibfnamefont {Daniel~A.}\ \bibnamefont
  {Lidar}},\ }\bibfield  {title} {\enquote {\bibinfo {title} {Demonstration of
  algorithmic quantum speedup},}\ }\href {\doibase
  10.1103/PhysRevLett.130.210602} {\bibfield  {journal} {\bibinfo  {journal}
  {Physical Review Letters}\ }\textbf {\bibinfo {volume} {130}},\ \bibinfo
  {pages} {210602--} (\bibinfo {year} {2023})}\BibitemShut {NoStop}%
\bibitem [{\citenamefont {Viola}\ and\ \citenamefont {Lloyd}(1998)}]{Viola:98}%
  \BibitemOpen
  \bibfield  {author} {\bibinfo {author} {\bibfnamefont {Lorenza}\ \bibnamefont
  {Viola}}\ and\ \bibinfo {author} {\bibfnamefont {Seth}\ \bibnamefont
  {Lloyd}},\ }\bibfield  {title} {\enquote {\bibinfo {title} {Dynamical
  suppression of decoherence in two-state quantum systems},}\ }\href
  {https://link.aps.org/doi/10.1103/PhysRevA.58.2733} {\bibfield  {journal}
  {\bibinfo  {journal} {Phys. Rev. A}\ }\textbf {\bibinfo {volume} {58}},\
  \bibinfo {pages} {2733--2744} (\bibinfo {year} {1998})}\BibitemShut {NoStop}%
\bibitem [{\citenamefont {Viola}\ \emph {et~al.}(1999)\citenamefont {Viola},
  \citenamefont {Knill},\ and\ \citenamefont {Lloyd}}]{Viola:99}%
  \BibitemOpen
  \bibfield  {author} {\bibinfo {author} {\bibfnamefont {Lorenza}\ \bibnamefont
  {Viola}}, \bibinfo {author} {\bibfnamefont {Emanuel}\ \bibnamefont {Knill}},
  \ and\ \bibinfo {author} {\bibfnamefont {Seth}\ \bibnamefont {Lloyd}},\
  }\bibfield  {title} {\enquote {\bibinfo {title} {Dynamical decoupling of open
  quantum systems},}\ }\href
  {http://link.aps.org/doi/10.1103/PhysRevLett.82.2417} {\bibfield  {journal}
  {\bibinfo  {journal} {Physical Review Letters}\ }\textbf {\bibinfo {volume}
  {82}},\ \bibinfo {pages} {2417--2421} (\bibinfo {year} {1999})}\BibitemShut
  {NoStop}%
\bibitem [{\citenamefont {Zanardi}(1999)}]{Zanardi:1999fk}%
  \BibitemOpen
  \bibfield  {author} {\bibinfo {author} {\bibfnamefont {Paolo}\ \bibnamefont
  {Zanardi}},\ }\bibfield  {title} {\enquote {\bibinfo {title} {Symmetrizing
  evolutions},}\ }\href
  {http://www.sciencedirect.com/science/article/pii/S0375960199003655}
  {\bibfield  {journal} {\bibinfo  {journal} {Physics Letters A}\ }\textbf
  {\bibinfo {volume} {258}},\ \bibinfo {pages} {77--82} (\bibinfo {year}
  {1999})}\BibitemShut {NoStop}%
\bibitem [{\citenamefont {Vitali}\ and\ \citenamefont
  {Tombesi}(1999)}]{Vitali:99}%
  \BibitemOpen
  \bibfield  {author} {\bibinfo {author} {\bibfnamefont {D.}~\bibnamefont
  {Vitali}}\ and\ \bibinfo {author} {\bibfnamefont {P.}~\bibnamefont
  {Tombesi}},\ }\bibfield  {title} {\enquote {\bibinfo {title} {Using parity
  kicks for decoherence control},}\ }\href
  {https://link.aps.org/doi/10.1103/PhysRevA.59.4178} {\bibfield  {journal}
  {\bibinfo  {journal} {Physical Review A}\ }\textbf {\bibinfo {volume} {59}},\
  \bibinfo {pages} {4178--4186} (\bibinfo {year} {1999})}\BibitemShut {NoStop}%
\bibitem [{\citenamefont {Duan}\ and\ \citenamefont {Guo}(1999)}]{Duan:98e}%
  \BibitemOpen
  \bibfield  {author} {\bibinfo {author} {\bibfnamefont {Lu-Ming}\ \bibnamefont
  {Duan}}\ and\ \bibinfo {author} {\bibfnamefont {Guang-Can}\ \bibnamefont
  {Guo}},\ }\bibfield  {title} {\enquote {\bibinfo {title} {Suppressing
  environmental noise in quantum computation through pulse control},}\ }\href
  {https://www.sciencedirect.com/science/article/pii/S0375960199005927}
  {\bibfield  {journal} {\bibinfo  {journal} {Physics Letters A}\ }\textbf
  {\bibinfo {volume} {261}},\ \bibinfo {pages} {139--144} (\bibinfo {year}
  {1999})}\BibitemShut {NoStop}%
\bibitem [{\citenamefont {Pokharel}\ \emph {et~al.}(2018)\citenamefont
  {Pokharel}, \citenamefont {Anand}, \citenamefont {Fortman},\ and\
  \citenamefont {Lidar}}]{Pokharel2018}%
  \BibitemOpen
  \bibfield  {author} {\bibinfo {author} {\bibfnamefont {Bibek}\ \bibnamefont
  {Pokharel}}, \bibinfo {author} {\bibfnamefont {Namit}\ \bibnamefont {Anand}},
  \bibinfo {author} {\bibfnamefont {Benjamin}\ \bibnamefont {Fortman}}, \ and\
  \bibinfo {author} {\bibfnamefont {Daniel~A.}\ \bibnamefont {Lidar}},\
  }\bibfield  {title} {\enquote {\bibinfo {title} {Demonstration of fidelity
  improvement using dynamical decoupling with superconducting qubits},}\ }\href
  {https://link.aps.org/doi/10.1103/PhysRevLett.121.220502} {\bibfield
  {journal} {\bibinfo  {journal} {Phys. Rev. Lett.}\ }\textbf {\bibinfo
  {volume} {121}},\ \bibinfo {pages} {220502} (\bibinfo {year}
  {2018})}\BibitemShut {NoStop}%
\bibitem [{\citenamefont {Souza}(2021)}]{souza2020process}%
  \BibitemOpen
  \bibfield  {author} {\bibinfo {author} {\bibfnamefont {Alexandre~M.}\
  \bibnamefont {Souza}},\ }\bibfield  {title} {\enquote {\bibinfo {title}
  {Process tomography of robust dynamical decoupling with superconducting
  qubits},}\ }\href {https://doi.org/10.1007/s11128-021-03176-z} {\bibfield
  {journal} {\bibinfo  {journal} {Quantum Information Processing}\ }\textbf
  {\bibinfo {volume} {20}},\ \bibinfo {pages} {237} (\bibinfo {year}
  {2021})}\BibitemShut {NoStop}%
\bibitem [{\citenamefont {Tripathi}\ \emph {et~al.}(2022)\citenamefont
  {Tripathi}, \citenamefont {Chen}, \citenamefont {Khezri}, \citenamefont
  {Yip}, \citenamefont {Levenson-Falk},\ and\ \citenamefont
  {Lidar}}]{tripathi2021suppression}%
  \BibitemOpen
  \bibfield  {author} {\bibinfo {author} {\bibfnamefont {Vinay}\ \bibnamefont
  {Tripathi}}, \bibinfo {author} {\bibfnamefont {Huo}\ \bibnamefont {Chen}},
  \bibinfo {author} {\bibfnamefont {Mostafa}\ \bibnamefont {Khezri}}, \bibinfo
  {author} {\bibfnamefont {Ka-Wa}\ \bibnamefont {Yip}}, \bibinfo {author}
  {\bibfnamefont {E.~M.}\ \bibnamefont {Levenson-Falk}}, \ and\ \bibinfo
  {author} {\bibfnamefont {Daniel~A.}\ \bibnamefont {Lidar}},\ }\bibfield
  {title} {\enquote {\bibinfo {title} {Suppression of crosstalk in
  superconducting qubits using dynamical decoupling},}\ }\href {\doibase
  10.1103/PhysRevApplied.18.024068} {\bibfield  {journal} {\bibinfo  {journal}
  {Physical Review Applied}\ }\textbf {\bibinfo {volume} {18}},\ \bibinfo
  {pages} {024068--} (\bibinfo {year} {2022})}\BibitemShut {NoStop}%
\bibitem [{\citenamefont {Jurcevic}\ \emph {et~al.}(2021)\citenamefont
  {Jurcevic}, \citenamefont {{Javadi-Abhari}}, \citenamefont {Bishop},
  \citenamefont {Lauer}, \citenamefont {Bogorin}, \citenamefont {Brink},
  \citenamefont {Capelluto}, \citenamefont {G{\"u}nl{\"u}k}, \citenamefont
  {Itoko}, \citenamefont {Kanazawa}, \citenamefont {Kandala}, \citenamefont
  {Keefe}, \citenamefont {Krsulich}, \citenamefont {Landers}, \citenamefont
  {Lewandowski}, \citenamefont {McClure}, \citenamefont {Nannicini},
  \citenamefont {Narasgond}, \citenamefont {Nayfeh}, \citenamefont {Pritchett},
  \citenamefont {Rothwell}, \citenamefont {Srinivasan}, \citenamefont
  {Sundaresan}, \citenamefont {Wang}, \citenamefont {Wei}, \citenamefont
  {Wood}, \citenamefont {Yau}, \citenamefont {Zhang}, \citenamefont {Dial},
  \citenamefont {Chow},\ and\ \citenamefont
  {Gambetta}}]{jurcevicDemonstrationQuantumVolume2021}%
  \BibitemOpen
  \bibfield  {author} {\bibinfo {author} {\bibfnamefont {Petar}\ \bibnamefont
  {Jurcevic}}, \bibinfo {author} {\bibfnamefont {Ali}\ \bibnamefont
  {{Javadi-Abhari}}}, \bibinfo {author} {\bibfnamefont {Lev~S.}\ \bibnamefont
  {Bishop}}, \bibinfo {author} {\bibfnamefont {Isaac}\ \bibnamefont {Lauer}},
  \bibinfo {author} {\bibfnamefont {Daniela~F.}\ \bibnamefont {Bogorin}},
  \bibinfo {author} {\bibfnamefont {Markus}\ \bibnamefont {Brink}}, \bibinfo
  {author} {\bibfnamefont {Lauren}\ \bibnamefont {Capelluto}}, \bibinfo
  {author} {\bibfnamefont {Oktay}\ \bibnamefont {G{\"u}nl{\"u}k}}, \bibinfo
  {author} {\bibfnamefont {Toshinari}\ \bibnamefont {Itoko}}, \bibinfo {author}
  {\bibfnamefont {Naoki}\ \bibnamefont {Kanazawa}}, \bibinfo {author}
  {\bibfnamefont {Abhinav}\ \bibnamefont {Kandala}}, \bibinfo {author}
  {\bibfnamefont {George~A.}\ \bibnamefont {Keefe}}, \bibinfo {author}
  {\bibfnamefont {Kevin}\ \bibnamefont {Krsulich}}, \bibinfo {author}
  {\bibfnamefont {William}\ \bibnamefont {Landers}}, \bibinfo {author}
  {\bibfnamefont {Eric~P.}\ \bibnamefont {Lewandowski}}, \bibinfo {author}
  {\bibfnamefont {Douglas~T.}\ \bibnamefont {McClure}}, \bibinfo {author}
  {\bibfnamefont {Giacomo}\ \bibnamefont {Nannicini}}, \bibinfo {author}
  {\bibfnamefont {Adinath}\ \bibnamefont {Narasgond}}, \bibinfo {author}
  {\bibfnamefont {Hasan~M.}\ \bibnamefont {Nayfeh}}, \bibinfo {author}
  {\bibfnamefont {Emily}\ \bibnamefont {Pritchett}}, \bibinfo {author}
  {\bibfnamefont {Mary~Beth}\ \bibnamefont {Rothwell}}, \bibinfo {author}
  {\bibfnamefont {Srikanth}\ \bibnamefont {Srinivasan}}, \bibinfo {author}
  {\bibfnamefont {Neereja}\ \bibnamefont {Sundaresan}}, \bibinfo {author}
  {\bibfnamefont {Cindy}\ \bibnamefont {Wang}}, \bibinfo {author}
  {\bibfnamefont {Ken~X.}\ \bibnamefont {Wei}}, \bibinfo {author}
  {\bibfnamefont {Christopher~J.}\ \bibnamefont {Wood}}, \bibinfo {author}
  {\bibfnamefont {Jeng-Bang}\ \bibnamefont {Yau}}, \bibinfo {author}
  {\bibfnamefont {Eric~J.}\ \bibnamefont {Zhang}}, \bibinfo {author}
  {\bibfnamefont {Oliver~E.}\ \bibnamefont {Dial}}, \bibinfo {author}
  {\bibfnamefont {Jerry~M.}\ \bibnamefont {Chow}}, \ and\ \bibinfo {author}
  {\bibfnamefont {Jay~M.}\ \bibnamefont {Gambetta}},\ }\bibfield  {title}
  {\enquote {\bibinfo {title} {Demonstration of quantum volume 64 on a
  superconducting quantum computing system},}\ }\href
  {https://iopscience.iop.org/article/10.1088/2058-9565/abe519} {\bibfield
  {journal} {\bibinfo  {journal} {Quantum Sci. Technol.}\ }\textbf {\bibinfo
  {volume} {6}},\ \bibinfo {pages} {025020} (\bibinfo {year}
  {2021})}\BibitemShut {NoStop}%
\bibitem [{\citenamefont {Ravi}\ \emph {et~al.}(2022)\citenamefont {Ravi},
  \citenamefont {Smith}, \citenamefont {Gokhale}, \citenamefont {Mari},
  \citenamefont {Earnest}, \citenamefont {Javadi-Abhari},\ and\ \citenamefont
  {Chong}}]{raviVAQEMVariationalApproach2021}%
  \BibitemOpen
  \bibfield  {author} {\bibinfo {author} {\bibfnamefont {G.~S.}\ \bibnamefont
  {Ravi}}, \bibinfo {author} {\bibfnamefont {K.~N.}\ \bibnamefont {Smith}},
  \bibinfo {author} {\bibfnamefont {P.}~\bibnamefont {Gokhale}}, \bibinfo
  {author} {\bibfnamefont {A.}~\bibnamefont {Mari}}, \bibinfo {author}
  {\bibfnamefont {N.}~\bibnamefont {Earnest}}, \bibinfo {author} {\bibfnamefont
  {A.}~\bibnamefont {Javadi-Abhari}}, \ and\ \bibinfo {author} {\bibfnamefont
  {F.~T.}\ \bibnamefont {Chong}},\ }\bibfield  {title} {\enquote {\bibinfo
  {title} {Vaqem: A variational approach to quantum error mitigation},}\
  }\bibfield  {booktitle} {\emph {\bibinfo {booktitle} {2022 IEEE International
  Symposium on High-Performance Computer Architecture (HPCA)}},\ }\href
  {\doibase 10.1109/HPCA53966.2022.00029} {\bibfield  {journal} {\bibinfo
  {journal} {2022 IEEE International Symposium on High-Performance Computer
  Architecture (HPCA)}\ ,\ \bibinfo {pages} {288--303}} (\bibinfo {year}
  {2022})}\BibitemShut {NoStop}%
\bibitem [{\citenamefont {Zhou}\ \emph {et~al.}(2023)\citenamefont {Zhou},
  \citenamefont {Sitler}, \citenamefont {Oda}, \citenamefont {Schultz},\ and\
  \citenamefont {Quiroz}}]{Zeyuan:22}%
  \BibitemOpen
  \bibfield  {author} {\bibinfo {author} {\bibfnamefont {Zeyuan}\ \bibnamefont
  {Zhou}}, \bibinfo {author} {\bibfnamefont {Ryan}\ \bibnamefont {Sitler}},
  \bibinfo {author} {\bibfnamefont {Yasuo}\ \bibnamefont {Oda}}, \bibinfo
  {author} {\bibfnamefont {Kevin}\ \bibnamefont {Schultz}}, \ and\ \bibinfo
  {author} {\bibfnamefont {Gregory}\ \bibnamefont {Quiroz}},\ }\bibfield
  {title} {\enquote {\bibinfo {title} {Quantum crosstalk robust quantum
  control},}\ }\href {\doibase 10.1103/PhysRevLett.131.210802} {\bibfield
  {journal} {\bibinfo  {journal} {Physical Review Letters}\ }\textbf {\bibinfo
  {volume} {131}},\ \bibinfo {pages} {210802--} (\bibinfo {year}
  {2023})}\BibitemShut {NoStop}%
\bibitem [{\citenamefont {B{\"a}umer}\ \emph {et~al.}(2023)\citenamefont
  {B{\"a}umer}, \citenamefont {Tripathi}, \citenamefont {Wang}, \citenamefont
  {Rall}, \citenamefont {Chen}, \citenamefont {Majumder}, \citenamefont
  {Seif},\ and\ \citenamefont {Minev}}]{baumer2023efficient}%
  \BibitemOpen
  \bibfield  {author} {\bibinfo {author} {\bibfnamefont {Elisa}\ \bibnamefont
  {B{\"a}umer}}, \bibinfo {author} {\bibfnamefont {Vinay}\ \bibnamefont
  {Tripathi}}, \bibinfo {author} {\bibfnamefont {Derek~S.}\ \bibnamefont
  {Wang}}, \bibinfo {author} {\bibfnamefont {Patrick}\ \bibnamefont {Rall}},
  \bibinfo {author} {\bibfnamefont {Edward~H.}\ \bibnamefont {Chen}}, \bibinfo
  {author} {\bibfnamefont {Swarnadeep}\ \bibnamefont {Majumder}}, \bibinfo
  {author} {\bibfnamefont {Alireza}\ \bibnamefont {Seif}}, \ and\ \bibinfo
  {author} {\bibfnamefont {Zlatko~K.}\ \bibnamefont {Minev}},\ }\href@noop {}
  {\enquote {\bibinfo {title} {Efficient long-range entanglement using dynamic
  circuits},}\ } (\bibinfo {year} {2023}),\ \Eprint
  {http://arxiv.org/abs/2308.13065} {arXiv:2308.13065 [quant-ph]} \BibitemShut
  {NoStop}%
\bibitem [{\citenamefont {Seif}\ \emph {et~al.}(2024)\citenamefont {Seif},
  \citenamefont {Liao}, \citenamefont {Tripathi}, \citenamefont {Krsulich},
  \citenamefont {Malekakhlagh}, \citenamefont {Amico}, \citenamefont
  {Jurcevic},\ and\ \citenamefont {Javadi-Abhari}}]{seif2024suppressing}%
  \BibitemOpen
  \bibfield  {author} {\bibinfo {author} {\bibfnamefont {Alireza}\ \bibnamefont
  {Seif}}, \bibinfo {author} {\bibfnamefont {Haoran}\ \bibnamefont {Liao}},
  \bibinfo {author} {\bibfnamefont {Vinay}\ \bibnamefont {Tripathi}}, \bibinfo
  {author} {\bibfnamefont {Kevin}\ \bibnamefont {Krsulich}}, \bibinfo {author}
  {\bibfnamefont {Moein}\ \bibnamefont {Malekakhlagh}}, \bibinfo {author}
  {\bibfnamefont {Mirko}\ \bibnamefont {Amico}}, \bibinfo {author}
  {\bibfnamefont {Petar}\ \bibnamefont {Jurcevic}}, \ and\ \bibinfo {author}
  {\bibfnamefont {Ali}\ \bibnamefont {Javadi-Abhari}},\ }\href@noop {}
  {\enquote {\bibinfo {title} {Suppressing correlated noise in quantum
  computers via context-aware compiling},}\ } (\bibinfo {year} {2024}),\
  \Eprint {http://arxiv.org/abs/2403.06852} {arXiv:2403.06852 [quant-ph]}
  \BibitemShut {NoStop}%
\bibitem [{\citenamefont {Shirizly}\ \emph {et~al.}(2024)\citenamefont
  {Shirizly}, \citenamefont {Misguich},\ and\ \citenamefont
  {Landa}}]{Shirizly:2024aa}%
  \BibitemOpen
  \bibfield  {author} {\bibinfo {author} {\bibfnamefont {Liran}\ \bibnamefont
  {Shirizly}}, \bibinfo {author} {\bibfnamefont {Gr{\'e}goire}\ \bibnamefont
  {Misguich}}, \ and\ \bibinfo {author} {\bibfnamefont {Haggai}\ \bibnamefont
  {Landa}},\ }\bibfield  {title} {\enquote {\bibinfo {title} {Dissipative
  dynamics of graph-state stabilizers with superconducting qubits},}\ }\href
  {\doibase 10.1103/PhysRevLett.132.010601} {\bibfield  {journal} {\bibinfo
  {journal} {Physical Review Letters}\ }\textbf {\bibinfo {volume} {132}},\
  \bibinfo {pages} {010601--} (\bibinfo {year} {2024})}\BibitemShut {NoStop}%
\bibitem [{\citenamefont {B{\"a}umer}\ \emph {et~al.}(2024)\citenamefont
  {B{\"a}umer}, \citenamefont {Tripathi}, \citenamefont {Seif}, \citenamefont
  {Lidar},\ and\ \citenamefont {Wang}}]{Baumer2024}%
  \BibitemOpen
  \bibfield  {author} {\bibinfo {author} {\bibfnamefont {Elisa}\ \bibnamefont
  {B{\"a}umer}}, \bibinfo {author} {\bibfnamefont {Vinay}\ \bibnamefont
  {Tripathi}}, \bibinfo {author} {\bibfnamefont {Alireza}\ \bibnamefont
  {Seif}}, \bibinfo {author} {\bibfnamefont {Daniel}\ \bibnamefont {Lidar}}, \
  and\ \bibinfo {author} {\bibfnamefont {Derek~S.}\ \bibnamefont {Wang}},\
  }\href@noop {} {\enquote {\bibinfo {title} {Quantum fourier transform using
  dynamic circuits},}\ } (\bibinfo {year} {2024}),\ \Eprint
  {http://arxiv.org/abs/2403.09514} {arXiv:2403.09514 [quant-ph]} \BibitemShut
  {NoStop}%
\bibitem [{\citenamefont {Evert}\ \emph {et~al.}(2024)\citenamefont {Evert},
  \citenamefont {Izquierdo}, \citenamefont {Sud}, \citenamefont {Hu},
  \citenamefont {Grabbe}, \citenamefont {Rieffel}, \citenamefont {Reagor},\
  and\ \citenamefont {Wang}}]{evert2024syncopated}%
  \BibitemOpen
  \bibfield  {author} {\bibinfo {author} {\bibfnamefont {Bram}\ \bibnamefont
  {Evert}}, \bibinfo {author} {\bibfnamefont {Zoe~Gonzalez}\ \bibnamefont
  {Izquierdo}}, \bibinfo {author} {\bibfnamefont {James}\ \bibnamefont {Sud}},
  \bibinfo {author} {\bibfnamefont {Hong-Ye}\ \bibnamefont {Hu}}, \bibinfo
  {author} {\bibfnamefont {Shon}\ \bibnamefont {Grabbe}}, \bibinfo {author}
  {\bibfnamefont {Eleanor~G.}\ \bibnamefont {Rieffel}}, \bibinfo {author}
  {\bibfnamefont {Matthew~J.}\ \bibnamefont {Reagor}}, \ and\ \bibinfo {author}
  {\bibfnamefont {Zhihui}\ \bibnamefont {Wang}},\ }\href
  {https://arxiv.org/abs/2403.07836} {\enquote {\bibinfo {title} {Syncopated
  dynamical decoupling for suppressing crosstalk in quantum circuits},}\ }
  (\bibinfo {year} {2024}),\ \Eprint {http://arxiv.org/abs/2403.07836}
  {arXiv:2403.07836 [quant-ph]} \BibitemShut {NoStop}%
\bibitem [{\citenamefont {Brown}\ and\ \citenamefont
  {Lidar}(2024)}]{brown2024efficient}%
  \BibitemOpen
  \bibfield  {author} {\bibinfo {author} {\bibfnamefont {Amy~F.}\ \bibnamefont
  {Brown}}\ and\ \bibinfo {author} {\bibfnamefont {Daniel~A.}\ \bibnamefont
  {Lidar}},\ }\href@noop {} {\enquote {\bibinfo {title} {Efficient
  chromatic-number-based multi-qubit decoherence and crosstalk suppression},}\
  } (\bibinfo {year} {2024}),\ \Eprint {http://arxiv.org/abs/2406.13901}
  {arXiv:2406.13901} \BibitemShut {NoStop}%
\bibitem [{\citenamefont {Tripathi}\ \emph {et~al.}(2024)\citenamefont
  {Tripathi}, \citenamefont {Goss}, \citenamefont {Vezvaee}, \citenamefont
  {Nguyen}, \citenamefont {Siddiqi},\ and\ \citenamefont
  {Lidar}}]{tripathi2024quditdynamicaldecouplingsuperconducting}%
  \BibitemOpen
  \bibfield  {author} {\bibinfo {author} {\bibfnamefont {Vinay}\ \bibnamefont
  {Tripathi}}, \bibinfo {author} {\bibfnamefont {Noah}\ \bibnamefont {Goss}},
  \bibinfo {author} {\bibfnamefont {Arian}\ \bibnamefont {Vezvaee}}, \bibinfo
  {author} {\bibfnamefont {Long~B.}\ \bibnamefont {Nguyen}}, \bibinfo {author}
  {\bibfnamefont {Irfan}\ \bibnamefont {Siddiqi}}, \ and\ \bibinfo {author}
  {\bibfnamefont {Daniel~A.}\ \bibnamefont {Lidar}},\ }\href
  {https://arxiv.org/abs/2407.04893} {\enquote {\bibinfo {title} {Qudit
  dynamical decoupling on a superconducting quantum processor},}\ } (\bibinfo
  {year} {2024}),\ \Eprint {http://arxiv.org/abs/2407.04893} {arXiv:2407.04893
  [quant-ph]} \BibitemShut {NoStop}%
\bibitem [{\citenamefont {Jozsa}(2001)}]{Jozsa:2001aa}%
  \BibitemOpen
  \bibfield  {author} {\bibinfo {author} {\bibfnamefont {R.}~\bibnamefont
  {Jozsa}},\ }\bibfield  {title} {\enquote {\bibinfo {title} {Quantum
  factoring, discrete logarithms, and the hidden subgroup problem},}\ }\href
  {\doibase 10.1109/5992.909000} {\bibfield  {journal} {\bibinfo  {journal}
  {Computing in Science \& Engineering}\ }\textbf {\bibinfo {volume} {3}},\
  \bibinfo {pages} {34--43} (\bibinfo {year} {2001})}\BibitemShut {NoStop}%
\bibitem [{\citenamefont {Zantema}(2022)}]{Zantema2022}%
  \BibitemOpen
  \bibfield  {author} {\bibinfo {author} {\bibfnamefont {Hans}\ \bibnamefont
  {Zantema}},\ }\bibfield  {title} {\enquote {\bibinfo {title} {{Complexity of
  Simon's problem in classical sense}},}\ }\href
  {http://arxiv.org/abs/2211.01776} {\bibfield  {journal} {\bibinfo  {journal}
  {arXiv e-prints}\ } (\bibinfo {year} {2022})},\ \Eprint
  {http://arxiv.org/abs/2211.01776} {2211.01776} \BibitemShut {NoStop}%
\bibitem [{\citenamefont {Cai}\ and\ \citenamefont {Qiu}(2018)}]{Cai2018}%
  \BibitemOpen
  \bibfield  {author} {\bibinfo {author} {\bibfnamefont {Guangya}\ \bibnamefont
  {Cai}}\ and\ \bibinfo {author} {\bibfnamefont {Daowen}\ \bibnamefont {Qiu}},\
  }\bibfield  {title} {\enquote {\bibinfo {title} {Optimal separation in exact
  query complexities for simon's problem},}\ }\href {\doibase
  10.1016/j.jcss.2018.05.001} {\bibfield  {journal} {\bibinfo  {journal}
  {Journal of Computer and System Sciences}\ }\textbf {\bibinfo {volume} {97}}
  (\bibinfo {year} {2018}),\ 10.1016/j.jcss.2018.05.001}\BibitemShut {NoStop}%
\bibitem [{\citenamefont {Sloane}(2001)}]{Erdos-Borwein}%
  \BibitemOpen
  \bibfield  {author} {\bibinfo {author} {\bibfnamefont {N.~J.~A.}\
  \bibnamefont {Sloane}},\ }\href {https://oeis.org/A065442} {\enquote
  {\bibinfo {title} {{The On-Line Encyclopedia of Integer Sequences:
  Erd\"os-Borwein constant}},}\ } (\bibinfo {year} {2001})\BibitemShut
  {NoStop}%
\bibitem [{\citenamefont {Lenstra}(2011)}]{Lenstra:2011aa}%
  \BibitemOpen
  \bibfield  {author} {\bibinfo {author} {\bibfnamefont {Arjen~K.}\
  \bibnamefont {Lenstra}},\ }\enquote {\bibinfo {title} {\textit{L}
  notation},}\ in\ \href {https://doi.org/10.1007/978-1-4419-5906-5_459} {\emph
  {\bibinfo {booktitle} {{Encyclopedia of Cryptography and Security}}}},\
  \bibinfo {editor} {edited by\ \bibinfo {editor} {\bibfnamefont {Henk C.~A.}\
  \bibnamefont {van Tilborg}}\ and\ \bibinfo {editor} {\bibfnamefont {Sushil}\
  \bibnamefont {Jajodia}}}\ (\bibinfo  {publisher} {Springer US},\ \bibinfo
  {address} {Boston, MA},\ \bibinfo {year} {2011})\ pp.\ \bibinfo {pages}
  {709--710}\BibitemShut {NoStop}%
\bibitem [{\citenamefont {Ezzell}\ \emph {et~al.}(2023)\citenamefont {Ezzell},
  \citenamefont {Pokharel}, \citenamefont {Tewala}, \citenamefont {Quiroz},\
  and\ \citenamefont {Lidar}}]{DD-survey}%
  \BibitemOpen
  \bibfield  {author} {\bibinfo {author} {\bibfnamefont {Nic}\ \bibnamefont
  {Ezzell}}, \bibinfo {author} {\bibfnamefont {Bibek}\ \bibnamefont
  {Pokharel}}, \bibinfo {author} {\bibfnamefont {Lina}\ \bibnamefont {Tewala}},
  \bibinfo {author} {\bibfnamefont {Gregory}\ \bibnamefont {Quiroz}}, \ and\
  \bibinfo {author} {\bibfnamefont {Daniel~A.}\ \bibnamefont {Lidar}},\
  }\bibfield  {title} {\enquote {\bibinfo {title} {Dynamical decoupling for
  superconducting qubits: A performance survey},}\ }\href {\doibase
  10.1103/PhysRevApplied.20.064027} {\bibfield  {journal} {\bibinfo  {journal}
  {Physical Review Applied}\ }\textbf {\bibinfo {volume} {20}},\ \bibinfo
  {pages} {064027--} (\bibinfo {year} {2023})}\BibitemShut {NoStop}%
\bibitem [{\citenamefont {Genov}\ \emph {et~al.}(2017)\citenamefont {Genov},
  \citenamefont {Schraft}, \citenamefont {Vitanov},\ and\ \citenamefont
  {Halfmann}}]{Genov:2017aa}%
  \BibitemOpen
  \bibfield  {author} {\bibinfo {author} {\bibfnamefont {Genko~T.}\
  \bibnamefont {Genov}}, \bibinfo {author} {\bibfnamefont {Daniel}\
  \bibnamefont {Schraft}}, \bibinfo {author} {\bibfnamefont {Nikolay~V.}\
  \bibnamefont {Vitanov}}, \ and\ \bibinfo {author} {\bibfnamefont {Thomas}\
  \bibnamefont {Halfmann}},\ }\bibfield  {title} {\enquote {\bibinfo {title}
  {Arbitrarily accurate pulse sequences for robust dynamical decoupling},}\
  }\href {\doibase 10.1103/PhysRevLett.118.133202} {\bibfield  {journal}
  {\bibinfo  {journal} {Physical Review Letters}\ }\textbf {\bibinfo {volume}
  {118}},\ \bibinfo {pages} {133202--} (\bibinfo {year} {2017})}\BibitemShut
  {NoStop}%
\bibitem [{\citenamefont {Quiroz}\ and\ \citenamefont
  {Lidar}(2013)}]{Quiroz:2013fv}%
  \BibitemOpen
  \bibfield  {author} {\bibinfo {author} {\bibfnamefont {Gregory}\ \bibnamefont
  {Quiroz}}\ and\ \bibinfo {author} {\bibfnamefont {Daniel~A.}\ \bibnamefont
  {Lidar}},\ }\bibfield  {title} {\enquote {\bibinfo {title} {Optimized
  dynamical decoupling via genetic algorithms},}\ }\href
  {http://link.aps.org/doi/10.1103/PhysRevA.88.052306} {\bibfield  {journal}
  {\bibinfo  {journal} {Phys. Rev. A}\ }\textbf {\bibinfo {volume} {88}},\
  \bibinfo {pages} {052306--} (\bibinfo {year} {2013})}\BibitemShut {NoStop}%
\bibitem [{\citenamefont {Akaike}(1974)}]{Akaike:1974aa}%
  \BibitemOpen
  \bibfield  {author} {\bibinfo {author} {\bibfnamefont {H.}~\bibnamefont
  {Akaike}},\ }\bibfield  {title} {\enquote {\bibinfo {title} {A new look at
  the statistical model identification},}\ }\href {\doibase
  10.1109/TAC.1974.1100705} {\bibfield  {journal} {\bibinfo  {journal} {IEEE
  Transactions on Automatic Control}\ }\textbf {\bibinfo {volume} {19}},\
  \bibinfo {pages} {716--723} (\bibinfo {year} {1974})}\BibitemShut {NoStop}%
\bibitem [{\citenamefont {Rines}\ and\ \citenamefont
  {Chuang}(2018)}]{rines2018high}%
  \BibitemOpen
  \bibfield  {author} {\bibinfo {author} {\bibfnamefont {Rich}\ \bibnamefont
  {Rines}}\ and\ \bibinfo {author} {\bibfnamefont {Isaac}\ \bibnamefont
  {Chuang}},\ }\href@noop {} {\enquote {\bibinfo {title} {High performance
  quantum modular multipliers},}\ } (\bibinfo {year} {2018}),\ \Eprint
  {http://arxiv.org/abs/1801.01081} {arXiv:1801.01081 [quant-ph]} \BibitemShut
  {NoStop}%
\bibitem [{\citenamefont {Oonishi}\ and\ \citenamefont
  {Kunihiro}(2023)}]{10262370}%
  \BibitemOpen
  \bibfield  {author} {\bibinfo {author} {\bibfnamefont {K.}~\bibnamefont
  {Oonishi}}\ and\ \bibinfo {author} {\bibfnamefont {N.}~\bibnamefont
  {Kunihiro}},\ }\bibfield  {title} {\enquote {\bibinfo {title} {Shor's
  algorithm using efficient approximate quantum fourier transform},}\ }\href
  {\doibase 10.1109/TQE.2023.3319044} {\bibfield  {journal} {\bibinfo
  {journal} {IEEE Transactions on Quantum Engineering}\ }\textbf {\bibinfo
  {volume} {4}},\ \bibinfo {pages} {1--16} (\bibinfo {year}
  {2023})}\BibitemShut {NoStop}%
\bibitem [{\citenamefont {Poblete}\ \emph {et~al.}(2006)\citenamefont
  {Poblete}, \citenamefont {Munro},\ and\ \citenamefont
  {Papadakis}}]{poblete2006binomial}%
  \BibitemOpen
  \bibfield  {author} {\bibinfo {author} {\bibfnamefont {Patricio~V}\
  \bibnamefont {Poblete}}, \bibinfo {author} {\bibfnamefont {J~Ian}\
  \bibnamefont {Munro}}, \ and\ \bibinfo {author} {\bibfnamefont {Thomas}\
  \bibnamefont {Papadakis}},\ }\bibfield  {title} {\enquote {\bibinfo {title}
  {The binomial transform and the analysis of skip lists},}\ }\href
  {https://www.sciencedirect.com/science/article/pii/S0304397505007711}
  {\bibfield  {journal} {\bibinfo  {journal} {Theoretical computer science}\
  }\textbf {\bibinfo {volume} {352}},\ \bibinfo {pages} {136--158} (\bibinfo
  {year} {2006})}\BibitemShut {NoStop}%
\bibitem [{\citenamefont {Temme}\ \emph {et~al.}(2017)\citenamefont {Temme},
  \citenamefont {Bravyi},\ and\ \citenamefont {Gambetta}}]{Temme:2017aa}%
  \BibitemOpen
  \bibfield  {author} {\bibinfo {author} {\bibfnamefont {Kristan}\ \bibnamefont
  {Temme}}, \bibinfo {author} {\bibfnamefont {Sergey}\ \bibnamefont {Bravyi}},
  \ and\ \bibinfo {author} {\bibfnamefont {Jay~M.}\ \bibnamefont {Gambetta}},\
  }\bibfield  {title} {\enquote {\bibinfo {title} {Error mitigation for
  short-depth quantum circuits},}\ }\href
  {https://link.aps.org/doi/10.1103/PhysRevLett.119.180509} {\bibfield
  {journal} {\bibinfo  {journal} {Physical Review Letters}\ }\textbf {\bibinfo
  {volume} {119}},\ \bibinfo {pages} {180509--} (\bibinfo {year}
  {2017})}\BibitemShut {NoStop}%
\bibitem [{\citenamefont {Srinivasan}\ \emph {et~al.}(2022)\citenamefont
  {Srinivasan}, \citenamefont {Pokharel}, \citenamefont {Quiroz},\ and\
  \citenamefont {Boots}}]{Srinivasan:22}%
  \BibitemOpen
  \bibfield  {author} {\bibinfo {author} {\bibfnamefont {Siddarth}\
  \bibnamefont {Srinivasan}}, \bibinfo {author} {\bibfnamefont {Bibek}\
  \bibnamefont {Pokharel}}, \bibinfo {author} {\bibfnamefont {Gregory}\
  \bibnamefont {Quiroz}}, \ and\ \bibinfo {author} {\bibfnamefont {Byron}\
  \bibnamefont {Boots}},\ }\bibfield  {title} {\enquote {\bibinfo {title}
  {Scalable measurement error mitigation via iterative bayesian unfolding},}\
  }\href {https://arxiv.org/abs/2210.12284} {\bibfield  {journal} {\bibinfo
  {journal} {arXiv e-prints}\ } (\bibinfo {year} {2022})},\ \Eprint
  {http://arxiv.org/abs/2210.12284} {2210.12284} \BibitemShut {NoStop}%
\bibitem [{\citenamefont {Nation}\ \emph {et~al.}(2021)\citenamefont {Nation},
  \citenamefont {Kang}, \citenamefont {Sundaresan},\ and\ \citenamefont
  {Gambetta}}]{Nation:2021aa}%
  \BibitemOpen
  \bibfield  {author} {\bibinfo {author} {\bibfnamefont {Paul~D.}\ \bibnamefont
  {Nation}}, \bibinfo {author} {\bibfnamefont {Hwajung}\ \bibnamefont {Kang}},
  \bibinfo {author} {\bibfnamefont {Neereja}\ \bibnamefont {Sundaresan}}, \
  and\ \bibinfo {author} {\bibfnamefont {Jay~M.}\ \bibnamefont {Gambetta}},\
  }\bibfield  {title} {\enquote {\bibinfo {title} {Scalable mitigation of
  measurement errors on quantum computers},}\ }\href {\doibase
  10.1103/PRXQuantum.2.040326} {\bibfield  {journal} {\bibinfo  {journal} {PRX
  Quantum}\ }\textbf {\bibinfo {volume} {2}},\ \bibinfo {pages} {040326--}
  (\bibinfo {year} {2021})}\BibitemShut {NoStop}%
\bibitem [{\citenamefont {Guntuboyina}\ \emph {et~al.}(2013)\citenamefont
  {Guntuboyina}, \citenamefont {Saha},\ and\ \citenamefont
  {Schiebinger}}]{guntuboyina2013sharp}%
  \BibitemOpen
  \bibfield  {author} {\bibinfo {author} {\bibfnamefont {Adityanand}\
  \bibnamefont {Guntuboyina}}, \bibinfo {author} {\bibfnamefont {Sujayam}\
  \bibnamefont {Saha}}, \ and\ \bibinfo {author} {\bibfnamefont {Geoffrey}\
  \bibnamefont {Schiebinger}},\ }\bibfield  {title} {\enquote {\bibinfo {title}
  {Sharp inequalities for $ f $-divergences},}\ }\href
  {https://ieeexplore.ieee.org/document/6655891} {\bibfield  {journal}
  {\bibinfo  {journal} {IEEE transactions on information theory}\ }\textbf
  {\bibinfo {volume} {60}},\ \bibinfo {pages} {104--121} (\bibinfo {year}
  {2013})}\BibitemShut {NoStop}%
\bibitem [{\citenamefont {Wagenmakers}\ and\ \citenamefont
  {Farrell}(2004)}]{Wagenmakers:2004aa}%
  \BibitemOpen
  \bibfield  {author} {\bibinfo {author} {\bibfnamefont {Eric-Jan}\
  \bibnamefont {Wagenmakers}}\ and\ \bibinfo {author} {\bibfnamefont {Simon}\
  \bibnamefont {Farrell}},\ }\bibfield  {title} {\enquote {\bibinfo {title}
  {{AIC model selection using Akaike weights}},}\ }\href {\doibase
  10.3758/BF03206482} {\bibfield  {journal} {\bibinfo  {journal} {Psychonomic
  Bulletin \& Review}\ }\textbf {\bibinfo {volume} {11}},\ \bibinfo {pages}
  {192--196} (\bibinfo {year} {2004})}\BibitemShut {NoStop}%
\bibitem [{\citenamefont {Burnham}\ and\ \citenamefont
  {Anderson}(2002)}]{Burnham:book}%
  \BibitemOpen
  \bibfield  {author} {\bibinfo {author} {\bibfnamefont {K.~P.}\ \bibnamefont
  {Burnham}}\ and\ \bibinfo {author} {\bibfnamefont {D.~R.}\ \bibnamefont
  {Anderson}},\ }\href {https://doi.org/10.1007/b97636} {\emph {\bibinfo
  {title} {Model Selection and Multimodel Inference: A Practical
  Information-Theoretic Approach}}},\ \bibinfo {edition} {2nd}\ ed.\ (\bibinfo
  {publisher} {Springer},\ \bibinfo {year} {2002})\BibitemShut {NoStop}%
\bibitem [{QCT()}]{QCTRL-IBM}%
  \BibitemOpen
  \href {https://docs.q-ctrl.com/q-ctrl-embedded/ibm/get-started} {\enquote
  {\bibinfo {title} {{Q-CTRL Embedded on IBM Quantum services overview}},}\
  }\BibitemShut {NoStop}%
\bibitem [{\citenamefont {Mundada}\ \emph {et~al.}(2023)\citenamefont
  {Mundada}, \citenamefont {Barbosa}, \citenamefont {Maity}, \citenamefont
  {Wang}, \citenamefont {Merkh}, \citenamefont {Stace}, \citenamefont
  {Nielson}, \citenamefont {Carvalho}, \citenamefont {Hush}, \citenamefont
  {Biercuk},\ and\ \citenamefont {Baum}}]{Mundada:2023aa}%
  \BibitemOpen
  \bibfield  {author} {\bibinfo {author} {\bibfnamefont {Pranav~S.}\
  \bibnamefont {Mundada}}, \bibinfo {author} {\bibfnamefont {Aaron}\
  \bibnamefont {Barbosa}}, \bibinfo {author} {\bibfnamefont {Smarak}\
  \bibnamefont {Maity}}, \bibinfo {author} {\bibfnamefont {Yulun}\ \bibnamefont
  {Wang}}, \bibinfo {author} {\bibfnamefont {Thomas}\ \bibnamefont {Merkh}},
  \bibinfo {author} {\bibfnamefont {T.~M.}\ \bibnamefont {Stace}}, \bibinfo
  {author} {\bibfnamefont {Felicity}\ \bibnamefont {Nielson}}, \bibinfo
  {author} {\bibfnamefont {Andre R.~R.}\ \bibnamefont {Carvalho}}, \bibinfo
  {author} {\bibfnamefont {Michael}\ \bibnamefont {Hush}}, \bibinfo {author}
  {\bibfnamefont {Michael~J.}\ \bibnamefont {Biercuk}}, \ and\ \bibinfo
  {author} {\bibfnamefont {Yuval}\ \bibnamefont {Baum}},\ }\bibfield  {title}
  {\enquote {\bibinfo {title} {Experimental benchmarking of an automated
  deterministic error-suppression workflow for quantum algorithms},}\ }\href
  {\doibase 10.1103/PhysRevApplied.20.024034} {\bibfield  {journal} {\bibinfo
  {journal} {Physical Review Applied}\ }\textbf {\bibinfo {volume} {20}},\
  \bibinfo {pages} {024034--} (\bibinfo {year} {2023})}\BibitemShut {NoStop}%
\end{thebibliography}

%

\end{document}